\documentclass[11pt,a4paper]{article}
\usepackage{amsfonts}
\usepackage[ansinew]{inputenc}
\usepackage{verbatim}
\usepackage{graphicx}
\usepackage{tabularx}
\usepackage{array}
\usepackage{delarray}
\usepackage{dcolumn}
\usepackage{float}
\usepackage{amsmath}
\usepackage{amstext}
\usepackage{amssymb}
\usepackage{hyperref}
\usepackage{url}
\usepackage{xr}
\usepackage{setspace}

\RequirePackage[round,authoryear,longnamesfirst]{natbib}
\newtheorem{lem}{Lemma}[section]
\newtheorem{thm}{Theorem}[section]

\newtheorem{alg}{Algorithm}[section]

\newtheorem{ass}{Assumption}
\newtheorem{prop}{Proposition}[section]

\newenvironment{proof}[1][Proof]{\noindent\textbf{#1.} }{\ \rule{0.5em}{0.5em}}
\DeclareMathOperator*{\argmax}{arg\,max}
\DeclareMathOperator*{\argmin}{arg\,min}

\newcommand{\eps}{\varepsilon}

\numberwithin{equation}{section}
\input{tcilatex}

\begin{document}

\title{Non-separable Models with High-dimensional Data\thanks{
First draft: February, 2017. We are grateful to Alex Belloni, Xavier D'Haultf%
œ{}uille, Michael Qingliang Fan, Bryan Graham, Yu-Chin Hsu, Yuya Sasaki, and
seminar participants at Academia Sinica, Duke, Asian Meeting of the
Econometric Society, China Meeting of the Econometric Society, and the 7th
Shanghai Workshop of Econometrics. Su acknowledges the funding support
provided by the Lee Kong Chian Fund for Excellence.}}
\author{Liangjun Su\thanks{
School of Economics, Singapore Management University, 90 Stamford Road,
Singapore 178903. E-mail: ljsu@smu.edu.sg.} \and Takuya Ura\thanks{
Department of Economics, University of California, Davis. One Shields
Avenue, Davis, CA 95616. E-mail: takura@ucdavis.edu.} \and Yichong Zhang%
\thanks{
School of Economics, Singapore Management University, 90 Stamford Road,
Singapore 178903. E-mail: yczhang@smu.edu.sg.}}
\maketitle

\begin{abstract}
This paper studies non-separable models with a continuous treatment when the
dimension of the control variables is high and potentially larger than the
effective sample size. We propose a three-step estimation procedure to
estimate the average, quantile, and marginal treatment effects. In the first
stage we estimate the conditional mean, distribution, and density objects by
penalized local least squares, penalized local maximum likelihood
estimation, and numerical differentiation, respectively,
where control variables are selected via a localized method of $L_{1}$%
-penalization at each value of the continuous treatment. In the second stage
we estimate the average and marginal distribution of the potential
outcome via the plug-in principle. In the third stage, we estimate the
quantile and marginal treatment effects by inverting the estimated
distribution function and using the local linear regression, respectively.
We study the asymptotic properties of these estimators and propose a
weighted-bootstrap method for inference. Using simulated and real datasets,
we demonstrate that the proposed estimators perform well in finite samples.
\bigskip

\noindent \textbf{Keywords:} Average treatment effect, High dimension, Least
absolute shrinkage and selection operator (Lasso), Nonparametric quantile
regression, Nonseparable models, Quantile treatment effect, Unconditional
average structural derivative\bigskip

\noindent \textbf{JEL codes:} C21, J62
\end{abstract}

\nopagebreak 

\section{Introduction\label{sec:intro}}

Non-separable models without additivity appear frequently in econometric
analyses, because economic theory motivates a nonlinear role of the
unobserved individual heterogeneity \citep{AM05} and its
multi-dimensionality \citep{BC07,CHH03,CHS10}. A large fraction of the
previous literature on non-separable models has used control variables to
achieve the unconfoundedness condition \citep{RR83}, that is, the
conditional independence between a regressor of interest (or a treatment)
and the unobserved individual heterogeneity given the control variables.
Although including high-dimensional control variables make unconfoundedness more
plausible, the estimation and inference become more challenging, as well. It
remains unanswered how to select control variables among potentially very
many variables and conduct proper statistical inference for parameters of
interest in non-separable models with a continuous treatment.

This paper proposes estimation and inference for unconditional parameters,\footnote{To be more specific, the parameters of interest are unconditional on covariates but conditional on the treatment level.}
including unconditional means of the potential outcomes, the unconditional
cumulative distribution function, the unconditional quantile function, and
the unconditional quantile partial derivative with the presence of both
continuous treatment and high-dimensional covariates.\footnote{%
We focus on unconditional parameters, in which (potentially
high-dimensional) covariates are employed to achieve the unconfoundedness
but the parameters of interest are unconditional on the covariates.
Unconditional parameters are simple to display and the simplicity is crucial
especially when the covariates are high dimensional. As emphasized in \cite%
{FM13} and \cite{P10}, unconditional parameters have two additional
attractive features. First, by definition, they capture all the individuals
in the sample at the same time instead of investigating the underlying
structure separately for each subgroup defined by the covariates $X$. The
treatmen effect for the whole population is more policy-relevant. Second, an
estimator for unconditional parameters can have better finite/large sample
properties.} The proposed method estimates the parameters of interest in
three stages. The first stage selects controls by the method of least
absolute shrinkage and selection operator (Lasso) and predicts reduced-form
parameters such as the conditional expectation and distribution of the
outcome given the variables and treatment level and the conditional density
of the treatment given the control variables. We allow for different control
variables to be selected at different values of the continuous treatment.
The second stage recovers the average and the marginal distribution of the
potential outcome by plugging the reduced-form parameters into doubly robust
moment conditions. The last stage recovers the quantile of the potential
outcome and its derivative with respect to the treatment by inverting the
estimated distribution function and using the local linear regression,
respectively. The inference is implemented via a weighted-bootstrap
without recalculating the first stage variable selections, which saves
considerable computation time.

To motivate our parameters of interest, we relate our estimands (the
population objects that our procedure aims to recover) with the structural
outcome function. Notably, we extend \cite{HM07} and \cite{S15} to
demonstrate that the unconditional derivative of the quantile of the
potential outcome with respect to the treatment is equal to the weighted
average of the marginal effects over individuals with same outcomes and
treatments. 

This paper contributes to two important strands of the econometric
literature. The first is the literature on non-separable models with a
continuous treatment, in which previous analyses have focused on a fixed and
small number of control variables; see, e.g., \cite{C03}, \cite{CIN07}, \cite%
{HM07}, \cite{IN09}, \cite{M94} and \cite{RM03}. The second is a growing
literature on recovering the causal effect from the high-dimensional data;
see, e.g., \cite{BCCH12}, \cite{BCH14jep}, \cite{CHS15}, \cite{CHSAnnual15}, 
\cite{F15}, \cite{AI15}, \cite{CCDDHN16}, \cite{BCH14}, \cite{AW16}, \cite%
{BCFH13}, and \cite{BCW17}. Our paper complements the previous works by
studying both the variable selection and post-selection inference of causal
parameters in a \textit{non-separable} model with a \textit{continuous}
treatment. Recently, \cite{CMN16}, \cite{CMM17}, and \cite{CJN17} have considered
the semiparametric estimation of the causal effect in a setting with many
included covariates and proposed novel bias-correction methods to conduct
valid inference. Comparing with them, we deal with the fully nonparametric
model with an ultra-high dimension of potential covariates, and rely on the
approximate sparsity to reduce dimensionality.

The treatment variable being continuous imposes difficulties in both
variable selection and post-selection inference. To address the former, we
use penalized local Maximum Likelihood and Least Square estimations
(hereafter, MLE and LS, respectively) to select control variables for
each value of the continuous treatment. The penalized local LS was previously studied by \cite{K15} and \cite{LM16}.\footnote{We thank the referee for the reference.} The local MLE complements the LS method by estimating a nonlinear and high-dimensional model with varying coefficients indexed by 
not only the continuous treatment variable but also a location variable. Our approach directly extend the distribution regression proposed in \cite{CFM13} to the high-dimensional varying coefficient setting. By relying on kernel smoothing
method, we require a different penalty loading than the traditional Lasso
method. \cite{CZW11} and \cite{NL14} develop general theories of
estimation, inference, and hypothesis testing of penalized (Pseudo) MLE. We
complement their results by considering the local likelihood with an $L_{1}$
penalty term. \cite{BCCW15} construct uniformly valid confidence bands
for the Z-estimators of unconditional moment equalities. Our results are not
covered by theirs, either, as our parameters are defined based on
conditional moment equalities. To prove the statistical properties of the
penalized local MLE, we establish a local version of the
compatibility condition (\citeauthor{BV11}, \citeyear{BV11}), which itself
is new to the best of our knowledge.

For the post-selection inference, we establish doubly robust moment
conditions for the continuous treatment effect model. Our parameters of
interest is irregularly identified by the definition in \cite{KT10}, as they
are identified by a thin-set. Therefore, by averaging observations only when
their treatment levels are close to the one of interest, the convergence rates
of our estimators are nonparametric, which is in contrast with the $\sqrt{n}$%
-rate obtained in \cite{BCFH13} and \cite{F15}. Albeit motivated by distinct
models, \cite{BCC16} also estimate the irregular identified parameters in
the high-dimensional setting. However, the irregularity faced by \cite{BCC16}
is not due to the continuity of the variable of interest. Consequently, \cite%
{BCC16} do not study the regularized estimator with localization as we do
in this paper. 

Estimation based on doubly robust moments is also related to the literature of semiparametric efficiency. The idea of doubly robust estimation can be traced back to the  nonparametric efficiency theory for functional estimation developed by \cite{B83}, \cite{P90}, \cite{B93}, and \cite{N94}. \cite{R01} and \cite{VR03} study the semiparametric doubly robust estimators by modeling both the treatment and outcome processes. \cite{VD03} allow for nonparametric modeling in causal inference problems.  When both processes are nonparametrically estimated, the doubly robust methods can achieve faster rates of convergence than their nuisance estimator, making the estimator less sensitive to the curse of dimensionality and model selection bias. Their use in causal inference is also considered by \cite{RR95}, \cite{H98}, \cite{VR03}, \cite{HIR03}, \cite{V06}, \cite{F07}, \cite{t07}, \cite{VR11}, \cite{K17}, and \cite{R17}, among others. 

Among the works above, our paper is most closely related to \cite{K17}, who consider the doubly robust estimation for the average treatment effect when the treatment variable is continuous. Our paper complements theirs in four aspects. First, the estimation procedures are different. \cite{K17} first estimate the efficient influence function for the weighted average of the mean effect over all treatment levels, and then, use kernel smoothing to estimate the mean effect at each treatment level. On the contrary, we directly consider the doubly-robust moment for the parameters of interest. Second, \cite{K17} mainly focus on the mean effect, while we also consider quantile and marginal treatment effects. We obtain linear expansions for our estimators uniformly over both the quantile index and the treatment variable. Third, \cite{K17} do not construct detailed estimators of their nuisance parameters, but instead, impose high-level assumptions. To verify such high-level assumptions in the high-dimensional setting is nontrivial. In contrast, we provide valid estimators for our nuisance parameters via both regularization and localization, and derive their statistical properties. Fourth, we take into account the fact that the dimension of covariates may increase with the sample size so that the complexity of our nuisance parameter estimator measured by the uniform entropy will diverge to infinity. Such a situation is ruled out by \cite{K17}.

To obtain uniformly valid results over values of the
continuous treatment, we derive linear expansions of the rearrangement
operator for a local process which is not tight, extending the existing
results in \cite{CFG10}. 

We study the finite sample performance of our estimation procedure via Monte
Carlo simulations and an empirical application. The simulations suggest that
the proposed estimators perform reasonably well in finite samples. In the
empirical exercise, we estimate the distributional effect of parental income
on son's income and intergenerational elasticity using the 1979 National
Longitudinal Survey of Youth (NLSY79). We control for a large dimension of
demographic variables. The quantiles of son's potential income are in
general upward slopping with respect to parental income, but for the subsample of blacks, the
intergenerational elasticities are not statistically significant. 

The rest of this paper is organized as follows. Section \ref{sec:model}
presents the model and the parameters of interest. Section \ref%
{sec:estimators} proposes an estimation method in the presence of
high-dimensional covariates. Section \ref{sec:inf} demonstrates the validity
of a bootstrap inference procedure. Section \ref{sec:sim} presents Monte
Carlo simulations. Section \ref{sec:app} illustrates the proposed estimator
using NLSY79. Section \ref{sec:concl} concludes. Proofs of the main theorems
and Lemma \ref{lem:localRE} are reported in the appendix. Proofs of the rest
of the lemmas are collected in an online supplement.

Throughout this paper, we adopt the convention that the capital letters,
such as $A$, $Y$, $X$, denote random elements while their corresponding
lower cases denote realizations. $C$ denotes an arbitrary positive constant that may not be the
same in different contexts. For a sequence of random variables $%
\{U_{n}\}_{n=1}^{\infty }$ and a random variable $U$, $U_{n}\rightsquigarrow
U$ indicates weak convergence in the sense of \cite{VW96}. When $U_{n}$ and $%
U$ are $k$-dimensional elements, the space of the sample path is $\Re ^{k}$
equipped with Euclidean norm. When $U_{n}$ and $U$ are stochastic processes,
the space of sample path is $L^{\infty }(\{v\in \Re^{k}:|v|<B\})$ for some
positive $B$ equipped with sup norm. The letters $\mathbb{P}_{n}
$, $\mathbb{P}$, and $\mathcal{U}_{n}$ denote the empirical process,
expectation, and U-process, respectively. In particular, $\mathbb{P}_{n}$
assigns probability $\frac{1}{n}$ to each observation and $\mathcal{U}_{n}$
assigns probability $\frac{1}{n(n-1)}$ to each pair of observations. $%
\mathbb{E}$ also denotes expectation. We use $\mathbb{P}$ and $\mathbb{E}$
exchangeably. For any positive (random) sequence $(u_{n},v_{n})$, if there
exists a positive constant $C$ independent of $n$ such that $u_{n}\leq
Cv_{n} $, then we write $u_{n}\lesssim v_{n}$. $||\cdot ||_{Q,q}$ denotes $%
L^{q}$ norm under measure $Q$, where $q=1,2,\infty $. If measure $Q$ is
omitted, the underlying measure is assumed to be the counting measure. For
any vector $\theta $, $||\theta ||_{0}$ denotes the number of its nonzero
coordinates. $\text{Supp}(\theta )$, the support of a $p$-dimensional vector 
$\theta $, is defined as $\{j:\theta _{j}\neq 0\}$. For $T\subset
\{1,2,\cdots ,p\}$, let $|T|$ be the cardinality of $T$, $T^{c}$ be the
complement of $T$, and $\theta _{T}$ be the vector in $\Re ^{p}$ that has
the same coordinates as $\theta $ on $T$ and zero coordinates on $T^{c}$.
Last, let $a\vee b=\max (a,b)$.

\section{Model and Parameters of Interest}

\label{sec:model} Econometricians observe an outcome $Y$, a continuous
treatment $T$, and a set of covariates $X$, which may be high-dimensional.
They are connected by a measurable function $\Gamma (\cdot )$, i.e., 
\begin{equation*}
Y=\Gamma (T,X,A),
\end{equation*}%
where $A$ is an unobservable random vector and may not be weakly separable
from observables $(T,X)$, and $\Gamma $ may not be monotone in either $T$ or 
$A$.

Let $Y(t)=\Gamma (t,X,A)$. We are interested in the average $\mathbb{E}Y(t)$%
, the marginal distribution $\mathbb{P}(Y(t)\leq u)$ for some $u\in \Re $, and the
quantile $q_{\tau }(t)$, where we denote $q_{\tau }(t)$ as the $\tau $-th
quantile of $Y(t)$ for some $\tau \in (0,1)$. We are also interested in the
causal effect of moving $T$ from $t$ to $t^{\prime }$, i.e., $\mathbb{E}%
(Y(t)-Y(t^{\prime }))$ and $q_{\tau }(t)-q_{\tau }(t^{\prime })$. Last, we
are interested in the average marginal effect $\mathbb{E}[\partial
_{t}\Gamma (t,X,A)] $ and quantile partial derivative $\partial_t q_\tau(t)$%
. Next, we specify conditions under which the above parameters are
identified.\medskip

\begin{ass}
The random variables $A$ and $T$ are conditionally independent given $%
X $.
\label{ass:unconfoundedness}\medskip
\end{ass}

Assumption \ref%
{ass:unconfoundedness} is known as the unconfoundedness condition, which
is commonly assumed in the treatment effect literature. See \cite{M10}, \cite%
{MM11}, \cite{HIR03} and \cite{F07} for the case of discrete treatment and 
\cite{GIG14}, \cite{GW14}, and \cite{HI04} for the case of continuous
treatment. It is also called the conditional independence assumption in \cite%
{HM07}, which is weaker than the full joint independence between $A$ and $%
(T,X)$. Note that $X$ can be arbitrarily correlated with the unobservables $A
$. This assumption is more plausible when we control for sufficiently many
and potentially high-dimensional covariates.

%


\begin{thm}
Suppose Assumption \ref{ass:unconfoundedness} holds and $\Gamma(\cdot)$ is
differentiable in its first argument. Then the marginal distribution of $Y(t)
$ and the average marginal effect $\partial_t \mathbb{E}Y(t)$ are
identified. In addition, if Assumption \ref{ass:reg_sasaki} in the Appendix
holds and $X$ is continuously distributed, then $\partial_tq_\tau(t) = 
\mathbb{E}_{\mu_{\tau,t}}[\partial_t\Gamma(t,X,A)]$, where, for $f_{(X,A)}$
denoting the joint density of $(X,A)$, $\mu_{\tau,t}$ is the probability
measure on $\{(x,a): \Gamma(t,x,a)=q_\tau(t)\}$ with density $\frac{%
f_{(X,A)}}{c_f\|\nabla_{(x,a)}\Gamma(t,\cdot,\cdot)\|}$, where 
$$c_f = \int_{(x,a): \Gamma(t,x,a)=q_\tau(t)} \frac{%
f_{(X,A)}(x,a)}{\|\nabla_{(x,a)}\Gamma(t,\cdot,\cdot)\|}dxda.$$

 \label{thm:uasf}
\end{thm}

Several comments are in order. First, because the marginal distribution of $%
Y(t)$ is identified, so be its average, quantile, average marginal effect,
and quantile partial derivative. As pointed out by \cite{IN09}, a
non-separable outcome with a general disturbance is equivalent to treatment
effect models. Therefore, we can view $Y(t)$ as the potential outcome. Under
unconfoundedness, the identification of the marginal distribution of the
potential outcome with a continuous treatment has already been established
in \cite{HI04} and \cite{GW14}. The first part of Theorem \ref{thm:uasf}
just re-states their results. Second, the second result indicates that the
partial quantile derivative identifies the weighted average marginal effect
for the subpopulation with the same potential outcome, i.e., $\{Y(t) =
q_\tau(t)\}.$ The result is closely related to, but different from \cite{S15}%
. We consider the unconditional quantile of $Y(t)$, whereas he considered
the conditional quantile of $Y(t)$ given $X$. Note that $q_\tau(t)$ is not
the average of the conditional quantile of $Y(t)$ given $X$. Third, we
require $X$ to be continuous just for the simplicity of derivation. If some
elements of $X$ are discrete, a similar result can be established in a
conceptually straightforward manner by focusing on the continuous covariates
within samples homogenous in the discrete covariates, at the expense of
additional notation. Finally, we do not require $X$ to be continuous when
establishing the estimation and inference results below.


\section{Estimation}

\label{sec:estimators} 
Let $f_t(x) = f_{T|X}(t|x)$ denote the conditional density of $T$ evaluated
at $t$ given $X=x$ and $d_t(\cdot )$ denote the Dirac function such that for
any function $g(\cdot )$, 
\begin{equation*}
\int g(s)d_t(s)ds=g(t).
\end{equation*}%
In addition, let $Y_{u}(t)=1\{Y(t)\leq u\}$ and $Y_{u}=1\{Y\leq u\}$ for
some $u\in \Re $. Then $\mathbb{E}(Y(t))$ and $\mathbb{E}(Y_{u}(t))$ can be
identified by the method of generalized propensity score as proposed in \cite%
{HI04}, i.e., 
\begin{equation}
\mathbb{E}(Y(t))=\mathbb{E}\biggl(\frac{Y d_t(T)}{f_{t}(X)}\biggr)\quad 
\text{and}\quad \mathbb{E}(Y_{u}(t))=\mathbb{E}\biggl(\frac{Y_{u}d_t(T)}{%
f_{t}(X)}\biggr).  \label{eq:mom1}
\end{equation}

There is a direct analogy between \eqref{eq:mom1} for the continuous
treatment and $\mathbb{E}(Y_u(t)) = \mathbb{E}(\frac{Y_u1\{T =t\}}{\mathbb{P}(T=t|X)}%
) $ when the treatment $T$ is discrete: the indicator function shrinks to a
Dirac function and the propensity score is replaced by the conditional
density. Following this analogy, \cite{HI04} called $f_t(X)$ the generalized
propensity.

\cite{BCFH13} and \cite{F15} considered the model with a discrete treatment
and high-dimensional control variables, and proposed to use the doubly
robust moment for inference. Following their lead, we propose the
corresponding doubly robust moment when the treatment status is continuous.
Let $\nu _{t}(x)=\mathbb{E}(Y|X=x,T=t)$ and $\phi _{t,u}(x)=\mathbb{E}%
(Y_{u}|X=x,T=t)$, then 
\begin{equation}
\mathbb{E}(Y(t))=\mathbb{E}\biggl[\biggl(\frac{(Y-\nu _{t}(X))d_t(T)}{%
f_{t}(X)}\biggr)+\nu _{t}(X)\biggr]  \label{eq:doublerobust_E}
\end{equation}%
and 
\begin{equation}
\mathbb{E}(Y_{u}(t))=\mathbb{E}\biggl[\biggl(\frac{(Y_{u}-\phi
_{t,u}(X))d_t(T)}{f_{t}(X)}\biggr)+\phi _{t,u}(X)\biggr].
\label{eq:doublerobust}
\end{equation}%
%
We propose the following three-stage procedure to estimate $\mu (t):= 
\mathbb{E}Y(t)$, $\alpha (t,u):= \mathbb{P}(Y(t)\leq u)$, $q_{\tau }(t)$, and $%
\partial _{t}q_{\tau }(t)$:

\begin{itemize}
\item[1.] Estimate $\nu_t(x)$, $\phi _{t,u}(x)$, and $f_{t}(x)$ by $\widehat{%
\nu}_t(x)$, $\widehat{\phi }_{t,u}(x)$ and $\hat{f}_t(x)$, respectively, using the first-stage bandwidth $h_1$. 

\item[2.] Estimate $\mu (t)$ and $\alpha (t,u)$ by 
\begin{equation*}
\hat{\mu}(t)=\frac{1}{n}\sum_{i=1}^{n}\biggl[\biggl(\frac{(Y-\widehat{\nu }%
_{t}(X_{i}))}{\hat{f}_{t}(X_{i})h_2}K(\frac{T_{i}-t}{h_2})\biggr)+\widehat{\nu }%
_{t}(X_{i})\biggr]
\end{equation*}%
and 
\begin{equation*}
\hat{\alpha}(t,u)=\frac{1}{n}\sum_{i=1}^{n}\biggl[\biggl(\frac{(Y_{u}-%
\widehat{\phi }_{t,u}(X_{i}))}{\hat{f}_{t}(X_{i})h_2}K(\frac{T_{i}-t}{h_2})%
\biggr)+\widehat{\phi }_{t,u}(X_{i})\biggr],\quad \text{respectively,}
\end{equation*}%
where $K(\cdot )$ and $h_2$ are a kernel function and the second-stage bandwidth, respectively. Then
rearrange $\hat{\alpha}(t,u)$ to obtain $\hat{\alpha}^{r}(t,u)$, which is monotone in $%
u$.

\item[3] Estimate $q_{\tau }(t)$ by inverting $\hat{a}^{r}(t,u)$ with
respect to (w.r.t.) $u$, i.e., $\hat{q}_{\tau }(t)=\inf \{u:\hat{a}%
^{r}(t,u)\geq \tau \};$ estimate $\partial_t \mu(t) = \mathbb{E}\partial
_{t}\Gamma(t,X,A)$ by $\breve{\beta}^{1}(t)$, which is the estimator of the
slope coefficient in the local linear regression of $\hat{\mu}(T_{i})$ on $%
T_{i}$; estimate $\partial _{t}q_{\tau }(t)$ by $\hat{\beta}_{\tau }^{1}(t)$%
, which is the estimator of the slope coefficient in the local linear
regression of $\hat{q}_{\tau }(T_{i})$ on $T_{i}$. 
\end{itemize}


\subsection{The First Stage Estimation}

\label{sec:1st}

In this section, we define the first stage estimators and derive their
asymptotic properties. Since $\nu_t(x)$, $\phi_{t,u}(x)$, and $f_t(x)$ are
local parameters w.r.t. $T=t$, in addition to using $L_1$ penalty to select
relevant covariates, we rely on a kernel function to implement the
localization. In particular, we propose to estimate $\nu_t(x)$, $%
\phi_{t,u}(x)$, and $f_t(x)$ by a penalized local LS, a penalized local MLE,
and numerical differentiation, respectively. 

\subsubsection{Penalized Local LS and MLE}

Recall $\nu _{t}(x)=\mathbb{E}(Y|X=x,T=t)$ and $\phi _{t,u}(x)=\mathbb{E}%
(Y_{u}|X=x,T=t)$ where $Y_{u}=1\{Y\leq u\}$. We approximate $\nu _{t}(x)$
and $\phi _{t,u}(x)$ by $b(x)^{\prime }\gamma_{t}$ and $\Lambda
(b(x)^{\prime }\theta_{t,u})$, respectively, where $\Lambda (\cdot )$ is
the logistic CDF and $b(X)$ is a $p\times 1$ vector of basis functions with
potentially large $p$. In the case of high-dimensional covariates, $b(X)$ is
just $X$, while in the case of nonparametric sieve estimation, $b(X)$ is a
series of bases of $X$. The approximation errors for $\nu _{t}(x)$ and $\phi
_{t,u}(x)$ are given by $r_{t}^{\nu }(x)=\nu _{t}(x)-b(x)^{\prime }\gamma
_{t}$ and $r_{t,u}^{\phi }(x)=\phi _{t,u}(x)-\Lambda (b(x)^{\prime }\theta
_{t,u}),$ respectively.

Note that we only approximate $\nu_t(x)$ and $\phi_{t,u}(x)$ by a linear
regression and a logistic regression, respectively, with the approximation
errors satisfying Assumption \ref{ass:approx} below. Assumption \ref{ass:approx} below puts a sparsity
structure on $\nu_t(x)$ and $\phi_{t,u}(x)$ so that the number of effective
covariates that can affect them is much smaller than $p$. If the effective
covariates are a few discrete variables that have a few categories, then we can saturate the regressions by low-dimensional dummy variables so that
there is no approximate error. If some of the effective covariates are
continuous, then we can include sieve bases in the linear regression so that
the approximation error can still satisfy Assumption \ref{ass:approx}. One possible scenario that the approximate sparsity condition may fail is when there are a substantial amount of discrete variables that are all on the same footing (e.g., job occupation dummies). In this case, it is hard to define a sparse approximation.\footnote{We thank the Associate Editor for this point.} Last,
the coefficients $\gamma_t$ and $\theta_{t,u}$ are both functional
parameters that can vary with their indexes. This provides additional
flexibility of our setup against misspecification.

We estimate $\nu _{t}(x)$ and $\phi _{t,u}(x)$ by $\widehat{\nu }%
_{t}(x)=b(x)^{\prime }\hat{\gamma}_{t}$ and $\widehat{\phi }%
_{t,u}(x)=\Lambda (b(x)^{\prime }\hat{\theta}_{t,u})$, respectively, where%
\begin{equation}
\hat{\gamma}_{t}=\argmin_{\gamma }\frac{1}{2n}%
\sum_{i=1}^{n}(Y_{i}-b(X_{i})^{\prime }\gamma )^{2}K(\frac{T_{i}-t}{h_1})+%
\frac{\lambda }{n}||\widehat{\Xi }_{t}\gamma ||_{1},  \label{eq:e}
\end{equation}

\begin{equation}
\hat{\theta}_{t,u}=\argmin_{\theta }\frac{1}{n}\sum_{i=1}^{n}M(1\{Y_{i}\leq
u\},X_{i};\theta )K(\frac{T_{i}-t}{h_1})+\frac{\lambda }{n}||\widehat{\Psi }%
_{t,u}\theta ||_{1},  \label{eq:m}
\end{equation}%
$\left\Vert \cdot \right\Vert _{1}$ denotes the $L_{1}$ norm, $h_1$ is the first-stage bandwidth, $\lambda =\ell
_{n}(\log (p\vee nh_1)nh_1)^{1/2}$ for some slowly diverging sequence $\ell _{n}$%
, and $M(y,x;g)=-[y\log (\Lambda (b(x)^{\prime }g))+(1-y)\log (1-\Lambda
(b(x)^{\prime }g))]$. Our penalty term $\lambda$ is different from the one
used in \cite{BCFH13} and \cite{BCK18}, i.e., $\lambda^* = 1.1
\Phi^{-1}(1-\gamma/p)n^{1/2}$, where $\gamma = o(1)$ is some user-supplied constant, and $%
\Phi(\cdot)$ is the standard normal CDF. \cite{BCFH13} suggest $\gamma = C/(n\log(n))$, which implies that  
\begin{equation*}
\Phi^{-1}(1-\gamma/p) \sim  [\log(1/C) + \log(p) + \log(n) + \log(\log(n))]^{1/2} \sim \sqrt{\log(p \vee n)}.
\end{equation*}
Therefore, our
penalty term $\lambda$ is of same order of magnitude of $\lambda^*$ if $nh_1$
is replaced with $n$ and $\ell_n$ is removed. We need to use $nh_1$ in our
penalty due to the presence of the kernel function in our estimation
procedure. In particular, the effective sample size is of the
same order of $nh_1$.\footnote{%
Note that $\log(n)$ and $\log(nh_1)$ are of the same order of magnitude.} We will specify the order of magnitude of $h_1$ in Assumption \ref{ass:approx}. The role
played by $\ell_n$ in our penalty is similar to that of $\gamma$ in $%
\lambda^*$, which is to control the selection error uniformly. We refer
readers to \citet[Equation (6.4)]{BCFH13} for a more detailed discussion on
this point. Since we do not use the advanced technique of self-normalized
process as in \cite{BCFH13}, we multiply the sequence $\ell_n$ with $%
\sqrt{\log(p \vee n)}$ while in $\lambda^*$, $\log(\gamma)$ is additive to $%
\log(pn)$ inside the square root. We propose a rule-of-thumb $\lambda$ in
Section \ref{sec:sim} and study the sensitivity of our inference method
against the choice of $\lambda$ in Section \ref{sec:addsim} of the
supplementary material. 


In \eqref{eq:e} and \eqref{eq:m}, $\widehat{\Xi }_{t}=\text{diag}(\tilde{l}%
_{t,1},\cdots ,\tilde{l}_{t,p})$ and $\widehat{\Psi }_{t,u}=\text{diag}%
(l_{t,u,1},\cdots ,l_{t,u,p})$ are generic penalty loading matrices. The
infeasible loading matrices we would like to use are $\widehat{\Xi }_{t,0}=%
\text{diag}(\tilde{l}_{t,0,1},\cdots ,\tilde{l}_{t,0,p})$ and $\widehat{\Psi 
}_{t,u,0}=\text{diag}(l_{t,u,0,1},\cdots ,l_{t,u,0,p})$ in which%
\begin{equation*}
\tilde{l}_{t,0,j}=\biggl|\biggl|(Y-\nu _{t}(X))b_{j}(X)K(\frac{T-t}{h_1}%
)h_1^{-1/2}\biggr|\biggr|_{\mathbb{P}_n,2}
\end{equation*}%
and 
\begin{equation*}
l_{t,u,0,j}=\biggl|\biggl|(Y_{u}-\phi _{t,u}(X))b_{j}(X)K(\frac{T-t}{h_1}%
)h_1^{-1/2}\biggr|\biggr|_{\mathbb{P}_n,2} ,
\end{equation*}%
respectively. Since $\nu _{t}(\cdot )$ and $\phi _{t,u}(\cdot )$ are not
known, we follow \cite{BCFH13} and propose an iterative algorithm to obtain
the feasible versions of the loading matrices. The statistical properties of
the feasible loading matrices are summarized in Lemma \ref{lem:E32} in the
Appendix.

\begin{alg}
\begin{enumerate}
\item Let $\widehat{\Xi }_{t}^{0}=\text{diag}(\tilde{l}_{t,1}^{0},\cdots ,%
\tilde{l}_{t,p}^{0})$ and $\widehat{\Psi }_{t,u}^{0}=\text{diag}%
(l_{t,u,1}^{0},\cdots ,l_{t,u,p}^{0})$, where $\tilde{l}%
_{t,j}^{0}=||Yb_{j}(X)K(\frac{T-t}{h_1})h_1^{-1/2}||_{\mathbb{P}_n,2}$ and $%
l_{t,u,j}^{0}=||Y_{u}b_{j}(X)K(\frac{T-t}{h_1})h_1^{-1/2}||_{\mathbb{P}_n,2}.$ Using $%
\widehat{\Xi }_{t}^{0}$ and $\widehat{\Psi }_{t,u}^{0}$, we can compute $%
\hat{\gamma}_{t}^{0}$ and $\hat{\theta}_{t,u}^{0}$ by \eqref{eq:e} and %
\eqref{eq:m}. Let $\widehat{\nu }_{t}^{0}(x)=b(x)^{\prime }\hat{\gamma}%
_{t}^{0}$ and $\widehat{\phi }_{t,u}^{0}(x)=\Lambda (b(x)^{\prime }\hat{%
\theta}_{t,u}^{0})$ for $x=X_{1},...,X_{n}.$

\item For $k=1,\cdots ,K$ for some fixed positive integer $K$, we compute $%
\widehat{\Xi }_{t}^{k}=\text{diag}(\tilde{l}_{t,1}^{k},\cdots ,\tilde{l}%
_{t,p}^{k})$ and $\widehat{\Psi }_{t,u}^{k}=\text{diag}(l_{t,u,1}^{k},\cdots
,l_{t,u,p}^{k}),$ where 
\begin{equation*}
\tilde{l}_{t,j}^{k}=\biggl|\biggl|(Y-\widehat{\nu }_{t}^{k-1}(X))b_{j}(X)K(%
\frac{T-t}{h_1})h_1^{-1/2}\biggr|\biggr|_{\mathbb{P}_n,2}
\end{equation*}%
and 
\begin{equation*}
l_{t,u,j}^{k}=\biggl|\biggl|(Y_{u}-\widehat{\phi }_{t,u}^{k-1}(X))b_{j}(X)K(%
\frac{T-t}{h_1})h_1^{-1/2}\biggr|\biggr|_{\mathbb{P}_n,2} .
\end{equation*}%
Using $\widehat{\Xi }_{t}^{k}$ and $\widehat{\Psi }_{t,u}^{k}$, we can
compute $\hat{\gamma}_{t}^{k}$ and $\hat{\theta}_{t,u}^{k}$ by \eqref{eq:e}
and \eqref{eq:m}. Let $\widehat{\nu }_{t}^{k}(x)=b(x)^{\prime }\hat{\gamma}%
_{t}^{k}$ and $\widehat{\phi }_{t,u}^{k}(x)=\Lambda (b(x)^{\prime }\hat{%
\theta}_{t,u}^{k})$ for $x=X_{1},...,X_{n}.$ The final penalty loading
matrices $\widehat{\Xi }_{t}^{K}$ and $\widehat{\Psi }_{t,u}^{K}$ will be
used for $\widehat{\Xi }_{t}$ and $\widehat{\Psi }_{t,u}$ in \eqref{eq:e}
and \eqref{eq:m}. \label{alg:psi1}
\end{enumerate}
\end{alg}

Let $\widetilde{\mathcal{S}}_t^\mu$ and $\widetilde{\mathcal{S}}_{t,u}$
contain the supports of $\hat{\gamma}_{t}$ and $\hat{\theta}_{t,u}$,
respectively, such that $|\widetilde{\mathcal{S}}_t^\mu|\lesssim \sup_{t\in 
\mathcal{T}}||\widehat{\gamma }_{t}||_{0}$, and $|\widetilde{\mathcal{S}}%
_{t,u}|\lesssim \sup_{(t,u)\in \mathcal{T}\mathcal{U}}||\widehat{\theta }%
_{t,u}||_{0}$. For each $(t,u)\in \mathcal{T}\mathcal{U}:= \mathcal{T}\times 
\mathcal{U}$ where $\mathcal{T}$ and $\mathcal{U}$ are compact subsets of
the supports of $T$ and $Y$, respectively, the post-Lasso estimator of $%
\gamma _{t}$ and $\theta _{t,u}$ based on the set of covariates $\widetilde{%
\mathcal{S}}_t^\mu$ and $\widetilde{\mathcal{S}}_{t,u}$ are defined as 
\begin{equation*}
\tilde{\gamma}_{t}\in \argmin_{\gamma }\sum_{i=1}^{n}(Y_{i}-b(X_{i})^{\prime
}\gamma )^{2}K(\frac{T_{i}-t}{h_1}),\quad s.t.\quad \text{Supp}(\gamma )\in 
\widetilde{\mathcal{S}}_t^\mu,
\end{equation*}%
and 
\begin{equation*}
\tilde{\theta}_{t,u}\in \argmin_{\theta }\sum_{i=1}^{n}M(1\{Y_{i}\leq
u\},X_{i};\theta )K(\frac{T_{i}-t}{h_1}),\quad s.t.\quad \text{Supp}(\theta
)\in \widetilde{\mathcal{S}}_{t,u}.
\end{equation*}%
The post-Lasso estimators of $\nu _{t}(x)$ and $\phi _{t,u}(X)$ are given by 
$\widetilde{\nu }_{t}(X)=b(X)^{\prime }\tilde{\gamma}_{t}$ and $\widetilde{%
\phi }_{t,u}(X)=\Lambda (b(X)^{\prime }\tilde{\theta}_{t,u})$, respectively.

\subsubsection{Conditional Density Estimation}

Following \cite{BCK18}, we propose to first estimate $F_t(X)$, the conditional CDF of $T$ given $X$, by the (logistic) distributional lasso regression studied in \cite{BCFH13} and then take the numerical derivative. Following \cite{BCFH13}, we approximate $F_t(X)$ by a Logistic CDF $\Lambda(b(X)'\beta_t)$ and the approximation error is denoted as $r_t^F(x) = F_t(x) - \Lambda(b(x)'\beta_t)$. We estimate $\beta_t$ by $\hat{\beta}_t$, which is computed as 
\begin{align}
\label{eq:pcd}
\hat{\beta}_t = \argmin_\beta \frac{1}{n}\sum_{i=1}^n M(1\{T_i \leq t\},X_i;\beta) + \frac{\tilde{\lambda}}{n}||\hat{\Psi}_t\beta||_1 \quad \text{and} \quad \hat{F}_t(x) = \Lambda(b(x)'\hat{ \beta}_t),
\end{align}
where $M(\cdot)$ is the logistic likelihood as defined previously, the penalty $$\tilde{\lambda} = 1.1
\Phi^{-1}(1-\gamma/\{p \vee nh_1\})n^{1/2}$$ 
is slightly modified from but of the same order of magnitude as $\lambda^*$ used in \cite{BCFH13} and \cite{BCK18}, for some $\gamma \rightarrow 0$ specified in Section \ref{sec:sim}, and the penalty loading $\hat{\Psi}_t$ is estimated in Algorithm \ref{alg:psi2} below, which is also due to \cite{BCFH13}:
\begin{alg}
	\begin{enumerate}
		\item Let $\widehat{\Psi }_{t}^{0}=\text{diag}(l_{t,1}^{0},\cdots 
		,l_{t,p}^{0})$ where $l_{t,j}^{0}=||1\{T \leq t\}b_{j}(X)||_{\mathbb{P}_n,2}.$ Using $\widehat{\Psi }_{t}^{0}$, we can compute $ 
		\hat{\beta}_{t}^{0}$ and $\hat{F}_t(X)$ by the (logistic) distributional lasso regression.
		
		\item For $k=1,\cdots ,K$, we compute $\widehat{\Psi }_{t}^{k}=\text{diag}
		(l_{t,1}^{k},\cdots ,l_{t,p}^{k})$ where  
		\begin{equation*}
			l_{t,j}^{k}=\biggl|\biggl|\biggl(1\{T \leq t\}-\hat{F}_t^{k-1}(X)\biggr)b_{j}(X)\biggr|\biggr|_{\mathbb{P}_n,2}.
		\end{equation*}
		Using $\widehat{\Psi }_{t}^{k}$, we can compute $\hat{\beta}_{t}^{k}$ and $ 
		\hat{F}_{t}^{k}(X)$ by the (logistic) distributional lasso regression. The 
		final penalty loading matrix $\widehat{\Psi }_{t}^{K}$ will be used as $\widehat{\Psi }_{t}$ in \eqref{eq:pcd}. \label{alg:psi2}
	\end{enumerate}
\end{alg}

Then, $f_t(X)$, the conditional density of $T=t$ give $X$ is computed as 
\begin{align*}
\hat{f}_t(X) = \frac{\hat{F}_{t+h_1}(X) - \hat{F}_{t-h_1}(X)}{2h_1},
\end{align*}
where $h_1$ is the first-stage bandwidth.

\subsubsection{Asymptotic Properties of the First Stage Estimators}

To study the asymptotic properties of the first stage estimators, we need
some assumptions.

\begin{ass}
Let $\mathcal{T}\mathcal{U}$ be a compact subset of the support of $(T,Y)$ and $\mathcal{X}$ be the support of $X$.

\begin{enumerate}
\item The sample $\{Y_i,T_i,X_i\}_{i=1}^n$ is i.i.d.
	
\item $||\max_{j \leq p}|b_j(X)|||_{\mathbb{P}, \infty} \leq \zeta_n$ and $\underline{C%
} \leq \mathbb{E}b_j(X)^2 \leq 1/\underline{C}$ $j=1,\cdots,p.$

\item $\sup_{(t,u)\in \mathcal{T}\mathcal{U}}\max(||\gamma_t||_0,||\beta_t||_0,||\theta_{t,u}||_0) \leq s$ for
some $s$ which possibly depends on the sample size $n$.

\item $\sup_{t \in \mathcal{T}}||r_t^F(X)||_{\mathbb{P}_n,2} = O_p((s \log(p \vee n)/(n))^{1/2}) $ and 
\begin{equation*}
\sup_{(t,u)\in \mathcal{T}\mathcal{U}}\left[||r_{t,u}^\nu(X)K(\frac{T-t}{h_1})^{1/2}||_{\mathbb{P}_n,2} + ||r_{t,u}^\phi(X)K(\frac{%
T-t}{h_1})^{1/2}||_{\mathbb{P}_n,2}\right] =O_p((s \log(p \vee n)/n)^{1/2}).
\end{equation*}

\item $\sup_{t\in \mathcal{T}}||r_t^F(X)||_{\mathbb{P}, \infty} =O((\log(p \vee n)s^2\zeta_n^2/(n))^{1/2})$
and

\begin{equation*}
\sup_{(t,u) \in \mathcal{T} \mathcal{U}}\left[||r_{t,u}^\nu(X)||_{\mathbb{P}, \infty} + ||r_{t,u}^\phi(X)||_{\mathbb{P}, \infty}\right] = O((\log(p
\vee n)s^2\zeta_n^2/(nh_1))^{1/2}).
\end{equation*}

	\item $f_t(x)$ is second-order differentiable w.r.t. $t$ with bounded derivatives uniformly over $(t,x) \in \mathcal{TX}$, where $\mathcal{T}$ is a compact subset of the support of $T$ and $\mathcal{X}$ is the support of $X$. 
\item $\zeta_n^2s^2 \ell_n^2 \log(p \vee n)/(nh_1) \rightarrow 0$, $nh_1^5/(\log(p \vee  n)) \rightarrow 0.$
\end{enumerate}
\label{ass:approx}
\end{ass}

Assumption \ref{ass:approx}.1 is common for cross-sectional observations. Assumption \ref{ass:approx}.2 is the same as Assumption 6.1(a) in \cite%
{BCFH13}. Assumption \ref{ass:approx}.3 requires that $\nu _{t}(x)$, $\phi
_{t,u}(x)$, and $F_{t}(x)$ are approximately sparse, i.e., they can be
well-approximated by using at most $s$ elements of $b(x)$. This approximate
sparsity condition is common in the literature on high-dimensional data
(see, e.g., \cite{BCFH13}). Assumption \ref{ass:approx}.4 and \ref%
{ass:approx}.5 specify how well the approximations are in terms of $%
L_{\mathbb{P}_n,2}$ and $L_{\mathbb{P}, \infty }$ norms. The exact rate for $r_t^F(X)$ follows \cite{BCFH13}. The rates for $r_{t,u}^\nu(X)$ and $r_{t,u}^\phi(X)$ are different from that for $r_t^F(X)$ because their approximations are local in $%
T=t$. If the models for $\nu_t(\cdot)$, $\phi_{t,u}(\cdot)$, and
$F_t(\cdot)$ are correctly specified and exactly sparse, i.e., the
coefficients for all but $s$ regressors are zero, then there are no
approximate errors. This implies $r_t^F(\cdot)$, $r_{t,u}^{\nu}(\cdot)$, and 
$r_{t,u}^\phi(\cdot)$ equal to zero so that Assumption \ref{ass:approx}.4
and \ref{ass:approx}.5 hold automatically. In the sieve estimation, $%
X$ is finite dimensional and $b(X)$ is just a sequence of sieve bases of $X$%
. Then $r_t^F(\cdot)$, $r_{t,u}^{\nu}(\cdot)$, and $%
r_{t,u}^\phi(\cdot)$ are the sieve approximation bias. Assumptions \ref%
{ass:approx}.3 and \ref{ass:approx}.4 can be verified under some smoothness
conditions (see, e.g., \cite{C07}). Therefore, Assumption \ref{ass:approx}.4
and \ref{ass:approx}.5 are in spirit close to the smoothness condition.
Assumption \ref{ass:approx}.6 is the smoothness of the true density, which is needed for the theoretical analysis of the numerical derivative. Because $\mathcal{T}$ needs not be the whole support of $T$, this condition is plausible. In a simple case, if $T = \mu(X) + U$, $|\mu(x)|$ is bounded uniformly over $x \in \mathcal{X}$, and $U$ is independent of $X$ and logistically distributed, then this condition holds. Assumption \ref{ass:approx}.7 imposes conditions on the rates at which $s$, $%
\zeta _{n}$, and $p$ grow with sample size $n$. It ensures that the first
stage nuisance parameters are estimated with sufficient accuracy. In particular, we require $s^2/(nh_1) \rightarrow 0$. Comparing with the condition that $s^2/n \rightarrow 0$ imposed in \cite{BCFH13}, our condition reflects the local nature of our estimation procedure in the sense that our effective sample size is of order of magnitude $nh_1$.

\begin{ass}
\begin{enumerate}
\item $K(\cdot )$ is a symmetric probability density function (PDF) with
\begin{equation*}
\int uK(u)du=0, \quad \text{and} \quad \kappa_{2} := \int u^{2}K(u)du<\infty.
\end{equation*}
There exists a positive constant $\overline{C}_{K}$ such that $%
\sup_{u}u^{l}K\left( u\right) \leq \overline{C}_{K}$ for $l=0,1.$

\item There exists some positive constant $\underline{C}<1$ such that $%
\underline{C}\leq f_{t}(x)\leq 1/\underline{C}$ uniformly over $(t,x)\in 
\mathcal{T}\mathcal{X}$.

\item $\nu _{t}(x)$ and $\phi _{t,u}(x)$ are three times differentiable
w.r.t. $t$, with all three derivatives being bounded uniformly over $%
(t,x,u)\in \mathcal{T}\mathcal{X}\mathcal{U}.$

\item For the same $\underline{C}$ as above, $\underline{C}\leq \mathbb{E}%
(Y_{u}(t)|X=x)\leq 1-\underline{C}$ uniformly over $(t,x,u)\in \mathcal{T}%
\mathcal{X}\mathcal{U}:= \mathcal{T}\mathcal{X}\times \mathcal{U}$.
\end{enumerate}

\label{ass:ker}
\end{ass}


Assumption \ref{ass:ker}.1 holds for many kernel functions, e.g., uniform and Gaussian kernels. Since $f_{T}(X)$ was referred to as the generalized
propensity by \cite{HI04}, Assumption \ref{ass:ker}.2 is analogous to the
overlapping support condition commonly assumed in the treatment effect
literature; see, e.g., \cite{HIR03} and \cite{F07}. Since the conditional
density also has the sparsity structure as assumed in Assumption \ref%
{ass:approx}, at most $s$ members of $X$'s affect the conditional density,
which makes Assumption \ref{ass:ker}.2 more plausible. Assumption \ref%
{ass:ker}.3 imposes some smoothness conditions that are widely assumed in
the nonparametric kernel literature. Assumption \ref{ass:ker}.4 holds if $%
\mathcal{XU}$ is compact.
\medskip

\begin{ass}
There exists a sequence $\ell _{n}\rightarrow \infty $ such that, with
probability approaching one,%
\begin{equation*}
0<\kappa ^{\prime }\leq \inf_{\delta \neq 0,||\delta ||_{0}\leq s\ell _{n}}%
\frac{||b(X)^{\prime }\delta ||_{\mathbb{P}_n,2}}{||\delta ||_{2}}\leq \sup_{\delta
\neq 0,||\delta ||_{0}\leq s\ell _{n}}\frac{||b(X)^{\prime }\delta
||_{\mathbb{P}_n,2}}{||\delta ||_{2}}\leq \kappa ^{^{\prime \prime }}<\infty .
\end{equation*}%
\label{ass:eigen}\medskip
\end{ass}

Assumption \ref{ass:eigen} is the restricted eigenvalue condition commonly
assumed in the high-dimensional data literature. Based on \cite{BRT09}, 
\begin{equation*}
\inf_{\delta \neq 0,||\delta ||_{0}\leq s\ell _{n}}\frac{||b(X)^{\prime
}\delta ||_{\mathbb{P}_n,2}}{||\delta ||_{2}} \quad \text{and} \quad \sup_{\delta
\neq 0,||\delta ||_{0}\leq s\ell _{n}}\frac{||b(X)^{\prime }\delta
||_{\mathbb{P}_n,2}}{||\delta ||_{2}}
\end{equation*}
are the minimal and maximal eigenvalues of Gram submatrices formed by any $%
s\ell_{n}$ components of $b(X)$. Because $p \gg n$, the matrix $b(X)^{\prime
}b(X)$ is not invertible. However, because $s \ell_{n} \ll n$, Assumption %
\ref{ass:eigen} implies that the Gram submatrices can still be invertible.
We refer interested readers to \cite{BRT09} for more details and \cite{BV11}
for a textbook treatment.

Since there is a kernel in the Lasso objective functions in \eqref{eq:e} and %
\eqref{eq:m}, the asymptotic properties of $\hat{\gamma}_{t}$ and $\hat{%
\theta}_{t,u}$ cannot be established by directly applying the results in 
\cite{BCFH13}. The key missing piece is the following local version of the
compatibility condition. Let $\mathcal{S}_{t,u}$ be an arbitrary subset of $%
\{1,\cdots ,p\}$ such that $\sup_{(t,u)\in \mathcal{TU}}|\mathcal{S}%
_{t,u}|\leq s$ and $\Delta _{c,t,u}=\{\delta :||\delta _{\mathcal{S}%
_{t,u}^{c}}||_{1}\leq c||\delta _{\mathcal{S}_{t,u}}||_{1}\}$ for some $%
c<\infty $ independent of $(t,u)$.


\begin{lem}
If Assumptions \ref{ass:unconfoundedness}--\ref{ass:eigen} hold, then there
exists $\underline{\kappa }=\kappa ^{\prime }\underline{C}^{1/2}/4>0$ such
that, w.p.a.1, 
\begin{equation*}
\inf_{(t,u)\in \mathcal{T}\mathcal{U}}\inf_{\delta \in \Delta _{c,t,u}}\frac{%
||b(X)^{\prime }\delta K(\frac{T-t}{h_1})^{1/2}||_{\mathbb{P}_n,2}}{||\delta _{%
\mathcal{S}_{t,u}}||_{2}\sqrt{h_1}}\geq \underline{\kappa }.
\end{equation*}%
\label{lem:localRE}
\end{lem}

Note $\mathcal{S}_{t,u}$ in Lemma \ref{lem:localRE} is either the support of 
$\theta _{t,u}$ or the support of $\gamma _{t}$. For the latter case, the
index $u$ is not needed. We refer to Lemma \ref{lem:localRE} as the local
compatibility condition because (1) there is a kernel function implementing
the localization; and (2) by the Cauchy inequality, Lemma \ref{lem:localRE}
implies 
\begin{equation*}
\inf_{(t,u)\in \mathcal{T}\mathcal{U}}\inf_{\delta \in \Delta _{c,t,u}}\frac{%
\sqrt{s}||b(X)^{\prime }\delta K(\frac{T-t}{h_1})^{1/2}||_{\mathbb{P}_n,2}}{||\delta _{%
\mathcal{S}_{t,u}}||_{1}\sqrt{h_1}}\geq \underline{\kappa }.
\end{equation*}

\citet[Lemma 4.2]{BRT09} show that, under Assumption \ref{ass:eigen}, we
have the following compatibility condition: 
\begin{align}  \label{eq:RE}
\inf_{(t,u)\in \mathcal{T}\mathcal{U}}\inf_{\delta \in \Delta _{c,t,u}}\frac{
\sqrt{s}||b(X)^{\prime }\delta ||_{\mathbb{P}_n,2}}{||\delta _{\mathcal{S}%
_{t,u}}||_{1}} \geq \inf_{(t,u)\in \mathcal{T}\mathcal{U}}\inf_{\delta \in
\Delta _{c,t,u}}\frac{ ||b(X)^{\prime }\delta ||_{\mathbb{P}_n,2}}{||\delta _{%
\mathcal{S}_{t,u}}||_{2}}\geq \underline{\kappa },
\end{align}
which is the key convertibility condition used in high-dimensional analysis.
We refer interested readers to \citet[Equation 6.4]{BV11}, the remarks after
that, and \citet[Section 6.13]{BV11} for more detailed discussions and
further references. Under Assumption \ref{ass:eigen} and some regularity
conditions assumed in the paper, Lemma \ref{lem:localRE} establishes a local
version of \eqref{eq:RE}. Based on Lemma \ref{lem:localRE}, we can establish
the following asymptotic probability bounds for the first stage estimators.

\begin{thm}
	\label{thm:rate2}
	Suppose Assumptions \ref{ass:unconfoundedness}--\ref{ass:approx}, \ref%
{ass:ker}.1--\ref{ass:ker}.3, and \ref{ass:eigen} hold. Then 
\begin{equation*}
\sup_{t\in \mathcal{T}}||(\widehat{\nu }_{t}(X)-\nu _{t}(X))||_{\mathbb{P}_n,2}=O_{p}(\ell _{n}(\log (p\vee n)s)^{1/2}(nh_1)^{-1/2}),
\end{equation*}%
\begin{equation*}
\sup_{t\in \mathcal{T}}||\widehat{\nu }_{t}(X)-\nu _{t}(X)||_{\mathbb{P}, \infty
}=O_{p}(\ell _{n}(\log (p\vee n)s^{2}\zeta _{n}^{2}/(nh_1))^{1/2}),
\end{equation*}%
\begin{equation*}
\sup_{t\in \mathcal{T}}||(\widetilde{\nu }_{t}(X)-\nu _{t}(X))||_{\mathbb{P}_n,2}=O_{p}(\ell _{n}(\log (p\vee n)s)^{1/2}(nh_1)^{-1/2}),
\end{equation*}%
\begin{equation*}
\sup_{t\in \mathcal{T}}||\widetilde{\nu }_{t}(X)-\nu _{t}(X)||_{\mathbb{P}, \infty
}=O_{p}(\ell _{n}(\log (p\vee n)s^{2}\zeta _{n}^{2}/(nh_1))^{1/2}),
\end{equation*}%
and $\sup_{t\in \mathcal{T}}||\hat{\gamma}_{t}||_{0}=O_{p}(s)$. If in
addition, Assumption \ref{ass:ker}.4 holds, then 
\begin{equation*}
\sup_{(t,u)\in \mathcal{T}\mathcal{U}}||(\widehat{\phi }_{t,u}(X)-\phi
_{t,u}(X))||_{\mathbb{P}_n,2}=O_{p}(\ell _{n}(\log (p\vee
n)s)^{1/2}(nh_1)^{-1/2}),
\end{equation*}%
\begin{equation*}
\sup_{(t,u)\in \mathcal{T}\mathcal{U}}||\widehat{\phi }_{t,u}(X)-\phi
_{t,u}(X)||_{\mathbb{P}, \infty }=O_{p}(\ell _{n}(\log (p\vee n)s^{2}\zeta
_{n}^{2}/(nh_1))^{1/2}),
\end{equation*}%
\begin{equation*}
\sup_{(t,u)\in \mathcal{T}\mathcal{U}}||(\widetilde{\phi }_{t,u}(X)-\phi
_{t,u}(X))||_{\mathbb{P}_n,2}=O_{p}(\ell _{n}(\log (p\vee
n)s)^{1/2}(nh_1)^{-1/2}),
\end{equation*}%
\begin{equation*}
\sup_{(t,u)\in \mathcal{T}\mathcal{U}}||\widetilde{\phi }_{t,u}(X)-\phi
_{t,u}(X)||_{\mathbb{P}, \infty }=O_{p}(\ell _{n}(\log (p\vee n)s^{2}\zeta
_{n}^{2}/(nh_1))^{1/2}),
\end{equation*}%
and $\sup_{(t,u)\in \mathcal{T}\mathcal{U}}||\hat{\theta}%
_{t,u}||_{0}=O_{p}(s).$ 
\end{thm}

Several comments are in order. First, due to the nonlinearity of the
logistic link function, Assumption \ref{ass:ker}.4 is needed for deriving
the asymptotic properties of the penalized local MLE estimators $\widehat{%
\phi}_{t,u}(x)$ and $\widetilde{\phi}_{t,u}(x)$. Second, the $L_{\mathbb{P}_n,2}$
bounds in Theorem \ref{thm:rate2} are faster than $(nh_1)^{-1/4}$ by
Assumption \ref{ass:rate2} below. This implies the estimators are
sufficiently accurate so that in the second stage, their second and higher
order impacts are asymptotically negligible. Last, the numbers of nonzero
coordinates of $\hat{\gamma}_t$ and $\hat{\theta}_{t,u}$ determine the
complexity of our first stage estimators, which are uniformly controlled
with a high probability.

For the conditional density estimation, we have the following results.

\begin{thm}
	\label{thm:rate1}
	Suppose Assumptions \ref{ass:unconfoundedness}--\ref{ass:approx}, \ref%
	{ass:ker}.1--\ref{ass:ker}.3, and \ref{ass:eigen} hold. Then 
	\begin{equation*}
	\sup_{t\in \mathcal{T}}||\hat{f}_{t}(X)-f_{t}(X)||_{\mathbb{P}_n,2}=O_{p}((\log (p\vee n)s/n)^{1/2}h_1^{-1}),
	\end{equation*}%
	\begin{equation*}
	\sup_{t\in \mathcal{T}}||\hat{f}_{t}(X)-f_{t}(X)||_{\mathbb{P}, \infty }=O_{p}((\log (p\vee n)s^{2}\zeta _{n}^{2}/n)^{1/2}h_1^{-1}),
	\end{equation*}%
	\begin{equation*}
	\sup_{t\in \mathcal{T}}||\tilde{f}_{t}(X)-f_{t}(X)||_{\mathbb{P}_n,2}=O_{p}((\log (p\vee n)s/n)^{1/2}h_1^{-1}),
	\end{equation*}%
	\begin{equation*}
	\sup_{t\in \mathcal{T}}||\tilde{f}_{t}(X)-f_{t}(X)||_{\mathbb{P}, \infty }=O_{p}((\log (p\vee n)s^{2}\zeta _{n}^{2}/n)^{1/2}h_1^{-1}),
	\end{equation*}%
	and $\sup_{t\in \mathcal{T}}||\hat{\beta}_{t}||_{0}=O_{p}(s).$ 
\end{thm}
The rates of convergence in Theorem \ref{thm:rate1} are the same as those derived in \citet[Section 8]{BCK18}. 


\subsection{The Second Stage Estimation}
\label{sec:2nd} 
Let $W=\{Y,T,X\}$ and $W_{u}=\{Y_{u},T,X\}$. For three generic functions $%
\breve{\nu}(\cdot )$, $\breve{\phi}(\cdot )$ and $\breve{f}(\cdot )$ of $X$,
denote 
\begin{equation*}
\Pi _{t}^{\prime }(W,\breve{\nu},\breve{f})=\frac{(Y-\breve{\nu}(X))}{\breve{%
		f}(X)h_2}K(\frac{T-t}{h_2})+\breve{\nu}(X)
\end{equation*}%
and 
\begin{equation*}
\Pi _{t,u}(W_{u},\breve{\phi},\breve{f})=\frac{(Y_{u}-\breve{\phi}(X))}{%
	\breve{f}(X)h_2}K(\frac{T-t}{h_2})+\breve{\phi}(X).
\end{equation*}%
Then the estimators $\hat{\mu}(t)$ and $\hat{\alpha}(t,u)$ can be written as 
\begin{equation*}
\hat{\mu}(t)=\mathbb{P}_{n}\Pi _{t}^{\prime }(W,\overline{\nu }_{t},%
\overline{f})\quad \text{and}\quad \hat{\alpha}(t,u)=\mathbb{P}_{n}\Pi
_{t,u}(W_{u},\overline{\phi }_{t,u},\overline{f}),
\end{equation*}%
where $\overline{\nu }_{t}(\cdot )$, $\overline{\phi }_{t,u}(\cdot )$, and $%
\overline{f}(\cdot )$ are either the Lasso estimators (i.e., $\widehat{%
	\nu }_{t}(\cdot )$, $\widehat{\phi }_{t,u}(\cdot )$, and $\hat{f}_{t}(\cdot
) $) or the post-Lasso estimators (i.e., $\widetilde{\nu }_{t}(\cdot )$, $%
\widetilde{\phi }_{t,u}(\cdot )$, and $\tilde{f}_{t}(\cdot )$) as defined in
Section \ref{sec:1st}.\medskip

\begin{ass}
	Let $h_2 = C_2 n^{-H_2}$ for some positive constant $C_2$.
	
	\begin{enumerate}
		\item $H_2 \in [1/5,1/3)$, $\log^2(n) s^2 \log^2(p \vee n)/(nh_2) \rightarrow 0$, and $\ell_n^2s^2\log^2(p \vee n)/(nh_1^2) \rightarrow 0$, and $\ell_n^2s^2\log^2(p \vee n)h_2/(nh_1^3) \rightarrow 0$.
		
		\item $H_2 \in (1/4,1/3) $, $\log^2(n) s^2 \log^2(p \vee n)/(nh_2^2) \rightarrow 0$, $\ell_n^2s^2\log^2(p \vee n)/(nh_1^2h_2) \rightarrow 0$, and $\ell_n^2s^2\log^2(p \vee n)/(nh_1^3) \rightarrow 0$.
	\end{enumerate}
	
	\label{ass:rate2}
\end{ass}

\begin{thm}
	Suppose Assumptions \ref{ass:unconfoundedness}--\ref{ass:eigen} and \ref%
	{ass:rate2}.1 hold. Then 
	\begin{equation*}
	\hat{\mu}(t)-\mu (t)=(\mathbb{P}_{n}-\mathbb{P})\Pi _{t}^{\prime }(W,\nu
	_{t},f_{t})+\mathcal{B}_\mu(t)h_2^2+R_{n}^{\prime }(t)
	\end{equation*}%
	and 
	\begin{equation*}
	\hat{\alpha}(t,u)-\alpha (t,u)=(\mathbb{P}_{n}-\mathbb{P})\Pi
	_{t,u}(W_{u},\phi _{t,u},f_{t})+\mathcal{B}_\alpha(t,u)h_2^2+R_{n}(t,u),
	\end{equation*}%
	where 
	$$\mathcal{B}_\mu(t) = \frac{\kappa_2}{2}\left[\mathbb{E}\left(\partial_t^2 \nu_{t}(X) + \frac{2\partial_t \nu_{t}(X)\partial_t f_{t}(X)}{f_{t}(X)}\right)\right],$$ 
	$$\mathcal{B}_\alpha(t,u) = \frac{\kappa_2}{2}\left[\mathbb{E}\left(\partial_t^2 \phi_{t,u}(X) + \frac{2\partial_t \phi_{t,u}(X)\partial_t f_{t}(X)}{f_{t}(X)}\right)\right],$$ 
	$\kappa_2 = \int u^2K(u)du$, $\sup_{t\in \mathcal{T}}|R_{n}^{\prime }(t)|=o_{p}((nh_2)^{-1/2})$ and $%
	\sup_{(t,u)\in \mathcal{T}\mathcal{U}}|R_{n}(t,u)|=o_{p}((nh_2)^{-1/2}).$ If
	Assumption \ref{ass:rate2}.1 is replaced by Assumption \ref{ass:rate2}.2,
	then 
	$$\sup_{t\in \mathcal{T}}|\mathcal{B}_\mu(t)h_2^2 +R_{n}^{\prime }(t)|=o_{p}(n^{-1/2})\quad \text{and} \quad%
	\sup_{(t,u)\in \mathcal{T}\mathcal{U}}|\mathcal{B}_\alpha(t,u)h_2^2+R_{n}(t,u)|=o_{p}(n^{-1/2}).$$\label%
	{thm:2nd}
\end{thm}

Theorem \ref{thm:2nd} presents the Bahadur representations of the
nonparametric estimators $\hat{\mu}(t)$ and $\hat{\alpha}(t,u)$ with a
uniform control on the remainder terms. For most purposes (e.g., to obtain
the asymptotic distributions of these intermediate estimators or to obtain
the results below), Assumption \ref{ass:rate2}.1 is sufficient.
Occasionally, one needs to impose Assumption \ref{ass:rate2}.2 to have a
better control on the remainder terms, say, when one conducts an $L_{2}$%
-type specification test. See the remark after Theorem \ref{thm:q} below. 

\subsection{The Third Stage Estimation}
\label{sec:3rd}
Recall that $q_{\tau }(t)$ denotes the $\tau $-th quantile of $Y(t)$, which
is the inverse of $\alpha (t,u)$ w.r.t. $u$. We propose to estimate $q_{\tau
}(t)$ by $\hat{q}_{\tau }(t)$ where $\hat{q}_{\tau }(t)=\inf \{u:\hat{\alpha}%
^{r}(t,u)\geq \tau \}$ and $\hat{\alpha}^{r}(t,u)$ is the rearrangement
of $\hat{\alpha}(t,u)$.

We rearrange $\hat{\alpha}(t,u)$ to make it monotonically increasing in $%
u\in \mathcal{U}$. Following \cite{CFG10}, for a generic function $Q(\cdot)$, we define $\overline{Q}=Q\circ
\psi ^{\leftarrow }$ where $\psi $ can be any increasing bijective mapping: $%
\mathcal{U}\mapsto \lbrack 0,1]$ and $\psi ^{\leftarrow }$ is the inverse of $\psi$. Then the rearrangement $\overline{Q}^{r}$
of $\overline{Q}$ is defined as 
\begin{equation*}
\overline{Q}^{r}(u)=F^{\leftarrow }(u)=\inf \{y:F(y)\geq u\},
\end{equation*}%
where $F(y)=\int_{0}^{1}1\{\overline{Q}(u)\leq y\}du$. Then the
rearrangement $Q^{r}$ for $Q$ is $Q^{r}=\overline{Q}^{r}\circ \psi (u).$

The rearrangement and inverse are two functionals operating on the process 
\begin{equation*}
\{\hat{\alpha}(t,u):(t,u)\in \mathcal{TU}\}
\end{equation*}%
and are shown to be Hadamard differentiable by \cite{CFG10} and \cite{VW96},
respectively. However, by Theorem \ref{thm:2nd}, 
\begin{equation*}
\sup_{(t,u)\in \mathcal{T}\mathcal{U}}(nh_2)^{1/2}(\hat{\alpha}(t,u)-\alpha
(t,u))=O_{p}(\log ^{1/2}(n)),
\end{equation*}%
which is not asymptotically tight. Therefore, the standard functional delta
method used in \cite{CFG10} and \cite{VW96} is not directly applicable. The
next theorem overcomes this difficulty and establishes the linear expansion
of the quantile estimator. Denote $\mathcal{TI}$, $\{q_{\tau }(t):\tau \in 
\mathcal{I}\}^{\varepsilon }$, $\overline{\{q_{\tau }(t):\tau \in \mathcal{I}%
\}^{\varepsilon }}$, and $\mathcal{U}_{t}$ as $\mathcal{T}\times \mathcal{I}$%
, the $\varepsilon $-enlarged set of $\{q_{\tau }(t):\tau \in \mathcal{I}\}$%
, the closure of $\{q_{\tau }(t):\tau \in \mathcal{I}\}^{\varepsilon }$, and
the projection of $\mathcal{TU}$ on $T=t$, respectively.

\begin{thm}
Suppose that Assumptions \ref{ass:unconfoundedness}--\ref{ass:eigen} and \ref%
{ass:rate2}.1 hold. If $\overline{\{q_{\tau }(t):\tau \in \mathcal{I}%
\}^{\varepsilon }}\subset \mathcal{U}_{t}$ for any $t\in \mathcal{T}$, then 
\begin{equation*}
\hat{q}_{\tau }(t)-q_{\tau }(t)=-(\mathbb{P}_{n}-\mathbb{P})\frac{\Pi
_{t,u}(W_{q_{\tau }(t)},\phi _{t,q_{\tau }(t)},f_{t})}{f_{Y(t)}(q_{\tau
}(t))} - \mathcal{\beta}_q(t,\tau)h_2^2 +R_{n}^{q}(t,\tau ),
\end{equation*}%
where $f_{Y(t)}$ is the density of $Y(t)$, $\mathcal{\beta}_q(t,\tau) = \frac{\mathcal{\beta}_\alpha(t,q_\tau(t))}{f_{Y(t)}(q_{\tau}(t))}$, and $\sup_{(t,\tau )\in \mathcal{T}%
\mathcal{I}}R_{n}^{q}(t,\tau )=o_{p}((nh_2)^{-1/2}).$ If Assumption \ref%
{ass:rate2}.1 is replaced by Assumption \ref{ass:rate2}.2, then 
$$\sup_{(t,\tau )\in \mathcal{T}\mathcal{I}}\left(|R_{n}^{q}(t,\tau
)| + \left|\mathcal{\beta}_q(t,\tau)\right|\right)h_2^2=o_{p}(n^{-1/2}). $$
\label{thm:q}
\end{thm}


Under Assumption \ref{ass:rate2}.2, the remainder term $R_{n}^{q}(t,\tau )$
is $o_{p}(n^{-1/2})$ uniformly in $(t,\tau )\in \mathcal{T}\mathcal{I}$.
This result is needed if one wants to establish an $L_{2}$-type
specification test of $q_{\tau }(t)$. For example, one may be interested in
testing the null hypotheses of the quantile partial derivative being
homogeneous across treatment. In this case, the null hypothesis can be
written as 
\begin{equation*}
H_{0}:q_{\tau }(t)=\beta _{0}(\tau )+\beta _{1}(\tau )t\text{ \ for all }%
(t,\tau )\in \mathcal{T}\mathcal{I},
\end{equation*}%
and the alternative hypothesis is the negation of $H_{0}$. One way to
conduct a consistent test for the above hypothesis is to employ the
residuals of the linear regression of $\hat{q}_{\tau }(T_{i})$ on $T_{i}$ to
construct the test statistic $\Upsilon _{n}(\tau )$, i.e., 
\begin{equation*}
\Upsilon _{n}(\tau )=\frac{1}{n}\sum_{i=1}^{n}(\hat{q}_{\tau }(T_{i})-\hat{%
\beta}_{0}-\hat{\beta}_{1}T_{i})^{2}1\{T_{i}\in \mathcal{T}\},
\end{equation*}%
where $(\hat{\beta}_{0},\hat{\beta}_{1})$ are the linear coefficient
estimators. This type of specification test has been previously studied by 
\cite{SC13}, \cite{LLS15}, \cite{SJZ15}, \cite{HSWY16}, and \cite{SH16} in
various contexts. One can follow them and apply the results in Theorem \ref%
{thm:q} to study the asymptotic distribution of $\Upsilon _{n}(\tau )$ for
each $\tau .$ In addition, one can also consider either an integrated or a
sup-version of $\Upsilon _{n}(\tau )$ and then study its asymptotic
properties. For brevity we do not study such a specification test in this
paper.

\medskip Given the estimators $\hat{\mu}(t)$ and $\hat{q}_{\tau }(t)$, we
can run local linear regressions of $\hat{\mu}(T_i)$ and $\hat{q}_{\tau
}(T_{i})$ on $\left( 1,T_{i}-t\right) $ and obtain estimators $\breve{\beta}%
^{1}(t) $ and $\hat{\beta}_{\tau }^{1}(t)$ of $\partial \mu(t)$ and $%
\partial _{t}q_{\tau }(t)$, respectively, as estimators of the linear
coefficients in the local linear regression.\footnote{%
Alternatively, one can consider the local quadratic or cubic regression.}
Specifically, we define 
\begin{equation*}
(\breve{\beta}^{0}(t),\breve{\beta}^{1}(t))=\argmax_{\beta ^{0},\beta ^{1}}\sum_{i=1}^n(%
\hat{\mu}(T_{i})-\beta ^{0}-\beta ^{1}(T_{i}-t))^{2}K(\frac{T_{i}-t}{h_2})
\end{equation*}%
and 
\begin{equation*}
(\hat{\beta}_{\tau }^{0}(t),\hat{\beta}_{\tau }^{1}(t))=\argmax_{\beta
^{0},\beta ^{1}}\sum_{i=1}^n(\hat{q}_{\tau }(T_{i})-\beta ^{0}-\beta ^{1}(T_{i}-t))^{2}K(%
\frac{T_{i}-t}{h_2}),
\end{equation*}%
where $h_2$ is the second-stage bandwidth. It is possible to use a third bandwidth $h_3$ in this step. Results similar to Theorem \ref{thm:qprime} below still holds if $h_3/h_2 = O(1)$. Note that the usual optimal bandwidth for the kernel estimator of the derivative is $O(n^{-1/7})$. However, because $h_2 = O(n^{-1/5})$, the requirement that  $h_3/h_2 = O(1)$ implies the optimal bandwidth is not achievable. The key reason is that, unlike the usual local linear regression, we need to plug in the estimates of $\mu(\cdot)$ and $q_\tau(\cdot)$. For simplicity, we just take $h_3 = h_2.$

The following theorem shows the asymptotic properties of $\breve{\beta}%
^{1}(t)$ and $\hat{\beta}_{\tau }^{1}(t)$. 
\begin{thm}
Suppose Assumptions \ref{ass:unconfoundedness}--\ref{ass:eigen}, and \ref%
{ass:rate2}.1. If $\overline{\{q_{\tau }(t):\tau \in \mathcal{I}%
\}^{\varepsilon }}\subset \mathcal{U}_{t}$ for any $t\in \mathcal{T}$, then 
\begin{equation*}
\breve{\beta}^{1}(t)-\partial _{t}\mu(t)=(\mathbb{P}_n - \mathbb{P})(\kappa
_{2}f_{t}(X_{j})h_2^{2})^{-1}\biggl[Y_{j}- \nu _{t}(X_{j})\biggr]\overline{K}(%
\frac{T_{j}-t}{h_2})+\breve{R}_{n}^{1}(t )
\end{equation*}%
and 
\begin{equation*}
\hat{\beta}_{\tau }^{1}(t)-\partial _{t}q_{\tau }(t)=-(\mathbb{P}_n - \mathbb{P})(\kappa _{2}f_{Y(t)}(q_{\tau }(t))f_{t}(X_{j})h_2^{2})^{-1}%
\biggl[Y_{q_{\tau }(t),j}-\phi _{t,q_{\tau }(t)}(X_{j})\biggr]\overline{K}(%
\frac{T_{j}-t}{h_2})+R_{n}^{1}(t,\tau ),
\end{equation*}%
where $\sup_{t \in \mathcal{T}}|\breve{R}_{n}^{1}(t )| + \sup_{(t,\tau )\in 
\mathcal{T}\mathcal{I}}|R_{n}^{1}(t,\tau )|=o_{p}((nh_2^{3})^{-1/2})$ and $%
\overline{K}(v)=\int wK(v-w)K(w)dw$. \label{thm:qprime}
\end{thm}

Theorem \ref{thm:qprime} presents the Bahadur representations for $\breve{%
\beta}^{1}(t)$ and $\hat{\beta}_{\tau }^{1}(t).$ Since they are estimators for
the first order derivatives $\partial_t \mu(t)$ and $\partial _{t}q_{\tau
}(t),$ respectively, we can show that they converge to the true values at
the $\left( nh_2^{3}\right) ^{1/2}$-rate. Such a rate is common for kernel estimations of the first-order derivative of the conditional expectation, i.e., \citet[Theorem 2.10]{LR07}. 

\section{Inference}

\label{sec:inf}

In this section, we study the inference for $\mu (t),$ $q_{\tau }(t),$ and $%
\partial _{t}q_{\tau }(t).$ We follow the lead of \cite{BCFH13} and consider
the weighted-bootstrap inference. Let $\{\eta _{i}\}_{i=1}^{n}$ be a
sequence of i.i.d. random variables generated from the distribution of $\eta 
$ such that it has sub-exponential tails and unit mean and variance.%
\footnote{
A random variable $\eta $ has sub-exponential tails if $P(|\eta |>x)\leq 
K\exp (-Cx)$ for every $x$ and some constants $K$ and $C$.} For example, $%
\eta $ can be a standard exponential random variable or a normal random
variable with unit mean and standard deviation. We conduct the bootstrap
inference based on the following procedure.

\begin{enumerate}
\item Obtain $\widehat{\nu}_{t}(x)$, $\widehat{\phi}_{t,u}(x)$, $\hat{f}_t(x)
$, $\widetilde{\nu}_{t}(x)$, $\widetilde{\phi}_{t,u}(x)$ and $\tilde{f}_t(x) 
$ from the first stage.

\item For the $b$-th bootstrap sample:

\begin{itemize}
\item Generate $\{\eta_i\}_{i=1}^n$ from the distribution of $\eta$.

\item Compute 
\begin{equation*}
\hat{\mu}^b(t) := \frac{1}{\sum_{i=1}^n \eta_i}\sum_{i=1}^n \eta_i
\Pi_{t}^{\prime }(W_i,\overline{\nu}_{t},\overline{f}_t) 
\end{equation*}
and 
\begin{equation*}
\hat{\alpha}^b(t,u) := \frac{1}{\sum_{i=1}^n \eta_i}\sum_{i=1}^n
\eta_i\Pi_{t,u}(W_{ui},\overline{\phi}_{t,u},\overline{f}_t),
\end{equation*}
where $(\overline{\phi}_{t,u}(\cdot),\overline{f}_t(\cdot))$ are either $(%
\widehat{\phi}_{t,u}(\cdot),\hat{f}_t(\cdot))$ or $(\widetilde{\phi}%
_{t,u}(\cdot),\tilde{f}_t(\cdot))$.

\item Rearrange $\hat{\alpha}^b(t,u)$ and obtain $\hat{\alpha}^{br}(t,u)$.

\item Invert $\hat{a}^{br}(t,u)$ w.r.t. $u$ and obtain $\hat{q}^b_\tau(t) =
\inf\{u:\hat{a}^{br}(t,u) \geq \tau \}$.

\item Compute $\breve{\beta}^{b1}(t)$ and $\hat{\beta}_\tau^{b1}(t)$ as the
slope coefficients of local linear regressions of $\eta_i\hat{\mu}^b(T_i)$
on $(\eta_i,\eta_i(T_i-t))$ and $\eta_i\hat{q}^b_\tau(T_i)$ on $%
(\eta_i,\eta_i(T_i-t))$, respectively.
\end{itemize}

\item We repeat the above step for $b=1,\cdots,B$ and obtain a bootstrap
sample of 
\begin{equation*}
\{\hat{\mu}^b(t),\hat{q}^b_\tau(t), \breve{\beta}^{b1}(t), \hat{\beta}%
_\tau^{b1}(t)\}_{b=1}^B.
\end{equation*}

\item Obtain $\widehat{Q}^{\mu }(\alpha )$, $\widehat{Q}^{0}(\alpha )$, $%
\widehat{Q}^{\mu1}(\alpha )$, and $\widehat{Q}^{1}(\alpha )$ as the $\alpha $%
-th quantile of the sequences $\{\hat{\mu}^{b}(t)-\hat{\mu}(t)\}_{b=1}^{B}$, 
$\{\hat{q}_{\tau }^{b}(t)-\hat{q}_{\tau }(t)\}_{b=1}^{B}$, $\{\breve{\beta}%
^{b1}(t) - \breve{\beta}^{1}(t)\}_{b=1}^B$, and $\{\hat{\beta}_{\tau
}^{b1}(t)-\hat{\beta}_{\tau }^{1}(t)\}_{b=1}^{B}$, respectively.
\end{enumerate}

The standard $100(1-\alpha)\%$ percentile bootstrap confidence interval for $%
q_\tau(t)$ is 
\begin{equation*}
(\hat{Q}^0(\alpha/2) + \hat{q}_\tau(t),\hat{Q}^0(1-\alpha/2) + \hat{q}%
_\tau(t)).
\end{equation*}
However, in our simulation study, we find that it slightly undercovers.
Instead, we use the fact that normal CDF is symmetric and propose to use the
modified percentile bootstrap confidence interval as follows: 
\begin{equation*}
(-\hat{Q}^{*0}(\alpha/2) + \hat{q}_\tau(t),\hat{Q}^{*0}(\alpha/2) + \hat{q}%
_\tau(t)),
\end{equation*}
where $\hat{Q}^{*0}(\alpha/2) = (-\hat{Q}^0(\alpha/2)) \vee \hat{Q}%
^0(1-\alpha/2)$. We define $\widehat{Q}^{*\mu }(\alpha/2)$, $\widehat{Q}%
^{*\mu1}(\alpha/2)$, and $\widehat{Q}^{*1}(\alpha/2)$ in the same manner.
The following theorem summarizes the main results in this section.

\begin{thm}
Suppose that Assumptions \ref{ass:unconfoundedness}--\ref{ass:eigen} and \ref%
{ass:rate2}.1 hold and $nh_2^5 \rightarrow 0$. Then 
\begin{equation*}
\mathbb{P}(-\widehat{Q}^{*\mu }(\alpha /2) + \hat{\mu}(t)\leq \mu (t)\leq \widehat{Q}%
^{*\mu }(\alpha /2)+\hat{\mu}(t))\rightarrow 1-\alpha ,
\end{equation*}%
\begin{equation*}
\mathbb{P}(-\widehat{Q}^{*0}(\alpha /2) + \hat{q}_\tau(t) \leq q_{\tau }(t)\leq 
\widehat{Q}^{*0}(\alpha /2) + \hat{q}_\tau(t))\rightarrow 1-\alpha ,
\end{equation*}%
\begin{equation*}
\mathbb{P}(-\widehat{Q}^{*\mu 1}(\alpha /2) + \breve{\beta}^1(t) \leq \partial
_{t}\mu(t)\leq \widehat{Q}^{*\mu 1}(\alpha /2) + \breve{\beta}%
^1(t))\rightarrow 1-\alpha,
\end{equation*}
and 
\begin{equation*}
\mathbb{P}(-\widehat{Q}^{*1}(\alpha /2) + \breve{\beta}^1_\tau(t) \leq \partial
_{t}q_{\tau }(t)\leq \widehat{Q}^{*1}(\alpha /2) + \breve{\beta}%
^1_\tau(t))\rightarrow 1-\alpha .
\end{equation*}%
\label{thm:infer}
\end{thm}

Theorem \ref{thm:infer} implies that, via under-smoothing, the $100(1-\alpha )\%$ bootstrap
confidence intervals for $\mu (t),$ $q_{\tau }(t),$ $\partial_t\mu(t)$, and $%
\partial _{t}q_{\tau }(t)$ have the correct asymptotic coverage probability $%
1-\alpha .$ We need to under-smooth because, regardless of under-smoothing, the bootstrap estimator is always center around the original estimator without the asymptotic bias. With more complicated notations and the arguments of strong approximation in \cite{CCK14} and \cite{CCK14-anti}, one can show that
the validity of bootstrap inference holds uniformly over $\left( t,\tau \right).$ One of the key ingredients to verify \citet[Condition H1]{CCK14-anti} is the linear expansions of the estimators with a uniform control of the reminder terms, which has already been established in Theorems \ref{thm:q} and \ref{thm:qprime}. 

\section{Monte Carlo Simulations}

\label{sec:sim}

This section presents the results of Monte Carlo simulations, which
demonstrate the finite sample performance of the estimation and inference
procedure. 
Let $Y$ be generated as 
\begin{eqnarray}
	\label{eq:YXT}
Y=\Lambda \left(\left(U+b(X)'\beta- \Phi^{-1}\left(0.5T+
0.25\right)\right) \exp\left(\left(T-0.5\right)^2\right)\right)
\end{eqnarray}
while $T$ be generated as 
\begin{equation}
\label{eq:T|X}
T = \Lambda(V-b(X)'\beta),
\end{equation}
where $U$ and $V$ are two standard logistic random variables such that $U \perp V$ and $(U,V) \perp X$, $%
\Lambda(\cdot)$ and $\Phi(\cdot)$ are the logistic and normal CDFs, respectively, $p=100$, $X$ is a $p$-dimensional
random variables whose distribution is the Gaussian copula with covariance
parameter $[{0.5^{|j-k|}}]_{jk}$, and $b(X)$ is a vector of basis functions constructed from $X$. Note that $T$ ranges from $0$ to $1$. The parameters of interest are $q_{\tau}(t)$ and $\partial_{t}q_{\tau}(t)$,
where $t = 0.25, 0.5, 0.75$ and $\tau \in (0.2,0.8)$. We consider the following three designs: 
\begin{enumerate}
	\item (Exact sparse) $\beta_j=\frac{\pi^2}{24}$ for $j=1,\cdots,4$ $\beta_j = 0$, $j \geq 5$, and $b(X_j) = X_j$, $j=1,\cdots,100$;
	\item (Approximate sparsity) $\beta_j=\frac{1}{j^2}$ for $j=1,\cdots,100$ and $b(X_j) = X_j$, $j=1,\cdots,100$;
	\item (Sieve basis) $\beta_1=\beta_2=\frac{\pi^2}{12}$ and $\beta_j = 0$, $j \geq 3$. We construct $b(X)$ as the cubic spline basis functions of $(X_1,X_2)$:
	\begin{align*}
	b(X) = & \left[1,X_1,X_1^2,X_1^3,\max(X_1-q^{(1)}(0.1),0)^3,\cdots,\max(X_1-q^{(1)}(0.9),0)^3\right] \\
	&\times \left[1,X_2,X_2^2,X_2^3,\max(X_2-q^{(2)}(0.1),0)^3,\cdots,\max(X_2-q^{(2)}(0.9),0)^3\right],
	\end{align*} 
	where $q^{(j)}(\tau)$ denotes the $\tau$-th empirical quantile of $X_j$, $j=1,2$. This results in 169 basis functions. We further remove the basis functions with variance less than $10^{-4}$. We end up with about 128 basis functions on average.\footnote{The number of basis functions slightly varies across simulations.} 
\end{enumerate}
Note that the sum of the coefficients are (approximately) $\pi^2/6$ for all three designs.  
We normalize the basis functions $b(X)$ by their sample means and standard errors.

We use Gaussian kernel function in all three stages. We have four tuning parameters: $\lambda$, $\tilde{\lambda}$, $h_1$, and $h_2$. As we discussed in Section \ref{sec:1st}, we use 
\begin{align*}
\lambda = \ell_n(\log(p \vee nh_1)nh_1)^{1/2} \quad \text{and} \quad \tilde{\lambda} = 1.1\Phi^{-1}(1-\gamma/\{p \vee nh_1\})n^{1/2},
\end{align*}
where $\ell_n = \sqrt{\log(\log(nh_1))}$ and $\gamma = 1/\log(n)$. We use the rule-of-thumb bandwidth for $h_1$, i.e., $h_1 = h^* = 1.06 \times sd(T) \times n^{-1/5}$. Last, we build
$h_2$ based on the rule-of-thumb
bandwidth for the local quantile regression suggested by \cite{yu1998local}.
In particular, \cite{yu1998local} propose the bandwidths $h_{RoT}(\tau)=C(\tau)\times h_{mean}$, where $C(\tau)$ is a constant dependent only on $\tau$, and $C(0.5) = 1.095$ and $C(0.25) = C(0.75) = 1.13$ and $h_{mean}$ is the bandwidth for the kernel estimation of $\mathbb{E}(Y|T)$.\footnote{We refer interested readers to \cite[Table 1]{yu1998local} for more details on $C(\tau)$.  In our simulation studies, as $C(\tau)$ is nearly constant over $\tau \in [0.25,0.75]$, we just choose $C(0.5) = 1.095$ for all the quantile index $\tau$.} We use the leave-one-out cross-validation to search for the optimal bandwidth of $h_{mean}$ over a grid in $(0.8h^*,1.2h^*)$. The resulting bandwidth is denoted as $h_{mean}^*$. In order to achieve under-smoothing, we define $h_2=n^{-1/10}\times C(\tau)\times h_{mean}^*$,
where our choice of the factor $n^{-1/10}$ follows %
\citet[][p.418]{Cai2012413}. 

\medskip We repeat the bootstrap inference 500 times and all the results are
based on 500 Monte Carlo simulations. The sample size is $n=500$.
Although the sample size is large compared to $p$, in this DGP, the first-stage
bandwidth is as small as $0.09$. The
effective sample size for the first-stage estimation is of order of magnitude of $nh_1\approx 45 < 100$. In fact, we obtained warning
signs of potential multi-collinearity and were unable to estimate the model
when implementing the traditional estimation procedures without variable
selection (i.e., without penalization).

\begin{figure}[H]
\centering
\includegraphics[scale = 0.565,angle=0]{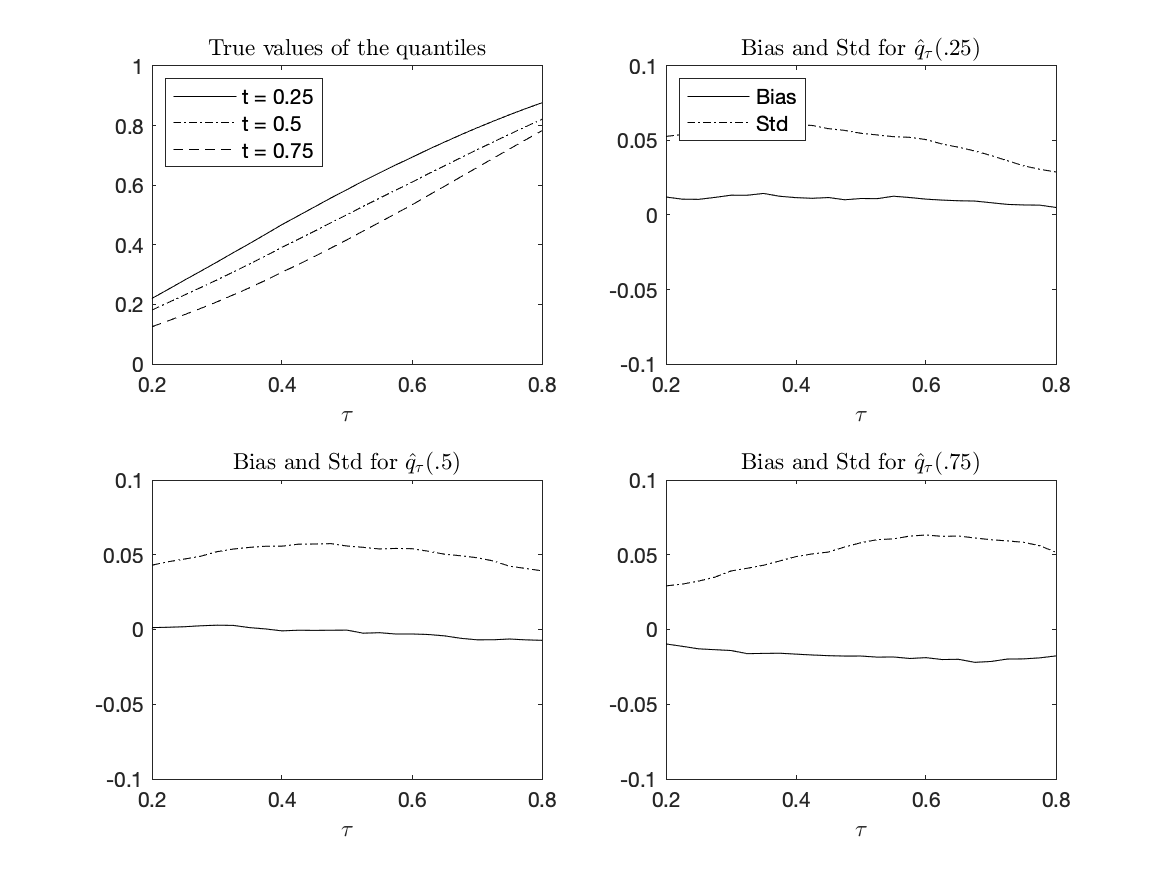}  
\caption{DGP1, finite sample performance of $\hat{q}_{\protect\tau }(t)$}
\label{fig:beta0_1}
\end{figure}

\begin{figure}[H]
\centering
\includegraphics[scale = 0.565,angle=0]{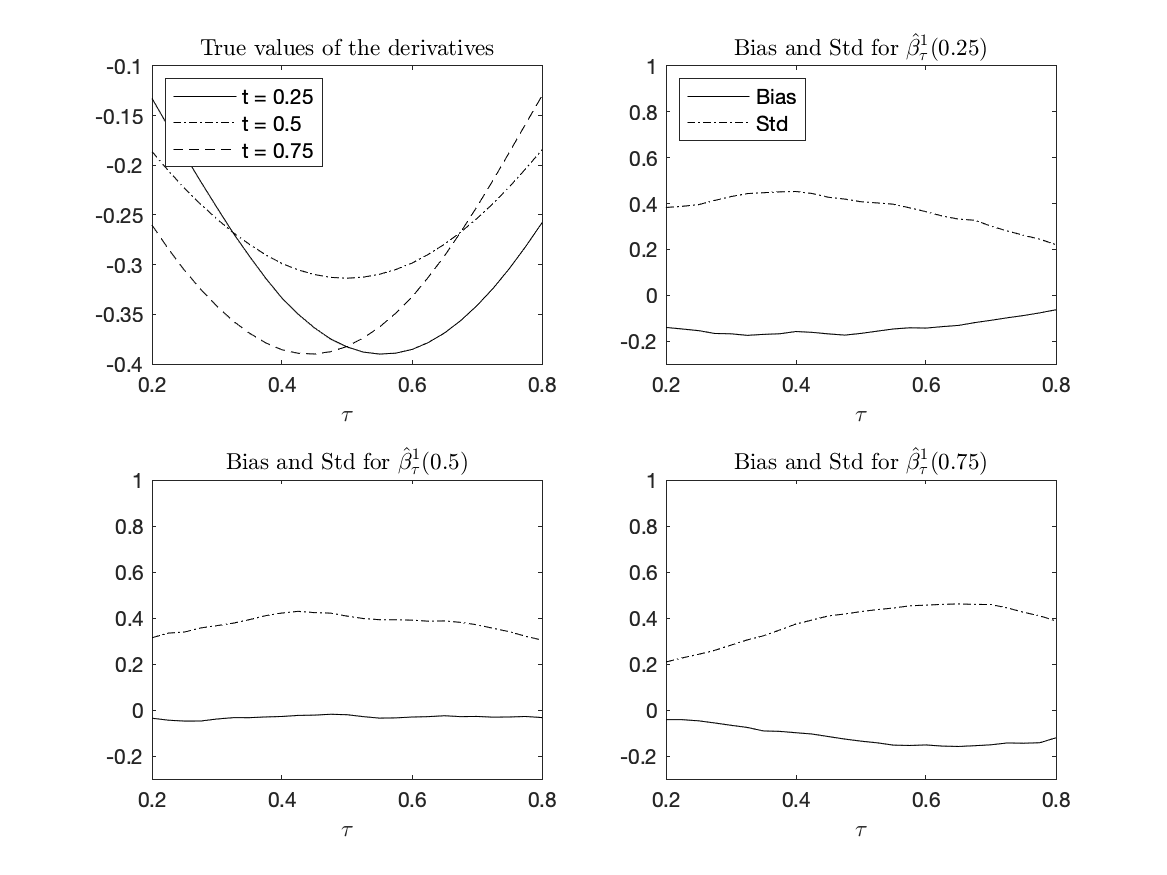}  
\caption{DGP1, finite sample performance of $\hat{\protect\beta}^1_\protect\tau(t)
$}
\label{fig:beta1_1}
\end{figure}

\begin{figure}[H]
\centering
\includegraphics[scale = 0.565,angle=0]{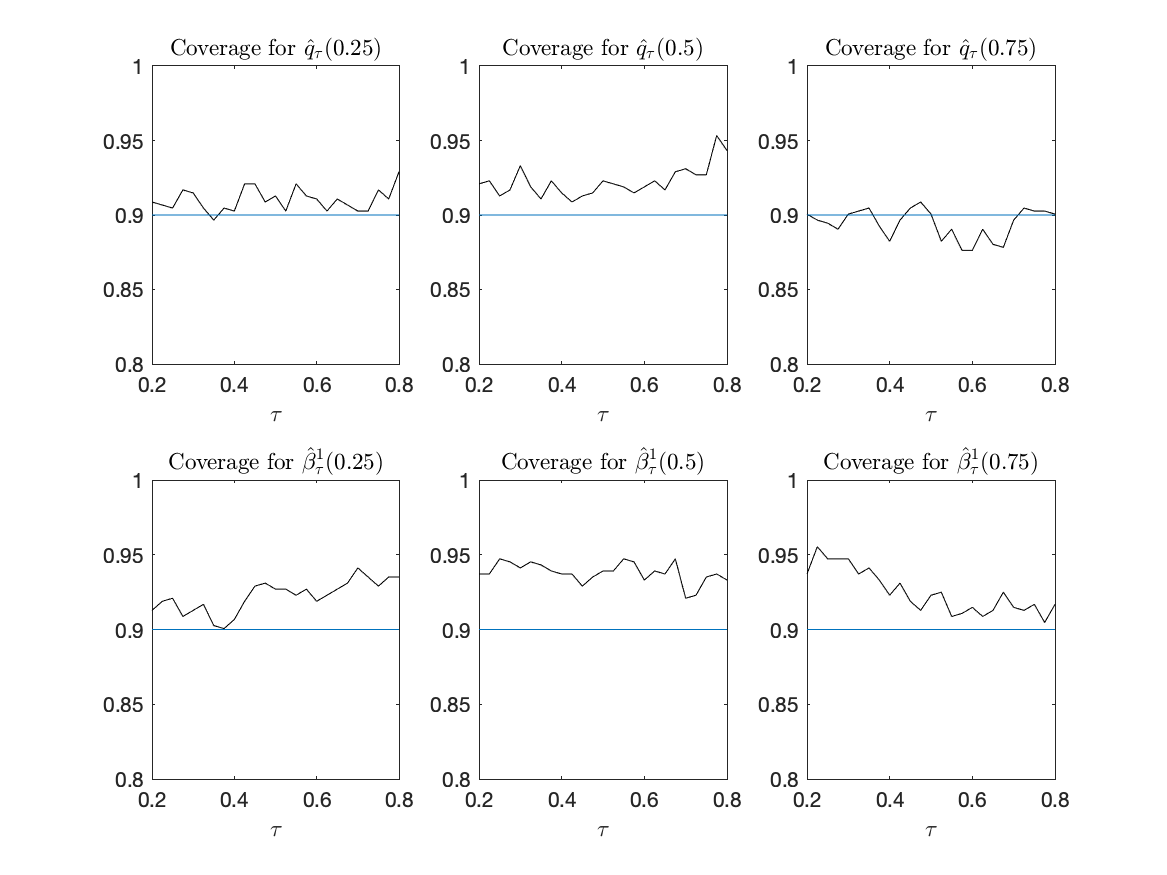}  
\caption{DGP1, coverage probability}
\label{fig:cov_1}
\end{figure}

\begin{figure}[H]
	\centering
	\includegraphics[scale = 0.565,angle=0]{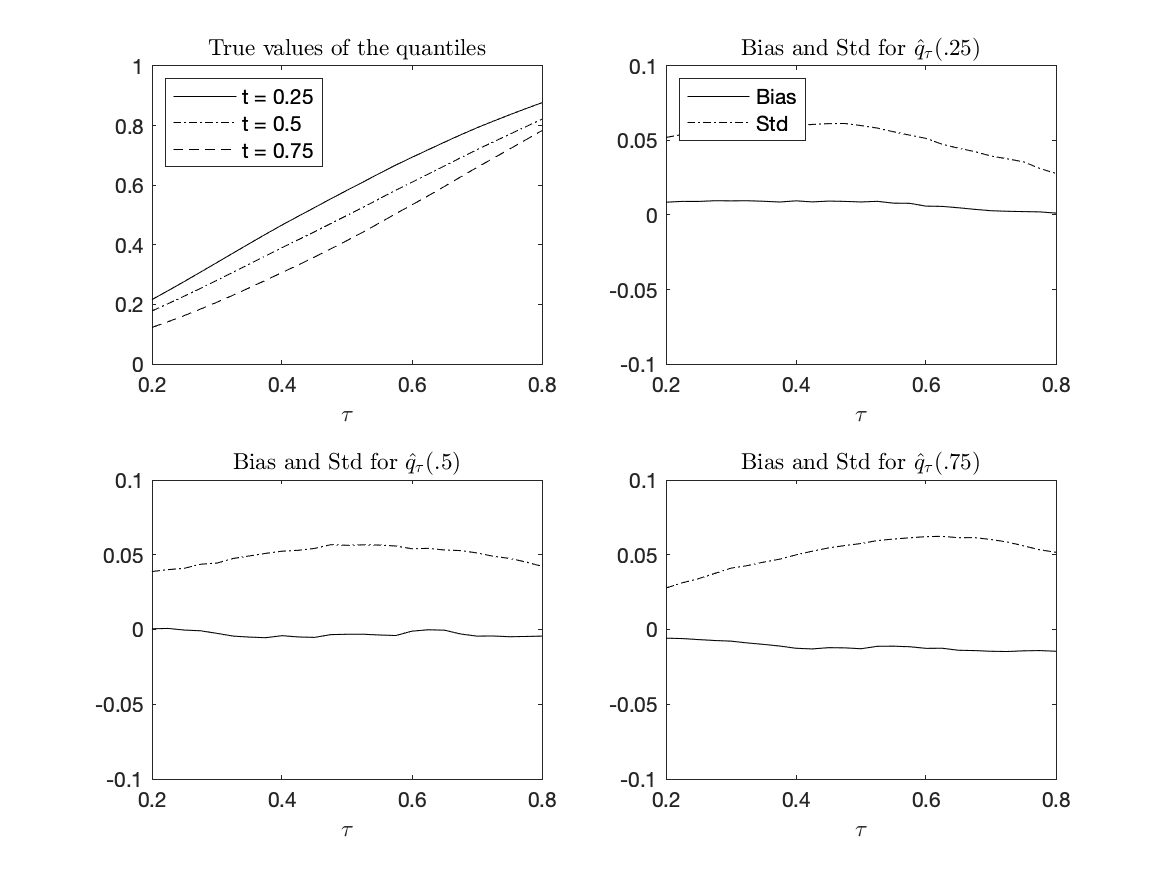}  
	\caption{DGP2, finite sample performance of $\hat{q}_{\protect\tau }(t)$}
	\label{fig:beta0_2}
\end{figure}

\begin{figure}[H]
	\centering
	\includegraphics[scale = 0.565,angle=0]{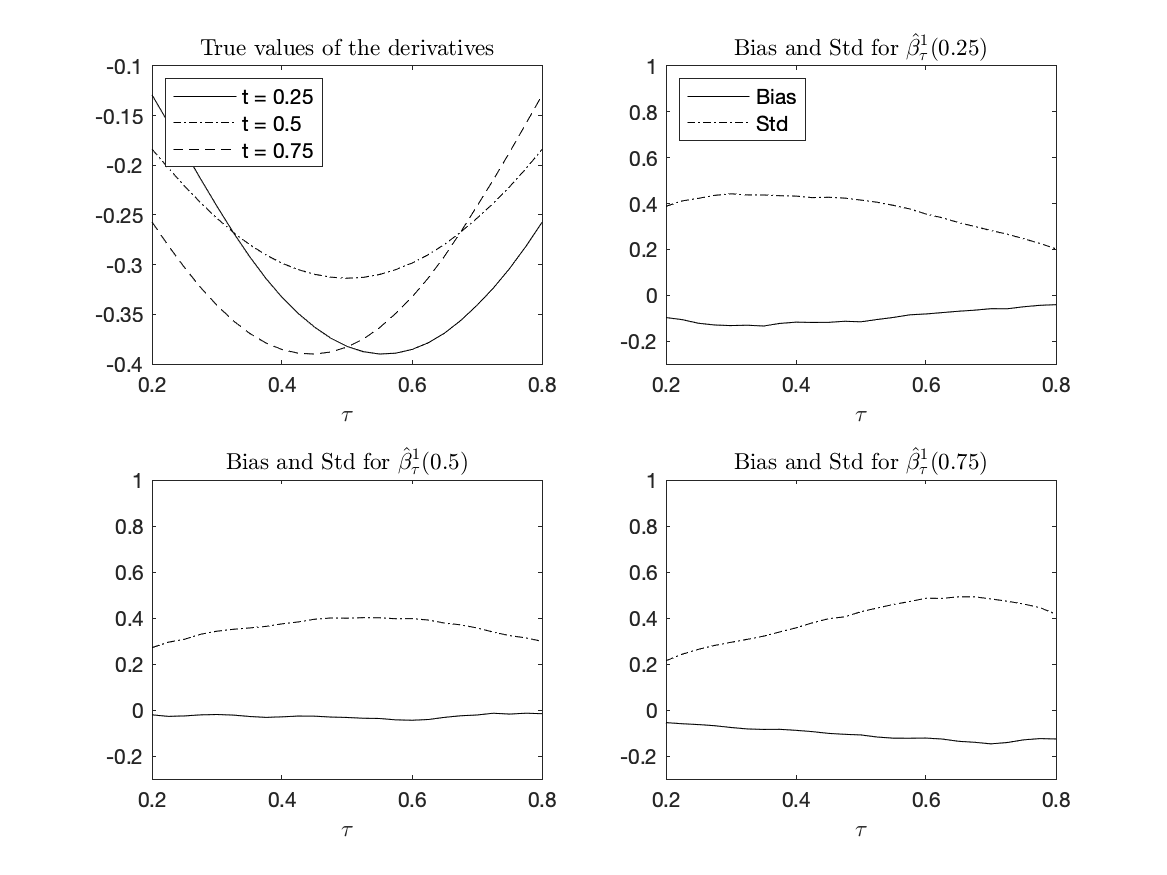}  
	\caption{DGP2, finite sample performance of $\hat{\protect\beta}^1_\protect\tau(t)
		$}
	\label{fig:beta1_2}
\end{figure}

\begin{figure}[H]
	\centering
	\includegraphics[scale = 0.565,angle=0]{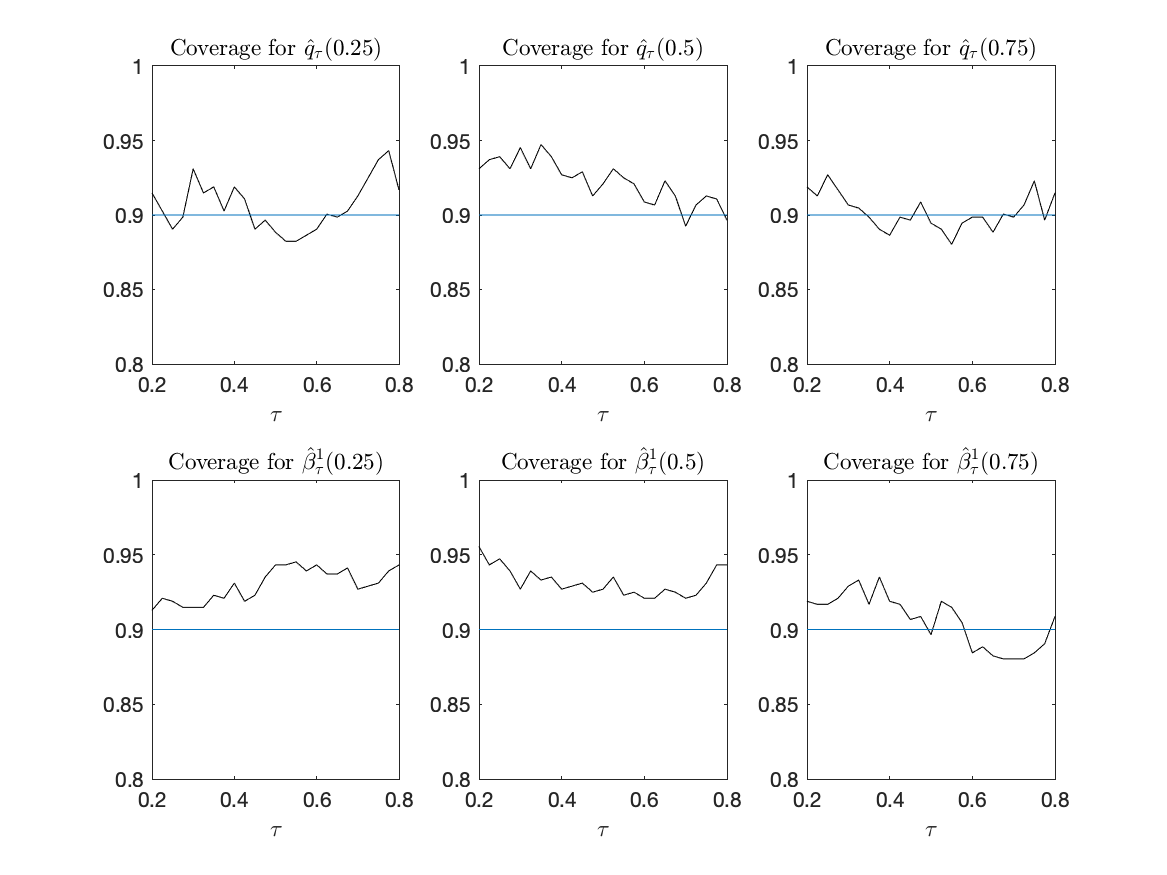}  
	\caption{DGP2, coverage probability}
	\label{fig:cov_2}
\end{figure}

\begin{figure}[H]
	\centering
	\includegraphics[scale = 0.565,angle=0]{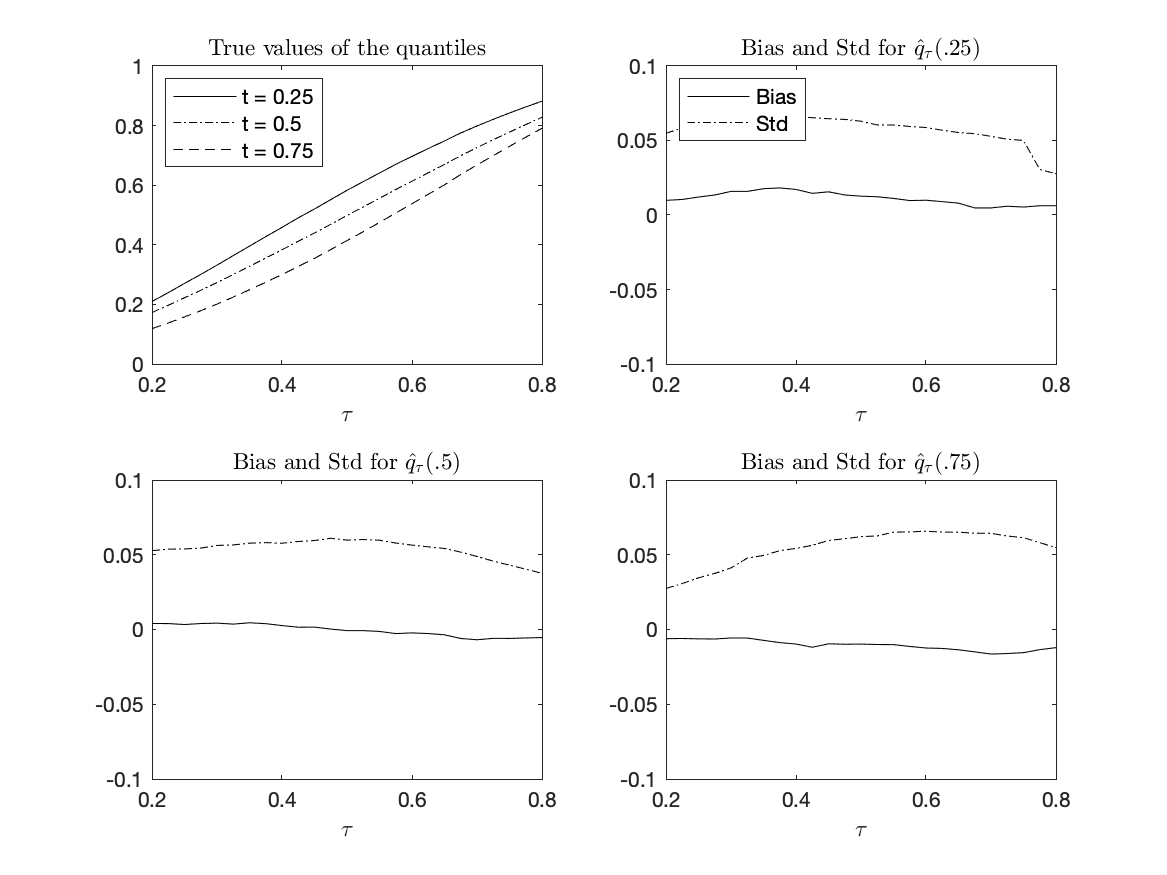}  
	\caption{DGP3, finite sample performance of $\hat{q}_{\protect\tau }(t)$}
	\label{fig:beta0_3}
\end{figure}

\begin{figure}[H]
	\centering
	\includegraphics[scale = 0.565,angle=0]{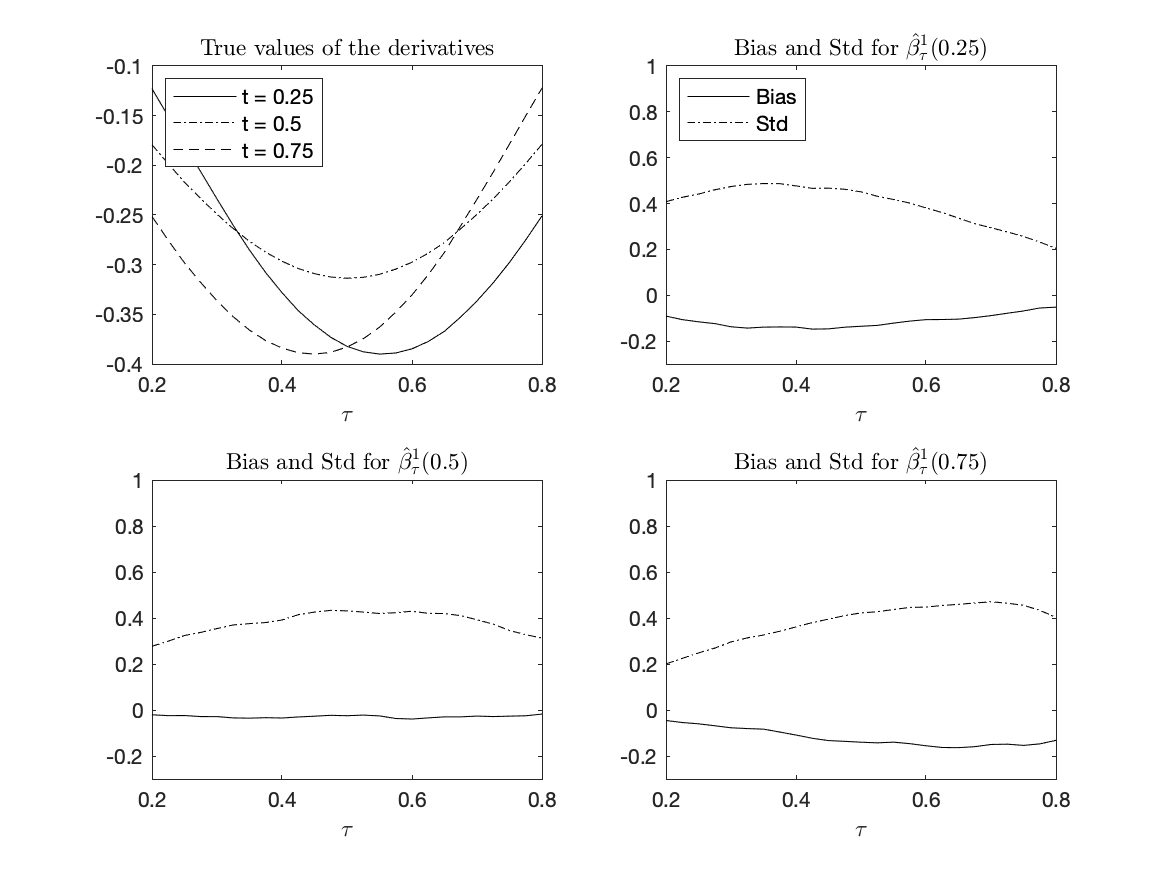}  
	\caption{DGP3, finite sample performance of $\hat{\protect\beta}^1_\protect\tau(t)
		$}
	\label{fig:beta1_3}
\end{figure}

\begin{figure}[H]
	\centering
	\includegraphics[scale = 0.565,angle=0]{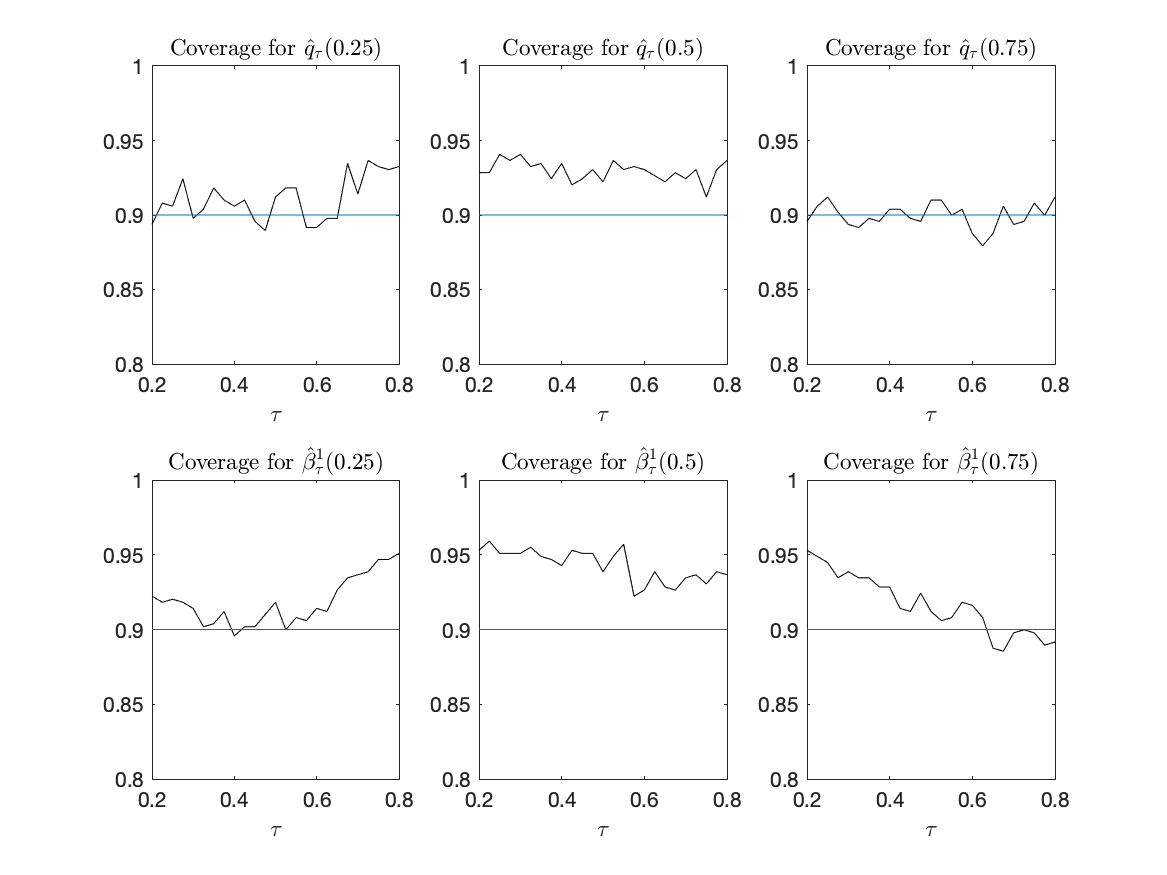}  
	\caption{DGP3, coverage probability}
	\label{fig:cov_3}
\end{figure}

The upper-left subplots in Figures \ref{fig:beta0_1}, \ref{fig:beta0_2}, \ref{fig:beta0_3} and \ref{fig:beta1_1}, \ref{fig:beta1_2}, \ref{fig:beta1_3}
report the true functions of $q_\tau(t)$ and $\partial_t q_\tau(t)$ for $%
t = 0.25,0.5,0.75$, $\tau \in (0.2,0.8)$ and DGP 1, 2, and 3, respectively. Both $q_\tau(t)$ and $%
\partial_t q_\tau(t)$ are heterogeneous across $\tau$ and $t$, which imposes
difficulties for estimation and inference. The rest of the subplots in
the above Figures show the estimation biases and standard errors. We observe that all the biases of our
estimators are of smaller order of magnitude than the standard error (std) and the root mean squared
error (rMSE), which indicates the doubly robust moments effectively remove
the selection bias induced by the Lasso method. The estimators of the
quantile functions are very accurate. The estimators of the quantile partial
derivatives are less so because they have slower convergence rates. Figures %
\ref{fig:cov_1}, \ref{fig:cov_2}, and \ref{fig:cov_3} show that the 90\% point-wise modified percentile bootstrap confidence intervals have reasonable
performance for both the quantile functions and their derivatives, across
all $\tau$ and $t$ values considered, with slight over-coverage for the
quantile derivative functions. The results of variable selections depend on the values
of $t$ and $(t,u)$ for conditional density estimation and
penalized local MLE, respectively, which are tedious to report, Thus, they are
omitted for brevity. Overall, 2 to 4 covariates are selected.

In Section \ref{sec:addsim} in the Appendix, we report the performance of oracle estimators for the three designs, in which oracle estimators are computed using the true conditional CDF and density functions. We also report the finite-sample performance of our mean potential outcome (i.e., $\mathbb{E}(Y(t))$) estimators, which is similar to that of the quantile effect estimates reported here. Last, we consider an extra design in which the approximate sparsity condition may be violated and show that our method breaks down. We use this design to illustrate the limitation of our method.

\section{Empirical Illustration}\label{sec:app} 

To investigate our proposed estimation and inference
procedures, we use the 1979 National Longitudinal Survey of Youth (NLSY79)
and consider the effect of father's income on son's income in the presence
of many control variables. Our analysis is based on \cite{BM11}. The data
consist of a nationally representative sample of individuals with age 14-22
years old as of 1979. We use only white and black males and discard the
individuals with missing values in the covariates we use. The resulting
sample size is 1,795, out of which 1,302 individuals are white and 493
individuals are black.

The treatment variable of interest is the logarithm of father's income, in
which father's income is computed as the average family income for 1978,
1979, and 1980. The outcome variable is the logarithm of son income, in
which son income is computed as the average family income for 1997, 1999,
2001 and 2003. We create control variables by interacting a list of
demographic variables with the cubic splines of the AFQT score and the years
of education.\footnote{The cubic splines for the AFQT score are constructed based on the normalized value by scaling the raw AFQT score into [0,1], where the knots are taken at the quantiles of the normalized AFQT score at $10\%, 20\%,\ldots, 90\%$. The cubic splines for the years of education are constructed in the same way. In this exercise, we do not interact the cubic splines for  the AFQT score and the years of education.} The list includes the age, the mother's education level, the
father's education level, the indicators of (i) living in urban areas at age
14, (ii) living in the south, (iii) speaking a foreign language at
childhood, and (iv) being born outside the U.S. 
We drop the variables whose variance is less than $10^{-4}$. 
The resulting numbers of control variables are 120 for whites and 145 for blacks.

\begin{figure}[H]
\centering
\includegraphics[scale = 0.65]{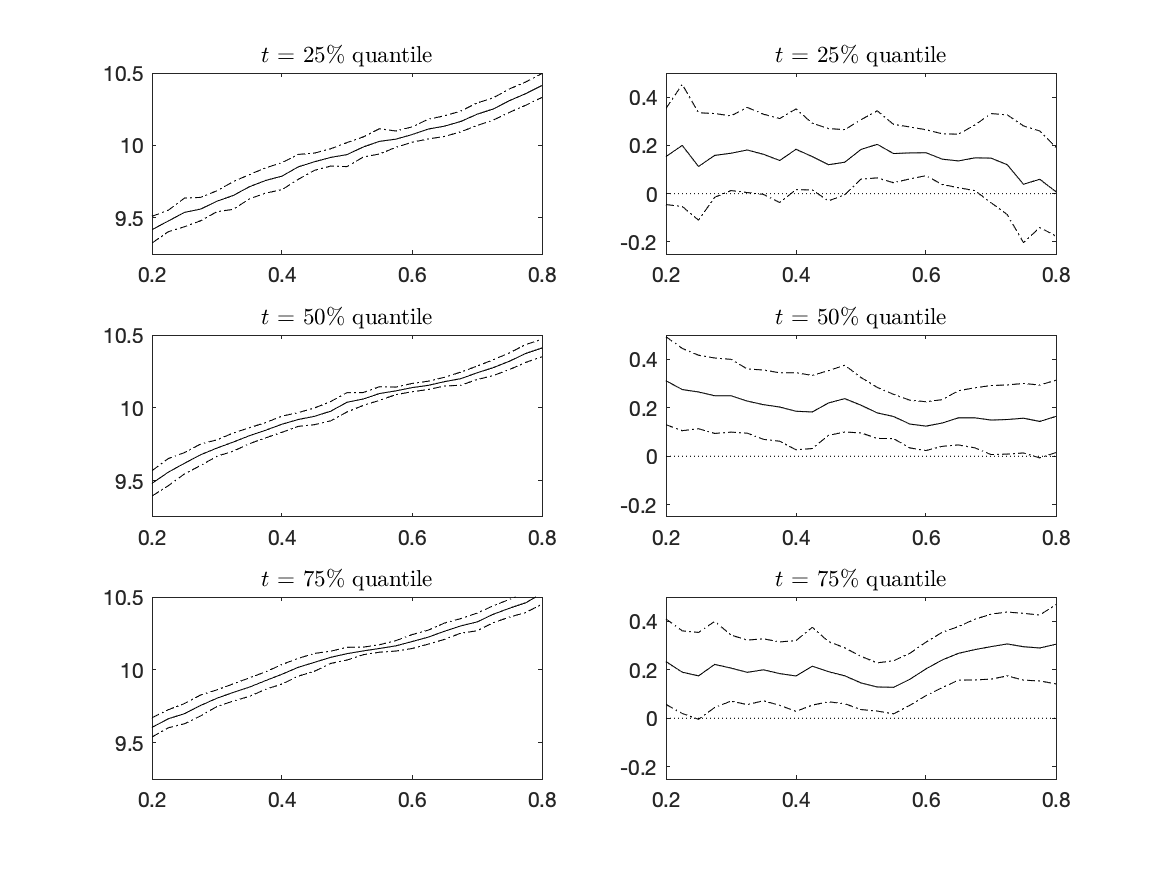}  
\caption{Whites. First column: the quantile index $\tau$ (X-axis), the son's log
income (Y-axis), the estimated unconditional quantile function at $\protect%
\tau$ (solid line), and its  (point-wise) $90\%$ confidence bands (dot-dash line).
Second column: the quantile index $\tau$ (X-axis), the intergenerational
elasticity (Y-axis), the estimated derivative of the unconditional quantile
function at $\protect\tau$ (solid line), and its (point-wise) $90\%$ confidence bands
(dot-dash line).}
\label{fig_white}
\end{figure}

We apply the proposed estimation and inference procedures for black and
white individuals separately. We use the same tuning parameter choices as in the
previous section.\footnote{In Section \ref{sec:add_emp} in the Appendix, we investigate the sensitivity of our estimation method with respect to the tuning parameters.}
As a result, our effective sample sizes are of orders of magnitude $nh_1\approx 462$ and $175$ for whites
and blacks, respectively. Figures \ref{fig_white} and \ref{fig_black} show the
estimated unconditional quantile functions and the estimated derivative, as
well as the point-wise 90\% confidence bands for $\tau \in [0.2,0.8]$ and $t$ taking values at the $25\%$, $50\%$, and $%
75\%$ quantiles of the empirical distribution of $T_i$. Under the context of intergenerational income mobility, the
unconditional quantile and its derivative represent the quantile of son's
potential log income indexed by father's log income and the
intergenerational elasticity, respectively. The unconditional quantile
functions have a slight upward trend and the estimated derivative is
positive in most parts of father's log income. The confidence bands for
the unconditional quantile functions are quite narrow for both black and
white individuals. For white individuals with the values of father's log income at the $50\%$ or $75\%$ quantile, we can reject the (locally) zero intergenerational elasticity for most of the values of $\tau\in [0.2,0.8]$.  
For the other cases, we cannot reject the (locally) zero intergenerational elasticity for almost all $\tau$'s. This is considered as the cost
of our fully nonparametric specification.

\begin{figure}[H]
\centering
\includegraphics[scale = 0.65]{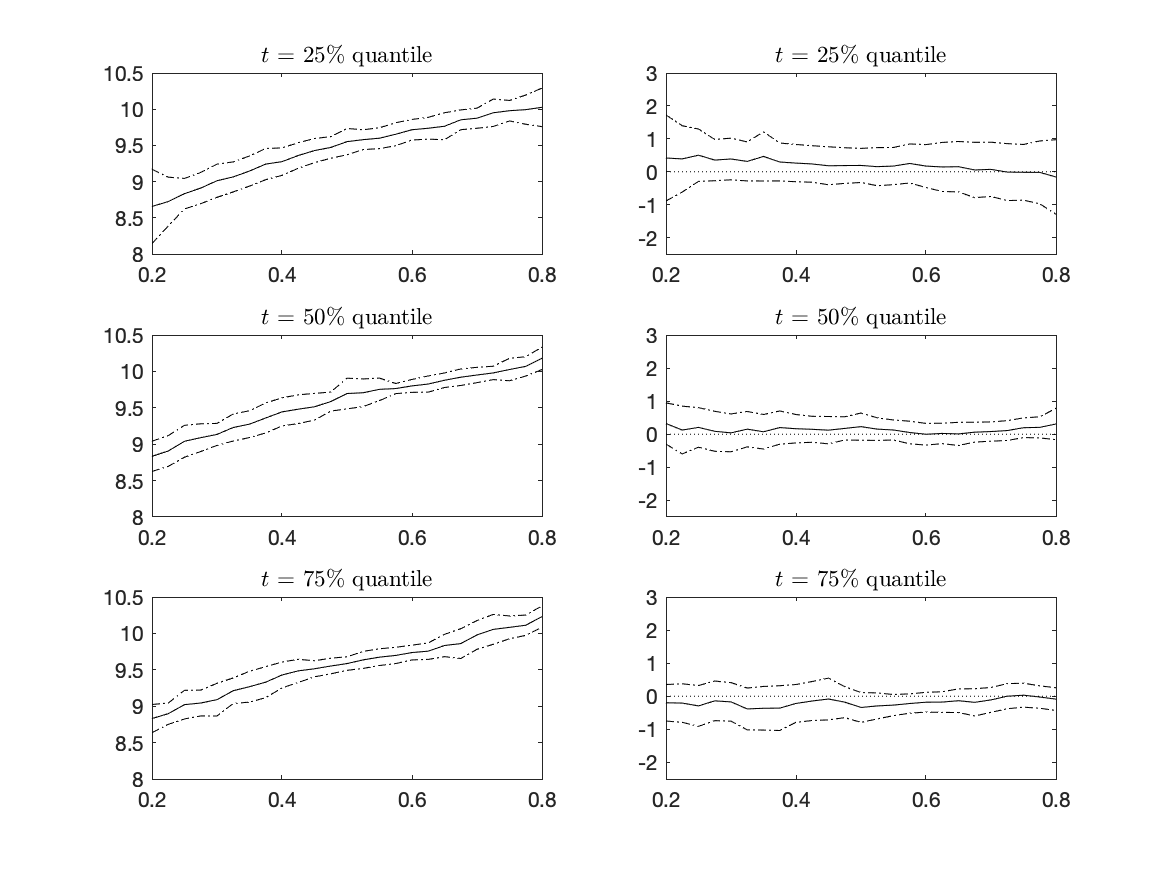}  
\caption{Blacks. First column: the quantile index $\tau$ (X-axis), the son's log
income (Y-axis), the estimated unconditional quantile function at $\protect%
\tau$ (solid line), and its  (point-wise) $90\%$ confidence bands (dot-dash line).
Second column: the quantile index $\tau$ (X-axis), the intergenerational
elasticity (Y-axis), the estimated derivative of the unconditional quantile
function at $\protect\tau$ (solid line), and its (point-wise) $90\%$ confidence bands
(dot-dash line).}
\label{fig_black}
\end{figure}

It is worthwhile to mention the variable selection in this application.  the years of education, the AFQT score,  the age,  the father's education
level, and the mother's education level are the leading
control variables selected.\footnote{More precisely, for whites, $dad\_educ*afqt$ and $mom\_educ$ are the two
most selected control variables for the density estimations. $age*educ$ and $age*afqt$ are the two most selected control variables for the penalized local MLE.
For blacks, $mom\_educ$ and 
$dad\_educ*educ$ are the two most selected control variables for the density estimations. $educ$ and $age*afqt$ are the two most selected control variables for the penalized local MLE.} 

\section{Conclusion}

\label{sec:concl} This paper studies non-separable models with a continuous
treatment and high-dimensional control variables. It extends the existing
results on the causal inference in non-separable models to the case with
both continuous treatment and high-dimensional covariates. It develops a
method based on localized $L_1$-penalization to select covariates at each
value of the continuous treatment. It then proposes a multi-stage estimation
and inference procedure for average, quantile, and marginal treatment
effects. The simulation and empirical exercises support the theoretical
findings in finite samples.

\newpage \appendix\label{sec:appendix}

\begin{center}
{\LARGE Appendix}
\end{center}

\section{Proof of the Main Results in the Paper}

Before proving the theorem, we first introduce some additional notation and
Assumption \ref{ass:reg_sasaki}, which is a restatement of %
\citet[][Assumptions 1 and 2]{S15} in our framework. Denote by $\dim_X$
(resp. $\dim_A$) the dimensionality of $X$ (resp. $A$). We define $\partial
V(y,t)=\{(x,a):\Gamma(t,x,a)=y\}$ and $\partial V(y,t)$ can be parametrized
as a mapping from a $(\dim_X+\dim_A-1)$-dimensional rectangle, denoted by $%
\Sigma$, to $\partial V(y,t)$. $H^{\dim_X+\dim_A-1}$ is the $%
(\dim_X+\dim_A-1)$-dimensional Hausdorff measure restricted from $\mathbb{R}%
^{\dim_X+\dim_A}$ to $(\partial V(y,t),\mathcal{B}(y,t))$, where $\mathcal{B}%
(y,t)$ is the set of the interactions between $\partial V(y,t)$ and a Borel
set in $\mathbb{R}^{\dim_X+\dim_A}$. $\partial v(y,\cdot;u)/\partial y$
(resp. $\partial v(\cdot,t;u)/\partial t$) is the velocity of $\partial
V(y,t)$ at $u$ with respect to $y$ (resp. $t$).

\begin{ass}
\begin{enumerate}

\item $\Gamma$ is continuously differentiable. 

\item $\|\nabla_{(x,a)}\Gamma(t,\cdot,\cdot)\|\ne 0$ on $\partial V(y,t)$. 

\item The conditional distribution of $(X,A)$ given $T$ is absolutely
continuous with respect to the Lebesgue measure, and $f_{(X,A)\mid T}$ is a
continuously differentiable function of $\mathcal{T}$ to $L^1(\mathbb{R}%
^{\dim_X+\dim_A})$. 

\item $\int_{\partial V(y,t)}f_{(X,A)\mid T}(x,a\mid t)dH^{\dim_X+\dim_A-1}(x,a)>0$. 

\item $t\mapsto\partial V(y,t)$ is a continuously differentiable function of 
$\Sigma\times\mathcal{T}$ to $\mathbb{R}^{\dim_X+\dim_A}$ for every $y$ and $%
y\mapsto\partial V(y,t)$ is a continuously differentiable function of $%
\Sigma\times\mathcal{Y}$ to $\mathbb{R}^{\dim_X+\dim_A}$ for every $t$. 

\item The mapping $\partial v(y,\cdot;\cdot)/\partial t$ is a continuously
differentiable function of $\mathcal{T}$ to $\mathbb{R}^{\dim_X+\dim_A}$ and 
$\partial v(\cdot,t;\cdot)/\partial y$ is a continuously differentiable
function of $\mathcal{Y}$ to $\mathbb{R}^{\dim_X+\dim_A}$. 

\item There is $p,q\geq 1$ with $\frac{1}{p}+\frac{1}{q}=1$ such that the
mapping $(x,a)\mapsto\|\nabla_{(x,a)}\Gamma(t,x,a)\|^{-1}$ is bounded in $%
L^p(\partial V(y,t),H^{\dim_X+\dim_A-1})$ and that the mapping $(x,a)\mapsto
f_{(X,A)}(x,a)$ is bounded in $L^q(\partial V(y,t),H^{\dim_X+\dim_A-1})$. 
\end{enumerate}

\label{ass:reg_sasaki}
\end{ass}

Assumption \ref{ass:reg_sasaki} is a combination of Assumptions 1 and 2 in 
\cite{S15}. We refer the readers to the paper for detailed
explanation. 

\begin{proof}[Proof of Theorem \protect\ref{thm:uasf}]
For the marginal distribution of $Y(t)$, we note that, by Assumption
\ref{ass:unconfoundedness}, $\mathbb{P}(Y(t)\leq u)=\mathbb{E}[\mathbb{E}(1\{Y(t)\leq
u\}|X)]=\mathbb{E}[\mathbb{E}(1\{Y(t)\leq
u\}|X,T=t)]=\mathbb{E}[\mathbb{E}(1\{Y\leq
u\}|X,T=t)].$ The first result follows as $\mathbb{E}(1\{Y\leq
u\}|X,T=t)$ is identified.

For the second result, consider a random variable $T^\ast$ which has
the same marginal distribution as $T$ and is independent of $(X,A)$.  Define
\begin{equation}
Y^\ast = \Gamma(T^\ast,X,A).  \label{exog_eq}
\end{equation}
Note that the (i) $(X,A)$ and $T^\ast$ are independent, and (ii) the $\tau$%
-th quantile of $Y^\ast$ given $T^\ast=t$ is $q_\tau(t)$ for all $t$,
because $\mathbb{P}(Y^\ast\leq q_\tau(t)\mid T^\ast=t)=\mathbb{P}(\Gamma(t,X,A)\leq
q_\tau(t))=\tau$.  Assumption \ref{ass:reg_sasaki} implies Assumptions 1 and
2 in \cite{S15} for $(Y^\ast,T^\ast,U^\ast)$ with $U^\ast=(X,A)$, and then
his Theorem 1 implies that the derivative of the $\tau$-th quantile of $%
Y^\ast$ given $T^\ast=t$ is equal to $\mathbb{E}_{\mu_{\tau,t}}[\partial_t%
\Gamma(t,X,A)]$. 
Therefore, $\partial_tq_{\tau}(t)=\mathbb{E}_{\mu_{\tau,t}}[\partial_t\Gamma(t,X,A)]$.
Note that Theorem 1 in \cite{S15} does not apply directly to $(Y,T,U^\ast)$, because our assumptions do not imply that $T$ and $U^\ast$ are independent. 
\end{proof}

Lemma \ref{lem:localRE} is the local version of the compatibility
condition, which is one of the key building blocks for Lemma \ref{lem:main}.
Then, Lemma \ref{lem:main} is used to prove Theorem \ref{thm:rate2}.
\medskip 

\begin{proof}[Proof of Lemma \protect\ref{lem:localRE}]
By Assumption \ref{ass:eigen}, we can work on the set 
\begin{equation*}
\biggl\{\{X_i\}_{i=1}^{n}:\sup_{|\delta |_{0}\leq s\ell _{n}}\frac{%
||b(X)^{\prime }\delta ||_{\mathbb{P}_n,2}}{||\delta ||_{2}}\leq \kappa ^{^{\prime
\prime }}<\infty\biggr\}.
\end{equation*}%
We use the same partition as in \cite{BRT09}. Let $\mathcal{S}_{0}=\mathcal{S%
}_{t,u}$ and $m\geq s$ be an integer which will be specified later.
Partition $\mathcal{S}_{t,u}^{c}$, the complement of $\mathcal{S}_{t,u}$, as 
$\sum_{l=1}^{L}\mathcal{S}_{l}$ such that $|\mathcal{S}_{l}|=m$ for $1\leq
l<L$, $|\mathcal{S}_{L}|\leq m$, where $\mathcal{S}_{l}$, for $l<L$,
contains the indexes corresponding to $m$ largest coordinates (in absolute
value) of $\delta $ outside $\cup _{j=0}^{l-1}\mathcal{S}_{j}$, and $\mathcal{%
S}_{L}$ collects the remaining indexes. Further denote $\delta _{j}=\delta _{%
\mathcal{S}_{j}}$ and $\delta _{01}=\delta _{\mathcal{S}_{0}\cup \mathcal{S}%
_{1}}$. Then 
\begin{equation}
||b(X)^{\prime }\delta K(\frac{T-t}{h_1})^{1/2}||_{\mathbb{P}_n,2}\geq ||b(X)^{\prime
}\delta _{01}K(\frac{T-t}{h_1})^{1/2}||_{\mathbb{P}_n,2}-\sum_{l=2}^{L}||b(X)^{\prime
}\delta _{l}K(\frac{T-t}{h_1})^{1/2}||_{\mathbb{P}_n,2}.  \label{eq:delta}
\end{equation}%

For the first term on the right hand side (r.h.s.) of (\ref{eq:delta}%
), we have 
\begin{equation}
\begin{split}
& ||b(X)^{\prime }\delta _{01}K(\frac{T-t}{h_1})^{1/2}||_{\mathbb{P}_n,2}^{2} \\
\geq & ||b(X)^{\prime }\delta _{01}K(\frac{T-t}{h_1})^{1/2}||_{P,2}^{2}-|(%
\mathbb{P}_{n}-\mathbb{P})(b(X)^{\prime }\delta _{01})^{2}K(\frac{T-t}{h_1})|
\\
\geq & \underline{C}h_1||b(X)^{\prime }\delta _{01}||_{P,2}^{2}-|(\mathbb{P}%
_{n}-\mathbb{P})(b(X)^{\prime }\delta _{01})^{2}K(\frac{T-t}{h_1})| \\
\geq & \underline{C}h_1||b(X)^{\prime }\delta _{01}||_{\mathbb{P}_n,2}^{2}-\underline{%
C}h|(\mathbb{P}_{n}-\mathbb{P})(b(X)^{\prime }\delta _{01})^{2}|-|(\mathbb{P}%
_{n}-\mathbb{P})(b(X)^{\prime }\delta _{01})^{2}K(\frac{T-t}{h_1})| \\
\geq & \underline{C}h_1||\delta _{01}||_{2}^{2}(\kappa ^{\prime })^{2}-%
\underline{C}h_1|(\mathbb{P}_{n}-\mathbb{P})(b(X)^{\prime }\delta
_{01})^{2}|-|(\mathbb{P}_{n}-\mathbb{P})(b(X)^{\prime }\delta _{01})^{2}K(%
\frac{T-t}{h_1})|
\end{split}
\label{A.3}
\end{equation}%
where the second inequality holds because 
\begin{equation*}
\mathbb{E}(b(X)^{\prime }\delta _{01})^{2}K(\frac{T-t}{h_1})=h_1\mathbb{E}%
(b(X)^{\prime }\delta _{01})^{2}\int f_{t+h_1v}(X)K(v)dv\geq \underline{C}h_1%
\mathbb{E}(b(X)^{\prime }\delta _{01})^{2}.
\end{equation*}%

We next bound the last term on the r.h.s. of (\ref{eq:delta}). The
second term can be bounded in the same manner. Let $\tilde{\delta}%
_{01}=\delta _{01}/||\delta _{01}||_{2}$. Then we have 
\begin{equation*}
|(\mathbb{P}_{n}-\mathbb{P})(b(X)^{\prime }\delta _{01})^{2}K(\frac{T-t}{h_1}%
)|=||\delta _{01}||_{2}^{2}|(\mathbb{P}_{n}-\mathbb{P})(b(X)^{\prime }\tilde{%
\delta}_{01})^{2}K(\frac{T-t}{h_1})|.
\end{equation*}%
Let $\{\eta _{i}\}_{i=1}^{n}$ be a sequence of Rademacher random variables
which is independent of the data and $\mathcal{F}=\{b(X)^{\prime
}\delta K(\frac{T-t}{h_1})^{1/2}:||\delta ||_{0}=m+s,||\delta ||_{2}=1,t\in \mathcal{%
T}\}$ with envelope $F=\overline{C}_{K}\zeta _{n}(m+s)^{1/2}$. Denote $\pi
_{1n}$ as $(\frac{\log (p\vee n)(s+m)^{2}\zeta _{n}^{2}}{nh_1})^{1/2}$ with $%
m=s\ell _{n}^{1/2}$. Then, 
\begin{align*}
& \mathbb{E}\sup_{||\tilde{\delta}_{01}||_{0}\leq m+s,||\tilde{\delta}%
_{01}||_{2}=1,t\in \mathcal{T}}|(\mathbb{P}_{n}-\mathbb{P})(b(X)^{\prime }%
\tilde{\delta}_{01})^{2}K(\frac{T-t}{h_1})| \\
\leq & 2\mathbb{E}\sup_{||\tilde{\delta}_{01}||_{0}\leq m+s,||\tilde{\delta}%
_{01}||_{2}=1,t\in \mathcal{T}}|\mathbb{P}_{n}\eta (b(X)^{\prime }\tilde{%
\delta}_{01})^{2}K(\frac{T-t}{h_1})| \\
\leq & 8\zeta _{n}\biggl(\sup_{||\tilde{\delta}_{01}||_{0}\leq m+s,||\tilde{%
\delta}_{01}||_{2}=1}||\tilde{\delta}_{01}||_{1}\biggr)\biggl(\mathbb{E}%
\sup_{f\in \mathcal{F}}|\mathbb{P}_{n}\eta f|\biggr) \\
\lesssim & 8\zeta _{n}(m+s)^{1/2}\biggl[\left(\frac{\log (p\vee n)(s+m)h_1}{n}%
\right)^{1/2}+\frac{\overline{C}_{K}\zeta _{n}(m+s)^{1/2}\log (p\vee n)(s+m)}{n}%
\biggr] \\
\lesssim & \biggl(\frac{\log (p\vee n)(s+m)^{2}h_1\zeta _{n}^{2}}{n}\biggr)%
^{1/2}=h_1\pi _{1n},
\end{align*}%
where the first inequality is by \citet[Lemma 2.3.1]{VW96}, the second
inequality is by \citet[Theorem 4.12]{LT13} and the remark thereafter, and the third one is by
applying Corollary 5.1 of \cite{CCK14} with $\sigma ^{2}=\sup_{f\in \mathcal{%
\ F}}\mathbb{E}f^{2}\lesssim h_1$ and, for some $A\geq e$, 
\begin{equation*}
\sup_{Q}N(\mathcal{F},e_{Q},\varepsilon ||F||_{Q,2})\leq \binom{p}{s+m}\left(\frac{A%
}{\varepsilon }\right)^{s+m}\lesssim \left(\frac{Ap}{\varepsilon }\right)^{s+m}.
\end{equation*}%

By Assumption \ref{ass:approx}, $\pi _{1n}\rightarrow 0.$ Then we
have, w.p.a.1., 
\begin{equation}
|(\mathbb{P}_{n}-\mathbb{P})(b(X)^{\prime }\delta _{01})^{2}K(\frac{T-t}{h_1}%
)|\leq 3h_1\underline{C}(\kappa ^{\prime })^{2}||\delta _{01}||_{2}^{2}/8.
\label{eq:delta011}
\end{equation}%

By the same token we can show that 
\begin{equation*}
\mathbb{E}\sup_{||\tilde{\delta}_{01}||_{0}\leq m+s,||\tilde{\delta}%
_{01}||_{2}=1,t\in \mathcal{T}}|(\mathbb{P}_{n}-\mathbb{P})(b(X)^{\prime }%
\tilde{\delta}_{01})^{2}|\lesssim \sqrt{h_1}\pi _{1n}\rightarrow 0.
\end{equation*}%
Therefore, we have, w.p.a.1., 
\begin{equation}
|(\mathbb{P}_{n}-\mathbb{P})(b(X)^{\prime }\delta _{01})^{2}|\leq 3(\kappa
^{\prime })^{2}||\delta _{01}||_{2}^{2}/8.  \label{eq:delta012}
\end{equation}%
Combining (\ref{A.3}), \eqref{eq:delta011}, and \eqref{eq:delta012} yields
that w.p.a.1., 
\begin{equation*}
||b(X)^{\prime }\delta _{01}K(\frac{T-t}{h_1})^{1/2}||_{\mathbb{P}_n,2}^{2}\geq
||\delta _{01}||_{2}^{2}h_1(\kappa ^{\prime })^{2}\underline{C}/4.
\end{equation*}%
Analogously, we can show that, w.p.a.1, 
\begin{equation*}
||b(X)^{\prime }\delta _{l}K(\frac{T-t}{h_1})^{1/2}||_{\mathbb{P}_n,2}^{2}\leq
4||\delta _{l}||_{2}^{2}\underline{C}^{-1}h_1(\kappa ^{\prime \prime })^{2}.
\end{equation*}%
Following \eqref{eq:delta}, we have, w.p.a.1, \ 
\begin{align*}
||b(X)^{\prime }\delta K(\frac{T-t}{h_1})^{1/2}||_{\mathbb{P}_n,2}\geq &
h_1^{1/2}||\delta _{01}||_{2}\kappa ^{\prime }\underline{C}^{1/2}/2-h_1^{1/2}%
\sum_{l=2}^{L}2||\delta _{l}||_{2}\kappa ^{^{\prime \prime }}\underline{C}%
^{-1/2} \\
\geq & h_1^{1/2}||\delta _{01}||_{2}\kappa ^{\prime }\underline{C}%
^{1/2}/2-h_1^{1/2}\sum_{l=2}^{L}2\kappa ^{^{\prime \prime }}\underline{C}%
^{-1/2}(||\delta _{l-1}||_{1}||\delta _{l}||_{1})^{1/2}/\sqrt{m} \\
\geq & h_1^{1/2}||\delta _{01}||_{2}\kappa ^{\prime }\underline{C}%
^{1/2}/2-2h_1^{1/2}\kappa ^{^{\prime \prime }}\underline{C}^{-1/2}||\delta
_{T^{c}}||_{1}/\sqrt{m} \\
\geq & h_1^{1/2}||\delta _{01}||_{2}\kappa ^{\prime }\underline{C}%
^{1/2}/2-2h_1^{1/2}\kappa ^{^{\prime \prime }}\underline{C}^{-1/2}c^{1/2}||%
\delta _{0}||_{1}/\sqrt{m} \\
\geq & h_1^{1/2}||\delta _{01}||_{2}\kappa ^{\prime }\underline{C}%
^{1/2}/2-2h_1^{1/2}\kappa ^{^{\prime \prime }}\underline{C}^{-1/2}c^{1/2}||%
\delta _{0}||_{2}\sqrt{s}/\sqrt{m} \\
\geq & h_1^{1/2}||\delta _{0}||_{2}\biggl[\kappa ^{\prime }\underline{C}%
^{1/2}/2-2\kappa ^{^{\prime \prime }}\underline{C}^{-1/2}c^{1/2}\sqrt{s}/%
\sqrt{m}\biggr],
\end{align*}%
where the second inequality holds because, by construction, $||\delta
_{l}||_{2}^{2}\leq ||\delta _{l-1}||_{1}||\delta _{l}||_{1}/\sqrt{m}.$ Since 
$m=s\ell _{n}^{1/2}$, $s/m=\ell _{n}^{-1/2}\rightarrow 0$, and thus, for $n$
large enough, the constant inside the brackets is greater than $\kappa
^{\prime }\underline{C}^{1/2}/4$ which is independent of $(t,u,n)$.
Therefore, we can conclude that, for $n$ large enough, 
\begin{equation*}
\inf_{(t,u)\in \mathcal{T}\mathcal{U}}\inf_{\delta \in \Delta _{2\tilde{c}%
,t,u}}\frac{||b(X)^{\prime }\delta K(\frac{T-t}{h_1})^{1/2}||_{\mathbb{P}_n,2}}{%
||\delta _{\mathcal{S}_{t,u}}||_{2}\sqrt{h_1}}\geq \kappa ^{\prime }\underline{%
C}^{1/2}/4:= \underline{\kappa }.
\end{equation*}%
This completes the proof of the lemma. \medskip {\ }
\end{proof}

We aim to prove the results with regard to $\widehat{\phi }%
_{t,u}(X) $ and $\hat{\theta}_{t,u}$ in Theorem \ref{thm:rate2}. The
derivations for the results regarding $\widetilde{\phi }_{t,u}(X)$ and $%
\widetilde{\theta }_{t,u}$ are exactly the same. We do not need to deal with
the nonlinear logistic link function when deriving the results regarding $%
\widehat{\nu }_{t}(X)$, $\widetilde{\nu }_{t}(X)$, $\hat{\gamma}_{t}$, and $%
\tilde{\gamma}_{t}$. Therefore, the corresponding results can be shown by
following the same proving strategy as below and treating $\omega _{t,u}$
defined below as $1$. The proofs for results regarding $\widehat{\nu }%
_{t}(X) $, $\widetilde{\nu }_{t}(X)$, $\hat{\gamma}_{t}$, and $\tilde{\gamma}%
_{t}$ are omitted for brevity.

\medskip 
Let $\tilde{r}_{t,u}^{\phi }=\Lambda ^{-1}(\mathbb{E}%
(Y_{u}|X,T=t))-b(X)^{\prime }\theta _{t,u}$, $\delta _{t,u}=\hat{\theta}%
_{t,u}-\theta _{t,u}$, $\hat{s}_{t,u}=||\hat{\theta}_{t,u}||_{0}$, $\omega
_{t,u}=\mathbb{E}(Y_{u}(t)|X)(1-\mathbb{E}(Y_{u}(t)|X))$, and $\widehat{%
\mathcal{S}}_{t,u}$ be the support of $\widehat{\theta }_{t,u}$. We need the
following four lemmas, whose proofs are relegated to the online supplement. 

\begin{lem}
If Assumptions \ref{ass:unconfoundedness}--\ref{ass:eigen} hold,
then 
\begin{equation*}
\sup_{(t,u)\in \mathcal{T}\mathcal{U}}||\omega _{t,u}^{1/2}b(X)^{\prime
}\delta _{t,u}K(\frac{T-t}{h_1})^{1/2}||_{\mathbb{P}_n,2}=O_{p}(\ell_n(\log (p\vee
n)s)^{1/2}n^{-1/2})
\end{equation*}%
and 
\begin{equation*}
\sup_{(t,u)\in \mathcal{T}\mathcal{U}}||\delta _{t,u}||_{1}=O_{p}(\ell_n(\log
(p\vee n)s^{2})^{1/2}(nh_1)^{-1/2}).
\end{equation*}%
\label{lem:main} 
\end{lem}

\begin{lem}
Suppose Assumptions \ref{ass:unconfoundedness}--\ref{ass:eigen}
hold. Let $\xi _{t,u}=Y_{u}-\phi _{t,u}(X).$ Then 
\begin{equation*}
\sup_{(t,u)\in \mathcal{T}\mathcal{U}}\biggl|\biggl|\widehat{\Psi }%
_{t,u}^{-1}\mathbb{P}_{n}\biggl[\xi _{t,u}K(\frac{T-t}{h_1})b(X)\biggr]\biggr|%
\biggr|_{\infty }=O_{p}((\log (p\vee n)h_1/n)^{1/2}).
\end{equation*}%
\label{lem:E2} 
\end{lem}

\begin{lem}
If the assumptions in Theorem \ref{thm:rate2} hold, then there
exists a constant $C_{\psi }\in (0,1)$ such that w.p.a.1, 
\begin{equation}
C_{\psi }/2\leq \inf_{(t,u)\in \mathcal{TU}, j=1,\cdots,p}l_{t,u,j}^0 \leq \sup_{(t,u)\in \mathcal{TU}, j=1,\cdots,p}l_{t,u,j}^0\leq 2/C_{\psi }.  \label{eq:psi20}
\end{equation}%
For any $k=0,1,\cdots ,K$ and $\widehat{\Psi }_{t,u}^{k}$ defined in
Algorithm \ref{alg:psi1}, there exists a constant $C_{k}\in (0,1)$ such
that, w.p.a.1, 
\begin{equation}
C_{k}/2\leq \inf_{(t,u)\in \mathcal{TU}, j=1,\cdots,p}l_{t,u,j}^k \leq \sup_{(t,u)\in \mathcal{TU}, j=1,\cdots,p}l_{t,u,j}^k \leq 2C_{k}.  \label{eq:psi2k}
\end{equation}%
In addition, for any $k=0,1,\cdots ,K$ and $\widehat{\Psi }_{t,u}^{k}$
defined in Algorithm \ref{alg:psi1}, there exist constants $l<1<L$
independent of $n$, $(t,u)$, and $k$ such that, element-wise and w.p.a.1, 
\begin{equation}
l\widehat{\Psi }_{t,u,0}\leq \widehat{\Psi }_{t,u}^{k}\leq L\widehat{\Psi }%
_{t,u,0}.  \label{eq:E32}
\end{equation}%
\label{lem:E32}
\end{lem}

\begin{lem}
If the assumptions in Theorem \ref{thm:rate2} hold, then w.p.a.1, 
\begin{equation*}
\sup_{t\in \mathcal{T},||\delta ||_{2}=1,||\delta ||_{0}\leq s\ell
_{n}}||b(X)^{\prime }\delta K(\frac{T-t}{h_1})^{1/2}||_{\mathbb{P}_n,2}h_1^{-1/2}\leq 2%
\underline{C}^{-1/2}\kappa ^{^{\prime \prime }}.
\end{equation*}%
\label{lem:localRE2}
\end{lem}

\begin{proof}[Proof of Theorem \protect\ref{thm:rate2}]
By the mean value theorem, there exist $\underline{\theta }_{t,u}\in
(\theta _{t,u},\hat{\theta}_{t,u})$ and $\overline{r}_{t,u}^{\phi }\in (0,%
\tilde{r}_{t,u}^{\phi })$ such that 
\begin{equation*}
|\phi _{t,u}(X)-\widehat{\phi }_{t,u}(X)|\leq \Lambda (b(X)^{\prime }%
\underline{\theta }_{t,u}+\overline{r}_{t,u}^{\phi })(1-\Lambda
(b(X)^{\prime }\underline{\theta }_{t,u}+\overline{r}_{t,u}^{\phi
}))(b(X)^{\prime }\delta _{t,u}+\tilde{r}_{t,u}^{\phi }),
\end{equation*}%
where $\delta_{t,u} =\hat{\theta}_{t,u} -\theta %
_{t,u}.$By the proof of Lemma \ref{lem:main}, we have, w.p.a.1, 
\begin{equation*}
|\tilde{r}_{t,u}^{\phi }|\leq \lbrack \underline{C}/2(1-\underline{C}%
/2)]^{-1}|r_{t,u}^{\phi }|.
\end{equation*}%
Therefore, by Lemma \ref{lem:main} and Assumptions \ref{ass:eigen} and \ref{ass:rate2}, we have%
\begin{equation*}
\begin{split}
& \sup_{(t,u)\in \mathcal{T}\mathcal{U}}|b(X)^{\prime }\underline{\theta }%
_{t,u}+\overline{r}_{t,u}^{\phi }-b(X)^{\prime }\theta _{t,u}-\tilde{r}%
_{t,u}^{\phi }| \\
\lesssim & \sup_{(t,u)\in \mathcal{T}\mathcal{U}}|b(X)^{\prime }\delta
_{t,u}|+\sup_{(t,u)\in \mathcal{T}\mathcal{U}}|r_{t,u}^{\phi }| \\
\lesssim & \zeta _{n}\sup_{(t,u)\in \mathcal{T}\mathcal{U}}||\delta
_{t,u}||_{1}+O((\log (p\vee n)s^{2}\zeta _{n}^{2}/(nh_1))^{-1/2})=o_{p}(1),
\end{split}%
\end{equation*}%
where the last equality is because $\sup_{(t,u)\in \mathcal{T}\mathcal{U}%
}||\delta _{t,u}||_{1}=O_{p}((\log (p\vee n)s^{2})^{1/2}(nh_1)^{-1/2})$ by
Lemma \ref{lem:main} and $\log (p\vee n)s^{2}\zeta _{n}^{2}/(nh_1)\rightarrow
0 $ by Assumption \ref{ass:rate2}. In addition, under Assumption \ref{ass:ker}.4 we have 
\begin{equation*}
\Lambda (b(X)^{\prime }\theta _{t,u}+\tilde{r}_{t,u}^{\phi })=\mathbb{E}%
(Y_{u}|X,T=t)\in \lbrack \underline{C},1-\underline{C}].
\end{equation*}%
Hence, there exist some positive constants $c$ and $c^{\prime }$ only
depending on $\underline{C}$ such that, w.p.a.1, 
\begin{equation*}
\Lambda (b(X)^{\prime }\underline{\theta }_{t,u}+\overline{r}_{t,u}^{\phi
})(1-\Lambda (b(X)^{\prime }\underline{\theta }_{t,u}+\overline{r}%
_{t,u}^{\phi }))\leq c
\end{equation*}%
and uniformly over $(t,u)\in \mathcal{T}\mathcal{U}$, 
\begin{equation}
|\phi _{t,u}(X)-\widehat{\phi }_{t,u}(X)|\leq c(b(X)^{\prime }\delta _{t,u}+%
\tilde{r}_{t,u}^{\phi })\leq c^{\prime }(b(X)^{\prime }\delta
_{t,u}+r_{t,u}^{\phi }).  \label{eq:phibound}
\end{equation}%
By Assumptions \ref{ass:ker}.3, \ref{ass:ker}.4, Lemma \ref{lem:main}, and the fact that $\omega
_{t,u} $ is bounded and bounded away from zero uniformly over $\mathcal{TU}$%
, we have, w.p.a.1, 
\begin{align}
\label{eq:311}
& \sup_{(t,u)\in \mathcal{T}\mathcal{U}}||(\phi _{t,u}(X)-\widehat{\phi }%
_{t,u}(X))K(\frac{T-t}{h_1})^{1/2}||_{\mathbb{P}_n,2} \notag \\
\leq & \sup_{(t,u)\in \mathcal{T}\mathcal{U}}c\biggl[||b(X)^{\prime }\delta
_{t,u}K(\frac{T-t}{h_1})^{1/2}||_{\mathbb{P}_n,2}+||r_{t,u}^{\phi }K(\frac{T-t}{h_1}%
)^{1/2}||_{\mathbb{P}_n,2}\biggr] \notag \\
=& O_{p}(\ell_n(\log (p\vee n)s/n)^{1/2})
\end{align}%
and 
\begin{align}
\label{eq:312}
\sup_{(t,u)\in \mathcal{T}\mathcal{U}}||\phi _{t,u}(X)-\widehat{\phi }%
_{t,u}(X)||_{\mathbb{P}, \infty }\lesssim & \zeta _{n}\sup_{(t,u)\in \mathcal{T}%
\mathcal{U}}||\delta _{t,u}||_{1}+O((\log (p\vee n)s^{2}\zeta
_{n}^{2}/(nh_1))^{1/2}) \notag \\
=& O_{p}(\ell_n(\log (p\vee n)s^{2}\zeta _{n}^{2}/(nh_1))^{1/2}).
\end{align}%

Next, recall that $\lambda =\ell _{n}(\log (p\vee n)nh)^{1/2}$. By
the first order conditions (FOC), for any $j\in \widehat{\mathcal{S}}_{t,u}$%
, we have 
\begin{equation*}
\biggl|\mathbb{P}_{n}\biggl[(Y_{u}-\Lambda (b(X)^{\prime }\hat{\theta}%
_{t,u}))b_{j}(X)K(\frac{T-t}{h_1})\biggr]\biggr|=\widehat{\Psi }_{t,u,jj}\frac{%
\lambda }{n}.
\end{equation*}%
Denote $\xi _{t,u}=Y_{u}-\phi _{t,u}(X)$. By Lemmas \ref{lem:main}, \ref%
{lem:E2} and \ref{lem:E32}, for any $\varepsilon >0$, with probability
greater than $1-\varepsilon $, there exist positive constants $C_{\lambda }$
and $C$, which only depend on $\varepsilon $ and are independent of $(t,u,n)$%
, such that 
\begin{align*}
\frac{\lambda \hat{s}_{t,u}^{1/2}}{n}=& \biggl|\biggl|\widehat{\Psi }%
_{t,u}^{-1}\biggl\{\mathbb{P}_{n}\biggl[(Y_{u}-\Lambda (b(X)^{\prime }\hat{%
\theta}_{t,u}))b(X)K(\frac{T-t}{h_1})\biggr]\biggr\}_{\widehat{\mathcal{S}}%
_{t,u}}\biggr|\biggr|_{2} \\
\leq & \sup_{||\theta ||_{0}\leq \hat{s}_{t,u},||\theta ||_{2}=1}||\theta
||_{1}\sup_{(t,u)\in \mathcal{TU}}||\widehat{\Psi }_{t,u}^{-1}(\mathbb{P}%
_{n}\xi _{t,u}b(X)K(\frac{T-t}{h_1}))||_{\infty } \\
& +\frac{||\widehat{\Psi }_{t,u,0}^{-1}||_{\infty }}{l}\sup_{||\theta
||_{0}\leq \hat{s}_{t,u},||\theta ||_{2}=1}\biggl|\biggl\{\mathbb{P}_{n}%
\biggl(\Lambda (b(X)^{\prime }\hat{\theta}_{t,u})-\Lambda (b(X)^{\prime
}\theta _{t,u})-r_{t,u}^{\phi }\biggr)b(X)^{\prime }\theta K(\frac{T-t}{h_1})%
\biggr\}\biggr| \\
\leq & \frac{C_{\lambda }\lambda \hat{s}_{t,u}^{1/2}}{n\ell _{n}}+\frac{%
c^{\prime }||\widehat{\Psi }_{t,u,0}^{-1}||_{\infty }}{l}||(b(X)^{\prime
}\delta _{t,u}+r_{t,u}^{\phi })K(\frac{T-t}{h_1})^{1/2}||_{\mathbb{P}_n,2} \\
& \times \sup_{||\theta ||_{0}\leq \hat{s}_{t,u},||\theta
||_{2}=1}||b(X)^{\prime }\theta K(\frac{T-t}{h_1})^{1/2}||_{\mathbb{P}_n,2} \\
\leq & \frac{\lambda \hat{s}_{t,u}^{1/2}}{2n}+C(\log (p\vee n)s/n)^{1/2}\phi
_{max}^{1/2}(\hat{s}_{t,u}) \\
\leq & \frac{\lambda \hat{s}_{t,u}^{1/2}}{2n}+\frac{C\lambda s^{1/2}}{%
nh_1^{1/2}}\phi _{max}^{1/2}(\hat{s}_{t,u})
\end{align*}%
where $\phi _{max}(s)=\sup_{||\theta ||_{0}\leq s,||\theta
||_{2}=1}||b(X)^{\prime }\theta K(\frac{T-t}{h_1})^{1/2}||_{\mathbb{P}_n,2}^{2}$ and
$r_{t,u}^{\phi }=r_{t,u}^{\phi }(X)$ $.$ This implies that
there exists a constant $C$ only depending on $\varepsilon $, such that,
with probability greater than $1-\varepsilon $, 
\begin{equation}
\hat{s}_{t,u}\leq Cs\phi _{max}(\hat{s}_{t,u})/h_1.  \label{eq:sphi}
\end{equation}%
Let $\mathcal{M}=\{m\in \mathbb{Z}:m>2Cs\phi _{max}(m)/h_1\}$. We claim that,
for any $m\in \mathcal{M}$, $\hat{s}_{t,u}\leq m$. Suppose not and there
exists $m_0 \in \mathcal{M}$ such that $m_0 < \hat{s}_{t,u}$. Then, 
\begin{align*}
\hat{s}_{t,u} \leq Cs\phi _{max}(\frac{\hat{s}_{t,u}}{m_0} m_0)/h_1 \leq
\lceil \frac{\hat{s}_{t,u}}{m_0} \rceil Cs\phi _{max}(m_0)/h_1 \leq \frac{\hat{%
s}_{t,u}}{m_0} \biggl[2Cs\phi _{max}(m_0)/h_1 \biggr] < \hat{s}_{t,u},
\end{align*}
where the second inequality holds because of \citet[Lemma 23]{BC11}, the
third inequality holds because $\lceil a \rceil \leq 2a$ for any $a>1$, and
the last inequality holds because $m_0 \in \mathcal{M}$. Therefore we reach
a contradiction. In addition, by Lemma \ref{lem:localRE2}, we can choose $%
C_{s}>4C\underline{C}^{-1}(\kappa ^{\prime \prime })^{2}$, which is
independent of $(t,u,n)$, such that

\begin{equation}
\label{eq:313}
2Cs\phi_{max}(C_s s)/h_1\leq 4C \underline{C}^{-1} (\kappa^{\prime \prime })^2
s< C_s s.
\end{equation}
This implies $C_s s \in \mathcal{M}$ and thus with probability greater than $%
1-\varepsilon$, $\hat{s}_{t,u} \leq C_s s$. This result holds uniformly over 
$(t,u) \in \mathcal{T} \mathcal{U}$. 

Last, we show that
\begin{equation*}
\sup_{(t,u)\in \mathcal{T}\mathcal{U}}||(\widehat{\phi }_{t,u}(X)-\phi
_{t,u}(X))||_{\mathbb{P}_n,2}=O_{p}(\ell _{n}(\log (p\vee
n)s)^{1/2}(nh_1)^{-1/2}).
\end{equation*}
Let $\varepsilon _{n}=(\log (p\vee n)s/(nh_1))^{1/2}$, $%
\delta _{n}=(\log (p\vee n)s^{2}\zeta _{n}^{2}/(nh_1))^{1/2}$, and
\begin{equation*}
\mathcal{J}_{t,u}=%
\begin{Bmatrix}
\Lambda (b(x)^{\prime }\theta ): & ||\theta ||_{0}\leq Ms,||(\Lambda
(b(X)^{\prime }\theta )-\phi _{t,u}(X))K(\frac{T-t}{h})^{1/2}||_{\mathbb{P}_n,2}%
\leq M\ell _{n}\varepsilon _{n}h_1^{1/2}, \\ 
& ||\Lambda (b(X)^{\prime }\theta )-\phi _{t,u}(X)||_{\mathbb{P}, \infty }\leq M\ell _{n}\delta
_{n}.%
\end{Bmatrix}%
\end{equation*}%

By \eqref{eq:311}, \eqref{eq:312}, and \eqref{eq:313}, for any $\varepsilon >0$,
there exists a constant $M$ such that, with probability greater than $%
1-\varepsilon $, $\widehat{\phi }_{t,u}(\cdot )\in \mathcal{J}_{t,u}$
uniformly in $(t,u)\in \mathcal{T}\mathcal{U}$.
Therefore, with probability greater than $1-\varepsilon$, 

\begin{align*}
& \left|\mathbb{P}_n(\hat{\phi}_{t,u}(X) - \phi_{t,u}(X))^2\biggl[K(\frac{T - t}{h_1}) - \mathbb{E}(K(\frac{T - t}{h_1})|X)\biggr]\right| \\
\leq & \sup_{(t,u)\in \mathcal{T}\mathcal{U}}\sup_{J \in \mathcal{J}_{t,u}}\left|\mathbb{P}_n(J(X) - \phi_{t,u}(X))^2\biggl[K(\frac{T - t}{h_1}) - \mathbb{E}(K(\frac{T_i - t}{h_1})|X)\biggr]\right| = ||\mathbb{P}_n - \mathbb{P}||_{\mathcal{F}}, 
\end{align*}
where $\mathcal{F} = \left\{(J(X) - \phi_{t,u}(X))^2\biggl[K(\frac{T - t}{h_1}) - \mathbb{E}(K(\frac{T_i - t}{h_1})|X)\biggr]: J \in \mathcal{J}_{t,u}, (t,u) \in \mathcal{TU}\right\}$ with bounded envelope. Note that, 
\begin{align*}
\sigma^2 \equiv \sup_{f \in \mathcal{F}}\mathbb{E}f^2 \leq  & \sup_{(t,u)\in \mathcal{T}\mathcal{U}}\sup_{J \in \mathcal{J}_{t,u}}\mathbb{E}(J(X) - \phi_{t,u}(X))^4K^2(\frac{T-t}{h_1}) \\
\lesssim & \ell _{n}^2\delta_n^2 \sup_{(t,u)\in \mathcal{T}\mathcal{U}}\sup_{J \in \mathcal{J}_{t,u}}\mathbb{E}(J(X) - \phi_{t,u}(X))^2K(\frac{T-t}{h_1}) \\
= & \ell _{n}^2\delta_n^2\sup_{(t,u)\in \mathcal{T}\mathcal{U}}\sup_{J \in \mathcal{J}_{t,u}}\mathbb{E} ||(J(X) - \phi_{t,u}(X))K^{1/2}(\frac{T-t}{h_1})||_{\mathbb{P}_n,2}^2\\
\lesssim & \ell _{n}^4\delta_n^2 \eps_n^2 h_1,
\end{align*}

In addition, we note that $\mathcal{F}$ is nested by 
\begin{equation*}
\overline{F}=%
\left\{\Lambda(b(X)'\theta) - \phi_{t,u}(X))^2\biggl[K(\frac{T - t}{h_1}) - \mathbb{E}(K(\frac{T_i - t}{h_1})|X)\biggr],  ||\theta ||_{0}\leq Ms,(t,u)\in \mathcal{T}\mathcal{U}  \right\},
\end{equation*}%
where 
\begin{equation*}
\sup_{Q}\log N(\mathcal{F},e_{Q},\varepsilon ||\overline{F}||_{Q,2})\lesssim s\log
(p\vee n)+s\log (\frac{1}{\varepsilon })\vee 0.
\end{equation*}%
Therefore, by \citet[Corollary 5.1]{CCK14}, we have
\begin{equation}
\label{eq:314}
\mathbb{E}||\mathbb{P}_{n}-\mathbb{P}||_{\mathcal{F}}\lesssim \ell _{n}^2\varepsilon
_{n}h_1^{1/2}\delta_n s^{1/2}\log ^{1/2}(p\vee n)n^{-1/2}+s \log(p \vee n) n^{-1} = o_p(\ell _{n}^2 \varepsilon _{n}^{2}h_1).
\end{equation}%
Therefore, 
\begin{align*}
& h_1 \mathbb{P}_n(\hat{\phi}_{t,u}(X) - \phi_{t,u}(X))^2 \\
\lesssim & \mathbb{P}_n h_1 \int f_{t+h_1v}(X)K(v)dv(\hat{\phi}_{t,u}(X) - \phi_{t,u}(X))^2 \\
= & \mathbb{P}_n (\hat{\phi}_{t,u}(X) - \phi_{t,u}(X))^2\mathbb{E}\left(K\left(\frac{T-t}{h_1}\right)|X\right) \\
\leq & \mathbb{P}_n (\hat{\phi}_{t,u}(X) - \phi_{t,u}(X))^2K(\frac{T-t}{h_1}) + \left|\mathbb{P}_n(\hat{\phi}_{t,u}(X) - \phi_{t,u}(X))^2\biggl[K(\frac{T - t}{h_1}) - \mathbb{E}(K(\frac{T - t}{h_1})|X)\biggr]\right| \\
= & O_p(\ell _{n}^2\eps_n^2 h_1),
\end{align*}
where the last equality holds due to \eqref{eq:311} and \eqref{eq:314}. Canceling the $h_1$'s on both sides, we obtain the desired the result. 
\end{proof}

\begin{proof}[Proof of Theorem \protect\ref{thm:rate1}]
By \citet[Theorem 6.2]{BCFH13}, we have 
\begin{align*}
\sup_t ||F_t(X) - \Lambda(b(X)'\hat{\beta}_t)||_{\mathbb{P}_n,2} \lesssim_p \sqrt{\frac{s\log(p \vee n)}{n}}
\end{align*}
and 
\begin{align*}
\sup_t ||F_t(X) - \Lambda(b(X)'\hat{\beta}_t)||_{\mathbb{P}, \infty} \lesssim_p \sqrt{\frac{\zeta_n^2 s^2\log(p \vee n)}{n}}
\end{align*}

Then, we have 
\begin{align*}
& ||\hat{f}_t(X) - f_t(X)||_{\mathbb{P}_n,2} \\
\leq & \left\Vert \frac{\Lambda(b(X)'\hat{\beta}_{t+h_1}) - F_{t+h_1}(X)}{2h_1}\right\Vert_{\mathbb{P}_n,2} + \left\Vert \frac{\Lambda(b(X)'\hat{\beta}_{t-h_1}) - F_{t-h_1}(X)}{2h_1}\right\Vert_{\mathbb{P}_n,2} \\
& + \left\Vert \frac{ F_{t+h_1}(X) - F_{t-h_1}(X)}{2h_1} -f_t(X)\right\Vert_{\mathbb{P}_n,2} \\
\lesssim_p & \frac{1}{h_1}\sqrt{\frac{s \log(p \vee n)}{n}} + h_1^2
\end{align*}
and similarly, 
\begin{align*}
||\hat{f}_t(X) - f_t(X)||_{\mathbb{P}, \infty} \lesssim_p \frac{1}{h_1}\sqrt{\frac{\zeta_n^2 s^2\log(p \vee n)}{n}} + h_1^2.
\end{align*} 
\end{proof}

\begin{proof}[Proof of Theorem \protect\ref{thm:2nd}]
Let $\hat{\alpha}^{\dagger}(t,u)=\mathbb{P}_{n}\eta \Pi _{t,u}(W_{u},%
	\widehat{\phi }_{t,u},\hat{f}_{t})$ where either $\eta =1$ or $\eta $ is a
	random variable that has sub-exponential tails with unit mean and variance.
	When $\eta =1$, $\hat{\alpha}^{\dagger }(t,u)=\hat{\alpha}(t,u)$, which is
	our original estimator. When $\eta $ is random, for $\bar{\eta} =
	\sum_{i=1}^n \eta_i/n$, 
	\begin{equation*}
	\hat{\alpha}^{b}(t,u) = \hat{\alpha}^{\dagger}(t,u)/\bar{\eta}
	\end{equation*}
	is the bootstrap estimator. In the following, we establish the linear
	expansion of $\hat{\alpha}^{\dagger}(t,u)$. 

Recall $\varepsilon _{n}=(\log (p\vee n)s/(nh_1))^{1/2}$ and $%
	\delta _{n}=(\log (p\vee n)s^{2}\zeta _{n}^{2}/(nh_1))^{1/2}.$ 	By Theorem \ref{thm:rate2} and \ref{thm:rate1}, for any $\varepsilon >0$,
	there exists a constant $M$ such that, with probability greater than $%
	1-\varepsilon $, $\hat{f}_{t}(\cdot )\in \mathcal{G}_{t}$ uniformly in $t\in 
	\mathcal{T}$ and $\widehat{\phi }_{t,u}(\cdot )\in \mathcal{J}_{t,u}$
	uniformly in $(t,u)\in \mathcal{T}\mathcal{U}$. Here, we denote 
	\begin{equation*}
	\mathcal{G}_{t}=%
	\begin{Bmatrix}
	& \tilde{f}_t(X) \equiv (\Lambda(b(X)^{\prime }\beta_{t+h_1}) - \Lambda(b(X)^{\prime }\beta_{t-h_1}))/(2h_1):\\ & ||\beta_{t+h_1} ||_{0} + ||\beta_{t-h_1} ||_{0}\leq Ms,\quad ||\tilde{f}_t(X)
	-f_{t}(X)||_{\mathbb{P}, \infty }\leq M\delta _{n}h_1^{-1/2}, \\ 
	& ||\tilde{f}_t(X) -f_{t}(X)||_{\mathbb{P}_n,2}\leq M\varepsilon _{n}h_1^{-1/2}%
	\end{Bmatrix}%
	\end{equation*}%
	and 
	\begin{equation*}
	\mathcal{J}_{t,u}=%
	\begin{Bmatrix}
	\Lambda (b(x)^{\prime }\theta ): & ||\theta ||_{0}\leq Ms,||(\Lambda
	(b(X)^{\prime }\theta )-\phi _{t,u}(X))||_{\mathbb{P}_n,2}%
	\leq M \ell_n\varepsilon _{n}, \\ 
	& ||\Lambda (b(X)^{\prime }\theta )-\phi _{t,u}(X)||_{\mathbb{P}, \infty }\leq M\ell_n\delta
	_{n}.%
	\end{Bmatrix}%
	\end{equation*}%
 We focus on the case in
	which $(\widehat{\phi }_{t,u},\hat{f}_{t})\in \mathcal{J}_{t,u} \times 
	\mathcal{G}_{t}$. Then 
	\begin{align*}
	\hat{\alpha}^{\dagger }(t,u)-\alpha (t,u)=& (\mathbb{P}_{n}-\mathbb{P})\eta
	\Pi _{t,u}(W_{u},\phi _{t,u},f_{t})+(\mathbb{P}_{n}-\mathbb{P})\biggl[\eta
	\Pi _{t,u}(W_{u},\overline{\phi },\overline{f})-\eta \Pi _{t,u}(W_{u},\phi
	_{t,u},f_{t})\biggr] \\
	+& \mathbb{P}\biggl[\eta \Pi _{t,u}(W_{u},\overline{\phi },\overline{f}%
	)-\eta \Pi _{t,u}(W_{u},\phi _{t,u},f_{t})\biggr]+\biggl[\mathbb{P}\eta \Pi
	_{t,u}(W_{u},\phi _{t,u},f_{t})-\alpha (t,u)\biggr] \\
	:=& I+II+III+IV,
	\end{align*}%
	where $(\overline{\phi },\overline{f})=(\widehat{\phi }_{t,u},\hat{f}_{t}).$

Below we fix $(\overline{\phi },\overline{f})\in \mathcal{J}%
_{t,u} \times \mathcal{G}_{t}.$ First,
	\begin{align*}
	\text{Term } IV = \frac{\kappa_2h_2^2}{2}\left[\mathbb{E}\left(\partial_t^2 \phi_{t,u}(X) + \frac{2\partial_t \phi_{t,u}(X)\partial_t f_{t}(X)}{f_{t}(X)}\right)\right] + o(h_2^2) = \mathcal{\beta}_\alpha(t,u)h_2^2 + o(h_2^2).
	\end{align*}	
where the $o(h_2^{2})$ term holds uniformly
	in $(t,u)\in \mathcal{TU}$. For term $III$, uniformly over $(t,u)\in 
	\mathcal{T}\mathcal{U}$, we have 
	\begin{align}
	& \mathbb{P}\eta \biggl[\Pi _{t,u}(W_{u},\overline{\phi },\overline{f})-\Pi
	_{t,u}(W_{u},\phi _{t,u},f_{t})\biggr] \notag \\
	=& \mathbb{E}\biggl(\overline{\phi }(X)-\phi _{t,u}(X)\biggr)\biggl(1-\frac{%
		\mathbb{E}(K(\frac{T-t}{h_2})|X)}{h_2f_{t}(X)}\biggr)+\mathbb{E}\biggl(\frac{%
		Y_{u}-\overline{\phi }(X)}{\overline{f}(X)f_{t}(X)}\biggr)\biggl(\frac{%
		f_{t}(X)-\overline{f}(X)}{h_2}\biggr)K(\frac{T-t}{h_2}) \notag \\
	=& O(\ell_n\delta _{n}h^{2}_2)+\mathbb{E}\biggl(\frac{Y_{u}-\overline{\phi }\left(
		X\right) }{\overline{f}(X)f_{t}(X)}\biggr)\biggl(\frac{f_{t}(X)-%
		\overline{f}(X)}{h_2}\biggr)K(\frac{T-t}{h_2}) \notag \\
	= & O(\ell_n\delta _{n}h^{2}_2) + \mathbb{E}\biggl[\frac{f_t(X) - \bar{f}_t(X)}{\bar{f}_t(X)f_t(X)h_2}\mathbb{E}\left((\phi_{T,u}(X) - \phi_{t,u}(X))K(\frac{T-t}{h_2})\biggl|X\right) \biggr]\notag \\
	& + \mathbb{E}\biggl[\frac{(f_t(X) - \bar{f}_t(X))( \phi_{t,u}(X) - \bar{\phi}(X))}{\bar{f}_t(X)f_t(X)h_2}\mathbb{E}K(\frac{T-t}{h_2}|X)\biggr] \notag \\
	=& O(\ell_n\delta _{n}h_1^{-1/2}h_2^{2})+O(||(\phi _{t,u}(X)-\overline{\phi }(X))||_{P,2}||(f_{t}(X)-\overline{f}(X))||_{P,2}) \notag \\
	=& O(\ell_n\delta_nh_1^{-1/2}h_2^2+\ell_n\varepsilon _{n}^{2}h_1^{-1/2}).
	\label{A.15}
	\end{align}%
	The second equality of (\ref{A.15}) follows because there exists a
	constant $c$ independent of $n$ such that 
	\begin{equation*}
	\sup_{(t,u)\in \mathcal{T}\mathcal{U}}\biggl|1-\frac{\mathbb{E}(K(\frac{T-t}{%
			h_2})|X)}{h_2f_{t}(X)}\biggr|\leq ch_2^{2}
	\end{equation*}%
	and then 
	\begin{align*}
	& \mathbb{E}\biggl(\overline{\phi }(X)-\phi _{t,u}(X)\biggr)\biggl(1-\frac{%
		\mathbb{E}(K(\frac{T-t}{h_2})|X)}{h_2f_{t}(X)}\biggr) \leq ch_2^{2}\mathbb{E}||%
	\overline{\phi }(X)-\phi _{t,u}(X)||_{\mathbb{P}, \infty}=O(\ell_n\delta _{n}h_2^{2}).
	\end{align*}%
	The third equality of (\ref{A.15}) holds because $\mathbb{E}(Y_u|X,T) = \phi_{T,u}(X)$. The fourth equality of \eqref{A.15} holds by the fact that $||\overline{f}%
	_{t}(X)-f_{t}(X)||_{\mathbb{P}, \infty }=O(\delta _{n}h_1^{-1/2}) = o(1)$, $f_{t}(x)$ is assumed
	to be bounded away from zero uniformly over $t,\tau $ and the Cauchy
	inequality. The fifth inequality of (\ref{A.15}) holds because 
	\begin{equation*}
	||(\phi _{t,u}(X)-\overline{\phi }(X))||_{P,2}=[%
	\mathbb{E}||(\phi _{t,u}(X)-\overline{\phi }(X))||_{\mathbb{P}_n,2}^{2}]^{1/2}=O(\ell_n\varepsilon _{n})
	\end{equation*}%
	and for some constant $c>0$ independent of $(t,u,n)$, 
	\begin{equation*}
	||(f_{t}(X)-\overline{f}(X))||_{P,2} = O(%
	\varepsilon _{n}h_1^{-1/2}).
	\end{equation*}%

For the term $II$, we have 
	\begin{equation*}
	\mathbb{E}(\mathbb{P}_{n}-\mathbb{P})\eta \biggl[\Pi _{t,u}(W_{u},\overline{%
		\phi },\overline{f})-\Pi _{t,u}(W_{u},\phi _{t,u},f_{t})\biggr]\leq \mathbb{E%
	}||\mathbb{P}_{n}-\mathbb{P}||_{\mathcal{F}}
	\end{equation*}%
	where 
	\begin{equation*}
	\mathcal{F}=\cup _{(t,u)\in \mathcal{T}\mathcal{U}}\mathcal{F}_{t,u}\quad 
	\text{and}\quad \mathcal{F}_{t,u}=%
	\begin{Bmatrix}
	\eta \biggl[\Pi _{t,u}(W_{u},\overline{\phi },\overline{f})-\Pi
	_{t,u}(W_{u},\phi _{t,u},f_{t})\biggr]:\overline{\phi }\in \mathcal{J}_{t,u},%
	\overline{f}\in \mathcal{G}_{t}%
	\end{Bmatrix}%
	.
	\end{equation*}%
	Note $\mathcal{F}$ has envelope $|\frac{\eta }{h_2}|$, 
	\begin{align*}
	\sigma ^{2}:= & \sup_{f\in \mathcal{F}}\mathbb{E}f^{2} \\
	\lesssim & \sup_{(t,u)\in \mathcal{T}\mathcal{U}, (\bar{\phi},\bar{f}) \in \mathcal{J}_{t,u} \times \mathcal{G}_t}\mathbb{E}\biggl[(\overline{%
		\phi }(X)-\phi_{t,u} (X))^{2}\biggl(1-\frac{K(\frac{T-t}{h_2})}{f_{t}(X)h_2}%
	\biggr)^{2}\biggr] \\
	& +\sup_{(t,u)\in \mathcal{T}\mathcal{U}, (\bar{\phi},\bar{f}) \in \mathcal{J}_{t,u} \times \mathcal{G}_t}\mathbb{E}\biggl[\frac{Y_{u}-%
		\overline{\phi }(X)}{\overline{f}(X)f_{t}(X)h_2}K(\frac{T-t}{h_2})\biggl(%
	f_{t}(X)-\overline{f}(X)\biggr)\biggr]^{2} \\
	\lesssim & \sup_{(t,u)\in \mathcal{T}\mathcal{U}, (\bar{\phi},\bar{f}) \in \mathcal{J}_{t,u} \times \mathcal{G}_t}\mathbb{E}\biggl[(\overline{%
		\phi }(X)-\phi_{t,u} (X))^{2}\biggr]\biggl[1+\frac{K^2(\frac{T-t}{h_2})}{h_2^2}%
	\biggr] \\
	& +\sup_{(t,u)\in \mathcal{T}\mathcal{U}, (\bar{\phi},\bar{f}) \in \mathcal{J}_{t,u} \times \mathcal{G}_t}\mathbb{E}\biggl[f_{t}(X)-%
	\overline{f}(X)\biggr]^{2}\frac{K^2(\frac{T-t}{h_2})}{h_2^2} \\
	\lesssim & \sup_{(t,u)\in \mathcal{T}\mathcal{U}, (\bar{\phi},\bar{f}) \in \mathcal{J}_{t,u} \times \mathcal{G}_t}h_2^{-1}\mathbb{E}\biggl[(%
	\overline{ \phi }(X)-\phi_{t,u} (X))^{2}\biggr] +
	h_2^{-1}\sup_{(t,u)\in \mathcal{T}\mathcal{U}, (\bar{\phi},\bar{f}) \in \mathcal{J}_{t,u} \times \mathcal{G}_t}\mathbb{E}\biggl[(\overline{
		f }(X)-f_{t} (X))^{2}\biggr] \\
	\lesssim & 
	h_2^{-1}\varepsilon_n^2h_1^{-1}.
	\end{align*}

The second last inequality in the above display holds because $%
	f_{t}(x)$ is bounded away from zero uniformly in $(t,x) $,
	where $t = T+h_2v$ belongs to some compact enlargement of $\mathcal{T}$. Furthermore, $\mathcal{F}$ is nested by 
	\begin{equation*}
	\overline{F}=%
	\begin{Bmatrix}
	& \Pi _{t,u}(W_{u},\Lambda (b(X)^{\prime }\theta ),b(X)^{\prime }\beta )-\Pi
	_{t,u}(W_{u},\phi _{t,u},f_{t}):(t,u)\in \mathcal{T}\mathcal{U}, \\ 
	& ||\theta ||_{0}\leq Ms,||\beta ||_{0}\leq Ms%
	\end{Bmatrix}%
	,
	\end{equation*}%
	where 
	\begin{equation*}
	\sup_{Q}\log N(\mathcal{F},e_{Q},\varepsilon ||\overline{F}||_{Q,2})\lesssim s\log
	(p\vee n)+s\log (\frac{1}{\varepsilon })\vee 0.
	\end{equation*}%
	In addition, we claim $||\max_{1\leq i\leq n}|\eta _{i}/h_2|||_{p,2}\lesssim
	\log (n)h_2^{-1}$. When $\eta =1$, the above claim holds trivially. When $\eta 
	$ has sub-exponential tail, and the claim holds by \citet[Lemma 2.2.2]{VW96}%
	. Therefore, by \citet[Corollary 5.1]{CCK14}, we have 
	\begin{equation*}
	\mathbb{E}||\mathbb{P}_{n}-\mathbb{P}||_{\mathcal{F}}\lesssim \varepsilon
	_{n}(nh_1h_2)^{-1/2}s^{1/2}\log ^{1/2}(p\vee n)+\log (n)(nh_2)^{-1}s\log (p\vee
	n).
	\end{equation*}%
	Combining the bounds for $II$, $III$, and $IV$, we have 
	\begin{equation*}
	\hat{\alpha}^{\dagger }(t,u)-\alpha (t,u)=(\mathbb{P}_{n}-\mathbb{P})\eta
	\Pi _{t,u}(W_{u},\phi _{t,u},f_{t})+\mathcal{\beta}_\alpha(t,u)h_2^2+R_{n}(t,u)
	\end{equation*}
	and 
	\begin{align*}
	\sup_{(t,u)\in \mathcal{TU}}|R_{n}(t,u)|=O_{p}( \varepsilon_{n}^2(h_2^{-1/2} + \ell_n h_1^{-1/2}) + \log(n)s\log(p \vee n )(nh_2)^{-1}) + o_p(h_2^2).
	\end{align*}
	
	Then, when $\eta = 1$, 
	\begin{align*}
	\hat{\alpha}(t,u)-\alpha (t,u) = &
	(\mathbb{P}_{n}-\mathbb{P}) \Pi _{t,u}(W_{u},\phi _{t,u},f_{t}) + \mathcal{B}_\alpha(t,u)h_2^2 + R_{n}(t,u)
	\\ = & (\mathbb{P}_{n}-\mathbb{P}) (\Pi _{t,u}(W_{u},\phi
	_{t,u},f_{t})-\alpha(t,u))+\mathcal{B}_\alpha(t,u)h_2^2+R_{n}(t,u).
	\end{align*}
	
	Then, Assumption \ref{ass:rate2} implies that $\sup_{(t,u) \in \mathcal{T} \mathcal{U}}|R_n(t,u)| = o_p((nh_2)^{-1/2})$. For the
	bootstrap estimator, we have 
	\begin{align}
	\label{eq:aboot} \hat{\alpha}^b(t,u)-\alpha (t,u) = &
	\hat{\alpha}^\dagger(t,u)/\bar{\eta}-\alpha (t,u) \notag \\ 
	= &
	(\hat{\alpha}^\dagger(t,u)-\alpha(t,u))/\bar{\eta} + \alpha
	(t,u)(1/\bar{\eta}-1) \notag \\ 
	= & (\mathbb{P}_{n}-\mathbb{P}) \eta \Pi
	_{t,u}(W_{u},\phi _{t,u},f_{t}) /\bar{\eta} + \alpha (t,u)(1/\bar{\eta}-1) +
	\mathcal{B}_\alpha(t,u)h_2^2/\bar{\eta}+R_n(t,u)/\bar{\eta} \notag \\ 
	= & (\mathbb{P}_{n}-\mathbb{P}) \eta \Pi
	_{t,u}(W_{u},\phi _{t,u},f_{t} - \alpha(t,u)) /\bar{\eta} + \mathcal{B}_\alpha(t,u)h_2^2/\bar{\eta}+
	R_n(t,u)/\bar{\eta} \notag \\ 
	= & (\mathbb{P}_{n}-\mathbb{P}) \eta (\Pi
	_{t,u}(W_{u},\phi _{t,u},f_{t}) - \alpha(t,u)) +\mathcal{B}_\alpha(t,u)h_2^2+R^b_{n}(t,u), \end{align}
	where $\sup_{(t,u)\in \mathcal{TU}}|R^b_{n}(t,u)|=O_{p}(\varepsilon_{n}^2(h_2^{-1/2} + \ell_nh_1^{-1/2}) + \log(n)s\log(p \vee n )(nh_2)^{-1}) + o_p(h_2^{2}).$ This is because of the fact that 
	\begin{equation*}
	\bar{\eta}- \mathbb{E}\eta = \bar{\eta}-1 = O_p(n^{-1/2}), 
	\end{equation*}
	\begin{align*}
	\sup_{(t,u)\in \mathcal{TU}}|R_{n}(t,u)|=O_{p}(\varepsilon_{n}^2(h_2^{-1/2} + \ell_nh_1^{-1/2}) + \log(n)s\log(p \vee n )(nh_2)^{-1}) + o_p(h_2^{2}),
	\end{align*}
	and the collection of functions 
	\begin{equation*}
	\{ \eta (\Pi _{t,u}(W_{u},\phi _{t,u},f_{t}) - \alpha(t,u)): (t,u) \in 
	\mathcal{TU} \}
	\end{equation*}
	satisfies 
	\begin{equation*}
	\sup_{(t,u)\in \mathcal{T}\mathcal{U}}|(\mathbb{P}_n - \mathbb{P})(\eta (\Pi
	_{t,u}(W_{u},\phi _{t,u},f_{t}) - \alpha(t,u)))| =
	O_p(\log^{1/2}(n)(nh_2)^{-1/2}).
	\end{equation*}

Therefore, 
	\begin{equation*}
	\begin{aligned} \hat{\alpha}^b(t,u)-\hat{\alpha} (t,u) = &
	(\mathbb{P}_{n}-\mathbb{P}) (\eta-1) (\Pi _{t,u}(W_{u},\phi _{t,u},f_{t}) -
	\alpha(t,u)) +R^b_{n}(t,u) - R_{n}(t,u), \end{aligned}
	\end{equation*}
	where 
	$$\sup_{(t,u) \in \mathcal{T} \mathcal{U}}|R^b_{n}(t,u) - R_{n}(t,u)| =
	O_{p}(\varepsilon_{n}^2(h_2^{-1/2} + \ell_n h_1^{-1/2}) + \log(n)s\log(p \vee n )(nh_2)^{-1})+ o_p(h_2^{2}) = o_p((nh_2)^{-1/2}).$$ 
\end{proof}

\bigskip 

\begin{proof}[Proof of Theorem \protect\ref{thm:q}]
Let $\hat{ \alpha}^{\ast }(t,u)$ be either the original or the
bootstrap estimator of $\alpha(t,u)$. We first derive the linear expansion
of the rearrangement of $\hat{\alpha}^{\ast }(t,u)$ defined in the proof of
Theorem \ref{thm:2nd}. For $z\in (0,1)$, let 
\begin{equation*}
F(t,z)=\int_{0}^{1}1\{\alpha (t,\psi ^{\leftarrow }(v))\leq z\}dv,\quad
F(t,z|d_{n})=\int_{0}^{1}1\{\hat{\alpha}^{\ast }(t,\psi ^{\leftarrow
}(v))\leq y\}dv,
\end{equation*}%
where $\psi(\cdot)$ is defined in Section \ref{sec:3rd}. Then, by Lemma \ref%
{lem:verify} in the online supplement, we have 
\begin{equation}
\frac{F(t,z|d_{n})-F(t,z)}{s_{n}}+\frac{d_{n}(t,\psi (q_{z}(t)))\psi
^{\prime }(q_{z}(t))}{f_{Y(t)}(q_{z}(t))}=o_{p}(\delta _{n})  \label{eq:cdf4}
\end{equation}%
and 
\begin{equation}
\frac{\hat{\alpha}^{\ast r}(t,u)-\alpha (t,u)}{s_{n}}+\frac{F(t,\alpha
(t,u)|d_{n})-F(t,\alpha (t,u))f_{Y(t)}(u)}{s_{n}\psi ^{\prime }(u)}%
=o_{p}(\delta _{n}).  \label{eq:inv1}
\end{equation}%
where $s_{n}=(nh_2)^{-1/2}$, $d_{n}(t,v)=(nh_2)^{1/2}(\hat{\alpha}^{\ast
}(t,\psi ^{\leftarrow }(v))-\alpha (t,\psi ^{\leftarrow }(v))),$ $%
f_{Y(t)}(\cdot )$ is the density of $Y(t)$, $q_{z}(t)$ is the $z$-th
quantile of $Y(t)$, and $\delta _{n}$ equals to either $1$ or $h_2^{1/2}$,
depending on either Assumption \ref{ass:rate2}.1 or \ref{ass:rate2}.2 is in
place. 

Combining \eqref{eq:cdf4} and \eqref{eq:inv1}, we have 
\begin{equation}
(nh_2)^{1/2}(\hat{\alpha}^{*r}(t,u) - \alpha(t,u)) = d_{n}(t,\psi(u)) +
o_p(\delta_n) = (nh_2)^{1/2}( \hat{\alpha}^*(t,u)-\alpha(t,u)) + o_p(\delta_n)
\label{eq:rearrange}
\end{equation}
uniformly over $(t,u) \in \mathcal{T} \mathcal{U}$. 

\medskip 
We can apply Lemma \ref{lem:verify} on $\hat{\alpha}%
^{*r}(t,u)$ again with $J_n(t,u) = (nh_2)^{1/2}(\hat{\alpha}^{*r}(t,u) -
\alpha(t,u))$, $F(t,u) = P(Y(t) \leq u) = \alpha(t,u)$, $f(t,u) =
f_{Y(t)}(u) $, and $F^\leftarrow(t,\tau) = q_\tau(t)$. Then, for $\delta_n$
equals $1$ or $h_2^{1/2}$ under either Assumption \ref{ass:rate2}.1 or \ref%
{ass:rate2}.2, respectively, we have, 
\begin{equation}
\frac{\hat{q}^*_\tau(t) - q_\tau(t)}{s_n} = - \frac{J_n(t,q_\tau(t))}{%
f_{Y(t)}(q_\tau(t))} + o_p(\delta_n) = - \frac{(nh_2)^{1/2}( \hat{\alpha}%
^{*r}(t,q_\tau(t))-\tau)}{f_{Y(t)}(q_\tau(t))} + o_p(\delta_n)
\label{eq:quantile}
\end{equation}
uniformly over $(t,\tau) \in \mathcal{T} \mathcal{I}.$ 

\medskip 
When $\eta = 1$, combining \eqref{eq:rearrange}, %
\eqref{eq:quantile}, and Theorem \ref{thm:2nd}, we have 
\begin{equation*}
\hat{q}_\tau(t) - q_\tau(t) = -(\mathbb{P}_{n}-\mathbb{P})\frac{\Pi
	_{t,u}(W_{q_{\tau }(t)},\phi _{t,q_{\tau }(t)},f_{t})}{f_{Y(t)}(q_{\tau
	}(t))} - \frac{\mathcal{\beta}_\alpha(t,q_\tau(t))h_2^2}{f_{Y(t)}(q_{\tau
	}(t))} + R_n(t,\tau) + o_p(\delta_n (nh_2)^{-1/2}).
\end{equation*}
By taking $\delta_n = 1$ and $\delta_n = h_2^{1/2}$ under Assumptions \ref%
{ass:rate2}.1 and \ref{ass:rate2}.2, respectively, we have establish the
desired results. For the bootstrap estimator, by \eqref{eq:aboot}, we have 
\begin{equation*}
\begin{aligned} \hat{q}^b_\tau(t) - q_\tau(t) = -(\mathbb{P}_{n}-\mathbb{P})\eta\frac{\Pi
	_{t,u}(W_{q_{\tau }(t)},\phi _{t,q_{\tau }(t)},f_{t})}{f_{Y(t)}(q_{\tau
	}(t))} - \frac{\mathcal{\beta}_\alpha(t,q_\tau(t))h_2^2}{f_{Y(t)}(q_{\tau
	}(t))} + R^b_n(t,\tau) + o_p(\delta_n (nh_2)^{-1/2}).
\end{aligned}
\end{equation*}
Then, 
\begin{equation*}
\begin{aligned} & \hat{q}^b_\tau(t) - \hat{q}_\tau(t) \\ = & -(\mathbb{P}_n
- \mathbb{P})(\eta-1)\Pi_{t,u}(W_{q_\tau(t)},\phi_{t,q_\tau(t)},f_t)
/f_{Y(t)}(q_\tau(t)) + R^b_n(t,\tau) - R_n(t,\tau) + o_p(\delta_n
(nh_2)^{-1/2}). \end{aligned}
\end{equation*}
By taking $\delta_n = 1$ and $\delta_n = h_2^{1/2}$ under Assumptions \ref%
{ass:rate2}.1 and \ref{ass:rate2}.2, respectively, we have establish the
linear expansion of the bootstrap estimator too. Last, note that the bootstrap estimator cannot preserve the asymptotic bias term. For the validity of bootstrap inference, we need to under-smooth and require $nh^5_2 \rightarrow 0$. This condition is assumed in Theorem \ref{thm:infer}.
\end{proof}

\bigskip 

\begin{proof}[Proof of Theorem \protect\ref{thm:qprime}]
We consider the general case in which the observations are weighted
by $\{\eta _{i}\}_{i=1}^{n}$ as above. For brevity, denote $\hat{\delta}:= (%
\hat{\delta}_{0},\hat{\delta}_{1})^{\prime }=(\hat{\beta}_{\tau }^{\ast
0}(t),\hat{\beta}_{\tau }^{\ast 1}(t))^{\prime }$ and $\delta := (\delta
_{0},\delta _{1})^{\prime }=(\beta _{\tau }^{0}(t),\beta _{\tau }^{1}(t)).$
For any variable $R_n := R_n(\tau,t)$ and some deterministic sequence $r_n$,
we write $R_n = O_p^*(r_n)$ (resp. $o_p^*(r_n)$) if $\sup_{(t,\tau) \in \mathcal{TI%
}}|R_n(\tau,t)| = O_p(r_n)$ (resp. $o_p(r_n)$). Then $\hat{\delta}=\widehat{\Sigma 
}_{2}^{-1}\widehat{\Sigma }_{1},$ where 
\begin{equation*}
\widehat{\Sigma }_{1}=%
\begin{pmatrix}
\frac{1}{n}\sum_{i=1}^{n}K(\frac{T_{i}-t}{h_2})\eta _{i}\hat{q}_{\tau }^{\ast
}(T_{i}) \\ 
\frac{1}{n}\sum_{i=1}^{n}K(\frac{T_{i}-t}{h_2})(T_{i}-t)\eta _{i}\hat{q}_{\tau
}^{\ast }(T_{i})%
\end{pmatrix}%
\end{equation*}%
and 
\begin{equation*}
\widehat{\Sigma }_{2}=%
\begin{pmatrix}
\frac{1}{n}\sum_{i=1}^{n}K(\frac{T_{i}-t}{h_2})\eta _{i} & \frac{1}{n}%
\sum_{i=1}^{n}K(\frac{T_{i}-t}{h_2})(T_{i}-t)\eta _{i} \\ 
\frac{1}{n}\sum_{i=1}^{n}K(\frac{T_{i}-t}{h_2})(T_{i}-t)\eta _{i} & \frac{1}{n}%
\sum_{i=1}^{n}K(\frac{T_{i}-t}{h_2})(T_{i}-t)^{2}\eta _{i}%
\end{pmatrix}%
.
\end{equation*}%

Let $\Sigma _{2}=%
\begin{pmatrix}
f(t) & 0 \\ 
\kappa _{2}f^{(1)}(t) & \kappa _{2}f(t)%
\end{pmatrix}%
$ and $G=%
\begin{pmatrix}
h_2^{-1} & 0 \\ 
0 & h_2^{-3}%
\end{pmatrix}%
$. Then we have 
\begin{equation*}
G\hat{\Sigma}_{2}-\Sigma _{2}=O_{p}^{\ast }(\log ^{1/2}(n)(nh_2^{3})^{-1/2}).
\end{equation*}%
In addition, note 
\begin{equation*}
\hat{q}_{\tau }^{\ast }(T_{i})=\delta _{0}+\delta _{1}(T_{i}-t)+(q_{\tau
}(T_{i})-\delta _{0}-\delta _{1}(T_{i}-t))+(\hat{q}_{\tau }^{\ast
}(T_{i})-q_{\tau }(T_{i}))
\end{equation*}%
and 
\begin{align*}
& \begin{pmatrix}
\frac{1}{nh_2}\sum_{i=1}^{n}K(\frac{T_{i}-t}{h_2})\eta _{i}\biggl(q_{\tau
}(T_{i})-\delta _{0}-\delta _{1}(T_{i}-t)\biggr) \\ 
\frac{1}{nh_2^{3}}\sum_{i=1}^{n}K(\frac{T_{i}-t}{h_2})(T_{i}-t)\eta _{i}\biggl(%
q_{\tau }(T_{i})-\delta _{0}-\delta _{1}(T_{i}-t)\biggr)%
\end{pmatrix} \\
=& \begin{pmatrix}
\frac{1}{2}q^{\prime \prime}_\tau(t)f_T(t)\kappa_2 h_2^2 \\
\frac{1}{6}\partial_t(q^{\prime \prime}_\tau(t)f_T(t))\kappa_4 h_2^2 \\
\end{pmatrix}+O_{p}^{\ast }(\sqrt{\frac{\log (n)h_2}{n}}) +o^*(h_2^{2}).
\end{align*}%
Therefore, 
\begin{align}
G\widehat{\Sigma }_{1}= & G\widehat{\Sigma }_{2}\delta +%
\begin{pmatrix}
\frac{1}{nh_2}\sum_{i=1}^{n}K(\frac{T_{i}-t}{h_2})\biggl(\hat{q}_{\tau
}(T_{i})-q_{\tau }(T_{i})\biggr)\eta _{i} \\ 
\frac{1}{nh_2^{3}}\sum_{i=1}^{n}K(\frac{T_{i}-t}{h_2})(T_{i}-t)\biggl(\hat{q}%
_{\tau }(T_{i})-q_{\tau }(T_{i})\biggr)\eta _{i}%
\end{pmatrix}%
+\begin{pmatrix}
\frac{1}{2}q^{\prime \prime}_\tau(t)f_T(t)\kappa_2 h_2^2 \\
\frac{1}{6}\partial_t(q^{\prime \prime}_\tau(t)f_T(t))\kappa_4 h_2^2 \\
\end{pmatrix} \notag \\
& +O_{p}^{\ast }(\sqrt{\frac{\log (n)h_2}{n}}) +o^*(h_2^{2}).  \label{eq:sigma1delta}
\end{align}%
Let $E(t,\tau) = \mathbb{E}\frac{Y_{q_{\tau }(t),j}-\phi
	_{t,q_{\tau }(t)}(X_{j})}{f_{t}(X_{j})h_2}K(\frac{T_{j}-t}{h_2}) + \tau$. By Theorem \ref{thm:q}, we have 
\begin{align}
& \hat{q}_{\tau }^{\ast }(t)-q_{\tau }(t) \notag \\
=& \frac{-1}{f_{Y_{t}}(q_{\tau }(t))}%
\frac{1}{n}\sum_{j=1}^{n}\eta _{j}\biggl(\frac{Y_{q_{\tau }(t),j}-\phi
_{t,q_{\tau }(t)}(X_{j})}{f_{t}(X_{j})h_2}K(\frac{T_{j}-t}{h_2})+\phi
_{t,q_{\tau }(t)}(X_{j})-E(t,\tau) \biggr) \notag \\
& - \mathcal{\beta}_q(t,\tau)h_2^2 +o_{p}^{\ast }((nh_2)^{-1/2}).
\label{eq:q2}
\end{align}%
Let $\Upsilon _{i}=(Y_{i},T_{i},X_{i},\eta _{i})$. Then, by plugging %
\eqref{eq:q2} in \eqref{eq:sigma1delta} and noticing that
\begin{equation*}
\begin{Bmatrix}
\sup_{t\in \mathcal{T}}\frac{1}{nh_2}\sum_{i=1}^{n}K(\frac{T_{i}-t}{h_2})\eta
_{i} \\ 
\sup_{t\in \mathcal{T}}\frac{1}{nh_2^{3}}\sum_{i=1}^{n}K(\frac{T_{i}-t}{h_2}%
)|T_{i}-t|\eta _{i}%
\end{Bmatrix}%
=%
\begin{Bmatrix}
O_{p}(1) \\ 
O_{p}(h_2^{-1})%
\end{Bmatrix}%
,
\end{equation*}%
we have 
\begin{align*}
G\widehat{\Sigma }_{1}= & G\widehat{\Sigma }_{2}\delta -\frac{1}{n(n-1)}%
\sum_{i\neq j}\eta _{i}\eta _{j}\Gamma (\Upsilon _{i},\Upsilon _{j};t,\tau ) - \begin{pmatrix}
f_T(t)\mathcal{\beta}_q(t,\tau)h_2^2 \\
f'_T(t)\mathcal{\beta}_q(t,\tau)h_2^2 \\
\end{pmatrix} \\
& +\begin{pmatrix}
\frac{1}{2}q^{\prime \prime}_\tau(t)f_T(t)\kappa_2 h_2^2 \\
\frac{1}{6}\partial_t(q^{\prime \prime}_\tau(t)f_T(t))\kappa_4 h_2^2 \\
\end{pmatrix}+%
\begin{Bmatrix}
o_{p}^{\ast }((nh_2)^{-1/2}) \\ 
o_{p}^{\ast }((nh_2^{3})^{-1/2}) \\ 
\end{Bmatrix}%
\end{align*}%
where $\Gamma (\Upsilon _{i},\Upsilon _{j};t,\tau )=(\Gamma _{0}(\Upsilon
_{i},\Upsilon _{j};t,\tau ),\Gamma _{1}(\Upsilon _{i},\Upsilon _{j};t,\tau
))^{\prime }$, and 
\begin{align*}
& \Gamma _{\ell }(\Upsilon _{i},\Upsilon _{j};t,\tau )\\
= & \frac{(T_{i}-t)^{\ell }%
}{h_2^{1+2\ell }f_{Y_{T_{i}}}(q_{\tau }(T_{i}))}K(\frac{T_{i}-t}{h_2})\biggl(%
\frac{Y_{q_{\tau }(T_{i}),j}-\phi _{T_{i},q_{\tau }(T_{i})}(X_{j})}{%
f_{T_{i}}(X_{j})h_2}K(\frac{T_{j}-T_{i}}{h_2})+\phi _{T_{i},q_{\tau
}(T_{i})}(X_{j})-E(T_i,\tau) \biggr)
\end{align*}%
for $\ell =0,1.$ Let $\Gamma ^{s}(\Upsilon _{i},\Upsilon
_{j};t,\tau )=(\Gamma (\Upsilon _{i},\Upsilon _{j};t,\tau )+\Gamma (\Upsilon
_{j},\Upsilon _{i};t,\tau ))/2$. Because $nh_2^7 \rightarrow 0$, we have 
\begin{equation}
\hat{\beta}_{\tau }^{1\ast }(t)-\beta _{\tau }^{1}(t)=-e_{2}^{\prime }(G%
\widehat{\Sigma }_{2})^{-1}U_{n}(t,\tau )+o_{p}^{\ast }((nh_2^{3})^{-1/2}),
\label{eq:delta2}
\end{equation}%
where $e_{2}=(0,1)^{\prime }$ and $U_{n}(t,\tau
)=(C_{n}^{2})^{-1}\sum_{1\leq i<j\leq n}\eta _{i}\eta _{j}\Gamma ^{s}(\cdot
,\cdot ;t,\tau )$ is a U-process indexed by $(t,\tau )$. By Lemma \ref%
{lem:43} in the online supplement, 
\begin{equation}
e_{2}^{\prime }(G\widehat{\Sigma }_{2})^{-1}U_{n}(t,\tau )=\frac{-1}{n}%
\sum_{j=1}^{n}\eta _{j}(\kappa _{2}f_{Y(t)}(q_{\tau
}(t))f_{t}(X_{j})h_2^{2})^{-1}\biggl[Y_{q_{\tau }(t),j}-\phi _{t,q_{\tau
}(t)}(X_{j})\biggr]\overline{K}(\frac{T_{j}-t}{h_2})+o_{p}^{\ast
}((nh_2^{3})^{-1/2}).  \label{eq:43}
\end{equation}%
Combining \eqref{eq:delta2} and \eqref{eq:43}, we have 
\begin{equation*}
\hat{\beta}_{\tau }^{1\ast }(t)-\beta _{\tau }^{1}(t)=\frac{-1}{n}%
\sum_{j=1}^{n}\eta _{j}(\kappa _{2}f_{Y_{t}}(q_{\tau
}(t))f_{t}(X_{j})h_2^{2})^{-1}\biggl[Y_{q_{\tau }(t),j}-\phi _{t,q_{\tau
}(t)}(X_{j})\biggr]\overline{K}(\frac{T_{j}-t}{h_2})+o_{p}^{\ast
}((nh_2^{3})^{-1/2}).
\end{equation*}%

\end{proof}

\bigskip

\begin{proof}[Proof of Theorem \protect\ref{thm:infer}]
By the proofs of Theorems \ref{thm:q} and \ref{thm:qprime}, we have 
\begin{equation*}
\hat{q}_{\tau }^{b}(t)-\hat{q}_{\tau }(t)=-(\mathbb{P}_{n}-\mathbb{P})(\eta
-1)\biggl(\Pi _{t,u}(W_{q_{\tau }(t)},\phi _{t,q_{\tau }(t)},f_{t})-\tau %
\biggr)/f_{t}(q_{\tau }(t))+o_{p}((nh_2)^{-1/2})
\end{equation*}%
and 
\begin{equation*}
\hat{\beta}_{\tau }^{1b}(t)-\hat{\beta}_{\tau }^{1}(t)=\frac{1}{n}%
\sum_{j=1}^{n}(\eta _{j}-1)(\kappa _{2}f_{Y(t)}(q_{\tau
}(t))f_{t}(X_{j})h^{2})^{-1}\biggl[Y_{q_{\tau }(t),j}-\phi _{t,q_{\tau
}(t)}(X_{j})\biggr]\overline{K}(\frac{T_{j}-t}{h_2})+o_{p}^{\ast
}((nh_2^{3})^{-1/2}).
\end{equation*}%
Then, it is straightforward to show that $\sqrt{nh_2}(\hat{q}_{\tau }^{b}(t)-%
\hat{q}_{\tau }(t))$ and $(nh_2^{3})^{1/2}(\hat{\beta}_{\tau }^{1b}(t)-\hat{%
\beta}_{\tau }^{1}(t))$ converge weakly to the limiting distribution of $%
\sqrt{nh_2}(\hat{q}_{\tau }(t)-q_{\tau }(t))$ and $(nh_2^{3})^{1/2}(\hat{\beta}%
_{\tau }^{1}(t)-\beta _{\tau }^{1}(t))$, respectively, conditional on data
in the sense of \citeauthor{VW96} (\citeyear{VW96}, Section 2.9). The
desired results then follow. 
\end{proof}

{\small
\bibliographystyle{elsarticle-harv}
\bibliography{QDR}
\newpage }

\setcounter{page}{1} \renewcommand\thesection{\Alph{section}} %
\setcounter{section}{1} \linespread{1.2}

\begin{center}
Supplementary Material for

\textbf{\textquotedblleft Non-separable Models with High-dimensional
Data\textquotedblright }

[NOT\ INTENDED FOR\ PUBLICATION]

\bigskip

Liangjun Su$^{a},$ Takuya Ura$^{b}$ and Yichong
Zhang$^{a}$

$^{a}${School of Economics, Singapore Management University}

$^{b}${Department of Economics, University of California,
Davis\medskip }
\end{center}

\noindent This supplement is composed of four parts. Appendix \ref%
{sec:lem} provides the proofs of some technical lemmas used in the proofs of
the main results in the paper. Appendix \ref{sec:rearrange} studies the
rearrangement operator on a local process. Appendix \ref{sec:addsim} and \ref{sec:add_emp} report some additional simulation and application results, respectively. 

\section{Proofs of the Technical Lemmas}

\label{sec:lem} Lemma \ref{lem:main} and Lemma \ref{lem:minor} below
are closely related to Lemmas J.6 and O.2 in \cite{BCFH13} with one major
difference: we have an additional kernel function which affects the rate of
convergence. We follow the proof strategies in \cite{BCFH13} in general, but
use the local compatibility condition established in Lemma \ref{lem:localRE}
when needed. We include these proofs mainly for completeness. Lemma \ref%
{lem:E2} is proved without referring to the theory of
moderate deviations for self-normalized sums, in contrast to the proof of
Lemma J.1 in \cite{BCFH13}. Consequently, we have the additional $\ell _{n}$
term but avoid one constraint on the rates of $p$, $s$, and $n$, as
well.

\medskip

\begin{proof}[Proof of Lemma \protect\ref{lem:main}]
We define the following three events:
\begin{equation*}
E_{1}=\{C_{r}(\log (p\vee n)s/n)^{1/2}\geq \sup_{(t,u)\in \mathcal{T}%
\mathcal{U}}||\frac{r_{t,u}^{\phi }}{\omega _{t,u}^{1/2}}K(\frac{T-t}{h_1}%
)^{1/2}||_{\mathbb{P}_n,2}\},
\end{equation*}%
\begin{equation*}
E_{2}=\biggl\{\frac{\lambda }{n}\geq \sup_{(t,u)\in \mathcal{T}\mathcal{U}%
}C_{\lambda }\biggl|\biggl|\widehat{\Psi }_{t,u,0}^{-1}\mathbb{P}_{n}\biggl[%
\xi _{t,u}K(\frac{T-t}{h_1})b(X)\biggr]\biggr|\biggr|_{\infty }\biggr\},
\end{equation*}%
and 
\begin{equation*}
E_{3}=\{l\widehat{\Psi }_{t,u,0}\leq \widehat{\Psi }_{t,u}\leq L\widehat{%
\Psi }_{t,u,0}\quad \text{and}\quad C_{\psi }/2\leq \inf_{(t,u)\in \mathcal{%
TU}}||\widehat{\Psi }_{t,u,0}||_{\infty }\leq \sup_{(t,u)\in \mathcal{TU}}||%
\widehat{\Psi }_{t,u,0}||_{\infty }\leq 2/C_{\psi }\},
\end{equation*}%
where $l$, $L$, and $C_{\psi }$ are defined in the statement of Lemma \ref%
{lem:E32} and the generic penalty loading matrix is $\widehat{\Psi }_{t,u}=%
\widehat{\Phi }_{t,u}^{k}$ for $k=0,\cdots ,K$.

By Assumption \ref{ass:approx}.4, for an arbitrary $\varepsilon >0$,
we can choose $C_{r}$ and $n$ sufficiently large so that $\mathbb{P}(E_{1})\geq
1-\varepsilon .$ By Lemma \ref{lem:E2} below and the fact that $\ell
_{n}\rightarrow \infty $, for any $\varepsilon >0$ and any $C_{\lambda }>0$,
for $n$ sufficiently large, we have $\mathbb{P}(E_{2})\geq 1-\varepsilon .$ In
particular, we choose $C_{\lambda }$ such that $C_{\lambda }l>1$. Last, by
Lemma \ref{lem:E32} below, $\mathbb{P}(E_{3})>1-\varepsilon _{n}$ for some
deterministic sequence $\varepsilon _{n}\downarrow 0$. 

From now on we assume $E_{1}$, $E_{2}$, and $E_{3}$ hold with
constants $C_{r}$, $C_{\lambda }$, $l$, and $L$, which occurs with
probability greater than $1-2\varepsilon -\varepsilon _{n}$. Let $\delta
_{t,u}=\hat{\theta}_{t,u}-\theta _{t,u}$ and $\mathcal{S}^0_{t,u}=\text{Supp}%
(\theta _{t,u})$. Let%
\begin{equation*}
\Gamma _{t,u}=||\omega _{t,u}^{1/2}b(X)^{\prime }\delta _{t,u}K(\frac{T-t}{h_1}%
)^{1/2}||_{\mathbb{P}_n,2},
\end{equation*}%
and 
\begin{equation*}
\tilde{c}=\max (4(LC_{\lambda }+1)(lC_{\lambda }-1)^{-1}C_{\psi }^{-2},1).
\end{equation*}%
Then, under $E_{3}$, 
\begin{equation*}
\tilde{c}\geq \max ((LC_{\lambda }+1)/(lC_{\lambda }-1)\sup_{(t,u)\in 
\mathcal{T}\mathcal{U}}||\widehat{\Psi }_{t,u,0}||_{\infty }||\widehat{\Psi }%
_{t,u,0}^{-1}||_{\infty },1)\geq 1.
\end{equation*}%
Let $Q_{t,u}(\theta )=\mathbb{P}_{n}M(Y_{u},X;\theta )K(\frac{T-t}{h_1})$. By
the fact that $\hat{\theta}_{t,u}$ solves the minimization problem
in \eqref{eq:m}, we have 
\begin{equation}
\begin{split}
Q_{t,u}(\hat{\theta}_{t,u})-Q_{t,u}(\theta _{t,u})\leq & \frac{\lambda }{n}||%
\widehat{\Psi }_{t,u}\theta _{t,u}||_{1}-\frac{\lambda }{n}||\widehat{\Psi }%
_{t,u}\hat{\theta}_{t,u}||_{1} \\
\leq & \frac{\lambda }{n}||\widehat{\Psi }_{t,u}(\delta _{t,u})_{\mathcal{S}%
^0_{t,u}}||_{1}-\frac{\lambda }{n}||\widehat{\Psi }_{t,u}(\hat{\theta}%
_{t,u})_{\mathcal{S}_{t,u}^{0c}}||_{1} \\
=& \frac{\lambda }{n}||\widehat{\Psi }_{t,u}(\delta _{t,u})_{\mathcal{S}%
^0_{t,u}}||_{1}-\frac{\lambda }{n}||\widehat{\Psi }_{t,u}(\delta _{t,u})_{%
\mathcal{S}_{t,u}^{0c}}||_{1} \\
\leq & \frac{\lambda L}{n}||\widehat{\Psi }_{t,u,0}(\delta _{t,u})_{\mathcal{%
S}^0_{t,u}}||_{1}-\frac{\lambda l}{n}||\widehat{\Psi }_{t,u,0}(\delta
_{t,u})_{\mathcal{S}_{t,u}^{0c}}||_{1}.
\end{split}
\label{B.1}
\end{equation}%
Because the kernel function $K(\cdot )$ is nonnegative, $Q_{t,u}(\theta )$
is convex in $\theta $. It follows that $Q_{t,u}(\hat{\theta}%
_{t,u})-Q_{t,u}(\theta _{t,u})\geq \partial _{\theta }Q_{t,u}(\theta
_{t,u})^{\prime }\delta _{t,u}.$

Let $D_{t,u}=-\mathbb{P}_{n}b(X)\xi _{t,u}K(\frac{T-t}{h_1})$ and $\xi
_{t,u}=Y_{u}-\phi _{t,u}(X)$. Then, 
\begin{equation}
\begin{split}
& |\partial _{\theta }Q_{t,u}(\theta _{t,u})^{\prime }\delta _{t,u}| \\
=& |\mathbb{P}_{n}(\Lambda (b(X)^{\prime }\theta _{t,u})-Y_{u})K(\frac{T-t}{h_1%
})b(X)^{\prime }\delta _{t,u}| \\
=& |\mathbb{P}_{n}r_{t,u}^{\phi }b(X)^{\prime }\delta _{t,u}K(\frac{T-t}{h_1}%
)+D_{t,u}^{\prime }\delta _{t,u}| \\
\leq & ||\widehat{\Psi }_{t,u,0}^{-1}D_{t,u}||_{\infty }||\widehat{\Psi }%
_{t,u,0}\delta _{t,u}||_{1}+||\frac{r_{t,u}^{\phi }K(\frac{T-t}{h_1})^{1/2}}{%
\omega _{t,u}^{1/2}}||_{\mathbb{P}_n,2}\Gamma _{t,u} \\
\leq & \frac{\lambda }{nC_{\lambda }}||\widehat{\Psi }_{t,u,0}\delta
_{t,u}||_{1}+||\frac{r_{t,u}^{\phi }K(\frac{T-t}{h_1})^{1/2}}{\omega
_{t,u}^{1/2}}||_{\mathbb{P}_n,2}\Gamma _{t,u} \\
\leq & \frac{\lambda }{nC_{\lambda }}||\widehat{\Psi }_{t,u,0}(\delta
_{t,u})_{\mathcal{S}^0_{t,u}}||_{1}+\frac{\lambda }{nC_{\lambda }}||\widehat{%
\Psi }_{t,u,0}(\delta _{t,u})_{\mathcal{S}_{t,u}^{0c}}||_{1}+||\frac{%
r_{t,u}^{\phi }K(\frac{T-t}{h_1})^{1/2}}{\omega _{t,u}^{1/2}}%
||_{\mathbb{P}_n,2}\Gamma _{t,u},
\end{split}
\label{B.2}
\end{equation}%
where $r_{t,u}^{\phi }=r_{t,u}^{\phi }(X)$. Combining \eqref%
{B.1} and \eqref{B.2}, we have 
\begin{equation*}
\frac{\lambda (lC_{\lambda }-1)}{nC_{\lambda }}||\widehat{\Psi }%
_{t,u,0}(\delta _{t,u})_{\mathcal{S}_{t,u}^{0c}}||_{1}\leq \frac{\lambda
(LC_{\lambda }+1)}{nC_{\lambda }}||\widehat{\Psi }_{t,u,0}(\delta _{t,u})_{%
\mathcal{S}^0_{t,u}}||_{1}+||\frac{r_{t,u}^{\phi }K(\frac{T-t}{h_1})^{1/2}}{%
\omega _{t,u}^{1/2}}||_{\mathbb{P}_n,2}\Gamma _{t,u}.
\end{equation*}%
Then 
\begin{align*}
||(\delta _{t,u})_{\mathcal{S}_{t,u}^{0c}}||_{1}\leq & \frac{LC_{\lambda }+1%
}{lC_{\lambda }-1}||\widehat{\Psi }_{t,u,0}^{-1}||_{\infty }||\widehat{\Psi }%
_{t,u,0}(\delta _{t,u})_{\mathcal{S}^0_{t,u}}||_{1}+\frac{nC_{\lambda }||%
\widehat{\Psi }_{t,u,0}^{-1}||_{\infty }}{\lambda (C_{\lambda }l-1)}||\frac{%
r_{t,u}^{\phi }K(\frac{T-t}{h_1})^{1/2}}{\omega _{t,u}^{1/2}}%
||_{\mathbb{P}_n,2}\Gamma _{t,u} \\
\leq & \tilde{c}||(\delta _{t,u})_{\mathcal{S}^0_{t,u}}||_{1}+\frac{%
nC_{\lambda }||\widehat{\Psi }_{t,u,0}^{-1}||_{\infty }}{\lambda (C_{\lambda
}l-1)}||\frac{r_{t,u}^{\phi }K(\frac{T-t}{h_1})^{1/2}}{\omega _{t,u}^{1/2}}%
||_{\mathbb{P}_n,2}\Gamma _{t,u}.
\end{align*}%
We will consider two cases: $\delta _{t,u}\notin \Delta _{2\tilde{c},t,u}$
and $\delta _{t,u}\in \Delta _{2\tilde{c},t,u}.$

First, if $\delta _{t,u}\notin \Delta _{2\tilde{c},t,u}$, i.e., $%
||(\delta _{t,u})_{\mathcal{S}_{t,u}^{0c}}||_{1}\geq 2\tilde{c}||(\delta
_{t,u})_{\mathcal{S}^0_{t,u}}||_{1}$, then 
\begin{align*}
||\delta _{t,u}||_{1}\leq & (1+\frac{1}{\tilde{2c}})||(\delta _{t,u})_{%
\mathcal{S}_{t,u}^{0c}}||_{1} \\
\leq & (\tilde{c}+\frac{1}{2})||(\delta _{t,u})_{\mathcal{S}%
^0_{t,u}}||_{1}+(1+\frac{1}{2\tilde{c}})\frac{nC_{\lambda }||\widehat{\Psi }%
_{t,u,0}^{-1}||_{\infty }}{\lambda (C_{\lambda }l-1)}||\frac{r_{t,u}^{\phi
}K(\frac{T-t}{h_1})^{1/2}}{\omega _{t,u}^{1/2}}||_{\mathbb{P}_n,2}\Gamma _{t,u} \\
\leq & (\frac{1}{2}+\frac{1}{4\tilde{c}})||(\delta _{t,u})_{\mathcal{S}%
_{t,u}^{0c}}||_{1}+(1+\frac{1}{2\tilde{c}})\frac{nC_{\lambda }||\widehat{%
\Psi }_{t,u,0}^{-1}||_{\infty }}{\lambda (C_{\lambda }l-1)}||\frac{%
r_{t,u}^{\phi }K(\frac{T-t}{h_1})^{1/2}}{\omega _{t,u}^{1/2}}%
||_{\mathbb{P}_n,2}\Gamma _{t,u} \\
\leq & (\frac{1}{2}+\frac{1}{4\tilde{c}})||\delta _{t,u}||_{1}+(1+\frac{1}{2%
\tilde{c}})\frac{nC_{\lambda }||\widehat{\Psi }_{t,u,0}^{-1}||_{\infty }}{%
\lambda (C_{\lambda }l-1)}||\frac{r_{t,u}^{\phi }K(\frac{T-t}{h_1})^{1/2}}{%
\omega _{t,u}^{1/2}}||_{\mathbb{P}_n,2}\Gamma _{t,u}.
\end{align*}%
Noting that $\tilde{c}\geq 1$, we have 
\begin{align*}
||\delta _{t,u}||_{1}\leq & \biggl[\frac{4\tilde{c}+2}{2\tilde{c}-1}\biggr]%
\frac{nC_{\lambda }||\widehat{\Psi }_{t,u,0}^{-1}||_{\infty }}{\lambda
(C_{\lambda }l-1)}||\frac{r_{t,u}^{\phi }K(\frac{T-t}{h_1})^{1/2}}{\omega
_{t,u}^{1/2}}||_{\mathbb{P}_n,2}\Gamma _{t,u} \\
\leq & 6\frac{nC_{\lambda }||\widehat{\Psi }_{t,u,0}^{-1}||_{\infty }}{%
\lambda (C_{\lambda }l-1)}||\frac{r_{t,u}^{\phi }K(\frac{T-t}{h_1})^{1/2}}{%
\omega _{t,u}^{1/2}}||_{\mathbb{P}_n,2}\Gamma _{t,u}:=I_{t,u}.
\end{align*}%

Now, we consider the case where $\delta _{t,u}\in \Delta _{2\tilde{c}%
,t,u}$. By Lemma \ref{lem:localRE}, we have 
\begin{equation*}
\underline{\kappa }\leq \inf_{(t,u)\in \mathcal{T}\mathcal{U}}\min_{\delta
\in \Delta _{2\tilde{c},t,u}}\frac{||b(X)^{\prime }\delta K(\frac{T-t}{h_1}%
)^{1/2}||_{\mathbb{P}_n,2}}{\sqrt{h_1}||\delta _{\mathcal{S}^0_{t,u}}||_{2}}.
\end{equation*}%
In addition, $\omega _{t,u}\in (\underline{C}(1-\underline{C}),1/4)$. If $%
\delta _{t,u}\in \Delta _{2\tilde{c},t,u}$, then 
\begin{equation*}
||(\delta _{t,u})_{\mathcal{S}^0_{t,u}}||_{1}\leq \frac{\sqrt{s}}{\underline{%
\kappa }\sqrt{h_1}\omega _{t,u}^{1/2}}\Gamma _{t,u}:=II_{t,u}.
\end{equation*}%
In this case, $||\delta _{t,u}||_{1}\leq (1+2\tilde{c})II_{t,u}$. 

In sum, we have 
\begin{equation}
||\delta _{t,u}||_{1}\leq I_{t,u}+(1+2\tilde{c})II_{t,u}\leq \biggl(6\frac{%
nC_{\lambda }||\widehat{\Psi }_{t,u,0}^{-1}||_{\infty }}{\lambda (C_{\lambda
}l-1)}||\frac{r_{t,u}^{\phi }K(\frac{T-t}{h_1})^{1/2}}{\omega _{t,u}^{1/2}}%
||_{\mathbb{P}_n,2}+\frac{(1+2\tilde{c})\sqrt{s}}{\underline{\kappa }\sqrt{h_1}\omega
_{t,u}^{1/2}}\biggr)\Gamma _{t,u}  \label{eq:delta1}
\end{equation}%
and $\delta _{t,u}\in A_{t,u}:= \Delta _{2\tilde{c},t,u}\cup \{\delta
:||\delta ||_{1}\leq I_{t,u}\}.$ 

Recall $\tilde{r}_{t,u}^{\phi }=\Lambda ^{-1}(\Lambda (b(X)^{\prime
}\theta _{t,u})+r_{t,u}^{\phi })-b(X)^{\prime }\theta _{t,u}$ and denote 
\begin{equation*}
\overline{q}_{A_{t,u}}=\inf_{\delta \in A_{t,u}}\frac{[\mathbb{P}_{n}\omega
_{t,u}|b(X)^{\prime }\delta |^{2}K(\frac{T-t}{h_1})]^{3/2}}{\mathbb{P}%
_{n}[\omega _{t,u}|b(X)^{\prime }\delta |^{3}K(\frac{T-t}{h_1})]}.
\end{equation*}%
Then, w.p.a.1., for some $\overline{r}_{t,u}^{\phi }$ between $0$ and $%
r_{t,u}^{\phi }$,

\begin{align*}
|\tilde{r}_{t,u}^{\phi }|=& \{[\Lambda (b(X)^{\prime }\theta _{t,u})+%
\overline{r}_{t,u}^{\phi }][1-\Lambda (b(X)^{\prime }\theta _{t,u})-%
\overline{r}_{t,u}^{\phi }]\}^{-1}|r_{t,u}^{\phi }| \\
\in & [4|r_{t,u}^{\phi }|,\{(\underline{C}/2)(1-\underline{C}%
/2)\}^{-1}|r_{t,u}^{\phi }|],
\end{align*}%
where the second line holds because $\sup_{(t,u)\in \mathcal{T}\mathcal{U}%
}||r_{t,u}^{\phi }||_{\mathbb{P}, \infty }\overset{p}{\longrightarrow }0$. In
addition, by Lemma \ref{lem:minor} below and equations (\ref{B.1})--%
\eqref{eq:delta1}, we have 
\begin{align*}
& \min (\frac{1}{3}\Gamma _{t,u}^{2},\frac{\overline{q}_{A_{t,u}}}{3}\Gamma
_{t,u}) \\
\leq & Q_{t,u}(\theta _{t,u}+\delta _{t,u})-Q_{t,u}(\theta _{t,u})-\partial
_{\theta }Q_{t,u}(\theta _{t,u})^{\prime }\delta _{t,u}+2||\frac{\tilde{r}%
_{t,u}^{\phi }K(\frac{T-t}{h_1})^{1/2}}{\omega _{t,u}^{1/2}}||_{\mathbb{P}_n,2}\Gamma
_{t,u} \\
\leq & \frac{\lambda }{n}(L+\frac{1}{C_{\lambda }})||\widehat{\Psi }%
_{t,u,0}(\delta _{t,u})_{\mathcal{S}^0_{t,u}}||_{1}-\frac{\lambda }{n}(l-%
\frac{1}{C_{\lambda }})||\widehat{\Psi }_{t,u,0}(\delta _{t,u})_{\mathcal{S}%
_{t,u}^{0c}}||_{1}+3||\frac{\tilde{r}_{t,u}^{\phi }K(\frac{T-t}{h_1})^{1/2}}{%
\omega _{t,u}^{1/2}}||_{\mathbb{P}_n,2}\Gamma _{t,u} \\
\leq & \frac{\lambda }{n}(L+\frac{1}{C_{\lambda }})||\widehat{\Psi }%
_{t,u,0}||_{\infty }||\delta _{t,u}||_{1}+3||\frac{\tilde{r}_{t,u}^{\phi }K(%
\frac{T-t}{h_1})^{1/2}}{\omega _{t,u}^{1/2}}||_{\mathbb{P}_n,2}\Gamma _{t,u} \\
\leq & \biggl(9\tilde{c}||\frac{\tilde{r}_{t,u}^{\phi }K(\frac{T-t}{h_1})^{1/2}%
}{\omega _{t,u}^{1/2}}||_{\mathbb{P}_n,2}+\frac{\lambda }{n}(L+\frac{1}{C_{\lambda }%
})||\widehat{\Psi }_{t,u,0}||_{\infty }\frac{(1+2\tilde{c})\sqrt{s}}{%
\underline{\kappa }\sqrt{h_1}}\biggr)\Gamma _{t,u},
\end{align*}%
where the last inequality holds because $|r_{t,u}^{\phi }|\leq |\tilde{r}%
_{t,u}^{\phi }|$. If 
\begin{equation}
\overline{q}_{A_{u,r}}>3\biggl\{9\tilde{c}||\frac{\tilde{r}_{t,u}^{\phi }K(%
\frac{T-t}{h_1})^{1/2}}{\omega _{t,u}^{1/2}}||_{\mathbb{P}_n,2}+\frac{\lambda }{n}(L+%
\frac{1}{C_{\lambda }})||\widehat{\Psi }_{t,u,0}||_{\infty }\frac{(1+2\tilde{%
c})\sqrt{s}}{\underline{\kappa }\sqrt{h_1}}\biggr\},  \label{eq:q}
\end{equation}%
then

\begin{equation}
\Gamma_{t,u} \leq 3\biggl\{9 \tilde{c}||\frac{\tilde{r}^\phi_{t,u}K(\frac{T-t%
}{h_1})^{1/2}}{\omega_{t,u}^{1/2}}||_{\mathbb{P}_n,2} + \frac{\lambda}{n}(L + \frac{1}{%
C_\lambda})||\widehat{\Psi}_{t,u,0}||_\infty \frac{(1+2\tilde{c})\sqrt{s}}{%
\underline{\kappa}\sqrt{h_1}} \biggr\}  \label{eq:gammarate1}
\end{equation}
and 
\begin{align}
\begin{split}
||\delta_{t,u}||_1 \leq & \biggl(6\frac{nC_\lambda ||\widehat{\Psi}%
^{-1}_{t,u,0}||_\infty}{\lambda(C_\lambda l-1)}||\frac{r^\phi_{t,u}K(\frac{%
T-t}{h_1})^{1/2}}{\omega_{t,u}^{1/2}}||_{\mathbb{P}_n,2} + \frac{(1+2\tilde{c})\sqrt{s}%
}{\underline{\kappa}\sqrt{h_1}} \biggr) \\
& \times 3\biggl\{9 \tilde{c}||\frac{\tilde{r}^\phi_{t,u}K(\frac{T-t}{h_1}%
)^{1/2}}{\omega_{t,u}^{1/2}}||_{\mathbb{P}_n,2} + \frac{\lambda}{n}(L + \frac{1}{%
C_\lambda})||\widehat{\Psi}_{t,u,0}||_\infty \frac{(1+2\tilde{c})\sqrt{s}}{%
\underline{\kappa}\sqrt{h_1}} \biggr\}.  \label{eq:gammarate2}
\end{split}%
\end{align}
Since $E_1$ holds, 
\begin{equation*}
\sup_{(t,u) \in \mathcal{T} \mathcal{U}}||\frac{\tilde{r}^\phi_{t,u}K(\frac{%
T-t}{h_1})^{1/2}}{\omega_{t,u}^{1/2}}||_{\mathbb{P}_n,2} \leq [\underline{C}/2(1-%
\underline{C}/2)]^{-1}C_r(\sqrt{\frac{\log(p \vee n)s}{n}}).
\end{equation*}
Further note that $\lambda = \ell_n(\log(p \vee n)nh_1)^{1/2}$. Hence, if %
\eqref{eq:q} holds, then \eqref{eq:gammarate1} and \eqref{eq:gammarate2}
imply that 
\begin{equation*}
\sup_{(t,u) \in \mathcal{T} \mathcal{U}}\Gamma_{t,u} \leq C_\Gamma \ell_n
(\log(p \vee n)s)^{1/2}n^{-1/2}
\end{equation*}
with $C_{\Gamma} = 3(9 \tilde{c}[\underline{C}/2(1-\underline{C}/2)]^{-1}C_r
+(LC_\lambda + 1)2C_\psi(1+2 \tilde{c})/\underline{\kappa})$ and 
\begin{equation*}
\sup_{(t,u) \in \mathcal{T} \mathcal{U}}||\delta_{t,u}||_1 \leq C_1 \ell_n
(\log(p \vee n)s^2)^{1/2}(nh_1)^{-1/2}
\end{equation*}
with $C_1 = \frac{2(1+2\tilde{c})}{\underline{\kappa}}C_\Gamma$, which are
the desired results.

Last, we verify \eqref{eq:q}. By Lemma \ref{lem:minor}, since $\ell
_{n}^{2}\log (p\vee n)s^{2}\zeta _{n}^{2}/(nh_1)\rightarrow 0$, 
\begin{equation*}
\frac{\overline{q}_{A_{u,r}}}{3\biggl\{9\tilde{c}||\frac{\tilde{r}%
_{t,u}^{\phi }K(\frac{T-t}{h_1})^{1/2}}{\omega _{t,u}^{1/2}}||_{\mathbb{P}_n,2}+\frac{%
\lambda }{n}(L+\frac{1}{C_{\lambda }})||\widehat{\Psi }_{t,u,0}||_{\infty }%
\frac{(1+2\tilde{c})\sqrt{s}}{\underline{\kappa }\sqrt{h_1}}\biggr\}}\geq c%
\sqrt{\frac{nh_1}{\log (p\vee n)s^{2}\zeta _{n}^{2}\ell _{n}^{2}}}\rightarrow
\infty .
\end{equation*}%
This concludes the proof.  
\end{proof}

\bigskip

\begin{proof}[Proof of Lemma \protect\ref{lem:E2}]
By Lemma \ref{lem:E32} below, $\widehat{\Psi }_{t,u}^{-1}$ is
bounded away from zero w.p.a.1, uniformly over $(t,u)$. Therefore, we can
just focus on bounding 
\begin{equation*}
\sup_{(t,u)\in \mathcal{T}\mathcal{U}}\biggl|\biggl|\mathbb{P}_{n}\biggl[\xi
_{t,u}K(\frac{T-t}{h_1})b(X)\biggr]\biggr|\biggr|_{\infty }.
\end{equation*}%
For $j$-th element, $1\leq j\leq p$, 
\begin{equation*}
|\mathbb{E}[\xi _{t,u}K(\frac{T-t}{h_1})b_{j}(X)]|\leq c\mathbb{E}%
|b_{j}(X)|h_1^{3}\leq c||b_{j}(X)||_{P,2}h_1^{3}\leq ch_1^{3}.
\end{equation*}%
where $c$ is a universal constant independent of $(j,t,u,n)$. In addition, 
\begin{equation*}
nh_1^{3}/(\log (p\vee n)h_1n)^{1/2}=(nh_1^{5}/\log (p\vee n))^{1/2}\rightarrow 0.
\end{equation*}%
Therefore, 
\begin{equation*}
\sup_{(t,u)\in \mathcal{T}\mathcal{U}}||\mathbb{E}[\xi _{t,u}K(\frac{T-t}{h_1}%
)b(X)]||_{\infty } = o((\log (p\vee n)h_1/n)^{1/2} ).
\end{equation*}
Next, We turn to the centered term: $\sup_{g\in \mathcal{G}}|(\mathbb{P}_{n}-%
\mathbb{P})g|,$ where $\mathcal{G}=\{\xi _{t,u}b_{j}(X)K(\frac{T-t}{h_1}%
):(t,u)\in \mathcal{T}\mathcal{U},1\leq j\leq p\}$ with envelope $G=%
\overline{C}_{K}\zeta _{n}$. Note that $\sup_{g\in \mathcal{G}}\mathbb{E}%
g^{2}\lesssim h_1$ and $\sup_{Q}N(\mathcal{G},e_{Q},\varepsilon ||G||)\leq p%
\biggl(\frac{A}{\varepsilon }\biggr)^{v}$ for some $A>e$ and $v>0$. So by
Corollary 5.1 of \cite{CCK14}, we have 
\begin{equation*}
\mathbb{E}\sup_{g\in \mathcal{G}}|(\mathbb{P}_{n}-\mathbb{P})g|\leq (\log
(p\vee n)h_1/n)^{1/2}+\log (p\vee n)\zeta _{n}/n\lesssim (\log (p\vee
n)h_1/n)^{1/2}
\end{equation*}%
because $\log (p\vee n)\zeta _{n}^{2}/(nh_1)\rightarrow 0$.  
\end{proof}

\bigskip

\begin{proof}[Proof of Lemma \protect\ref{lem:E32}]
For the first result, we have
\begin{equation*}
\mathbb{E}(Y_{u}-\phi _{t,u}(X))^{2}b_{j}^{2}(X)K(\frac{T-t}{h_1})^{2}h_1^{-1}.
\end{equation*}%
Let $\kappa _{1}=\int K(u)^{2}du$. Then, 
\begin{align*}
& \mathbb{E}(Y_{u}-\phi _{t,u}(X))^{2}b_{j}^{2}(X)K(\frac{T-t}{h_1})^{2}h_1^{-1}
\\
=& \mathbb{E}\int \biggl[\phi _{t+h_1v,u}(X)-2\phi _{t+h_1v,u}(X)\phi
_{t,u}(X)+\phi _{t,u}^{2}(X)\biggr]f_{t+h_1v}(X)K(v)^{2}dvb_{j}^{2}(X) \\
\geq & \underline{C}\mathbb{E}\int \biggl[\phi _{t,u}(X)(1-\phi
_{t,u}(X))-h_1|\partial _{t}\phi _{\tilde{t},u}(X)v|\biggr]%
K(v)^{2}dvb_{j}^{2}(X) \\
\geq & \kappa _{1}\underline{C}^{2}(1-\underline{C})\mathbb{E}%
b_{j}^{2}(X)/2\geq C_{\psi }.
\end{align*}%
Similarly, 
\begin{align*}
& \mathbb{E}(Y_{u}-\phi _{t,u}(X))^{2}b_{j}^{2}(X)K(\frac{T-t}{h_1})^{2}h_1^{-1}
\\
=& \mathbb{E}\int \biggl[\phi _{t+h_1v,u}(X)-2\phi _{t+h_1v,u}(X)\phi
_{t,u}(X)+\phi _{t,u}^{2}(X)\biggr]f_{t+h_1v}(X)K(v)^{2}dvb_{j}^{2}(X) \\
\leq & \underline{C}\mathbb{E}\int \biggl[\phi _{t,u}(X)(1-\phi
_{t,u}(X))+h_1|\partial _{t}\phi _{\tilde{t},u}(X)v|\biggr]%
K(v)^{2}dvb_{j}^{2}(X) \\
\leq & 2\kappa _{1}\underline{C}\mathbb{E}b_{j}^{2}(X)\leq 1/C_{\psi }.
\end{align*}

In addition, denote $\mathcal{F}=\{\frac{1}{h_1}K(\frac{%
		T-t}{h_1})^2(Y_u - \phi_{t,u}(X))^2b_{j}^{2}(X):(t,u)\in \mathcal{TU},j=1,\cdots ,p\}$ with
	envelope $C\zeta _{n}^{2}/h_1$. The entropy of $\mathcal{F}$ is bounded by $p(\frac{A}{%
		\varepsilon })^{v}$. In addition, $\sup_{f\in \mathcal{F}}\mathbb{E}%
	f^{2}\lesssim \zeta _{n}^{2}/h_1$. Therefore, 
	\begin{equation*}
	||\mathbb{P}_{n}-\mathbb{P}||_{\mathcal{F}}^{2}\lesssim O_{p}(\log (p\vee
	n)\zeta _{n}^{2}/(nh_1)) = o_p(1).
	\end{equation*}%
	Therefore, w.p.a.1, 
	\begin{align*}
	C_{\psi}/2\leq & \inf_{(t,u)\in \mathcal{TU}, j=1,\cdots,p}\mathbb{P}_n(Y_{u}-\phi _{t,u}(X))^{2}b_{j}^{2}(X)K(\frac{T-t}{h_1})^{2}h_1^{-1}\\
	\leq & \sup_{(t,u)\in \mathcal{TU}, j=1,\cdots,p}\mathbb{P}_n(Y_{u}-\phi _{t,u}(X))^{2}b_{j}^{2}(X)K(\frac{T-t}{h_1})^{2}h_1^{-1}\leq 2/C_{\psi}.
	\end{align*}%

For $k=0$, we let $\mathcal{F}=\{\frac{1}{h_1}K(\frac{%
	T-t}{h_1})^2Y^2_ub_{j}^{2}(X):(t,u)\in \mathcal{TU},j=1,\cdots ,p\}$ with envelope $C\zeta_n^2/h_1$. By the same argument as
	above, we can show that, w.p.a.1, 
	\begin{align*}
C_{0}/2\leq & \inf_{(t,u)\in \mathcal{TU}, j=1,\cdots,p}\mathbb{P}_nY^2_{u}b_{j}^{2}(X)K(\frac{T-t}{h_1})^{2}h_1^{-1}\\
\leq & \sup_{(t,u)\in \mathcal{TU}, j=1,\cdots,p}\mathbb{P}_nY^2_{u}b_{j}^{2}(X)K(\frac{T-t}{h_1})^{2}h_1^{-1}\leq 2/C_{0}.
\end{align*}%
	For $k\geq 1$, we have, w.p.a.1, 
	\begin{align*}
	& \sup_{{t,u}\in \mathcal{TU},j=1,\cdots ,p}\mathbb{P}_{n}(Y_{u}-\hat{\phi}^{k-1} _{t,u}(X))^{2}b_{j}^{2}(X)K(\frac{T-t}{h_1})^{2}h_1^{-1} \\
	\leq & 1.5 \sup_{{t,u}\in \mathcal{TU},j=1,\cdots ,p}\mathbb{P}_{n}(Y_{u}-\phi _{t,u}(X))^{2}b_{j}^{2}(X)K(\frac{T-t}{h_1})^{2}h_1^{-1}\\
	& +3\sup_{{t,u}\in \mathcal{TU},j=1,\cdots
		,p}\mathbb{P}_{n}(\hat{\phi}^{k-1} _{t,u}(X)-\phi _{t,u}(X))^{2}b_{j}^{2}(X)K(\frac{T-t}{h_1})^{2}h_1^{-1} \\
	\leq & 1.5\sup_{{t,u}\in \mathcal{TU},j=1,\cdots ,p}l_{t,u,0,j}^{2}+o_p(1) \\
	\leq & 2/C_{\psi}.
	\end{align*}%
	Similarly, we can show that w.p.a.1. 
	$$\inf_{{t,u}\in \mathcal{TU},j=1,\cdots ,p}\mathbb{P}_{n}(Y_{u}-\hat{\phi}^{k-1} _{t,u}(X))^{2}b_{j}^{2}(X)K(\frac{T-t}{h_1})^{2}h_1^{-1} \geq C_{\psi}/2.$$ 
	This concludes the second result with $C_{k}=C_{\psi}$ for $k=1,\cdots ,K$. The last result holds with $l=\min (C_{0}C_{\psi}/4,\cdots ,C_{k}C_{\psi}/4,1)$ and $L=\max
	(4/(C_{0}C_{\psi}),\cdots ,4/(C_{k}C_{\psi}),1).$ 
\end{proof}

\bigskip

\begin{proof}[Proof of Lemma \protect\ref{lem:localRE2}]
Following the same arguments as used in the proof of Lemma \ref%
{lem:localRE} and by Assumption \ref{ass:rate2}, we have, w.p.a.1, 
\begin{align*}
& \sup_{t\in \mathcal{T},||\delta ||_{2}=1,||\delta ||_{0}\leq s\ell
_{n}}||b(X)^{\prime }\delta K(\frac{T-t}{h_1})^{1/2}||_{\mathbb{P}_n,2}^{2} \\
\leq & \sup_{t\in \mathcal{T},||\delta ||_{2}=1,||\delta ||_{0}\leq s\ell
_{n}}|(\mathbb{P}_{n}-\mathbb{P})(b(X)^{\prime }\delta )^{2}K(\frac{T-t}{h_1}%
)| +\sup_{t\in \mathcal{T},||\delta ||_{2}=1,||\delta ||_{0}\leq s\ell _{n}}|%
\mathbb{P}(b(X)^{\prime }\delta )^{2}K(\frac{T-t}{h_1})| \\
\leq & O_{p}(h_1\pi _{n1})+\underline{C}^{-1}h_1\sup_{t\in \mathcal{T},||\delta
||_{2}=1,||\delta ||_{0}\leq s\ell _{n}}|\mathbb{P}(b(X)^{\prime }\delta
)^{2}| \\
\leq & o_{p}(h_1)+\underline{C}^{-1}h_1(\sup_{t\in \mathcal{T},||\delta
||_{2}=1,||\delta ||_{0}\leq s\ell _{n}}|(\mathbb{P}_{n}-\mathbb{P}%
)(b(X)^{\prime }\delta )^{2}|+\sup_{t\in \mathcal{T},||\delta
||_{2}=1,||\delta ||_{0}\leq s\ell _{n}}|\mathbb{P}_{n}(b(X)^{\prime }\delta
)^{2}|) \\
\leq & o_{p}(h_1)+\underline{C}^{-1}h_1(O_{p}(\pi _{n2})+\kappa ^{^{\prime
\prime }2}) \\
\leq & 2\underline{C}^{-1}\kappa ^{^{\prime \prime }2}h_1,
\end{align*}%
where the second inequality holds because 
\begin{equation*}
\mathbb{E}(b(X)^{\prime }\delta )^{2}K(\frac{T-t}{h_1})=\mathbb{E}%
(b(X)^{\prime }\delta )^{2}\int f_{t+h_1u}(X)K(u)du\leq \frac{\mathbb{E}%
(b(X)^{\prime }\delta )^{2}}{\underline{C}}.
\end{equation*}%
\end{proof}

\bigskip

\begin{lem}
Recall that $Q_{t,u}(\theta )=\mathbb{P}_{n}M(Y_{u},X;\theta )K(%
\frac{T-t}{h_1})$. Let $\overline{q}_{A_{t,u}}=\inf_{\delta \in A_{t,u}}\frac{[%
\mathbb{P}_{n}\omega _{t,u}|b(X)^{\prime }\delta |^{2}K(\frac{T-t}{h_1})]^{3/2}%
}{\mathbb{P}_{n}\omega _{t,u}|b(X)^{\prime }\delta |^{3}K(\frac{T-t}{h_1})},$ $%
\Gamma _{t,u}^{\delta }=||\omega _{t,u}^{1/2}b(X)^{\prime }\delta K(\frac{T-t%
}{h_1})^{1/2}||_{\mathbb{P}_n,2}$, and $s_{t,u}=||\theta _{t,u}||_{0}$. Let events $%
E_{1}$, $E_{2}$, and $E_{3}$ defined in the proof of Lemma \ref{lem:main}
hold. Then, for any $(t,u)\in \mathcal{T}\mathcal{U}$ and $\delta \in A_{t,u}
$, we have 
\begin{align*}
F_{t,u}(\delta ):=Q_{t,u}(\theta _{t,u}+\delta )& -Q_{t,u}(\theta
_{t,u})-\partial _{\theta }Q_{t,u}(\theta _{t,u})^{\prime }\delta +2||\frac{%
\tilde{r}_{t,u}^{\phi }K(\frac{T-t}{h_1})^{1/2}}{\omega _{t,u}^{1/2}}%
||_{\mathbb{P}_n,2}\Gamma _{t,u}^{\delta } \\
\geq & \min (\frac{1}{3}||\omega _{t,u}^{1/2}b(X)^{\prime }\delta K(\frac{T-t%
}{h_1})^{1/2}||_{\mathbb{P}_n,2}^{2},\frac{1}{3}\overline{q}_{A_{t,u}}\Gamma
_{t,u}^{\delta })
\end{align*}%
and w.p.a.1, 
\begin{equation*}
\overline{q}_{A_{t,u}}\geq \frac{1}{\zeta _{n}}\min \biggl(\frac{\underline{%
\kappa }\sqrt{h_1}}{\sqrt{s_{t,u}}(1+2\tilde{c})},\frac{(\lambda
/n)(lC_{\lambda }-1)}{6c||\widehat{\Psi }_{t,u,0}^{-1}||_{\infty }||\frac{%
r_{t,u}^{\phi }K(\frac{T-t}{h_1})^{1/2}}{\omega _{t,u}^{1/2}}||_{\mathbb{P}_n,2}}%
\biggr).
\end{equation*}%
\label{lem:minor} 
\end{lem}

\bigskip

\begin{proof}
The proof follows closely from that of Lemma O.2 in \cite{BCFH13}.
Note that 
\begin{equation*}
Q_{t,u}(\theta _{t,u}+\delta )-Q_{t,u}(\theta _{t,u})-\partial _{\theta
}Q_{t,u}(\theta _{t,u})^{\prime }\delta =\mathbb{P}_{n}[\tilde{g}_{t,u}(1)-%
\tilde{g}_{t,u}(0)-\tilde{g}_{t,u}^{\prime }(0)],
\end{equation*}%
where $\tilde{g}_{t,u}(s)=\log [1+\exp (b(X)^{\prime }(\theta _{t,u}+s\delta
))]K(\frac{T-t}{h_1})$. Let $g_{t,u}(s)=\log [1+\exp (b(X)^{\prime }(\theta
_{t,u}+s\delta )+\tilde{r}_{t,u}^{\phi })]K(\frac{T-t}{h_1})$. Then 
\begin{equation*}
g_{t,u}^{\prime }(0)=(b(X)^{\prime }\delta )\mathbb{E}(Y_{u}|X,T=t)K(\frac{%
T-t}{h_1}),
\end{equation*}%
\begin{equation*}
g_{t,u}^{^{\prime \prime }}(0)=(b(X)^{\prime }\delta )^{2}\mathbb{E}%
(Y_{u}|X,T=t)(1-\mathbb{E}(Y_{u}|X,T=t))K(\frac{T-t}{h_1}),
\end{equation*}%
and 
\begin{equation*}
g_{t,u}^{^{\prime \prime \prime }}(0)=(b(X)^{\prime }\delta )^{3}\mathbb{E}%
(Y_{u}|X,T=t)(1-\mathbb{E}(Y_{u}|X,T=t))(1-2\mathbb{E}(Y_{u}|X,T=t))K(\frac{%
T-t}{h_1}).
\end{equation*}%
By Lemmas O.3 and O.4 in \cite{BCFH13}, 
\begin{equation*}
g_{t,u}(1)-g_{t,u}(0)-g_{t,u}^{\prime }(0)\geq \omega _{t,u}K(\frac{T-t}{h_1})%
\biggl[\frac{(b(X)^{\prime }\delta )^{2}}{2}-\frac{|b(X)^{\prime }\delta
|^{3}}{6}\biggr].
\end{equation*}%
Let $\Upsilon _{t,u}(s)=\tilde{g}_{t,u}(s)-g_{t,u}(s).$ Then 
\begin{equation*}
|\Upsilon _{t,u}^{\prime }(s)|\leq |\omega _{t,u}^{1/2}b(X)^{\prime }\delta
K(\frac{T-t}{h_1})^{1/2}|\biggl|\frac{\tilde{r}_{t,u}^{\phi }K(\frac{T-t}{h_1}%
)^{1/2}}{\omega _{t,u}^{1/2}}\biggr|.
\end{equation*}%
It follows that 
\begin{align*}
& \mathbb{P}_{n}|\tilde{g}_{t,u}(1)-g_{t,u}(1)-(\tilde{g}%
_{t,u}(0)-g_{t,u}(0))-(\tilde{g}_{t,u}^{\prime }(0)-g_{t,u}^{\prime }(0))| \\
=& \mathbb{P}_{n}|\Upsilon _{t,u}(1)-\Upsilon _{t,u}(0)-\Upsilon
_{t,u}^{\prime }(0)| \\
\leq & 2\mathbb{P}_{n}|\omega _{t,u}^{1/2}b(X)^{\prime }\delta K(\frac{T-t}{h_1%
})^{1/2}|\biggl|\frac{\tilde{r}_{t,u}^{\phi }K(\frac{T-t}{h_1})^{1/2}}{\omega
_{t,u}^{1/2}}\biggr| \\
\leq & 2\Gamma _{t,u}^{\delta }\biggl|\biggl|\frac{\tilde{r}_{t,u}^{\phi }K(%
\frac{T-t}{h_1})^{1/2}}{\omega _{t,u}^{1/2}}\biggr|\biggr|_{\mathbb{P}_n,2},
\end{align*}%
and 
\begin{equation*}
F_{t,u}(\delta )\geq \frac{1}{2}\mathbb{P}_{n}\omega _{t,u}(b(X)^{\prime
}\delta )^{2}K(\frac{T-t}{h_1})-\frac{1}{6}\mathbb{P}_{n}\omega
_{t,u}|b(X)^{\prime }\delta |^{3}K(\frac{T-t}{h_1}).
\end{equation*}%
We consider two cases: $\Gamma _{t,u}^{\delta }\leq \overline{q}_{A_{t,u}}$
and $\Gamma _{t,u}^{\delta }>\overline{q}_{A_{t,u}}.$

First, if $\Gamma _{t,u}^{\delta }\leq \overline{q}_{A_{t,u}}$, we
have 
\begin{equation*}
\mathbb{P}_{n}\omega _{t,u}|b(X)^{\prime }\delta |^{3}K(\frac{T-t}{h_1})\leq
||\omega _{t,u}^{1/2}b(X)^{\prime }\delta K(\frac{T-t}{h_1}%
)^{1/2}||_{\mathbb{P}_n,2}^{2}
\end{equation*}%
and 
\begin{equation*}
F_{t,u}(\delta )\geq \frac{1}{3}(\Gamma _{t,u}^{\delta })^{2}.
\end{equation*}%
When $\Gamma _{t,u}^{\delta }>\overline{q}_{A_{t,u}}$, we let $\tilde{\delta}%
=\delta \overline{q}_{A_{t,u}}/\Gamma _{t,u}^{\delta }\in A_{t,u}$. Then by
the convexity of $F_{t,u}(\delta )$ and the fact that $F_{t,u}(0)=0$, we
have 
\begin{equation*}
F_{t,u}(\delta )\geq \frac{\Gamma _{t,u}^{\delta }}{\overline{q}_{A_{t,u}}}%
F_{t,u}(\tilde{\delta})\geq \frac{\Gamma _{t,u}^{\delta }}{\overline{q}%
_{A_{t,u}}}\biggl(\frac{1}{3}||\omega _{t,u}^{1/2}b(X)^{\prime }\tilde{\delta%
}K(\frac{T-t}{h_1})^{1/2}||_{\mathbb{P}_n,2}^{2}\biggr)=\frac{1}{3}\overline{q}%
_{A_{t,u}}\Gamma _{t,u}^{\delta }.
\end{equation*}%
Consequently, we have $F_{t,u}(\delta )\geq \min (\frac{1}{3}(\Gamma
_{t,u}^{\delta })^{2},\frac{\overline{q}_{A_{t,u}}}{3}\Gamma _{t,u}^{\delta
}).$

For the second result, note that 
\begin{equation*}
\overline{q}_{A_{t,u}}\geq \inf_{\delta \in A_{t,u}}\frac{||\omega
_{t,u}^{1/2}b(X)^{\prime }\delta K(\frac{T-t}{h_1})^{1/2}||_{\mathbb{P}_n,2}}{\zeta
_{n}||\delta ||_{1}}.
\end{equation*}%
If $\delta \in \Delta _{2\tilde{c},t,u}$, then by Lemma 3.1 
\begin{equation*}
\frac{\Gamma _{t,u}^{\delta }}{\zeta _{n}||\delta ||_{1}}\geq \frac{||\omega
_{t,u}^{1/2}b(X)^{\prime }\delta K(\frac{T-t}{h_1})^{1/2}||_{\mathbb{P}_n,2}}{\zeta
_{n}||\delta _{\mathcal{S}^0_{t,u}}||_{2}(1+2\tilde{c})s_{t,u}^{1/2}}\geq 
\frac{1}{\zeta _{n}}\frac{\underline{\kappa }\sqrt{h_1}}{\sqrt{s_{t,u}}(1+2%
\tilde{c})}.
\end{equation*}%
If $||\delta ||_{1}\leq I_{t,u}$, where $I_{t,u}$ is defined in the proof of
Lemma \ref{lem:main}, then \ 
\begin{equation*}
\frac{\Gamma _{t,u}^{\delta }}{\zeta _{n}||\delta ||_{1}}\geq \frac{||\omega
_{t,u}^{1/2}b(X)^{\prime }\delta K(\frac{T-t}{h_1})^{1/2}||_{\mathbb{P}_n,2}}{\zeta
_{n}I_{t,u}}\geq \frac{1}{\zeta _{n}}\frac{(\lambda /n)(lC_{\lambda }-1)}{%
6c||\widehat{\Psi }_{t,u,0}^{-1}||_{\infty }||\frac{r_{t,u}^{\phi }K(\frac{%
T-t}{h_1})^{1/2}}{\omega _{t,u}^{1/2}}||_{\mathbb{P}_n,2}}.
\end{equation*}%
Combining the above two results, we obtain that 
\begin{equation*}
\overline{q}_{A_{t,u}}\geq \frac{1}{\zeta _{n}}\min \biggl(\frac{\underline{%
\kappa }\sqrt{h_1}}{\sqrt{s_{t,u}}(1+2\tilde{c})},\frac{(\lambda
/n)(lC_{\lambda }-1)}{6c||\widehat{\Psi }_{t,u,0}^{-1}||_{\infty }||\frac{%
r_{t,u}^{\phi }K(\frac{T-t}{h_1})^{1/2}}{\omega _{t,u}^{1/2}}||_{\mathbb{P}_n,2}}%
\biggr).
\end{equation*}%

\end{proof}

\bigskip

\begin{lem}
Let $q_{y}(t)$ be the $y$-th quantile of $Y(t)$, $f_{Y(t)}(\cdot )$
the unconditional density of $Y(t)$, 
\begin{equation*}
F(t,y)=\int_{0}^{1}1\{\alpha (t,\psi ^{\leftarrow }(v))\leq y\}dv,\quad
F(t,y|d_{n})=\int_{0}^{1}1\{\hat{\alpha}^{\ast }(t,\psi ^{\leftarrow
}(v))\leq y\}dv,
\end{equation*}%
$s_{n}=(nh_2)^{-1/2}$, $d_{n}(t,v)=(nh_2)^{1/2}(\hat{\alpha}^{\ast }(t,\psi
^{\leftarrow }(v))-\alpha (t,\psi ^{\leftarrow }(v))),$ and $J_{n}(t,y)=%
\frac{F(t,y|d_{n})-F(t,y)}{s_{n}}.$ Then, for $\delta _{n}$ being either $1$
or $h_2^{1/2}$, depending on either Assumption \ref{ass:rate2}.1 or \ref%
{ass:rate2}.2 is in place, 
\begin{equation}
\frac{F(t,y|d_{n})-F(t,y)}{s_{n}}+\frac{d_{n}(t,\psi (q_{y}(t)))\psi
^{\prime }(q_{y}(t))}{f_{Y(t)}(q_{y}(t))}=o_{p}(\delta _{n})  \label{eq:cdf3}
\end{equation}%
and 
\begin{equation}
\frac{\hat{\alpha}^{\ast r}(t,u)-\alpha (t,u)}{s_{n}}+\frac{F(t,\alpha
(t,u)|d_{n})-F(t,\alpha (t,u))f_{Y(t)}(u)}{s_{n}\psi ^{\prime }(u)}%
=o_{p}(\delta _{n}).  \label{eq:inv2}
\end{equation}%
uniformly over $(t,y)\in \{(t,y):y=\alpha (t,\psi ^{\leftarrow
}(v)),(t,v)\in \mathcal{T}\times \lbrack 0,1]\}.$ \label{lem:verify}
\end{lem}

\bigskip

\begin{proof}
Let $Q(t,v)=\alpha (t,\psi ^{\leftarrow }(v))$ for $v\in \lbrack
0,1]$. Then, we have 
\begin{equation*}
F(t,y)=\int_{0}^{1}1\{Q(t,v)\leq y\}dv\quad \text{and}\quad
F(t,y|d_{n})=\int_{0}^{1}1\{Q(t,v)+s_{n}d_{n}\leq y\}dv.
\end{equation*}%
We prove the lemma by applying Propositions \ref{lem:cdf} and \ref{lem:inv}
in Appendix C??.

First, we verify Assumption \ref{ass:cdf} with $(\delta
_{n},\varepsilon _{n})=(1,(nh_2)^{-1/2}\log (n))$ and $(\delta
_{n},\varepsilon _{n})=(h_2^{1/2},(nh_2)^{-1/2}\log (n))$ under Assumptions \ref%
{ass:rate2}.1 and \ref{ass:rate2}.2, respectively, in order to apply
Proposition \ref{lem:cdf} to prove (\ref{eq:cdf3}). We only consider the
case in which $\delta _{n}=h_2^{1/2}$ as the $\delta _{n}=1$ case can be
studied similarly. Note that $Q(t,v)=\alpha (t,\psi ^{\leftarrow }(v))$, $%
\partial _{u}\alpha (t,u)=f_{Y(t)}(u)>0$ uniformly over $(t,u)\in \mathcal{T}%
\mathcal{U}$, and $\psi (\cdot )$ can be chosen such that $\partial _{v}\psi
^{\leftarrow }(v)>0$ uniformly over $v\in \lbrack 0,1]$. This verifies
Assumption \ref{ass:cdf}.1. 

For Assumption \ref{ass:cdf}.2, by Theorem \ref{thm:2nd}, $%
\sup_{(t,v)\in \mathcal{T}\times \lbrack 0,1]}|d_{n}(t,v)|=O_{p}(\log
^{1/2}(n))$. So we can take $\varepsilon _{n}=(nh_2)^{-1/2}\log (n)$. In
addition, $\sup_{(t,v)\in \mathcal{T}\times \lbrack
0,1]}|d_{n}^{2}(t,v)|s_{n}=O_{p}(\log (n)(nh_2)^{-1/2})=o_{p}(h_2^{1/2})$
because $nh_2^{2}/\log ^{2}(n)$ $\rightarrow \infty $. So we only need to show 
\begin{equation}
\sup_{(t,v,v^{\prime })\in \mathcal{T}\times \lbrack 0,1]^{2},|v-v^{\prime
}|\leq \varepsilon _{n}}|d_{n}(t,v)-d_{n}(t,v^{\prime })|=o_{p}(h_2^{1/2}).
\label{eq:h}
\end{equation}%
Let 
\begin{equation*}
\mathcal{G}=%
\begin{Bmatrix}
& \eta \Pi _{t,u}(W_{u},\phi _{t,u},f_{t})-\Pi _{t,u^{\prime
}}(W_{u}^{\prime },\phi _{t,u^{\prime }},f_{t}):u=\psi ^{\leftarrow
}(v),u^{\prime }=\psi ^{\leftarrow }(v^{\prime }), \\ 
& (t,v,v^{\prime })\in \mathcal{T}\times \lbrack 0,1]^{2},|v-v^{\prime
}|\leq \varepsilon _{n}%
\end{Bmatrix}%
\end{equation*}%
with envelope $c\eta h_2^{-1}$. By Theorem \ref{thm:2nd}, we have 
\begin{equation*}
d_{n}(t,v)-d_{n}(t,v^{\prime })=(\mathbb{P}_{n}-\mathbb{P})g+R_{n}(t,\psi
^{\leftarrow }(v))-R_{n}(t,\psi ^{\leftarrow }(v^{\prime })).
\end{equation*}%
$\sup_{(t,v)\in \mathcal{T}\times \lbrack 0,1]}R_{n}(t,\psi ^{\leftarrow
}(v))=o_{p}(\delta _{n})$. So we only have to show that 
\begin{equation*}
\sup_{g\in \mathcal{G}}|(\mathbb{P}_{n}-\mathbb{P})g|=o_{p}(h_2^{1/2}).
\end{equation*}%
We know that $\mathcal{G}$ is VC-type with fixed VC index and that $%
\sup_{g\in \mathcal{G}}\mathbb{E}g^{2}\leq \varepsilon _{n}h_2^{-1}.$ In
addition, as shown in the proof of Theorem \ref{thm:q}, $||\max_{1\leq i\leq
n}|\eta _{i}h_2^{-1}|||_{P,2}\leq \log (n)/h_2$. Therefore, by Corollary 5.1 of 
\cite{CCK14}, we have 
\begin{equation*}
(nh_2)^{1/2}||\mathbb{P}_{n}-\mathbb{P}||_{\mathcal{G}}=O_{p}((\log
(n)\varepsilon _{n})^{1/2}).
\end{equation*}%
Given $\varepsilon _{n}=(nh_2)^{-1/2}\log (n)$, $(\log (n)\varepsilon
_{n})^{1/2}=o(h_2^{1/2})$ because $h_2=C_{2}n^{-H_2}$ for some $H_2<1/3$. This
establishes \eqref{eq:h}. Then \eqref{eq:cdf3} follows by Proposition \ref%
{lem:cdf}.

To prove\eqref{eq:inv2}, we apply Proposition \ref{lem:inv} by
verifying Assumption \ref{ass:inv}. We note that $\hat{\alpha}^{\ast
r}(t,u)=F^{\leftarrow }(t,\psi (u)|d_{n})$ and $J_{n}(t,y)=\frac{%
F(t,y|d_{n})-F(t,y)}{s_{n}}$. Furthermore, notice that $\alpha ^{\ast
r}(t,u)=\alpha (t,u)=F^{\leftarrow }(t,\psi (u))$, $F^{\leftarrow
}(t,v)=\alpha (t,\psi ^{\leftarrow }(v))$, 
\begin{equation*}
F(t,y)=\int_{0}^{1}1\{Q(t,v)\leq y\}dv=\int_{0}^{1}1\{v\leq \psi
(q_{y}(t))\}dv=\psi (q_{y}(t)),
\end{equation*}%
and 
\begin{equation*}
\partial _{y}F(t,y)=-\psi ^{\prime }(q_{y}(t))/f_{Y(t)}(q_{y}(t)).
\end{equation*}%
Because $f_{Y(t)}(q_{y}(t))$ is bounded and bounded away from zero uniformly
over $(t,y)\in \mathcal{TY}$, so be $\partial _{y}F(t,y)$. In addition, 
\begin{equation*}
\partial _{yy}^{2}F(t,y)=-f^{\prime
\prime}(q_{y}(t))/f_{Y(t)}^{2}(q_{y}(t))+\phi ^{\prime
}(q_{y}(t))f_{Y(t)}^{\prime }(q_{y}(t))/f_{Y(t)}^{3}(q_{y}(t)),
\end{equation*}%
which is bounded because $f_{Y(t)}^{\prime }(q_{y}(t))$ is bounded. This
verifies Assumption \ref{ass:inv}.2.

For Assumption \ref{ass:inv}.3, we note that 
\begin{equation*}
J_{n}(t,y)=\frac{F(t,y|d_{n})-F(t,y|d_{n})}{s_{n}}=-\frac{d_{n}(t,\psi
(q_{y}(t)))\psi ^{\prime }(q_{y}(t))}{f_{Y(t)}(q_{y}(t))}+o_{p}(\delta _{n}),
\end{equation*}%
where the $o_{p}(\delta _{n})$ is uniform over $(t,y)\in \mathcal{TY}$. In
addition, by definition, $(t,q_{y}(t))\in \mathcal{T}\mathcal{U}$, $%
f_{Y(t)}(q_{y}(t))$ is bounded away from zero, and we can choose $\psi $
such that $\psi ^{\prime }(q_{y}(t))$ is bounded. Therefore, by Theorem \ref%
{thm:2nd} , 
\begin{equation*}
\sup_{(t,y)\in \mathcal{TY}}|J_{n}(t,y)|=O_{p}(\sup_{(t,u)\in \mathcal{T}%
\mathcal{U}}|d_{n}(t,\psi (u))|)+o_{p}(\delta _{n})=O_{p}(\log ^{1/2}(n)).
\end{equation*}%
We can choose $\varepsilon _{n}=s_{n}\log (n)$. In addition, $\sup_{(t,y)\in 
\mathcal{TY}}|J_{n}(t,y)|^{2}s_{n}=o_{p}(h_2^{1/2})$ because $%
nh_2^{3}\rightarrow \infty $. So we only need to show that 
\begin{equation*}
\sup_{(t,y,y^{\prime })\in \mathcal{TYY},|y-y^{\prime }|\leq \max
(\varepsilon _{n},s_{n}\delta _{n})}|J_{n}(t,y)-J_{n}(t,y^{\prime
})|=o_{p}(\delta _{n}).
\end{equation*}%

Note that, for $v=\psi (Q_{Y_{t}}(y))$ and $v^{\prime }=\psi
(Q_{Y_{t}}(y^{\prime }))$ 
\begin{equation*}
|J_{n}(t,y)-J_{n}(t,y^{\prime })|\lesssim |d_{n}(t,v)-d_{n}(t,v^{\prime
})|+o_{p}(\delta _{n}).
\end{equation*}%
In addition, $\phi (Q_{Y_{t}}(y))$ is Lipschitz uniformly over $(t,y)\in 
\mathcal{TY}$. Thus, 
\begin{align*}
& \sup_{(t,y,y^{\prime })\in \mathcal{TYY},|y-y^{\prime }|\leq \max
(\varepsilon _{n},s_{n}\delta _{n})}|J_{n}(t,y)-J_{n}(t,y^{\prime })| \\
\leq & \sup_{(t,v,v^{\prime })\in \mathcal{T}\times \lbrack
0,1]^{2},|v-v^{\prime }|\leq C\varepsilon _{n}}|d_{n}(t,v)-d_{n}(t,v^{\prime
})|=o_{p}(\delta _{n}),
\end{align*}%
given that $h_2=C_{2}n^{-H_2}$ for some $H<1/3$. This completes the verification
of Assumption \ref{ass:inv}.2.

Last, it is essentially the same as above to verify Assumption \ref%
{ass:inv} for $J_{n}(t,u)=(nh_2)^{1/2}(\hat{\alpha}^{\ast r}(t,u)-\alpha
(t,u)) $. The proof is omitted. 
\end{proof}

\bigskip 

\begin{lem}
Suppose the conditions in Theorem \ref{thm:qprime} hold. Then 
\begin{align*}
& e_{2}^{\prime }(G\widehat{\Sigma }_{2})^{-1}U_{n}(t,\tau ) \\
=& \frac{1}{n}\sum_{j=1}^{n}\eta _{j}(\kappa _{2}f_{Y(t)}(q_{\tau
}(t))f_{t}(X_{j})h_2^{2})^{-1}\biggl[Y_{q_{\tau }(t),j}-\phi _{t,q_{\tau
}(t)}(X_{j})\biggr]\overline{K}(\frac{T_{j}-t}{h_2})+o_{p}^{\ast
}((nh_2^{3})^{-1/2}).
\end{align*}%
\label{lem:43}
\end{lem}

\bigskip 

\begin{proof}
Note that 
\begin{equation}
U_{n}(t,\tau )=\frac{2}{n}\sum_{j=1}^{n}\eta _{j}\mathbb{P}\Gamma ^{s}(\cdot
,\Upsilon _{j};t,\tau )+\mathcal{U}_{n}H(\cdot ,\cdot ;t,\tau ),
\label{eq:U}
\end{equation}%
where $\mathcal{U}_{n}$ assigns probability $\frac{1}{n(n-1)}$ to each pair
of observations and 
\begin{equation*}
H(\Upsilon _{i},\Upsilon _{j};t,\tau )=\eta _{i}\eta _{j}\Gamma
^{s}(\Upsilon _{i},\Upsilon _{j};t,\tau )-\eta _{i}\mathbb{P}\Gamma
^{s}(\cdot ,\Upsilon _{j};t,\tau )-\eta _{j}\mathbb{P}\Gamma ^{s}(\Upsilon
_{i},\cdot ;t,\tau )+\mathbb{P}\Gamma ^{s}(\cdot ,\cdot ;t,\tau ).
\end{equation*}%

Let $\mathcal{H}=\{H(\cdot ,\cdot
;t,\tau ),(t,$ $\tau )\in \mathcal{T}\mathcal{I}\}.$ Note that $\mathcal{H}$ is nested by a VC-class with fixed VC-index and has envelop $(C\sup_{i \neq j} |\eta_i\eta_j|h_2^{-2},C\sup_{i \neq j} |\eta_i\eta_j|h_2^{-3})^{\prime
}$ for some large constant $C$. Then, by \citet[Corollary 5.6]{CK18}, there exist some constants $A \geq e$ and $v \geq 1$ such that  
\begin{equation*}
\sup_{(t,\tau )\in \mathcal{T}\mathcal{I}}\mathbb{E}|\mathcal{U}_{n}H(\cdot ,\cdot
;t,\tau )| \leq  (\frac{Cv\log(A)}{nh_2^2},\frac{Cv\log(A)}{nh_2^3})^{\prime},
\end{equation*}%
which implies that 

\begin{equation}
\sup_{(t,\tau )\in \mathcal{T}\mathcal{I}}\mathcal{U}_{n}H(\cdot ,\cdot
;t,\tau )=(O_{p}(\frac{1}{nh_2^2}),O_{p}(\frac{1}{nh_2^{3}}))^{\prime }.
\label{eq:H}
\end{equation}%

Now we compute $\frac{2}{n}\sum_{j=1}^n \eta_j\mathbb{P}%
\Gamma^s(\cdot,\Upsilon_j;t,\tau)$, whose first and second elements are
\begin{align*}
\int \frac{f_T(t+h_2v)}{f_{Y_{t+h_2v}}(q_\tau(t+h_2v))}&\biggl(\frac{%
	Y_{q_\tau(t+h_2v),j} - \phi_{t+h_2v,q_\tau(t+h_2v)}(X_j)}{f_{t+h_2v}(X_j)h_2}K(\frac{%
	T_j-t-h_2v}{h}) \\
& + \phi_{t+h_2v,q_\tau(t+h_2v)}(X_j) -E(t+h_2v,\tau) \biggr)K(v)dv 
\end{align*}
and 
\begin{align*}
\int \frac{vf_T(t+h_2v)}{h_2f_{Y_{t+h_2v}}(q_\tau(t+h_2v))}& \biggl(\frac{%
	Y_{q_\tau(t+h_2v),j} - \phi_{t+h_2v,q_\tau(t+h_2v)}(X_j)}{f_{t+h_2v}(X_j)h_2}K(\frac{%
	T_j-t-h_2v}{h_2}) \\
&+ \phi_{t+h_2v,q_\tau(t+h_2v)}(X_j) -E(t+h_2v,\tau) \biggr)K(v)dv,
\end{align*}
respectively. By the usual maximal inequality, 
\begin{align*}
\sup_{(t,\tau) \in \mathcal{T} \mathcal{I}}&\biggl|\frac{2}{n}\sum_{j=1}^n
\eta_j \int \frac{f_T(t+h_2v)}{f_{Y_{t+h_2v}}(q_\tau(t+h_2v))}\biggl(\frac{%
Y_{q_\tau(t+h_2v),j} - \phi_{t+h_2v,q_\tau(t+h_2v)}(X_j)}{f_{t+h_2v}(X_j)h_2}K(\frac{%
T_j-t-h_2v}{h_2}) \\
& + \phi_{t+h_2v,q_\tau(t+h_2v)}(X_j) -E(t+h_2v,\tau) \biggr)K(v)dv \biggr| =
O_p(\log^{1/2}(n)(nh_2)^{-1/2}).
\end{align*}

For the second element in $\mathbb{P}\Gamma ^{s}(\cdot ,\Upsilon
_{j};t,\tau )$, we first note that 
\begin{equation*}
\mathbb{E}\int \frac{vf_{T}(t+h_2v)}{%
h_2f_{Y_{t+h_2v}}(q_{\tau }(t+h_2v))}\biggl(\phi _{t+h_2v,q_{\tau
}(t+h_2v)}(X_{j})-\tau \biggr)K(v)dv=0
\end{equation*}%
and 
\begin{equation*}
\sup_{(t,\tau )\in \mathcal{T}\mathcal{I}}\mathbb{E}\biggl[\int \frac{vf_{T}(t+h_2v)}{%
	h_2f_{Y_{t+h_2v}}(q_{\tau }(t+h_2v))}\biggl(\phi _{t+h_2v,q_{\tau
	}(t+h_2v)}(X_{j})-\tau \biggr)K(v)dv\biggr]^2 \lesssim h_2^{-2}.
\end{equation*}%
Therefore, by the usual maximal inequality, 
\begin{align*}
\sup_{(t,\tau )\in \mathcal{T}\mathcal{I}}\frac{1}{n}\sum_{i=1}^n\biggl[\int \frac{vf_{T}(t+h_2v)}{%
	h_2f_{Y_{t+h_2v}}(q_{\tau }(t+h_2v))}\biggl(\phi _{t+h_2v,q_{\tau
	}(t+h_2v)}(X_{j})-\tau \biggr)K(v)dv\biggr] = O_p(\log^{1/2}(n)(nh_2^2)^{-1/2}). 
\end{align*}

Next, we turn to 
\begin{align*}
& \int v\overline{f}(t+h_2v,\tau)\biggl[\frac{Y_{q_{\tau }(t+h_2v),j}-\phi
_{t+h_2v,q_{\tau }(t+h_2v)}(X_{j})}{h_2^{2}f_{t+h_2v}(X_{j})}K(\frac{T_{j}-t-h_2v}{h_2}) - \frac{E(t+h_2v,\tau)-\tau}{h_2}\biggr]K(v)dv,
\end{align*}%
which has zero mean. Note that  
\begin{align*}
& \sup_{(t,\tau )\in \mathcal{T}\mathcal{I}}\mathbb{E}\biggl[\int v\biggl[%
\frac{\overline{f}(t+h_2v,\tau)}{f_{t+h_2v}(X_j)}-\frac{\overline{f}(t,\tau)}{f_{t}(X_j)}\biggr]\frac{Y_{q_{\tau
}(t+h_2v),j}-\phi _{t+h_2v,q_{\tau }(t+h_2v)}(X_{j})}{h_2^{2}}K(\frac{T_{j}-t-h_2v}{h_2}%
)K(v)dv\biggr]^{2} \\
\lesssim & \sup_{t\in \mathcal{T}}\int h_2^{-2}v^{2}\mathbb{E}K^{2}(\frac{%
T_{j}-t-h_2v}{h_2})K(v)dv\lesssim h_2^{-1}.
\end{align*}%
Therefore, by \citet[Corollary 5.1]{CCK14}, we have 
\begin{align*}
& \sup_{(t,\tau )\in \mathcal{T}\mathcal{I}}\biggl|(\mathbb{P}_n - \mathbb{P})\eta
_{j}\int v[%
\frac{\overline{f}(t+h_2v,\tau)}{f_{t+h_2v}(X_j)}-\frac{\overline{f}(t,\tau)}{f_{t}(X_j)}]\frac{Y_{q_{\tau
}(t+h_2v),j}-\phi _{t+h_2v,q_{\tau }(t+h_2v)}(X_{j})}{h_2^{2}}K(\frac{T_{j}-t-h_2v}{h_2}%
)K(v)dv\biggr| \\
=& O_{p}(\log ^{1/2}(n)(nh_2)^{-1/2}).
\end{align*}%
and 
\begin{align*}
& \int v\overline{f}(t+h_2v,\tau)\biggl[\frac{Y_{q_{\tau }(t+h_2v),j}-\phi
	_{t+h_2v,q_{\tau }(t+h_2v)}(X_{j})}{h_2^{2}f_{t+h_2v}(X_{j})}K(\frac{T_{j}-t-h_2v}{h_2}) - \frac{E(t+h_2v,\tau)-\tau}{h_2}\biggr]K(v)dv \\
= & (\mathbb{P}_n - \mathbb{P})\eta
_{j}\int v\frac{\overline{f}(t,\tau)}{f_{t}(X_j)}\frac{Y_{q_{\tau
		}(t+h_2v),j}-\phi _{t+h_2v,q_{\tau }(t+h_2v)}(X_{j})}{h_2^{2}}K(\frac{T_{j}-t-h_2v}{h_2}%
)K(v)dv + O_p^*(\log ^{1/2}(n)(nh_2)^{-1/2}).
\end{align*}
In addition, note that 
\begin{align*}
& \mathbb{E}\biggl\{\int \frac{v\overline{f}(X_{j};t)}{h_2^{2}}\biggl[%
Y_{q_{\tau }(t+h_2v),j}-\phi _{t+h_2v,q_{\tau }(t+h_2v)}(X_{j})-(Y_{q_{\tau
}(t),j}-\phi _{t,q_{\tau }(t)}(X_{j}))\biggr]K(\frac{T_{j}-t-h_2v}{h_2})K(v)dv%
\biggr\}^{2} \\
\lesssim & \int \mathbb{E}v^{2}h_2^{-4}\biggl(|\phi _{T_{j},q_{\tau
}(t+h_2v)}(X_{j})-\phi _{T_{j},q_{\tau }(t)}(X_{j})|+(\phi _{t+h_2v,q_{\tau
}(t+h_2v)}(X_{j})-\phi _{t,q_{\tau }(t)}(X_{j}))^{2}\biggr) \\
& \times K^{2}(\frac{T_{j}-t-h_2v}{h_2})K(v)dv\lesssim h_2^{-2}.
\end{align*}%
Therefore, by \citet[Corollary 5.1]{CCK14}, 
\begin{align*}
& (\mathbb{P}_n - \mathbb{P})\eta _{j}\int \frac{v\overline{f}(X_{j};t)}{h_2^{2}}%
[Y_{q_{\tau }(t+h_2v),j}-\phi _{t+h_2v,q_{\tau }(t+h_2v)}(X_{j})-(Y_{q_{\tau
}(t),j}-\phi _{t,q_{\tau }(t)}(X_{j}))]K(\frac{T_{j}-t-h_2v}{h_2})K(v)dv \\
=& O_{p}^{\ast }(\log ^{1/2}(n)(nh_2^{2})^{-1/2}).
\end{align*}
and 
\begin{align*}
& (\mathbb{P}_n - \mathbb{P})\eta
_{j}\int v\frac{\overline{f}(t,\tau)}{f_{t}(X_j)}\frac{Y_{q_{\tau
		}(t+h_2v),j}-\phi _{t+h_2v,q_{\tau }(t+h_2v)}(X_{j})}{h_2^{2}}K(\frac{T_{j}-t-h_2v}{h_2}%
)K(v)dv \\
= & (\mathbb{P}_n - \mathbb{P})\eta _{j}\int \frac{v\overline{f}(X_{j};t)}{h_2^{2}}%
[Y_{q_{\tau
	}(t),j}-\phi _{t,q_{\tau }(t)}(X_{j})]K(\frac{T_{j}-t-h_2v}{h_2})K(v)dv + O_{p}^{\ast }(\log ^{1/2}(n)(nh_2^{2})^{-1/2})
\end{align*}

Combining the above results and denoting $\overline{K}(u)=\int vK(u-v)K(v)dv$%
, we have 
\begin{align*}
& \int v\overline{f}(t+h_2v,\tau)\biggl[\frac{Y_{q_{\tau }(t+h_2v),j}-\phi
	_{t+h_2v,q_{\tau }(t+h_2v)}(X_{j})}{h_2^{2}f_{t+h_2v}(X_{j})}K(\frac{T_{j}-t-h_2v}{h_2}) - \frac{E(t+h_2v,\tau)-\tau}{h_2}\biggr]K(v)dv \\
=& (\mathbb{P}_n - \mathbb{P})\eta _{j}\frac{\overline{f}(t,\tau)}{f_t(X_j)h_2^{2}}%
[Y_{q_{\tau }(t),j}-\phi _{t,q_{\tau }(t)}(X_{j})]\overline{K}(\frac{T_{j}-t%
}{h_2})+O_{p}^{\ast }(\log ^{1/2}(n)(nh_2^{2})^{-1/2})
\end{align*}%
and%
\begin{equation}
\frac{2}{n}\sum_{j=1}^{n}\eta _{j}\mathbb{P}\Gamma ^{s}(\cdot ,\Upsilon
_{j};t,\tau )=%
\begin{Bmatrix}
O_{p}^{\ast }(\log ^{1/2}(n)(nh_2)^{-1/2}) \\ 
(\mathbb{P}_n - \mathbb{P})\biggl[\eta _{j}\frac{\overline{f}(t,\tau)}{f_t(X_j)h_2^{2}}%
[Y_{q_{\tau }(t),j}-\phi _{t,q_{\tau }(t)}(X_{j})]\overline{K}(\frac{T_{j}-t%
}{h_2})\biggr]+o_{p}^{\ast }((nh_2^{3})^{-1/2})%
\end{Bmatrix}
\label{B.13}
\end{equation}%
Combining \eqref{eq:U}, \eqref{eq:H}, and (\ref{B.13}), we have the desired
results.
\end{proof}

\section{Rearrangement Operator on A Local Process}

\label{sec:rearrange} 
The rearrangement operator has been previously
studied by \cite{CFG10}, in which they required the underlying process to be
tight to apply the continuous mapping theorem. However, the local processes
encountered in our paper are not tight due to the presence of the kernel
function. Therefore, the original results on the rearrangement operate
cannot directly apply to our case. Instead, in this section, we extend the
results in \cite{CFG10} to the case that the
underlying process is not tight.

Let $Q(t,v)$ be a generic monotonic function in $v\in \lbrack 0,1]$.
The functional $\Psi$ maps $Q(t,v)$ to $F(t,y)$ as
follows: 
\begin{equation*}
\Psi (Q)(t,y):= F(t,y)=\int_{0}^{1}1\{Q(t,v)\leq y\}dv.
\end{equation*}%
We want to derive a linear expansion of $\Psi (Q+s_{n}d_{n})-\Psi (Q)$ where 
$s_{n}\downarrow 0$ as the sample size $n\rightarrow \infty $ and $%
d_{n}(t,v) $ is some perturbation function. 

\begin{ass}
\begin{enumerate}
\item $Q(t,v)$ is twice differentiable w.r.t. $v$ with both
derivatives bounded. In addition, $\partial _{v}Q(t,v)>c$ for some positive
constant $c$, uniformly over $(t,v)\in \mathcal{T}\times \lbrack 0,1]$. 

\item There exist two vanishing sequences $\varepsilon _{n}$ and $%
\delta _{n}$ such that 
\begin{equation*}
\sup_{(t,v,v^{\prime })\in \mathcal{T}\times \lbrack 0,1]^{2},|v-v^{\prime
}|\leq \varepsilon _{n}}|d_{n}(t,v)-d_{n}(t,v^{\prime })|=o(\delta _{n}),
\end{equation*}%
\begin{equation*}
\sup_{(t,v)\in \mathcal{T}\times \lbrack
0,1]}|d_{n}(t,v)|s_{n}=o(\varepsilon _{n}),\quad \text{and}\quad
\sup_{(t,v)\in \mathcal{T}\times \lbrack 0,1]}|d_{n}(t,v)|^{2}s_{n}=o(\delta
_{n}).
\end{equation*}
\end{enumerate}

\label{ass:cdf}
\end{ass}
The following proposition extends the first part of Proposition 2 in 
\cite{CFG10}.

\begin{prop}
{\small Let $(t,y)\in \mathcal{TY}:= \{(t,y):y=Q(t,v),(t,v)\in \mathcal{T}%
\times \lbrack 0,1]\}$, $F(t,y|d_{n})=\int_{0}^{1}1\{Q(t,v)+s_{n}d_{n}(t,v)%
\leq y\}dv$, and $y=Q(t,v^{y})$. If Assumption \ref{ass:cdf} holds, then 
\begin{equation*}
\frac{F(t,y|d_{n})-F(t,y)}{s_{n}}-(\frac{-d_{n}(t,v^{y})}{\partial
_{v}Q(t,v^{y})})=o(\delta _{n})
\end{equation*}%
uniformly over $(t,y)\in \mathcal{TY}$. \label{lem:cdf} }
\end{prop}

\begin{proof}
Consider $(t_{n},y_{n})\rightarrow (t_{0},y_{0})$ and denote $v_{n}$
as $y_{n}=Q(t_{n},v_{n})$. Note that \ 
\begin{align*}
F(t_{n},y_{n}|d_{n})=& \int_{0}^{1}1\{Q(t_{n},v)+s_{n}d_{n}(t_{n},v)\leq
y_{n}\}dv \\
=&
\int_{0}^{1}1%
\{Q(t_{n},v)+s_{n}(d_{n}(t_{n},v_{n})+d_{n}(t_{n},v)-d_{n}(t_{n},v_{n}))\leq
y_{n}\}dv.
\end{align*}%
Let $\mathbb{B}_{\varepsilon }(v)=\{v^{\prime }:\left\vert v-v^{\prime
}\right\vert \leq \varepsilon \}$. For fixed $n$, if $v\in \mathbb{B}%
_{\varepsilon _{n}}(v_{n})\cap \lbrack 0,1]$, by Assumption \ref{ass:cdf}, 
\begin{equation*}
d_{n}(t_{n},v)-d_{n}(t_{n},v_{n})=o(\delta _{n}).
\end{equation*}%
Then for any $\delta >0$, there exists $n_{1}$ such that if $n\geq n_{1}$ , $%
|d_{n}(t_{n},v)-d_{n}(t_{n},v_{n})|\leq \delta \delta _{n}$ and 
\begin{equation*}
F(t_{n},y_{n}|d_{n})\leq
\int_{0}^{1}1\{Q(t_{n},v)+s_{n}(d_{n}(t_{n},v_{n})-\delta \delta _{n})\leq
y_{n}\}dv.
\end{equation*}%
If $v\notin \mathbb{B}_{\varepsilon _{n}}(v_{n})$, then there exists $n_{2}$
such that for $n\geq n_{2}$, 
\begin{equation}
|Q(t_{n},v)-y_{n}|\geq c\varepsilon _{n}.  \label{eq:cdf1}
\end{equation}%
Furthermore, by Assumption \ref{ass:cdf}, 
\begin{equation*}
s_{n}d_{n}(t_{n},v)\leq \sup_{(t,v)\in \mathcal{T}\times \lbrack
0,1]}|d_{n}(t,v)|s_{n}=o(\varepsilon _{n}).
\end{equation*}%
Therefore, 
\begin{equation*}
F(t_{n},y_{n}|d_{n})=\int_{0}^{1}1\{Q(t_{n},v) + s_{n}(d_{n}(t_{n},v_{n})-\delta \delta _{n})\leq y_{n}\}dv
\end{equation*}%
and 
\begin{equation}
\begin{split}
& \frac{F(t_{n},y_{n}|d_{n})-F(t_{n},y_{n})}{s_{n}}-(\frac{%
-d_{n}(t_{n},v_{n})}{\partial _{v}Q(t_{n},v_{n})}) \\
\leq & \int_{\mathbb{B}_{\varepsilon _{n}}(v_{n})}\frac{1}{s_{n}}\biggl(%
1\{Q(t_{n},v)+s_{n}(d_{n}(t_{n},v_{n})-\delta \delta _{n})\leq
y_{n}\}-1\{Q(t_{n},v)\leq y_{n}\}\biggr)dv+(\frac{d_{n}(t_{n},v_{n})}{%
\partial _{v}Q(t_{n},v_{n})}) \\
=& \int_{\mathbb{J}_{n}\cap \lbrack
y_{n},y_{n}-s_{n}(d_{n}(t_{n},v_{n})-\delta \delta _{n})]}\frac{dy}{%
s_{n}\partial _{v}Q(t_{n},v_{n}(y))}+\frac{d_{n}(t_{n},v_{n})}{\partial
_{v}Q(t_{n},v_{n})},
\end{split}
\label{C.2}
\end{equation}%
where the equality follows by the change of variables: $y=Q(t_{n},v)$, $%
v_{n}(y)=Q^{\leftarrow }(t_{n},\cdot )(y)$, and $\mathbb{J}_{n}$ is the
image of $\mathbb{B}_{\varepsilon _{n}}(v_{n})$. By \eqref{eq:cdf1} and
Assumption \ref{ass:cdf}.2, $[y_{n},y_{n}-s_{n}(d_{n}(t_{n},v_{n})-\delta
\delta _{n})]$ is nested by $\mathbb{J}_{n}$ for $n$ sufficiently large. In
addition, since $\partial _{v}Q(t,v)>c$ uniformly over $\mathcal{T}\times
\lbrack 0,1]$, for $y\in \lbrack y_{n},y_{n}-s_{n}(d_{n}(t_{n},v_{n}))]$, 
\begin{equation*}
|v_{n}(y)-v_{n}|=|Q^{\leftarrow }(t_{n},\cdot )(y)-Q^{\leftarrow
}(t_{n},\cdot )(y_{n})|\leq Cs_{n}(\sup_{(t,v)\in \mathcal{T}\times \lbrack
0,1]}|d_{n}(t,v)|).
\end{equation*}%
Then the r.h.s. of (\ref{C.2}) is bounded from above by \ 
\begin{align*}
& \frac{\delta \delta _{n}}{\partial _{v}Q(t_{n},\tilde{v}_{n}(\tilde{y}))}%
+\int_{[y_{n},y_{n}-s_{n}d_{n}(t_{n},v_{n})]}(\frac{1}{\partial
_{v}Q(t_{n},v_{n}(y))}-\frac{1}{\partial _{v}Q(t_{n},v_{n})})\frac{dy}{s_{n}}
\\
\leq & C\delta \delta _{n}+Cs_{n}(\sup_{(t,v)\in \mathcal{T}\times \lbrack
0,1]}|d_{n}^{2}(t,v)|)\leq C^{\prime }\delta \delta _{n},
\end{align*}
where $\tilde{y} \in (y_n - s_{n}d_{n}(t_{n},v_{n}), y_n -
s_{n}(d_{n}(t_{n},v_{n})-\delta\delta_n))$. Since $\delta $ is arbitrary, by
letting $\delta \rightarrow 0$, we obtain that 
\begin{equation*}
\frac{F(t_{n},y_{n}|d_{n})-F(t_{n},y_{n})}{s_{n}}-(\frac{-d_{n}(t_{n},v_{n})%
}{\partial _{v}Q(t_{n},v_{n})})\leq o(\delta _{n}).
\end{equation*}%
Similarly, we can show that 
\begin{equation*}
\frac{F(t_{n},y_{n}|d_{n})-F(t_{n},y_{n})}{s_{n}}-(\frac{-d_{n}(t_{n},v_{n})%
}{\partial _{v}Q(t_{n},v_{n})})\geq o(\delta _{n}).
\end{equation*}%
Therefore, we have proved that 
\begin{equation*}
\frac{F(t_{n},y_{n}|d_{n})-F(t_{n},y_{n})}{s_{n}}-(\frac{-d_{n}(t_{n},v_{n})%
}{\partial _{v}Q(t_{n},v_{n})})=o(\delta _{n}).
\end{equation*}%
Since the above result holds for any sequence of $(t_{n},y_{n})$, then by
Lemma 1 \cite{CFG10}, we have that uniformly over $(t,y)\in \mathcal{TY}$, 
\begin{equation*}
\frac{F(t,y|d_{n})-F(t,y)}{s_{n}}-(\frac{-d_{n}(t,v^{y})}{\partial
_{v}Q(t,v^{y})})=o(\delta _{n}).
\end{equation*}%
This completes the proof of the proposition.
\end{proof}

\bigskip

Let $F(t,y)$ and $F^{\leftarrow }(t,u)$ be a monotonic function and
its inverse w.r.t. $y$, respectively. Next, we consider the linear expansion
of the inverse functional: 
\begin{equation*}
(F+s_{n}J_{n})^{\leftarrow }-F^{\leftarrow }
\end{equation*}%
where $s_{n}\downarrow 0$ as the sample size $n\rightarrow \infty $ and $%
J_{n}(t,y)$ is some perturbation function.

\begin{ass}
\begin{enumerate}
\item $F(t,y)$ has a compact support $\mathcal{TY} = \{(t,y): y =
Q(t,v), (t,v) \in \mathcal{TV} := \mathcal{T} \times \mathcal{V}\}$. Denote $%
\mathcal{V}_\varepsilon$, $\mathcal{TY}_\varepsilon$, $\mathcal{Y}
_{t\varepsilon}$, and $\underline{y}_t$ as a compact subset of $\mathcal{V}$
, $\{(t,y): y = Q(t,v), (t,v) \in \mathcal{T} \times \mathcal{V}
_\varepsilon\}$, the projection of $\mathcal{TY}_\varepsilon$ on $T=t$, and
the lower bound of $\overline{(\mathcal{Y}_{\varepsilon t})^\varepsilon}$,
respectively. Then for any $t \in \mathcal{T}$, $\underline{y}_t> -\infty$
and $\overline{(\mathcal{Y}_{\varepsilon t})^\varepsilon} \subset \mathcal{Y}
_t $.

\item $F(t,y)$ is monotonic and twice continuously differentiable
w.r.t. $y$. The first and second derivatives are denoted as $f(t,y)$ and $%
f^{\prime }(t,y)$ respectively. Then both $f(t,y)$ and $f^{\prime }(t,y)$
are bounded and $f(t,y)$ is also bounded away from zero, uniformly over $%
\mathcal{TY}$.

\item Let $\mathcal{T}\mathcal{Y}\mathcal{Y}=\{(t,y,y^{\prime
}):y=Q(t,v),y^{\prime }=Q(t,v^{\prime }),(t,v,v^{\prime })\in \mathcal{T}%
\times \mathcal{V}\times \mathcal{V}\}$. Then, there exist two vanishing
sequences $\varepsilon _{n}$ and $\delta _{n}$ such that 
\begin{equation*}
\sup_{(t,y,y^{\prime })\in \mathcal{T}\mathcal{Y}\mathcal{Y},|y-y^{\prime
}|\leq \max (\varepsilon _{n},s_{n}\delta
_{n})}|J_{n}(t,y)-J_{n}(t,y^{\prime })|=o(\delta _{n}),
\end{equation*}%
\begin{equation*}
\sup_{(t,y)\in \mathcal{TY}}|J_{n}(t,y)|s_{n}=o(\varepsilon _{n}),\quad 
\text{ and}\quad \sup_{(t,y)\in \mathcal{TY}}|J_{n}(t,y)|^{2}s_{n}=o(\delta
_{n}).
\end{equation*}%

\end{enumerate}

\label{ass:inv}
\end{ass}

\begin{prop}
If Assumption \ref{ass:inv} holds, then 
\begin{equation*}
\frac{(F + s_n J_n)^{\leftarrow}(t,v) - F^{\leftarrow}(t,v)}{s_n} + \frac{
J_n(t,F^{\leftarrow}(t,v))}{f(t,F^{\leftarrow}(t,v))} = o(\delta_n)
\end{equation*}
uniformly over $(t,v) \in \mathcal{TV}_\varepsilon$. \label{lem:inv}
\end{prop}

\begin{proof}
Without loss of generality, we assume $F(t,y)$ is monotonically
increasing in $y$. Let $\xi (t,v)=F^{\leftarrow }(t,v)$ and $\xi
_{n}(t,v)=(F+s_{n}J_{n})^{\leftarrow }(t,v).$ Since for $n$ sufficiently
large, $\sup_{(t,v)\in \mathcal{TV}_{\varepsilon }}s_{n}|J_{n}^{\leftarrow
}(t,v)|<\varepsilon $ and by the definition of $V_{\varepsilon }$, we can
choose $\xi (t,v)\in \mathcal{Y}_{t}$ and $\xi _{n}(t,v)\in \mathcal{Y}_{t}$%
. In addition, since $F$ is differentiable, we have $F(t,\xi (t,v))=v$.
Denote $\eta _{n}(t,v)=\min (s_{n}\delta _{n}^{2},\xi _{n}(t,v)-\underline{y}%
_{t})$. Then, the definition of the inverse function implies that 
\begin{equation}
(F+s_{n}J_{n})(t,\xi _{n}(t,v)-\eta _{n}(t,v))\leq v\leq
(F+s_{n}J_{n})(t,\xi _{n}(t,v)).  \label{eq:u1}
\end{equation}%
Since $f(t,y)$ is bounded uniformly in $(t,y)\in \mathcal{TY}$, we have 
\begin{equation*}
F(t,\xi _{n}(t,v)-\eta _{n}(t,v))-v=F(t,\xi _{n}(t,v))-F(t,\xi
(t,v))+o(s_{n}\delta _{n})
\end{equation*}%
and 
\begin{equation*}
|s_{n}J_{n}(t,\xi _{n}(t,v)-\eta _{n}(t,v))|\leq \sup_{(t,y)\in \mathcal{TY}%
}s_{n}|J_{n}(t,y)|.
\end{equation*}%
Therefore, \eqref{eq:u1} implies that 
\begin{equation*}
-\sup_{(t,y)\in \mathcal{TY}}s_{n}|J_{n}(t,y)|\leq F(t,\xi
_{n}(t,v))-F(t,\xi (t,v))\leq \sup_{(t,y)\in \mathcal{TY}%
}s_{n}|J_{n}(t,y)|+o(s_{n}\delta _{n}).
\end{equation*}%
Since $f(t,y)$ is bounded and bounded away from zero, we have 
\begin{equation*}
|\xi _{n}(t,v)-\xi (t,v)|=O(\sup_{(t,y)\in \mathcal{TY}%
}s_{n}|J_{n}(t,y)|)+o(s_{n}\delta _{n})=o(\max (\varepsilon _{n},s_{n}\delta
_{n})).
\end{equation*}%
Then, 
\begin{align*}
& F(t,\xi _{n}(t,v)-\eta _{n}(t,v))-F(t,\xi (t,v))+s_{n}J_{n}(t,\xi
_{n}(t,v)-\eta _{n}(t,v)) \\
\geq & F(t,\xi _{n}(t,v))-F(t,\xi (t,v))-o(s_{n}\delta
_{n})+s_{n}J_{n}(t,\xi (t,v))-s_{n}\sup |J_{n}(t,y)-J_{n}(t,y^{\prime })| \\
\geq & f(t,\xi (t,v))(\xi _{n}(t,v)-\xi (t,v))+s_{n}J_{n}(t,\xi
(t,v))-O(\sup_{(t,y)\in \mathcal{TY}}s_{n}^{2}|J_{n}(t,y)|^{2})-o(s_{n}^{2}%
\delta _{n}^{2})-o(s_{n}\delta _{n}) \\
\geq & f(t,\xi (t,v))(\xi _{n}(t,v)-\xi (t,v))+s_{n}J_{n}(t,\xi
(t,v))-o(s_{n}\delta _{n}),
\end{align*}%
where the supremum in the second line is taken over $(t,y,y^{\prime })\in 
\mathcal{T}\mathcal{Y}\mathcal{Y}$, $|y-y^{\prime }|\leq \max (\varepsilon
_{n},s_{n}\delta _{n})$, and the third line is because $f^{\prime }(t,y)$ is
bounded uniformly in $(t,y)\in \mathcal{TY}$.

On the other hand, by \eqref{eq:u1}, 
\begin{equation*}
F(t,\xi _{n}(t,v)-\eta _{n}(t,v))-F(t,\xi (t,v))+s_{n}J_{n}(t,\xi
_{n}(t,v)-\eta _{n}(t,v))\leq 0.
\end{equation*}%
Therefore, we have 
\begin{equation}
\frac{(\xi _{n}(t,v)-\xi (t,v))}{s_{n}}+\frac{J_{n}(t,\xi (t,v))}{f(t,\xi
(t,v))}\leq o(\delta _{n}).  \label{eq:xi1}
\end{equation}%
Similarly, we can show that \ 
\begin{align*}
& F(t,\xi _{n}(t,v))-F(t,\xi (t,v))+s_{n}J_{n}(t,\xi _{n}(t,v)) \\
\leq & f(t,\xi (t,v))(\xi _{n}(t,v)-\xi (t,v))+s_{n}J_{n}(t,\xi
(t,v))+o(s_{n}\delta _{n}).
\end{align*}%
The r.h.s. of \eqref{eq:u1} implies that 
\begin{equation*}
F(t,\xi _{n}(t,v))-F(t,\xi (t,v))+s_{n}J_{n}(t,\xi _{n}(t,v))\geq 0.
\end{equation*}%
Therefore, 
\begin{equation}
\frac{(\xi _{n}(t,v)-\xi (t,v))}{s_{n}}+\frac{J_{n}(t,\xi (t,v))}{f(t,\xi
(t,v))}\geq -o(\delta _{n}).  \label{eq:xi2}
\end{equation}%
\eqref{eq:xi1} and \eqref{eq:xi2} imply that 
\begin{equation*}
\frac{(\xi _{n}(t,v)-\xi (t,v))}{s_{n}}+\frac{J_{n}(t,\xi (t,v))}{f(t,\xi
(t,v))}=o(\delta _{n})
\end{equation*}%
uniformly over $(t,v)\in \mathcal{TV}$. 
\end{proof}

\section{Additional Simulation Results}
\label{sec:addsim} This section investigates the sensitivity of
bootstrap confidence intervals against the tuning parameters $h_1$, $\tilde{\lambda}$, and $\lambda$, reports the finite sample performance for the oracle estimator and the estimator for the mean potential outcomes, and illustrates limitation of our method.

\subsection{Sensitivity Analysis}
We check the sensitivity of our estimation method with respect to three tuning parameters: $h_1$, $\tilde{\lambda}$, and $\lambda$. We focus on the first design in Section \ref{sec:sim}. Figures \ref{fig:comp1} and \ref{fig:comp2} show the coverage probabilities of $q_{\tau}(t)$ and $\beta _{\tau}^{1}(t)$ with $%
h_1'=0.8h_{1}$ and $%
h_1'=1.2h_{1}$, respectively. Figures \ref{fig:comp3} and \ref{fig:comp4} show the coverage probabilities of $q_{\tau}(t)$ and $\beta _{\tau}^{1}(t)$ with $%
\tilde{\lambda}'=0.8\tilde{\lambda}$ and $%
\tilde{\lambda}'=1.2\tilde{\lambda}$, respectively, where $\tilde{\lambda}$ is the penalty used to estimate the conditional density $f_t(X)$. Last, Figures \ref{fig:comp5} and \ref{fig:comp6} show the coverage probability $q_{\tau}(t)$ and $\beta _{\tau}^{1}(t)$ with $%
\lambda'=0.8\lambda$ and $%
\lambda'=1.2\lambda$, respectively, where $\lambda$ is the penalty used to estimate the conditional CDF $\phi_{t,u}(X)$. We observe that the coverage probabilities are in general not sensitive to
the choice of tuning parameters.

\begin{figure}[H]
\centering
\includegraphics[scale = 0.55,angle=0]{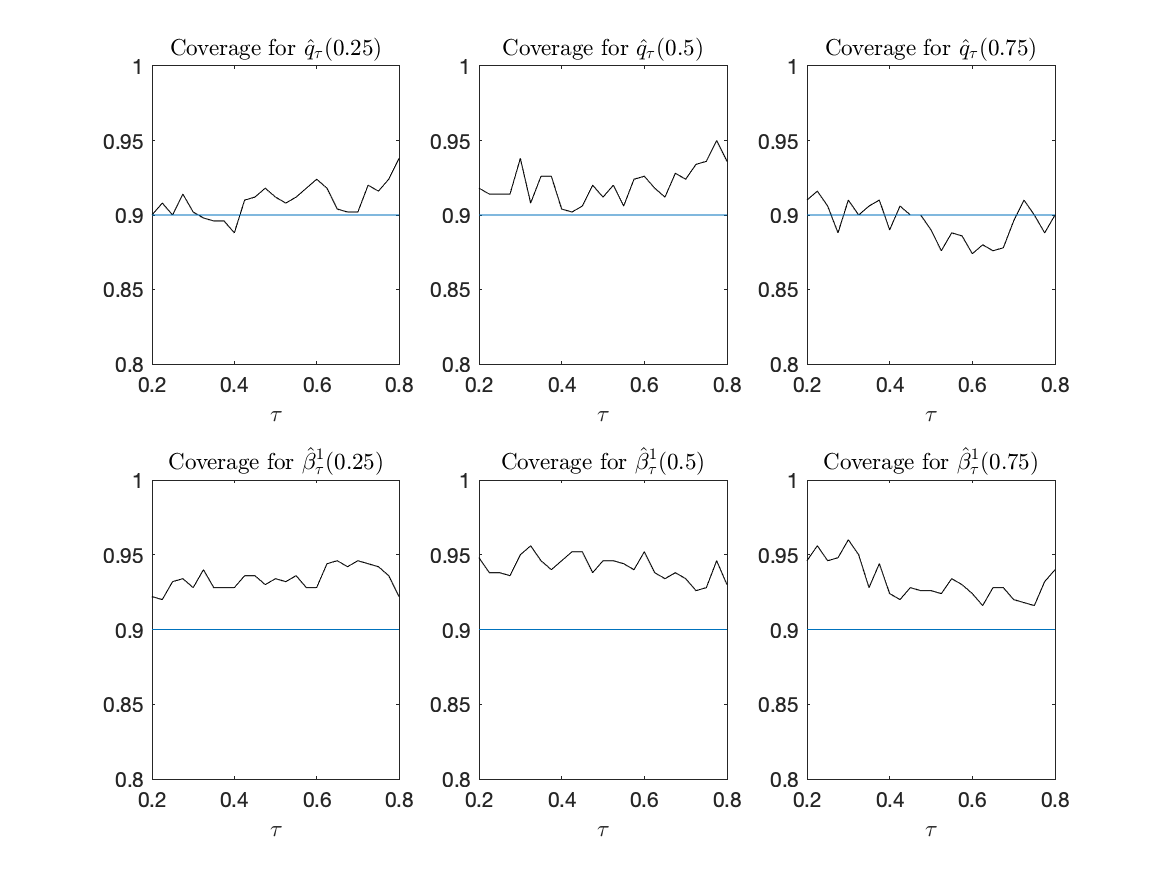}   
\caption{coverage probability for small $h_1$}
\label{fig:comp1}
\end{figure}

\begin{figure}[H]
 \centering
\includegraphics[scale = 0.55,angle=0]{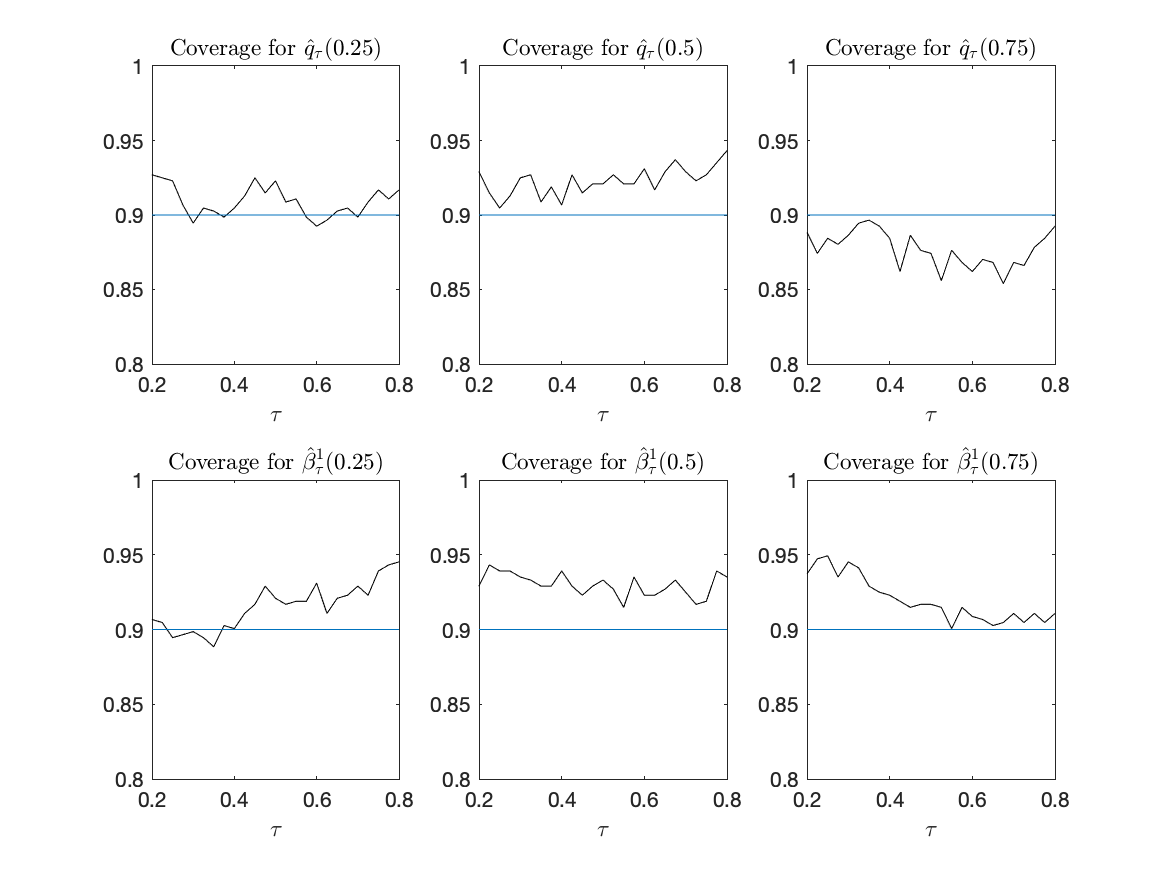}   
\caption{coverage probability for large $h_1$}
\label{fig:comp2}
\end{figure}

\begin{figure}[H]
	\centering
	\includegraphics[scale = 0.55,angle=0]{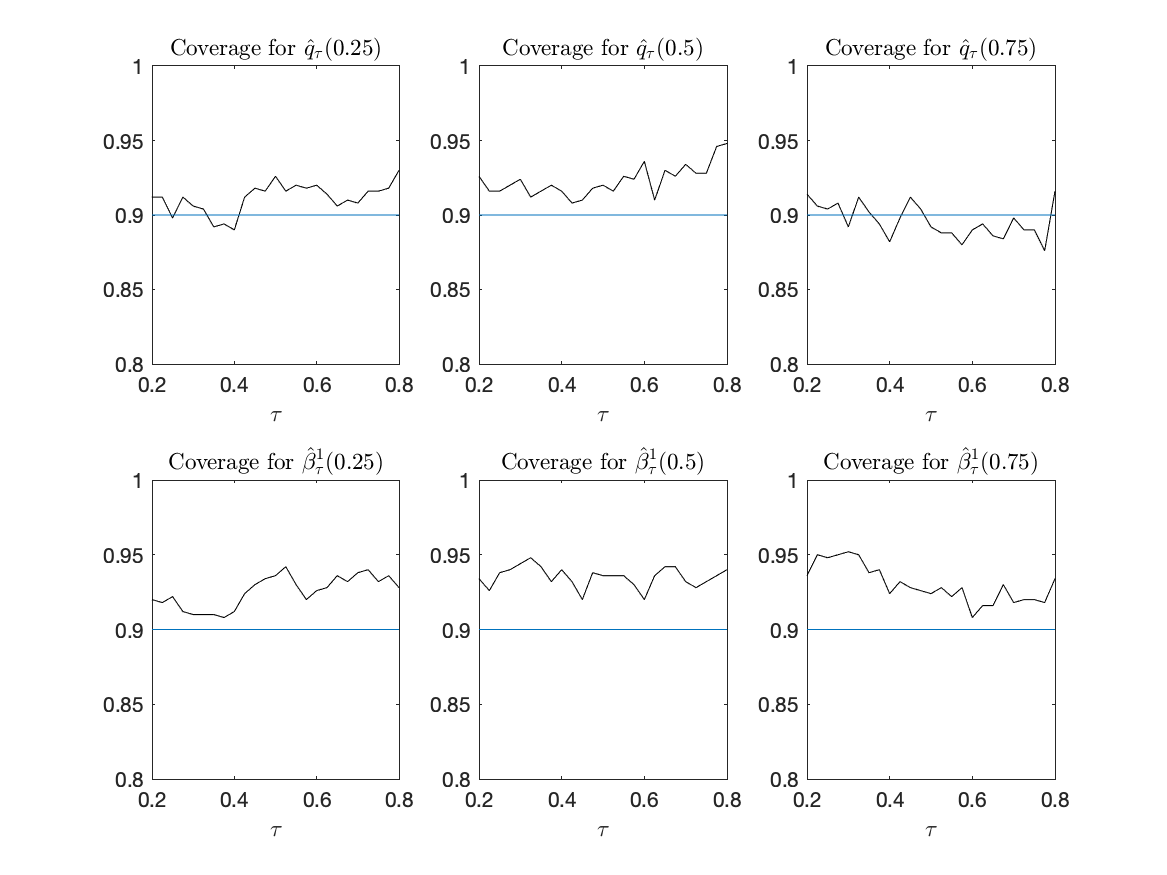}   
	\caption{coverage probability for small $\tilde{\lambda}$}
	\label{fig:comp3}
\end{figure}

\begin{figure}[H]
	\centering
	\includegraphics[scale = 0.55,angle=0]{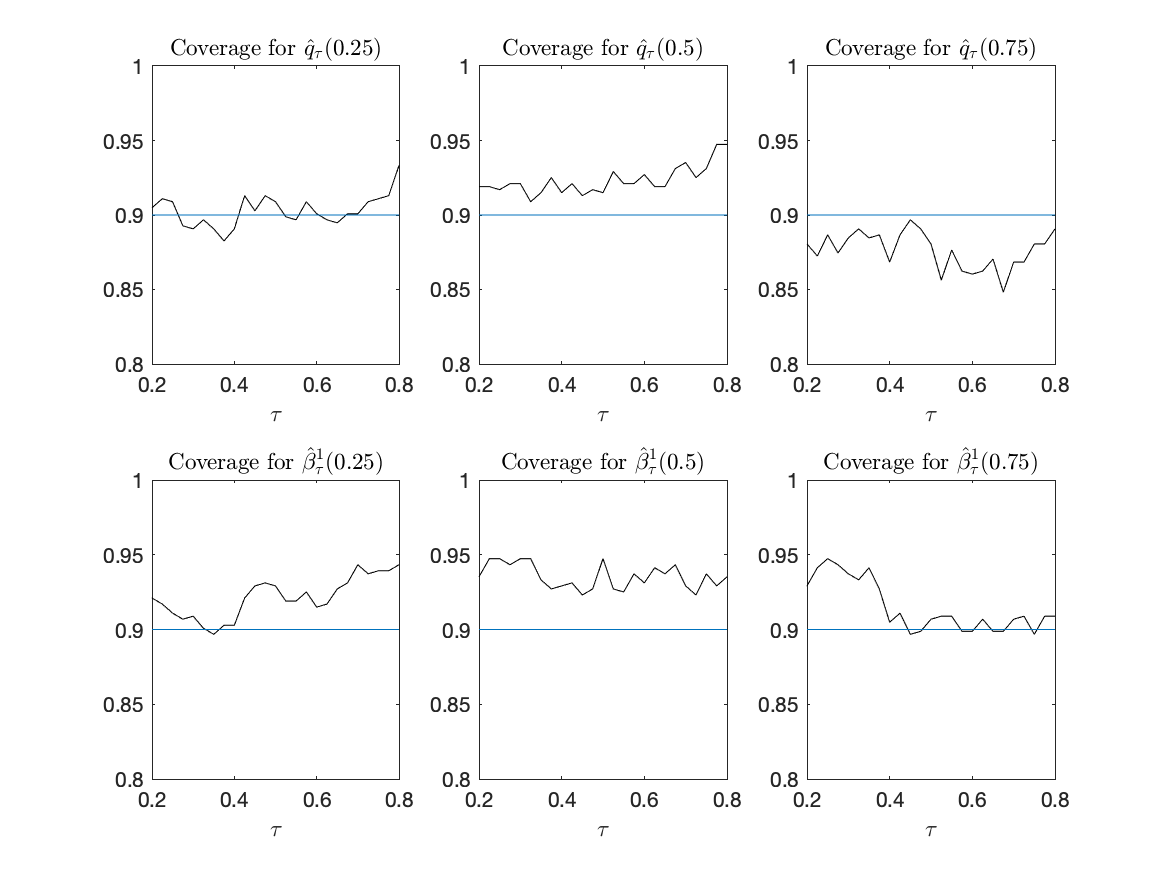}   
	\caption{coverage probability for large $\tilde{\lambda}$}
	\label{fig:comp4}
\end{figure}

\begin{figure}[H]
	\centering
	\includegraphics[scale = 0.55,angle=0]{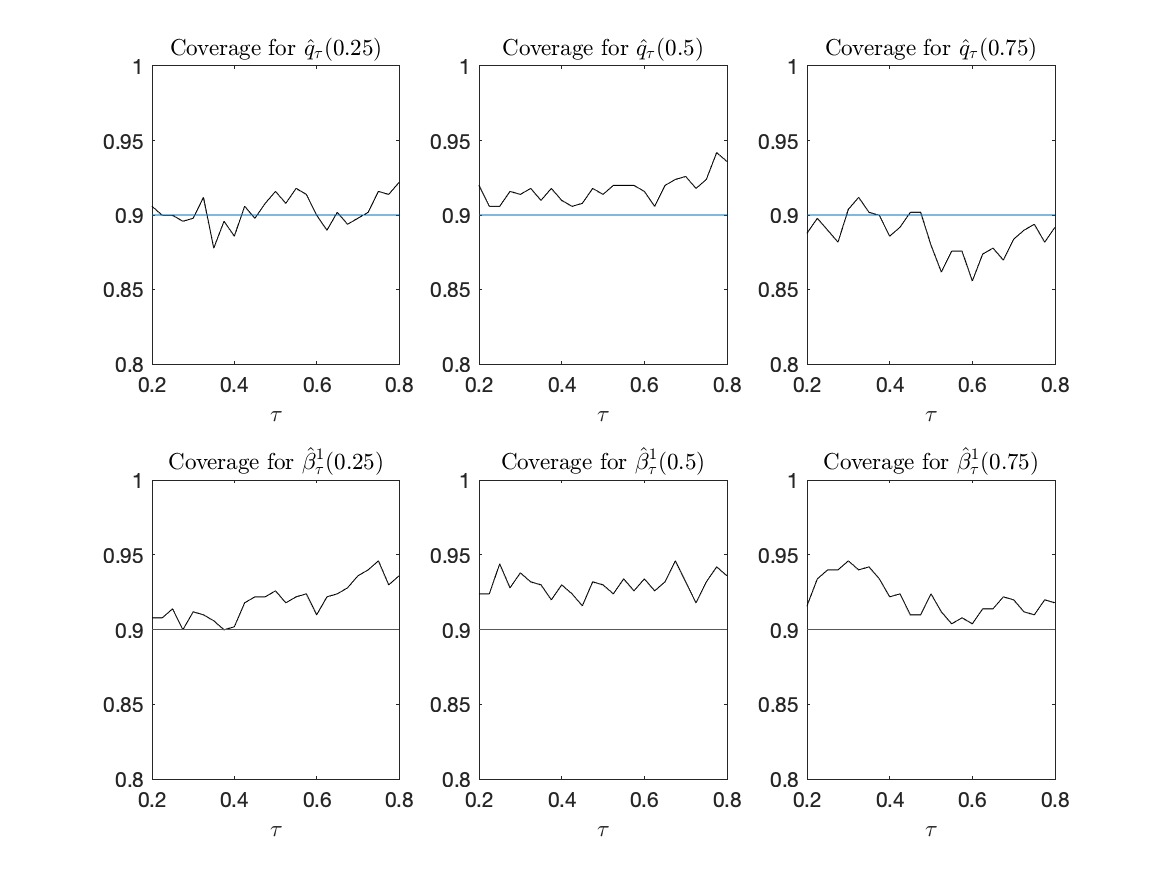}   
	\caption{coverage probability for small $\lambda$}
	\label{fig:comp5}
\end{figure}

\begin{figure}[H]
	\centering
	\includegraphics[scale = 0.55,angle=0]{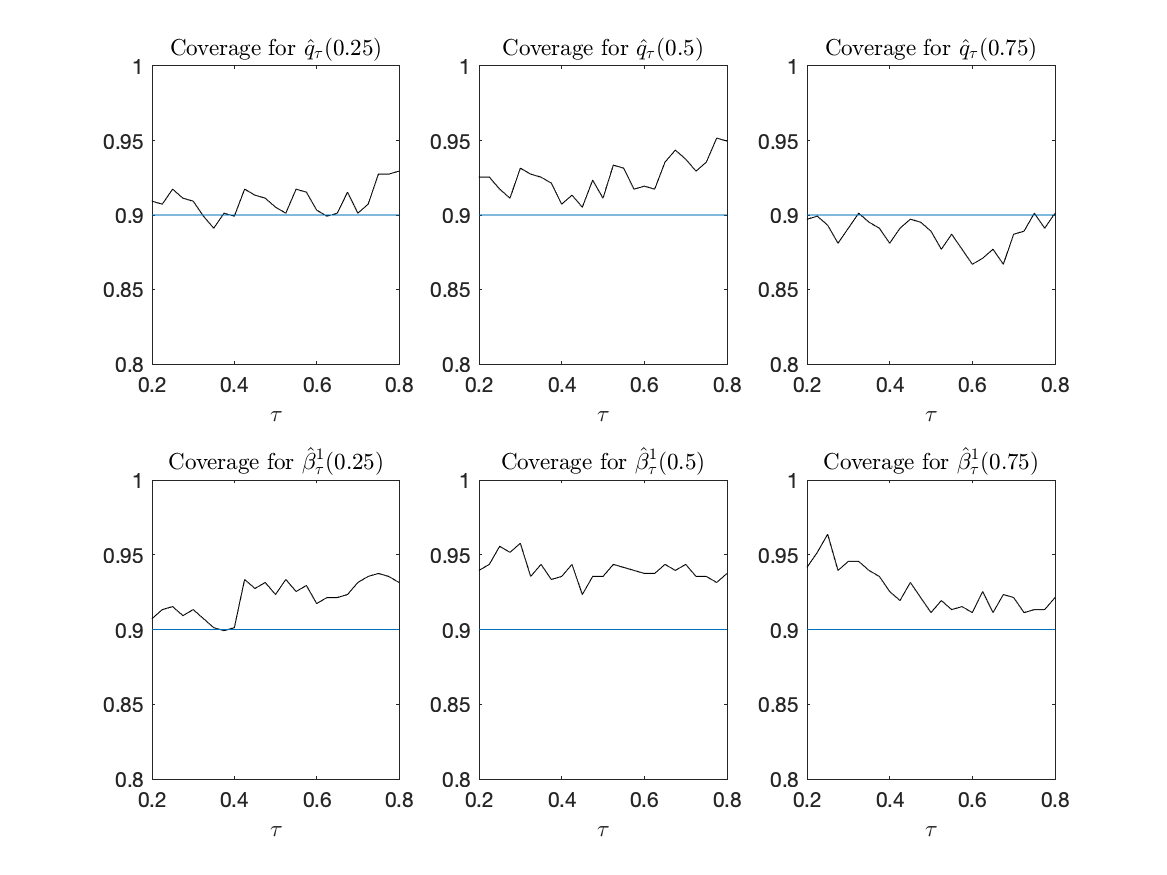}   
	\caption{coverage probability for large $\lambda$}
	\label{fig:comp6}
\end{figure}

\subsection{Oracle Estimators}
Next, we show the coverage probabilities for the oracle estimators in which the infinite-dimensional nuisance parameters are assumed to be known.
\begin{figure}[H]
	\centering
	\includegraphics[scale = 0.55,angle=0]{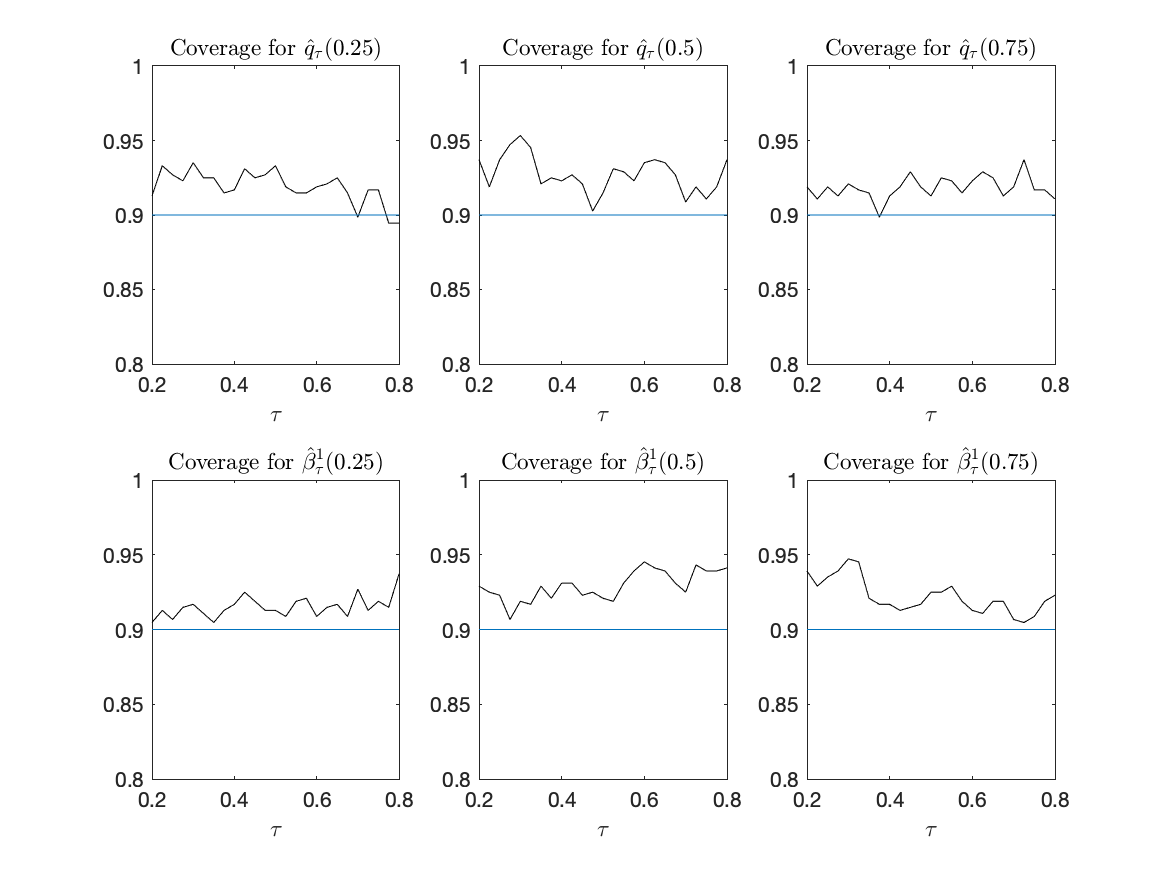}   
	\caption{DGP1, coverage probability for the oracle estimator}
	\label{fig:oracle1}
\end{figure}
\begin{figure}[H]
	\centering
	\includegraphics[scale = 0.55,angle=0]{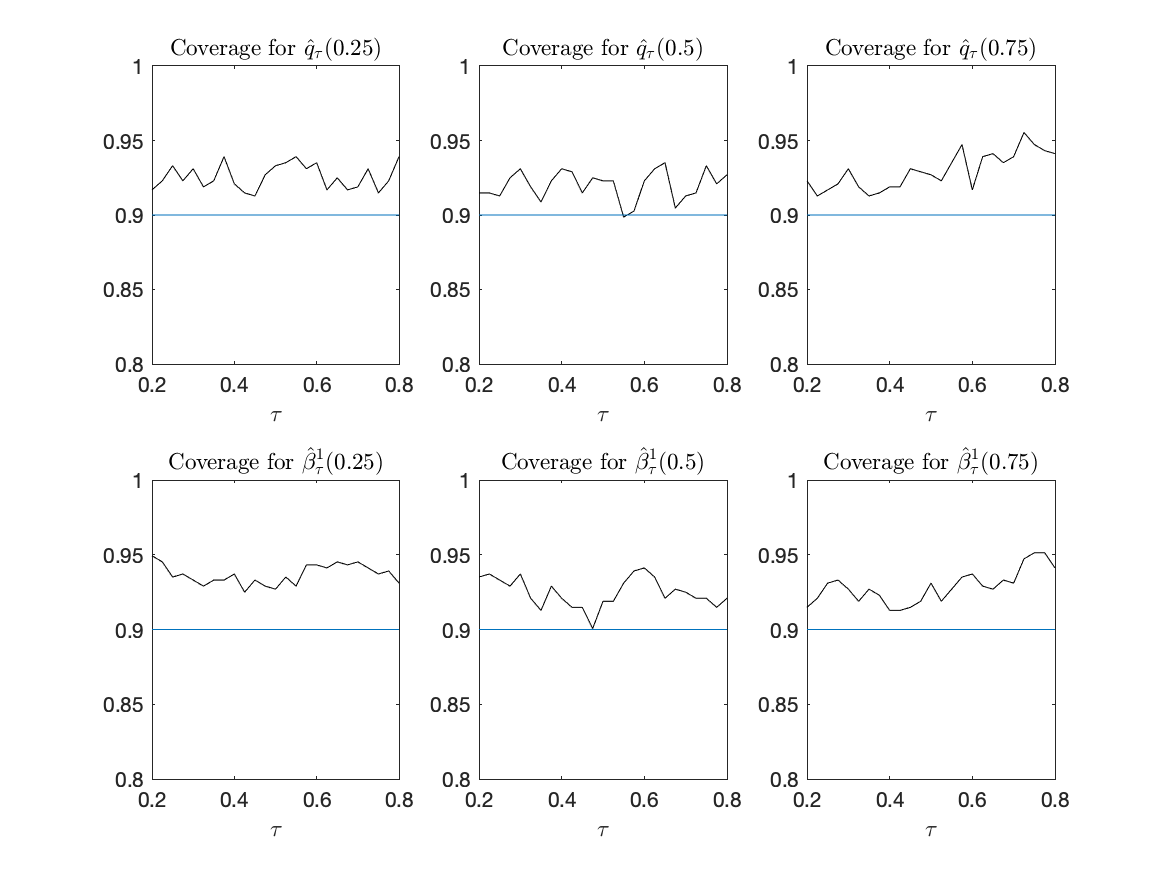}   
	\caption{DGP2, coverage probability for the oracle estimator}
	\label{fig:oracle2}
\end{figure}
\begin{figure}[H]
	\centering
	\includegraphics[scale = 0.55,angle=0]{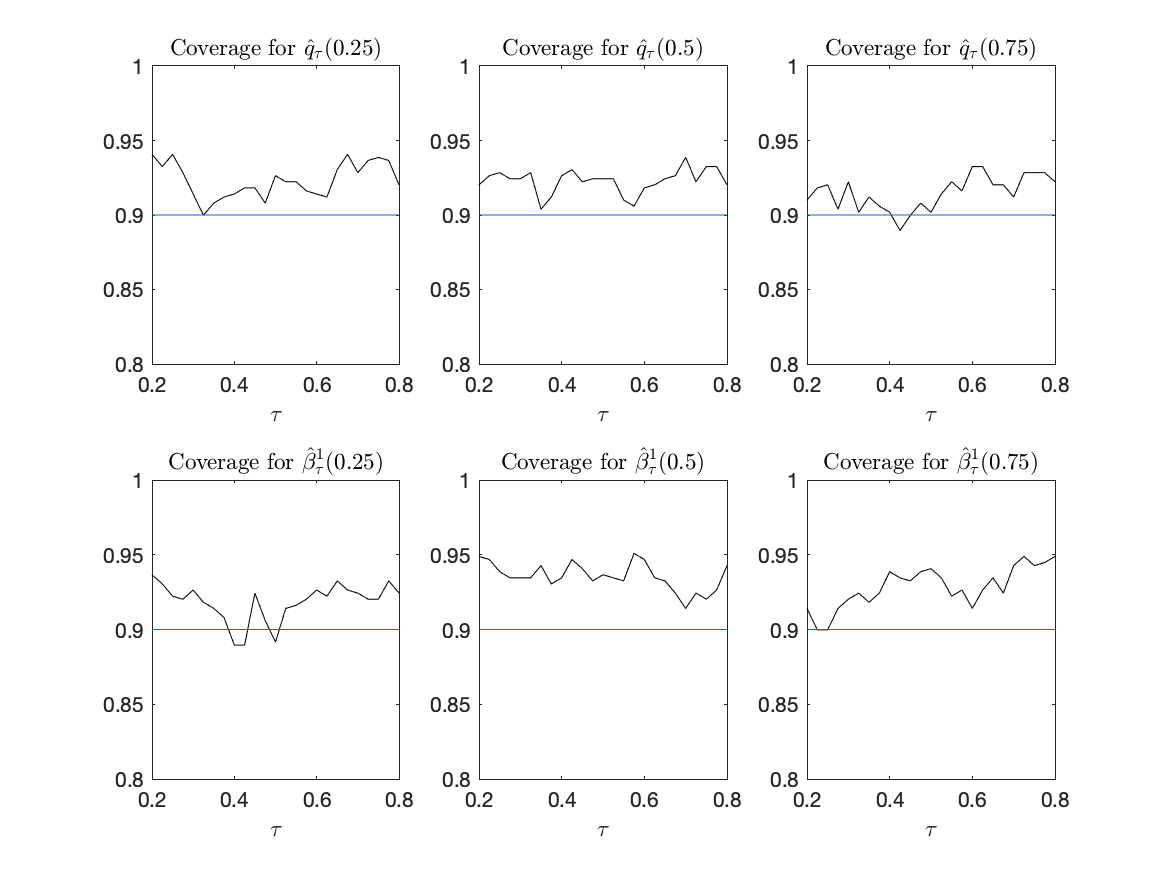}   
	\caption{DGP3, coverage probability for the oracle estimator}
	\label{fig:oracle3}
\end{figure}
We see that the coverage rates for the oracle estimators are conservative, which is due to the way we construct the confidence intervals. However, we can also see that for some values of $t$, the coverage rates are still very close to the nominal level 90\% and most coverage rates do not exceed 95\%.

\subsection{The Mean of the Potential Outcome}

We report the finite sample performance for the estimators for the mean of the potential outcome for $t \in [0.25,0.75]$.  
\begin{figure}[H]
	\centering
	\includegraphics[scale = 0.55,angle=0]{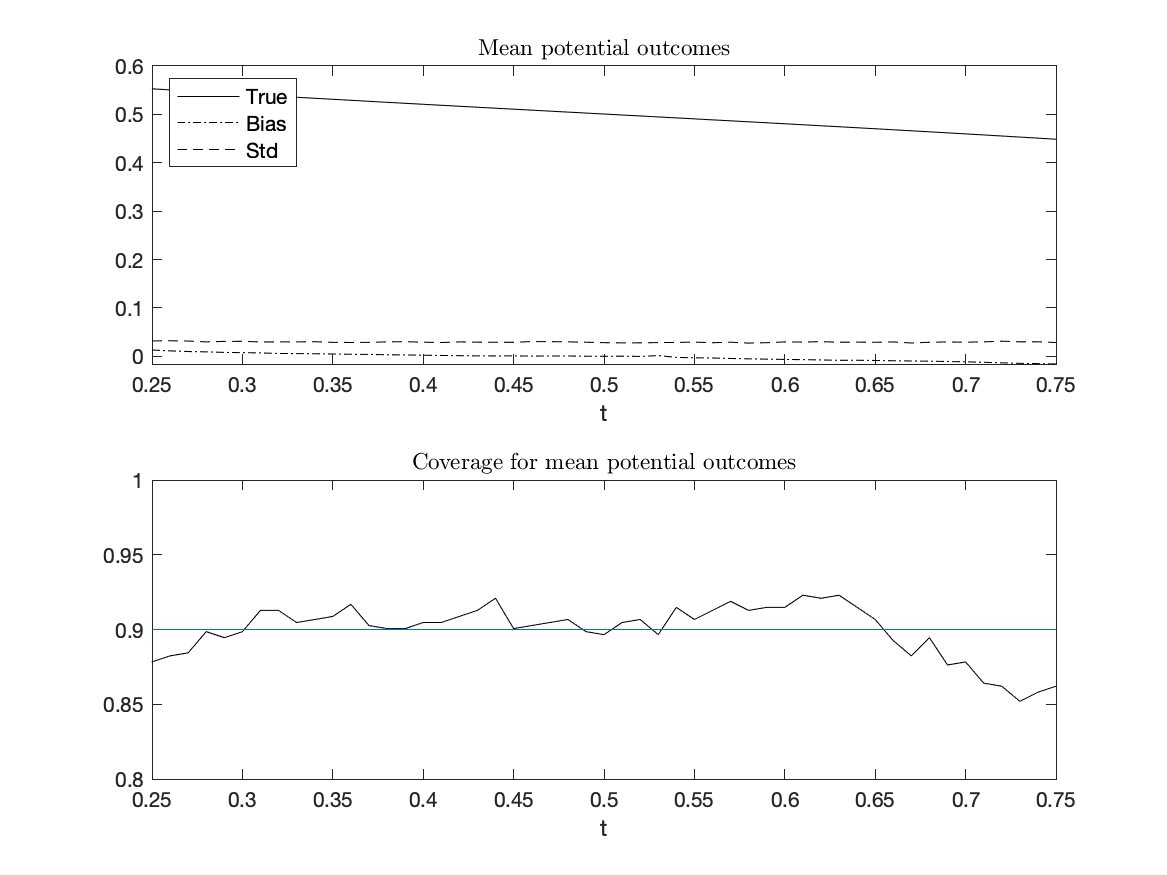}   
	\caption{DGP1, coverage probability}
	\label{fig:EY1}
\end{figure}
\begin{figure}[H]
	\centering
	\includegraphics[scale = 0.55,angle=0]{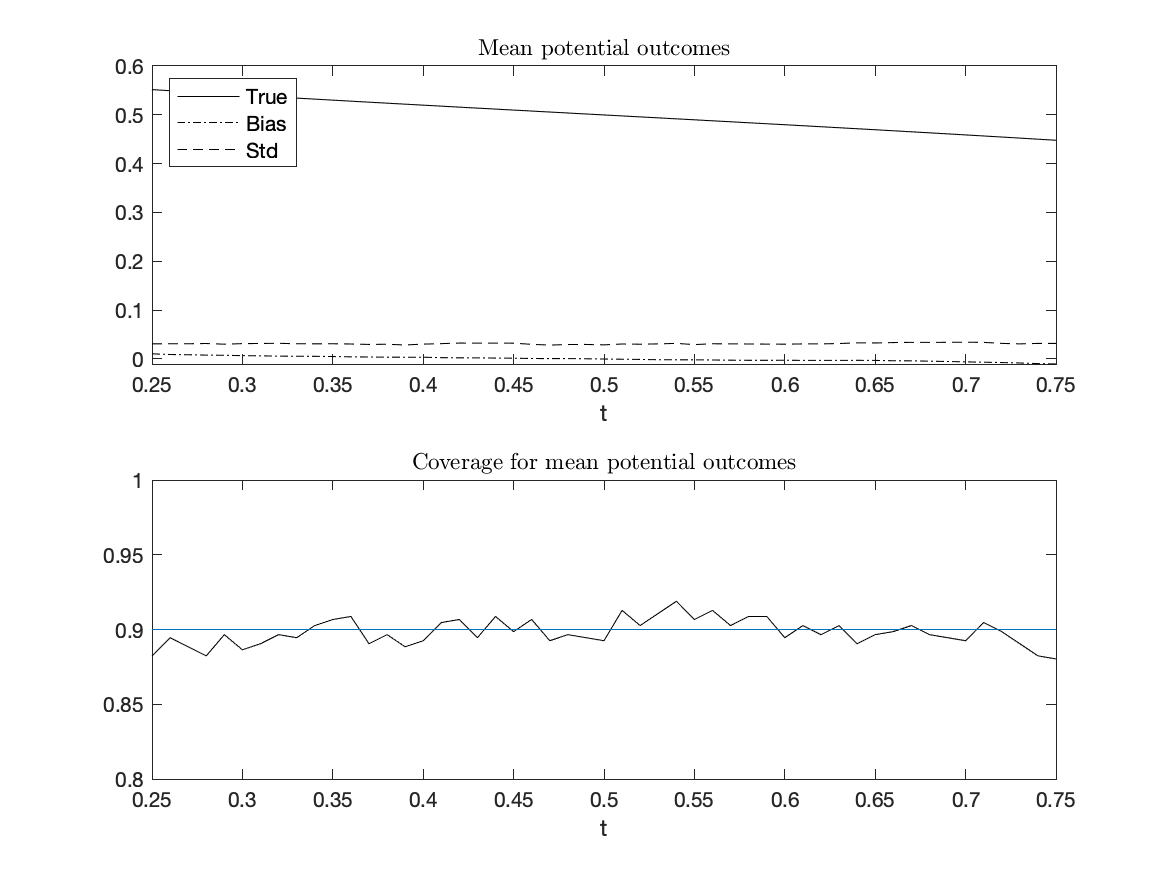}   
	\caption{DGP2, coverage probability}
	\label{fig:EY2}
\end{figure}
\begin{figure}[H]
	\centering
	\includegraphics[scale = 0.55,angle=0]{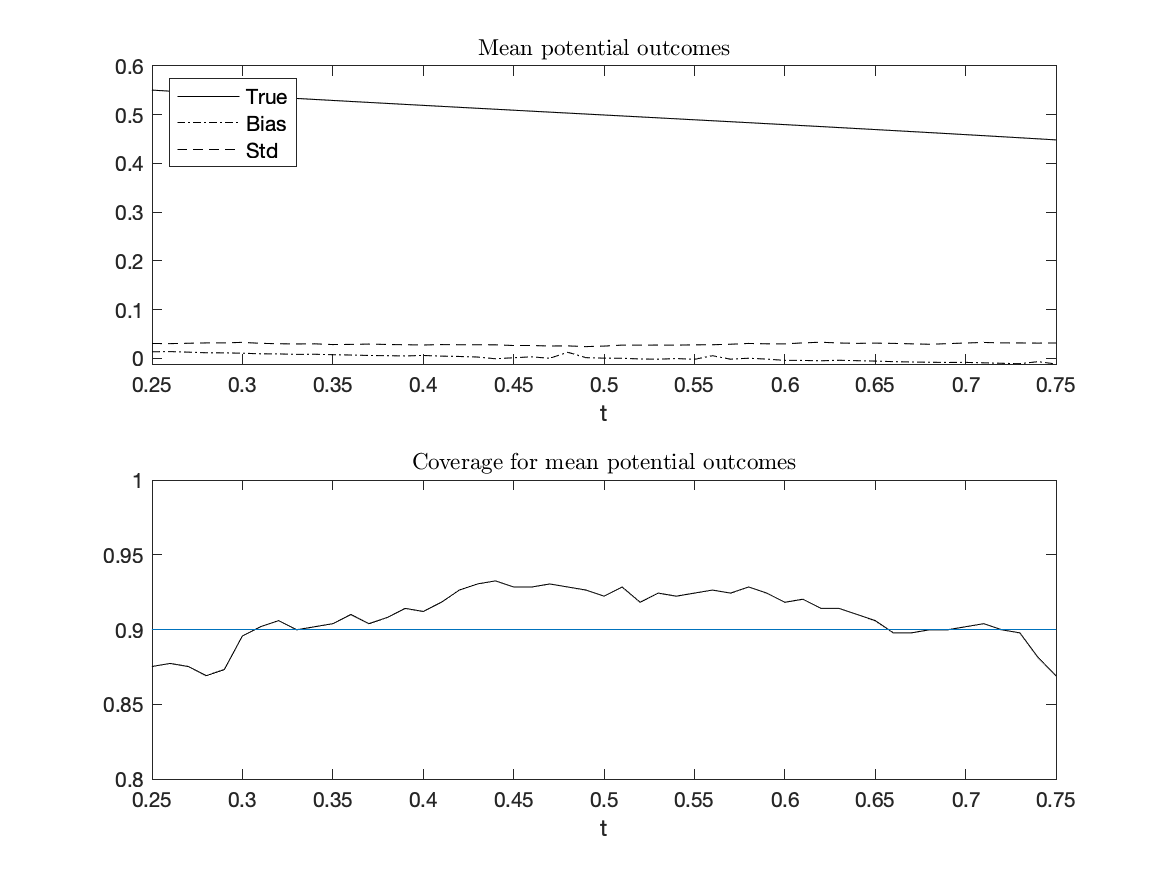}   
	\caption{DGP3, coverage probability}
	\label{fig:EY3}
\end{figure}
We observe that the estimators are quite accurate in terms of bias and variance. The coverage rates are reasonable for $t \in [0.25,0.75]$ in general. However, they are below the nominal rate $90\%$ when $t$ is close to $0.25$ and $0.75$. Comparing with the oracle results reported below, we see that the drop of coverage rates is mainly due to the variable selection, which has a larger effect for $t$ that is closer to the boundary. 
\begin{figure}[H]
	\centering
	\includegraphics[scale = 0.55,angle=0]{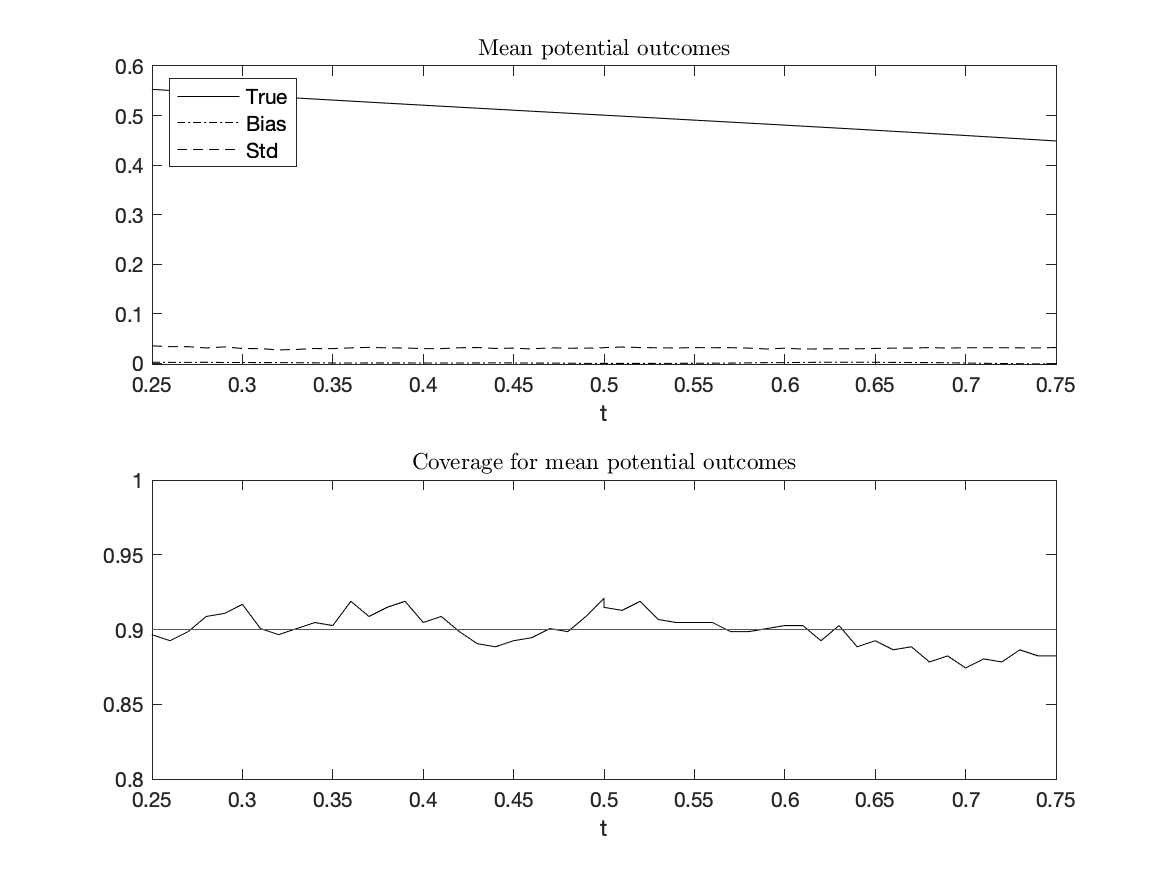}   
	\caption{DGP1, coverage probability for the oracle estimator}
	\label{fig:oracle_EY1}
\end{figure}
\begin{figure}[H]
	\centering
	\includegraphics[scale = 0.55,angle=0]{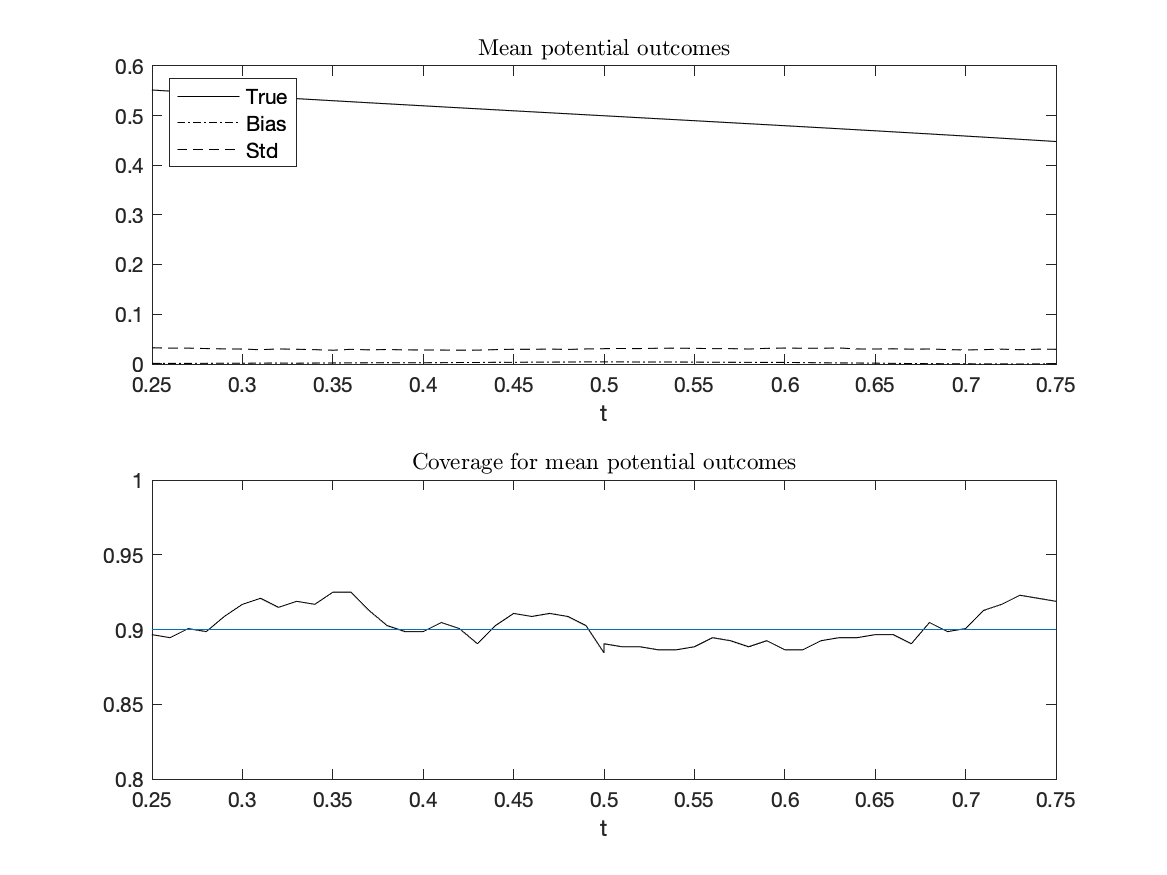}   
	\caption{DGP2, coverage probability for the oracle estimator}
	\label{fig:oracle_EY2}
\end{figure}
\begin{figure}[H]
	\centering
	\includegraphics[scale = 0.55,angle=0]{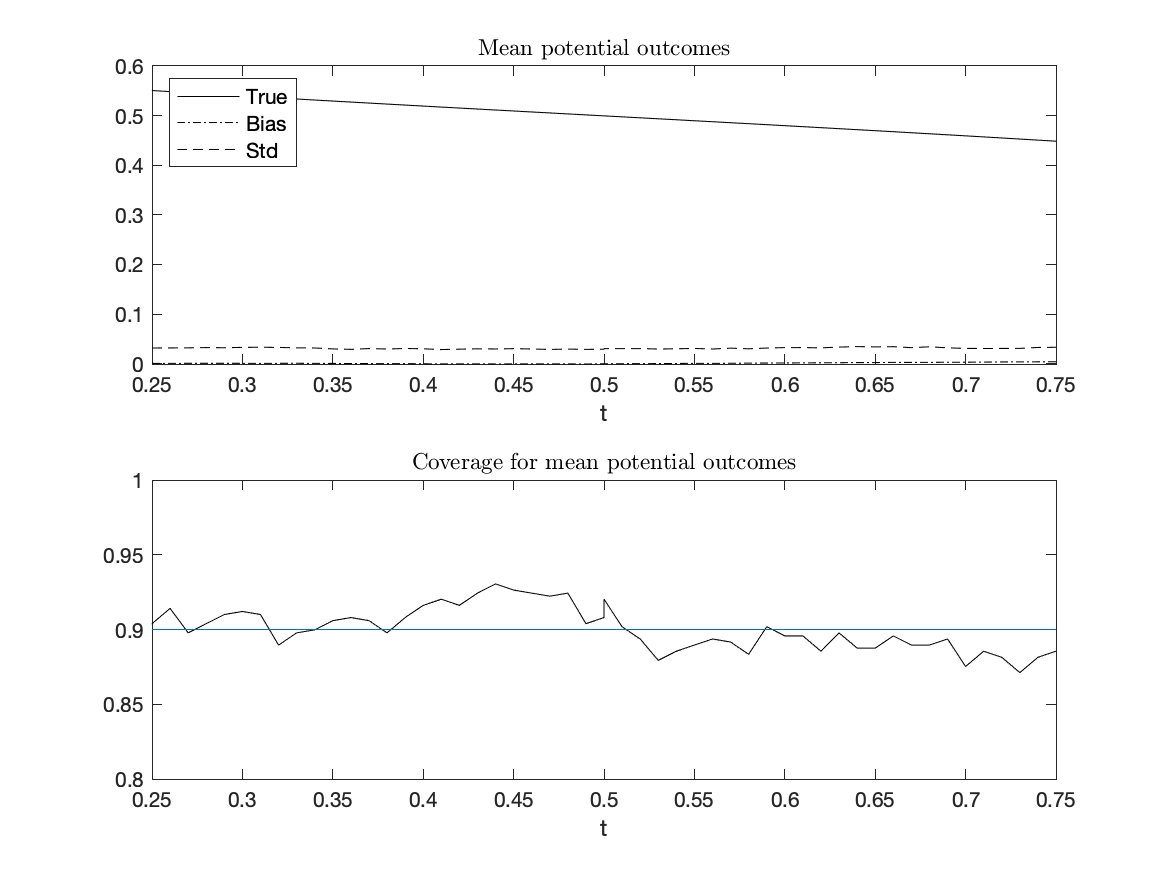}   
	\caption{DGP3, coverage probability for the oracle estimator}
	\label{fig:oracle_EY3}
\end{figure}

\subsection{An Additional Design}

Last, we consider a design that violates the approximate sparsity condition. The outcome and treatment equations are the same as \eqref{eq:YXT} and \eqref{eq:T|X}, respectively. We let $\beta_j = \frac{\pi^2}{24}$ for $j=1,\cdots,10$, $\beta_j = 0$, $j = 11,\cdots,100$, and $b(X) = X$. In this case, $s = 10$. Recall that we have $nh_1 \approxeq 47$. However, our theory requires that $s/\sqrt{nh_1} \rightarrow 0$. Such a condition is violated in this design.  
\begin{figure}[H]
	\centering
	\includegraphics[scale = 0.55,angle=0]{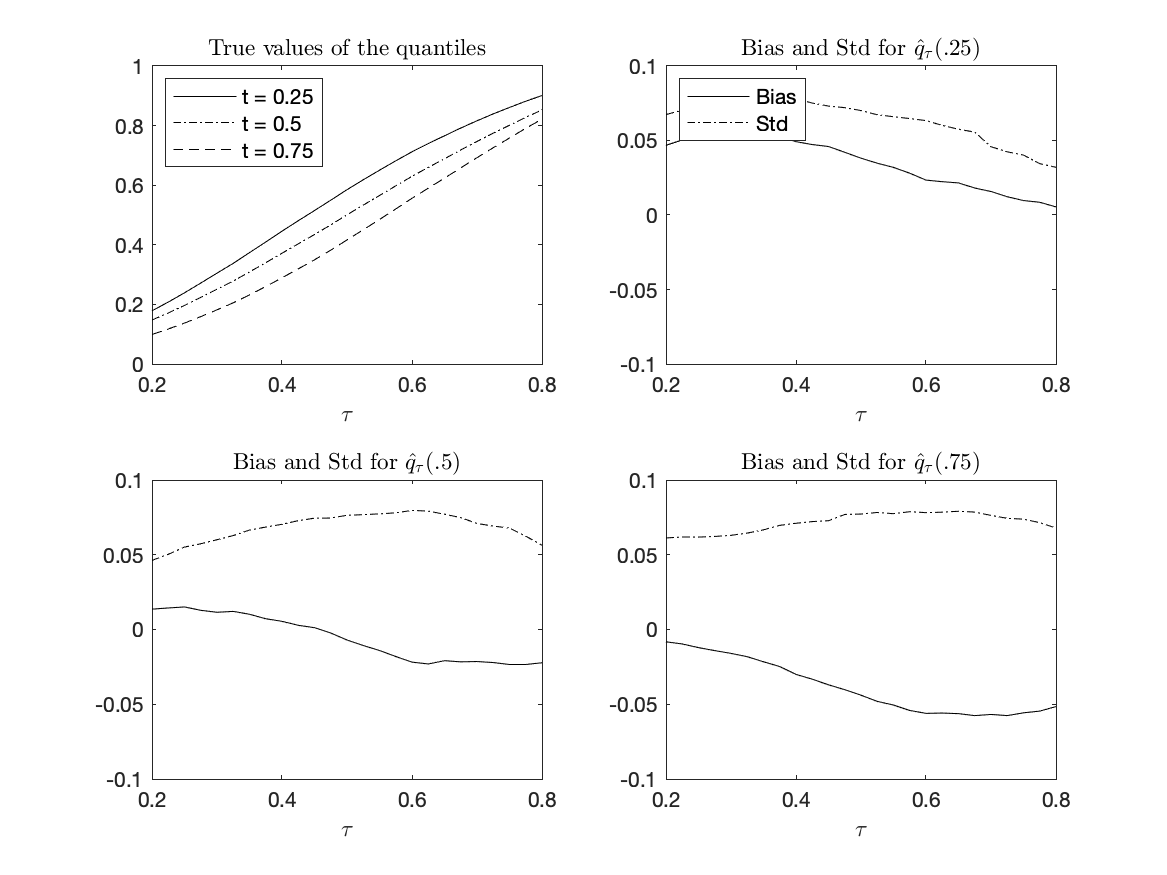}   
	\caption{DGP4, finite sample performance of $\hat{q}_{\protect\tau }(t)$}
	\label{fig:beta0_4}
\end{figure}
\begin{figure}[H]
	\centering
	\includegraphics[scale = 0.565,angle=0]{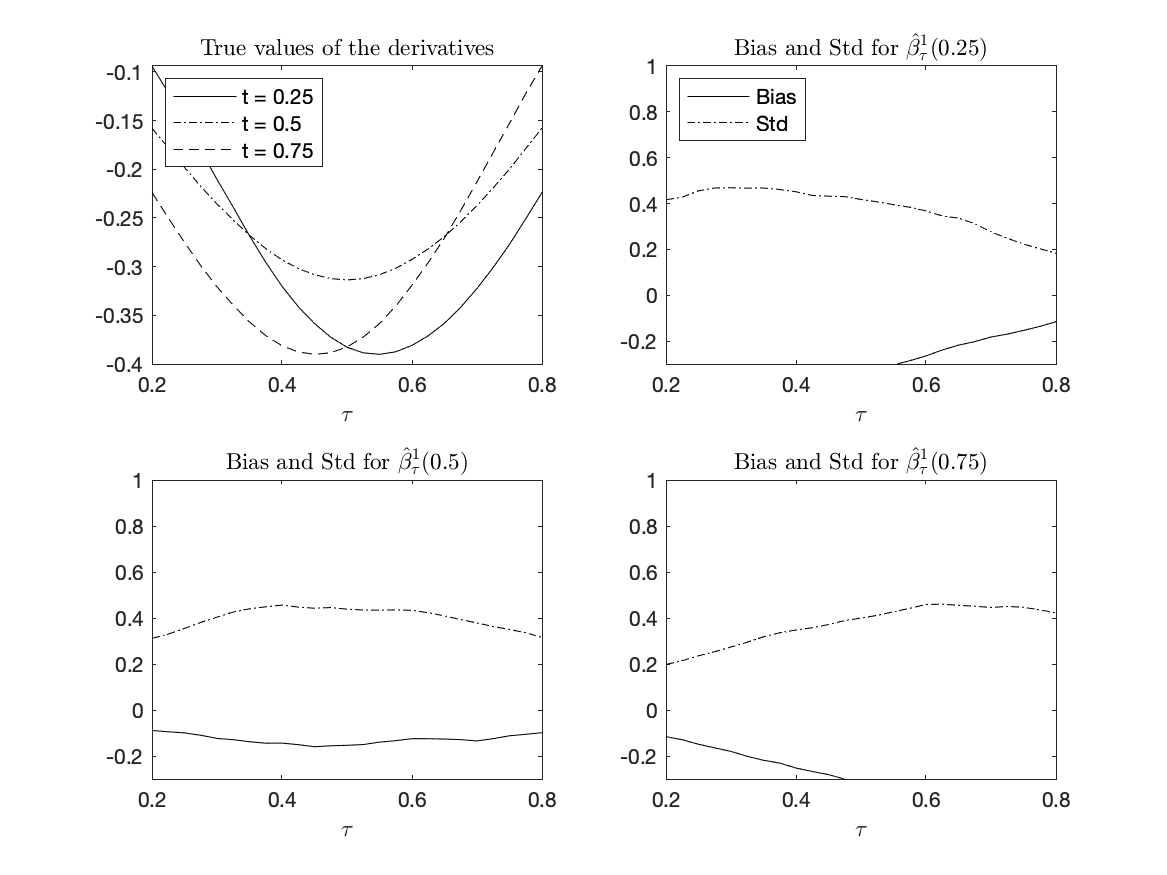}  
	\caption{DGP4, finite sample performance of $\hat{\protect\beta}^1_\protect\tau(t)
		$}
	\label{fig:beta1_4}
\end{figure}

\begin{figure}[H]
	\centering
	\includegraphics[scale = 0.565,angle=0]{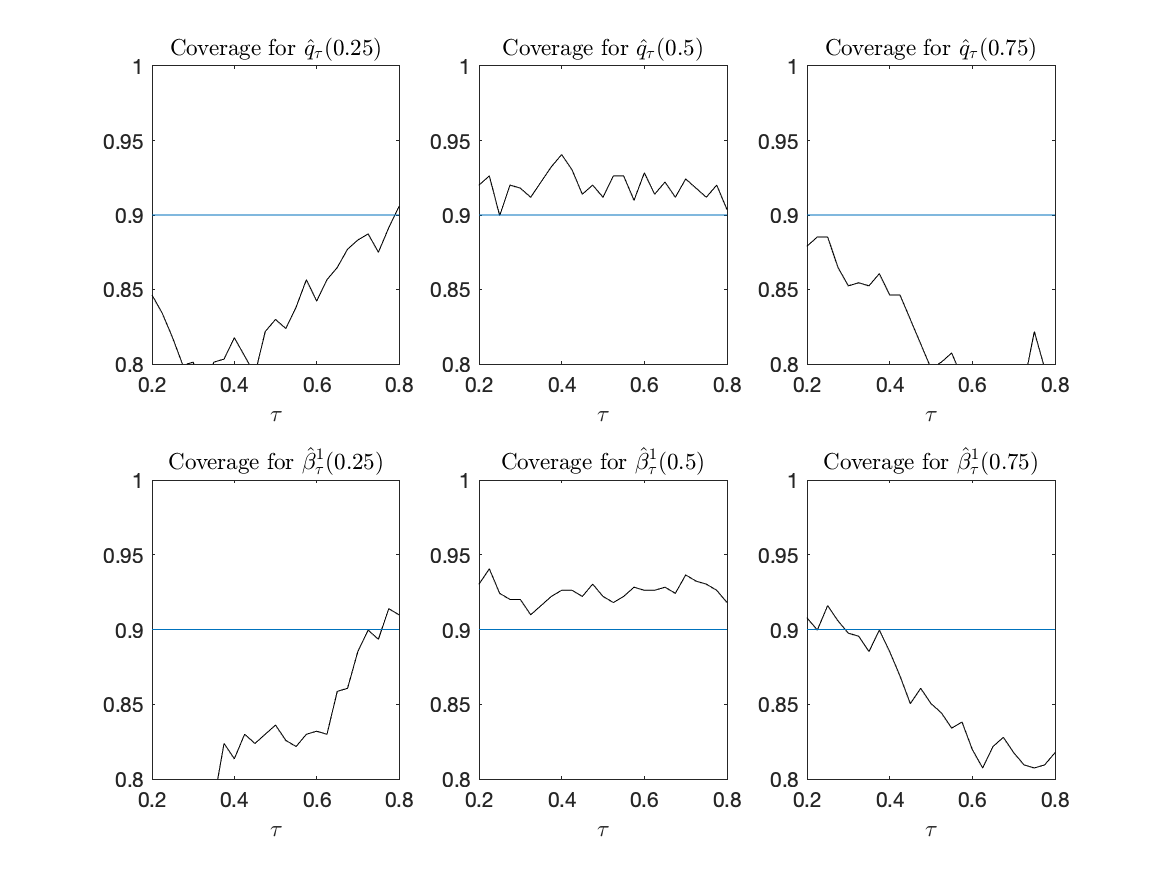}  
	\caption{DGP4, coverage probability}
	\label{fig:cov_4}
\end{figure}
We see that the coverage rates when $t=0.5$ are satisfactory. For $t = 0.25$ and $t=0.75$, the coverage rates are below the nominal 90\%. On the other hand, the coverage rates for the oracle estimators reported below perform quite well. This implies that the drop of coverage rates for our estimators is mainly due to the variable selection, which may have a larger effect when $t$ is away from the center.\footnote{Again, the cross-fitting technique promoted in \cite{CC18} may be helpful for eliminating the variable selection bias.}

\begin{figure}[H]
	\centering
	\includegraphics[scale = 0.55,angle=0]{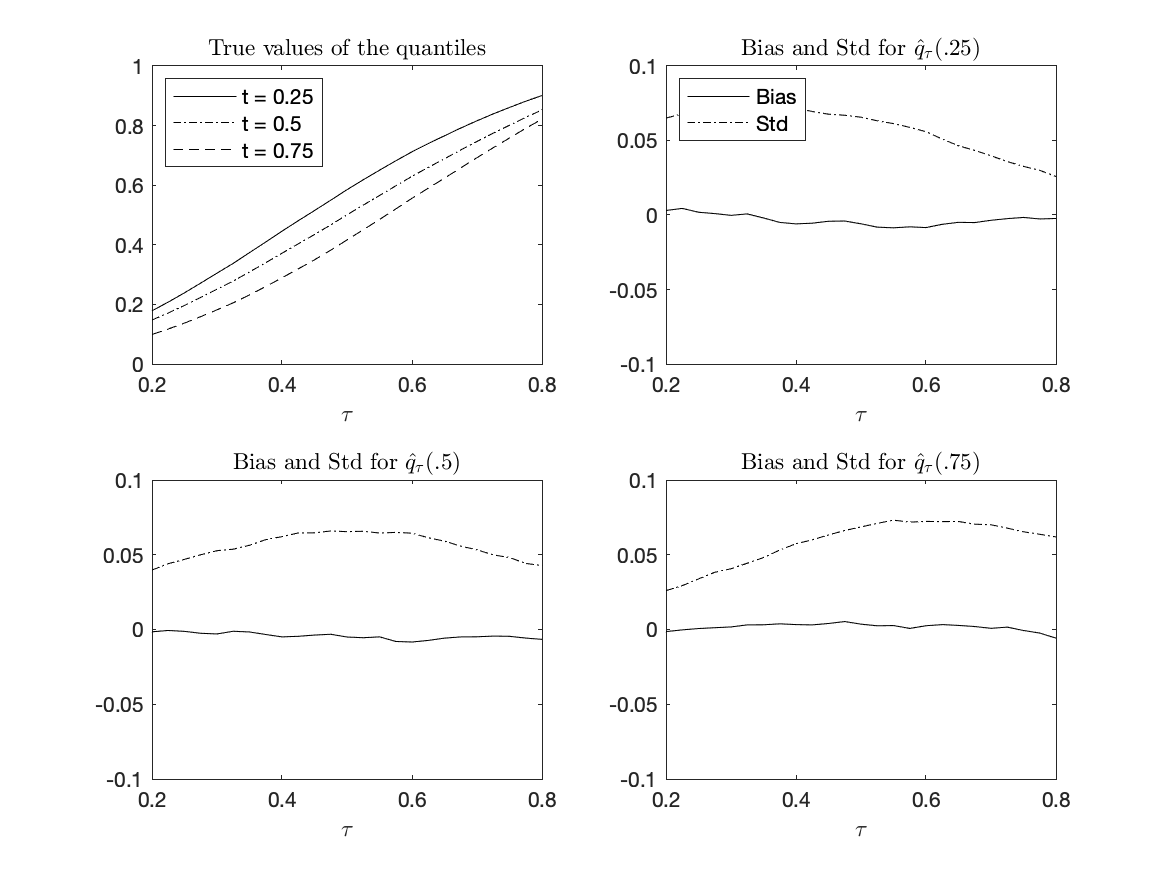}   
	\caption{DGP4, Finite sample performance for the oracle estimator of $q_{\protect\tau }(t)$}
	\label{fig:oracle_beta0_4}
\end{figure}
\begin{figure}[H]
	\centering
	\includegraphics[scale = 0.565,angle=0]{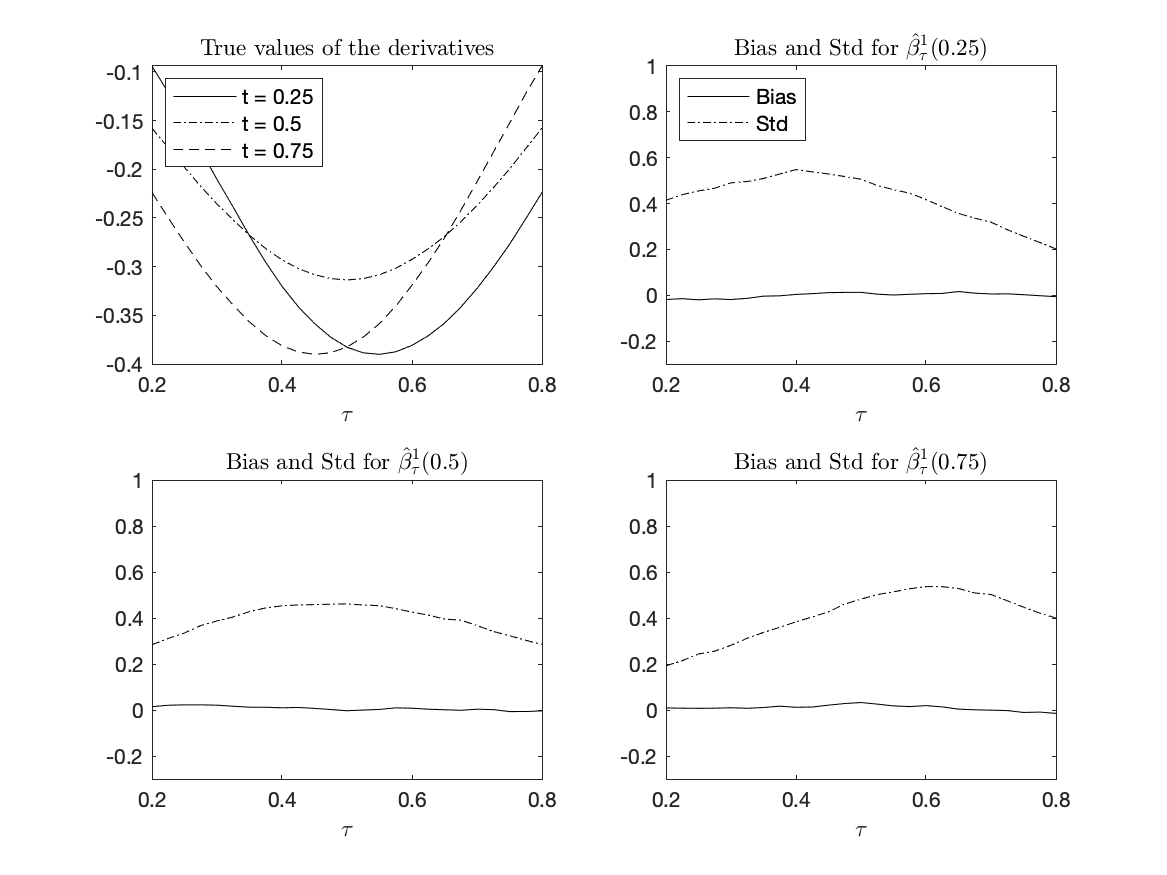}  
	\caption{DGP4, Finite sample performance  for the oracle estimator of $\beta^1_\protect\tau(t)
		$}
	\label{fig:oracle_beta1_4}
\end{figure}

\begin{figure}[H]
	\centering
	\includegraphics[scale = 0.565,angle=0]{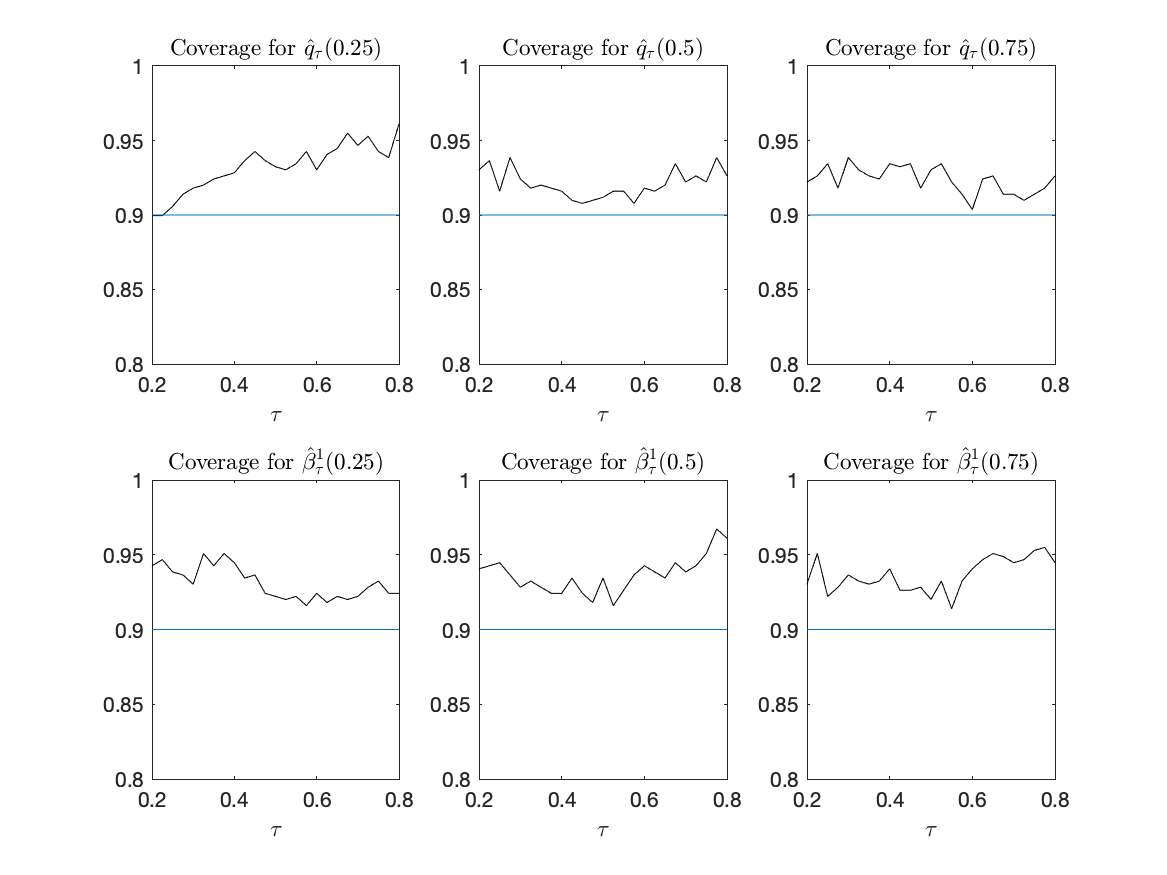}  
	\caption{DGP4, Coverage probability for the oracle estimators}
	\label{fig:oracle_cov_4}
\end{figure}

\section{Additional Empirical Illustration Results}
\label{sec:add_emp}
This section investigates the sensitivity of our empirical application results with respect to three tuning parameters: $h_1$, $\tilde{\lambda}$, and $\lambda$. We use the same model and dataset as in Section \ref{sec:app}. Figures \ref{fig:comp1_app_white}-\ref{fig:comp6_app_white} are about the white individuals, and Figures \ref{fig:comp1_app_black}-\ref{fig:comp6_app_black} are about the black individuals.  
The captions for these figures are the same as in Figures \ref{fig_white} and \ref{fig_black}. 
Figures \ref{fig:comp1_app_white} and \ref{fig:comp2_app_white} show the estimation results for $q_{\tau}(t)$ and $\beta _{\tau}^{1}(t)$ with $%
h_1'=0.8h_{1}$ and $%
h_1'=1.2h_{1}$, respectively. Figures \ref{fig:comp3_app_white} and \ref{fig:comp4_app_white} show the estimation results for $q_{\tau}(t)$ and $\beta _{\tau}^{1}(t)$ with $%
\tilde{\lambda}'=0.8\tilde{\lambda}$ and $%
\tilde{\lambda}'=1.2\tilde{\lambda}$, respectively, where $\tilde{\lambda}$ is the penalty used to estimate the conditional density $f_t(X)$. Last, Figures \ref{fig:comp5_app_white} and \ref{fig:comp6_app_white} show the estimation results for  $q_{\tau}(t)$ and $\beta _{\tau}^{1}(t)$ with $%
\lambda'=0.8\lambda$ and $%
\lambda'=1.2\lambda$, respectively, where $\lambda$ is the penalty used to estimate the conditional CDF $\phi_{t,u}(X)$.

\subsection{Sensitivity results for the white individuals}

\begin{figure}[H]
\centering
\includegraphics[scale = 0.55,angle=0]{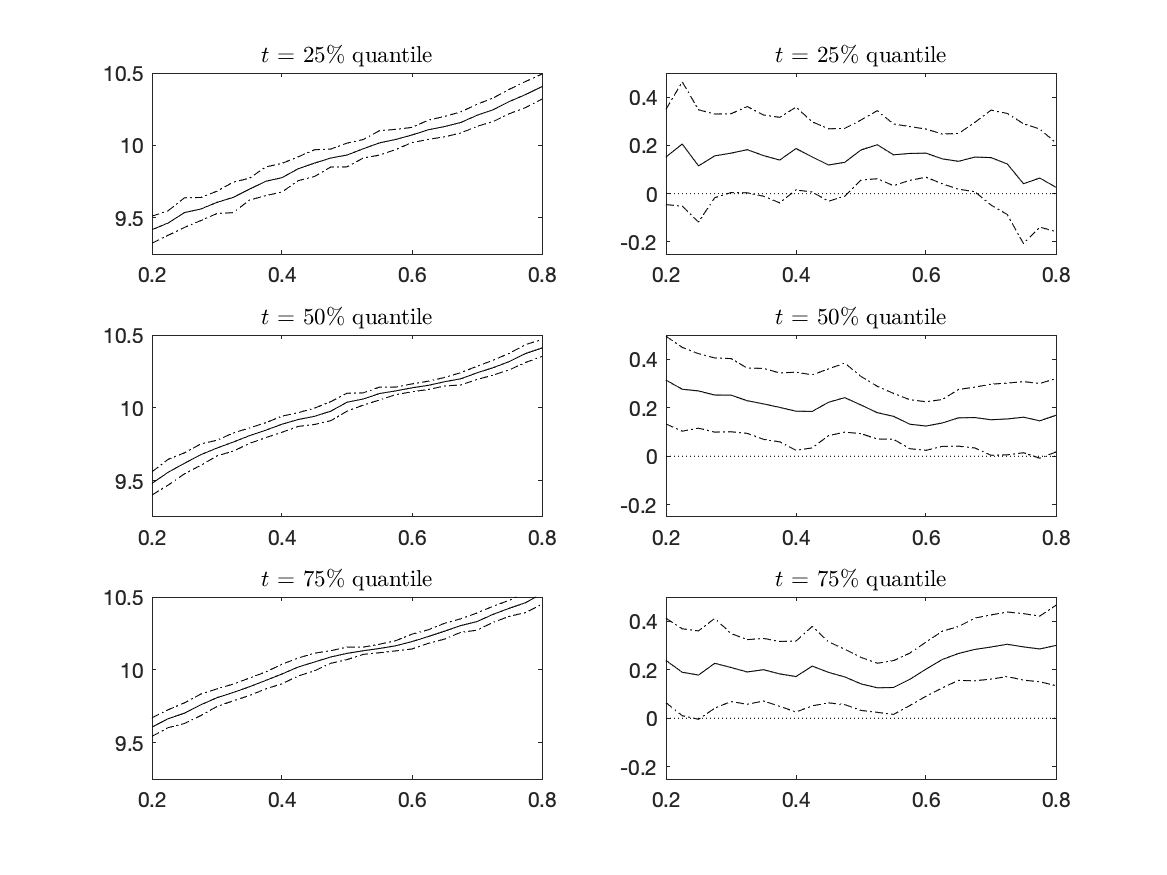}   
\caption{Empirical results for whites with small $h_1$}
\label{fig:comp1_app_white}
\end{figure}

\begin{figure}[H]
 \centering
\includegraphics[scale = 0.55,angle=0]{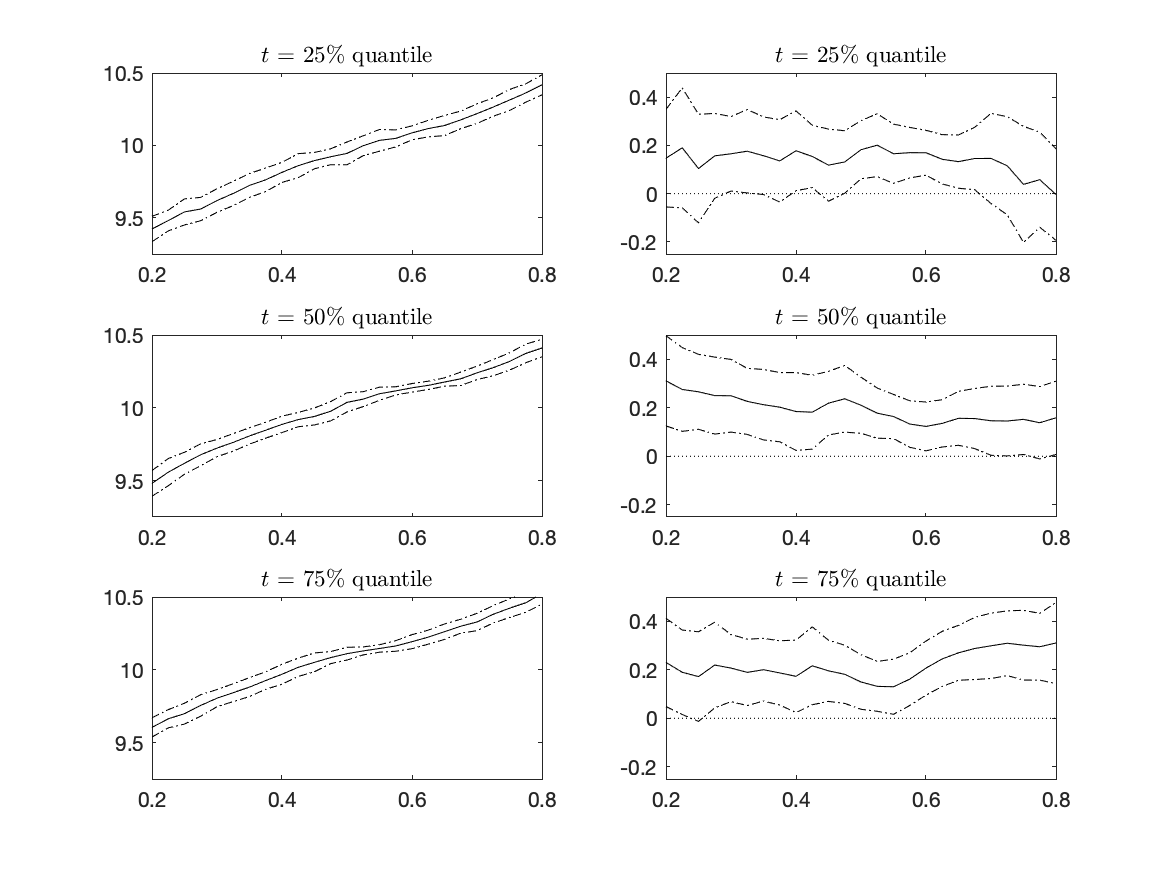}   
\caption{Empirical results for whites with large $h_1$}
\label{fig:comp2_app_white}
\end{figure}

\begin{figure}[H]
	\centering
\includegraphics[scale = 0.55,angle=0]{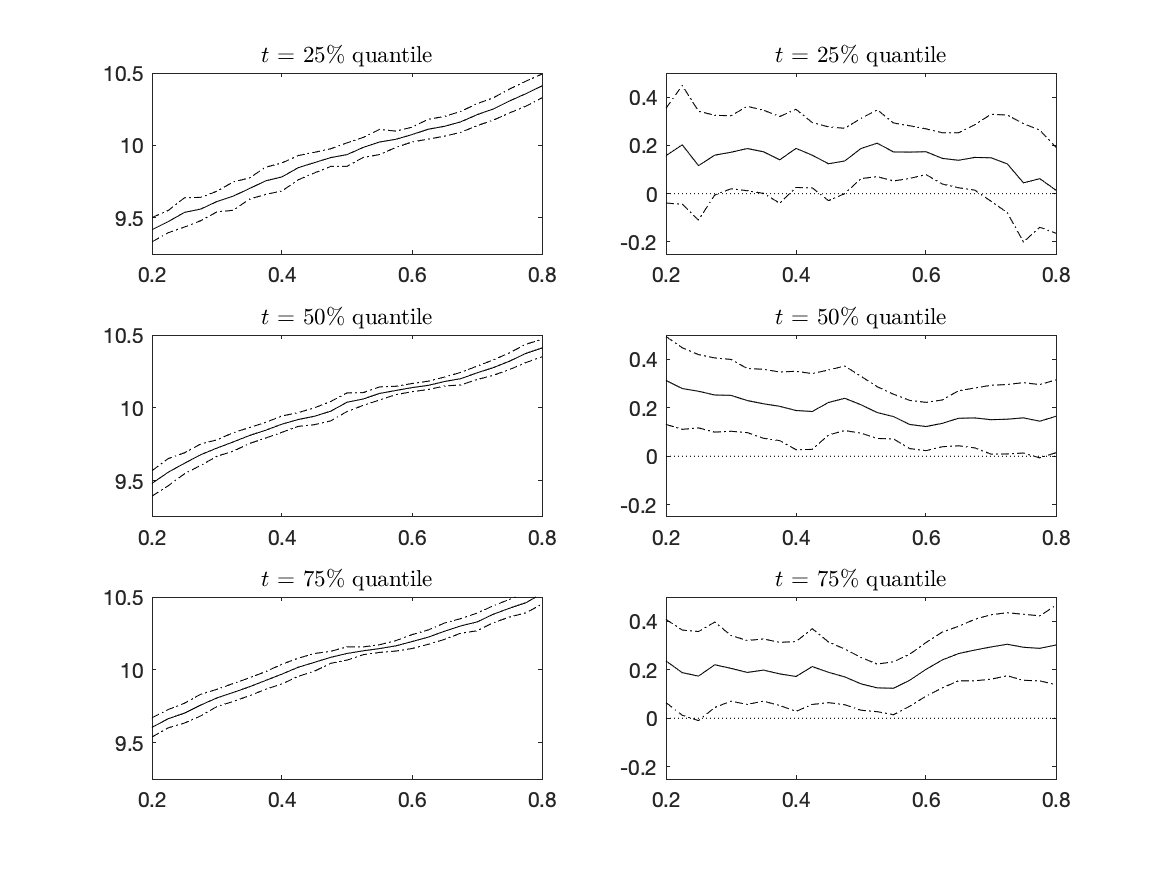}   
	\caption{Empirical results for whites with small $\tilde{\lambda}$}
\label{fig:comp3_app_white}
\end{figure}

\begin{figure}[H]
	\centering
\includegraphics[scale = 0.55,angle=0]{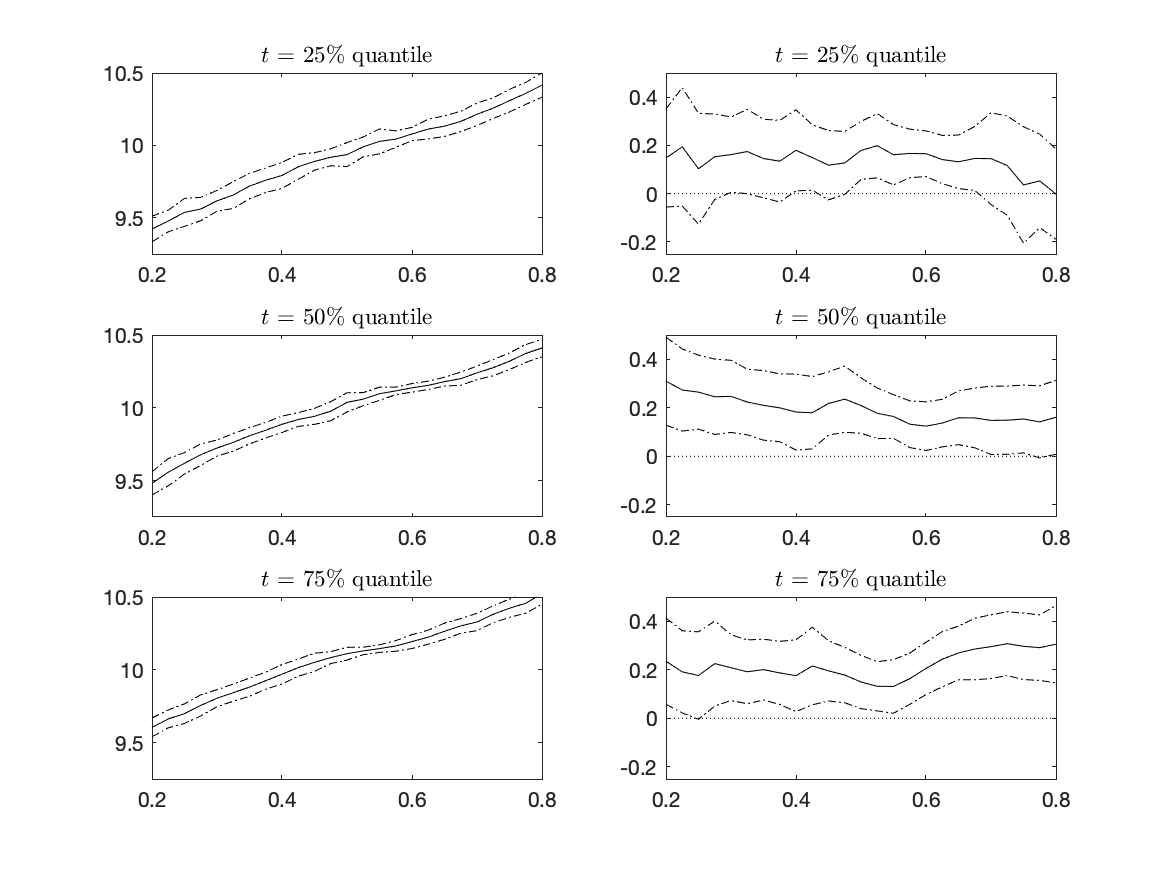}   
	\caption{Empirical results for whites with large $\tilde{\lambda}$}
\label{fig:comp4_app_white}
\end{figure}

\begin{figure}[H]
	\centering
\includegraphics[scale = 0.55,angle=0]{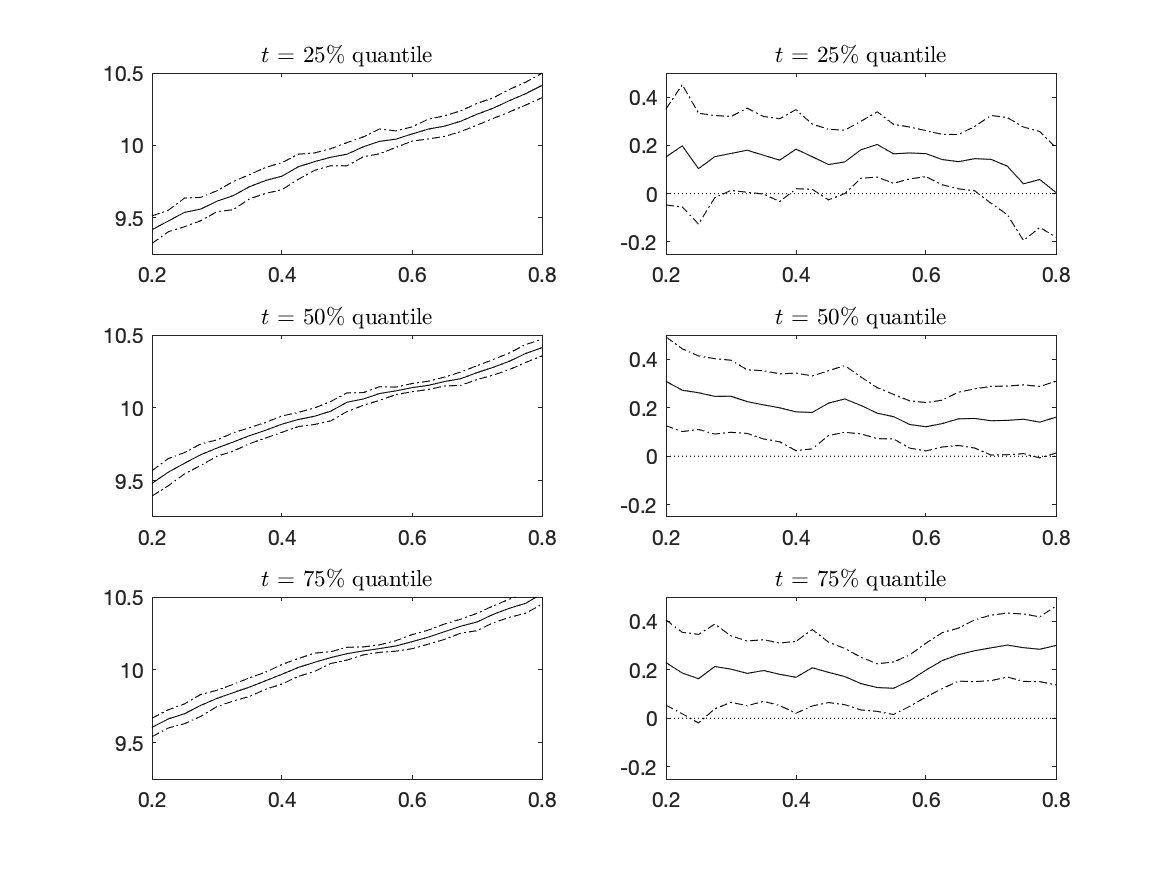}   
	\caption{Empirical results for whites with small $\lambda$}
\label{fig:comp5_app_white}
\end{figure}

\begin{figure}[H]
	\centering
\includegraphics[scale = 0.55,angle=0]{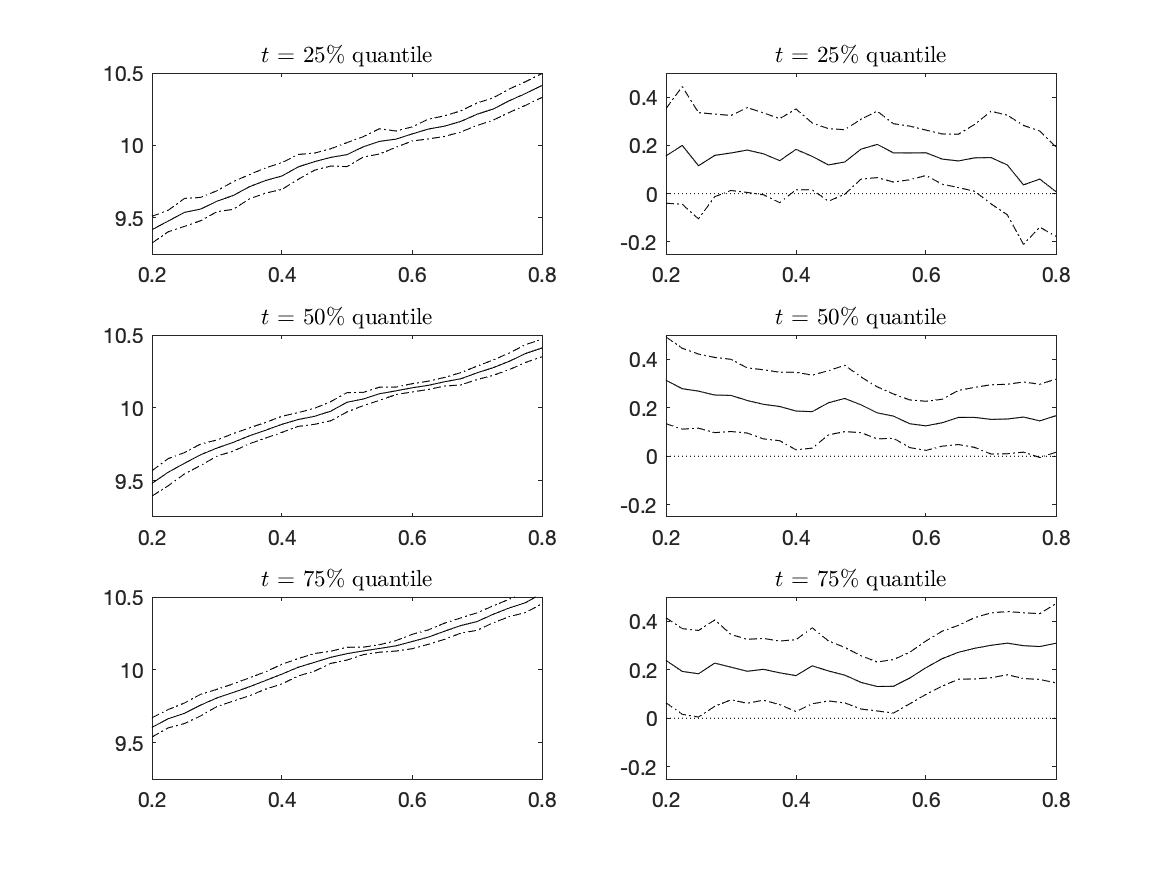}   
	\caption{Empirical results for whites with large $\lambda$}
\label{fig:comp6_app_white}
\end{figure}

\subsection{Sensitivity results for the black individuals}

\begin{figure}[H]
\centering
\includegraphics[scale = 0.55,angle=0]{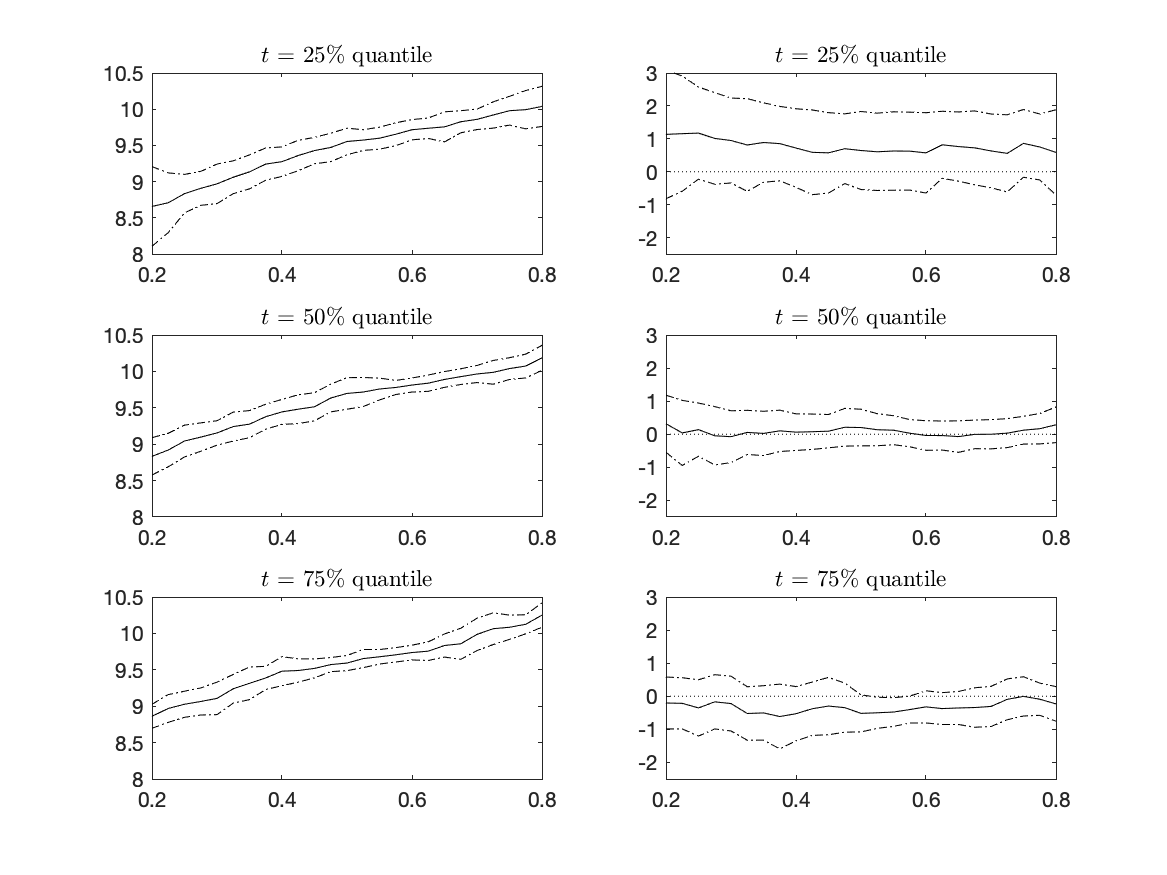}   
\caption{Empirical results for blacks with small $h_1$}
\label{fig:comp1_app_black}
\end{figure}

\begin{figure}[H]
 \centering
\includegraphics[scale = 0.55,angle=0]{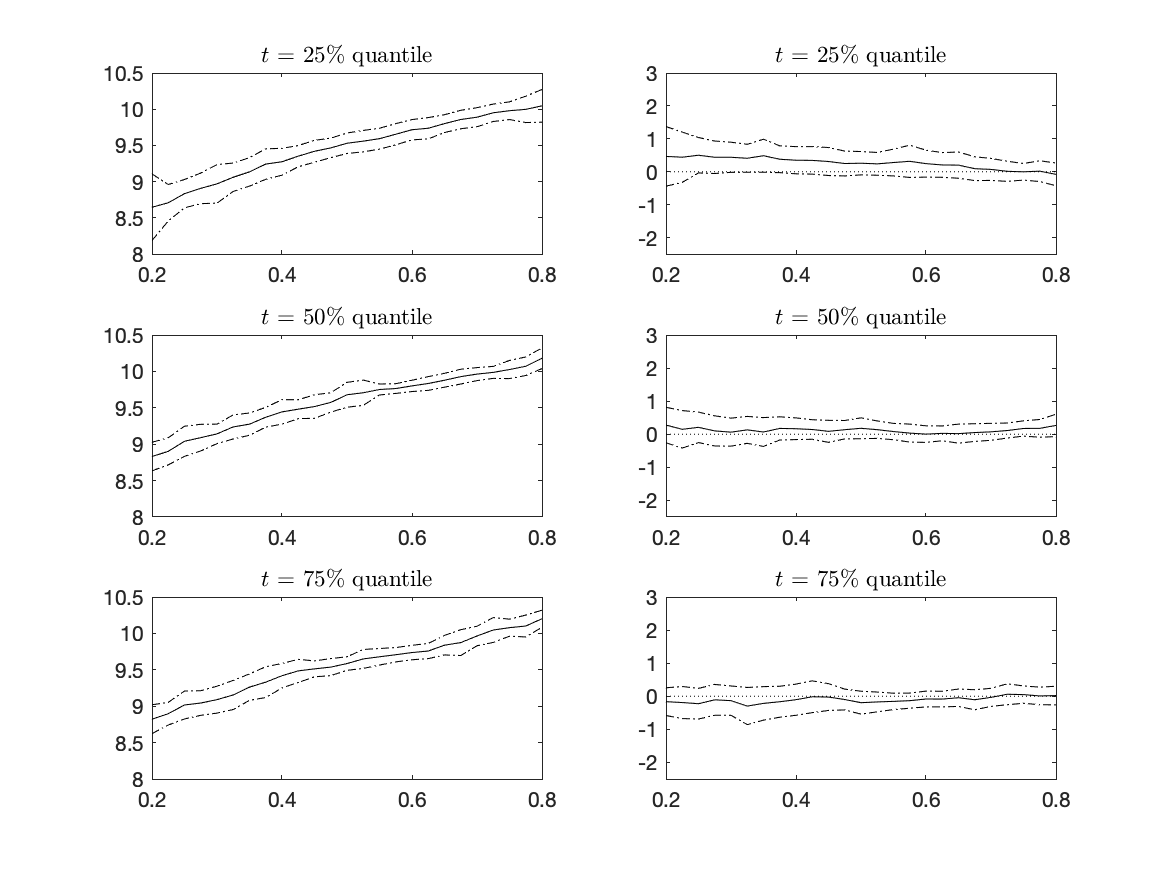}   
\caption{Empirical results for blacks with large $h_1$}
\label{fig:comp2_app_black}
\end{figure}

\begin{figure}[H]
	\centering
\includegraphics[scale = 0.55,angle=0]{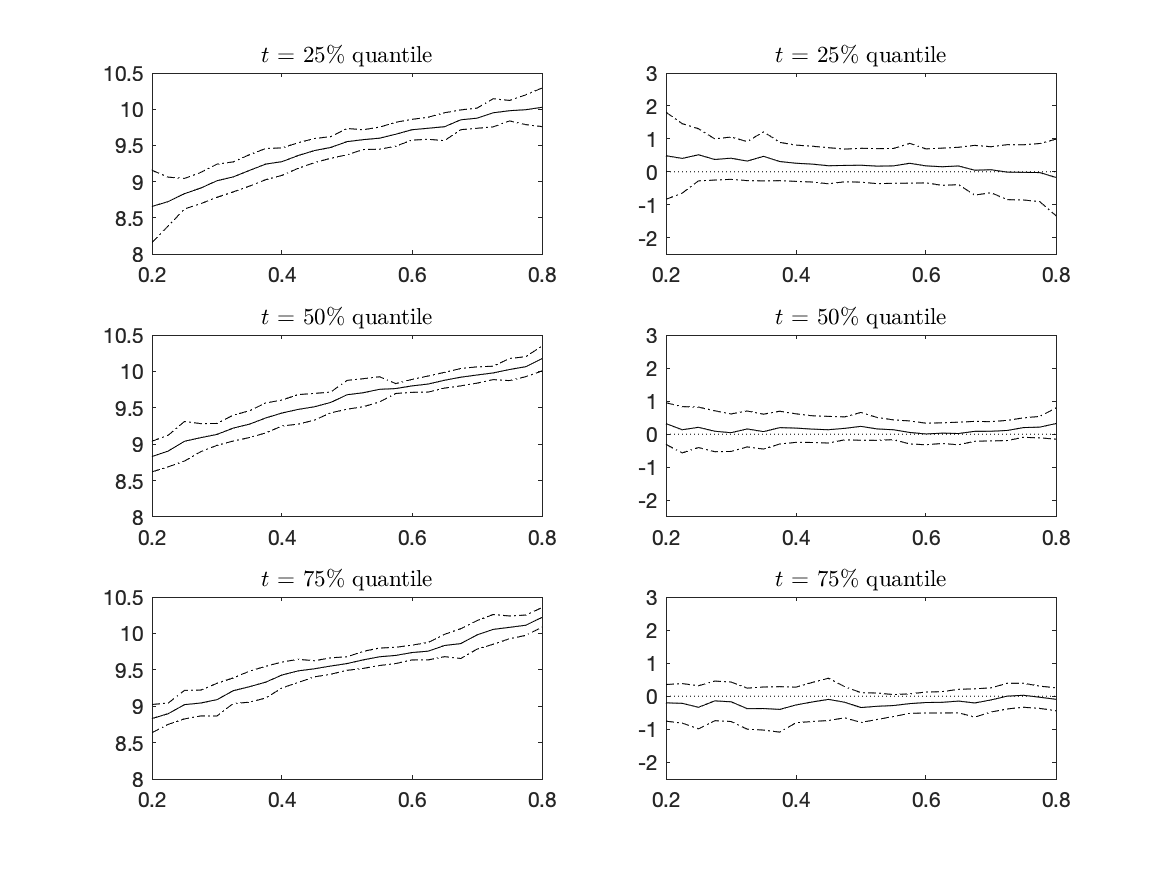}   
	\caption{Empirical results for blacks with small $\tilde{\lambda}$}
\label{fig:comp3_app_black}
\end{figure}

\begin{figure}[H]
	\centering
\includegraphics[scale = 0.55,angle=0]{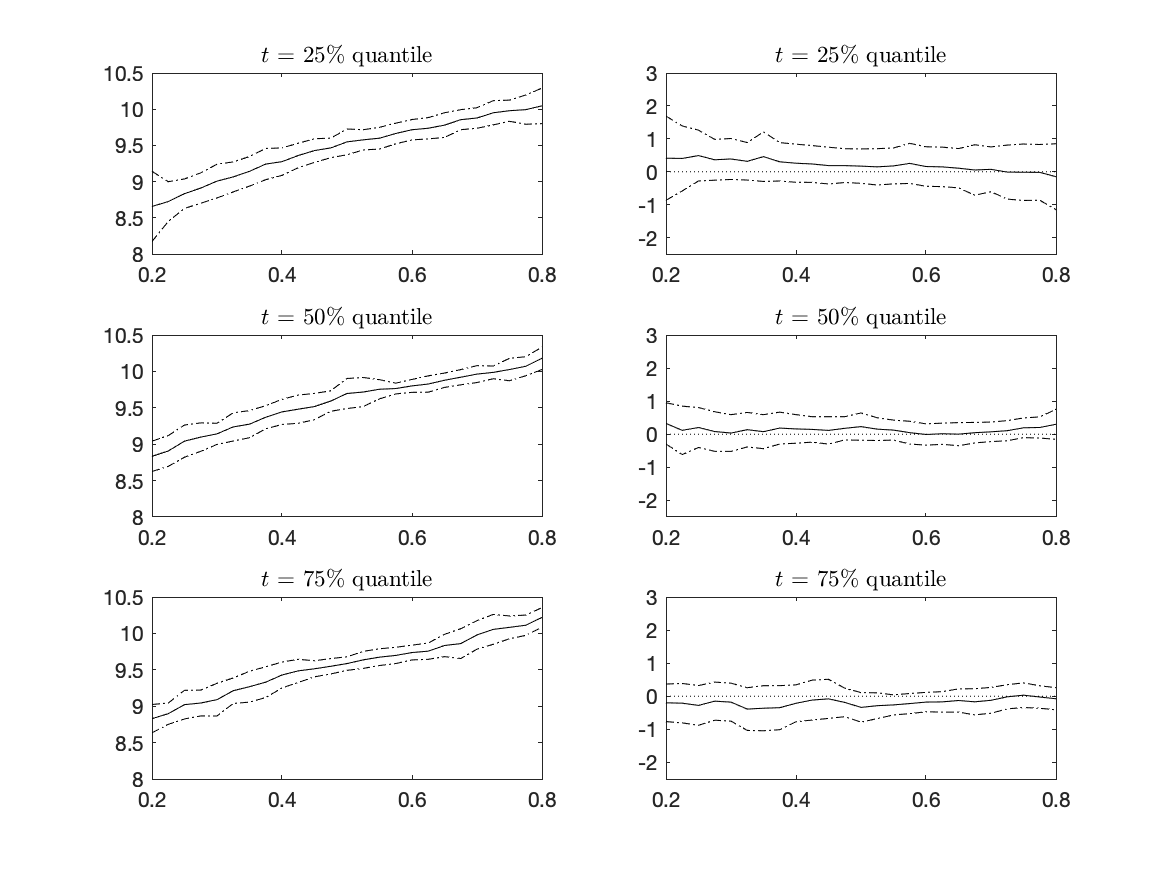}   
	\caption{Empirical results for blacks with large $\tilde{\lambda}$}
\label{fig:comp4_app_black}
\end{figure}

\begin{figure}[H]
	\centering
\includegraphics[scale = 0.55,angle=0]{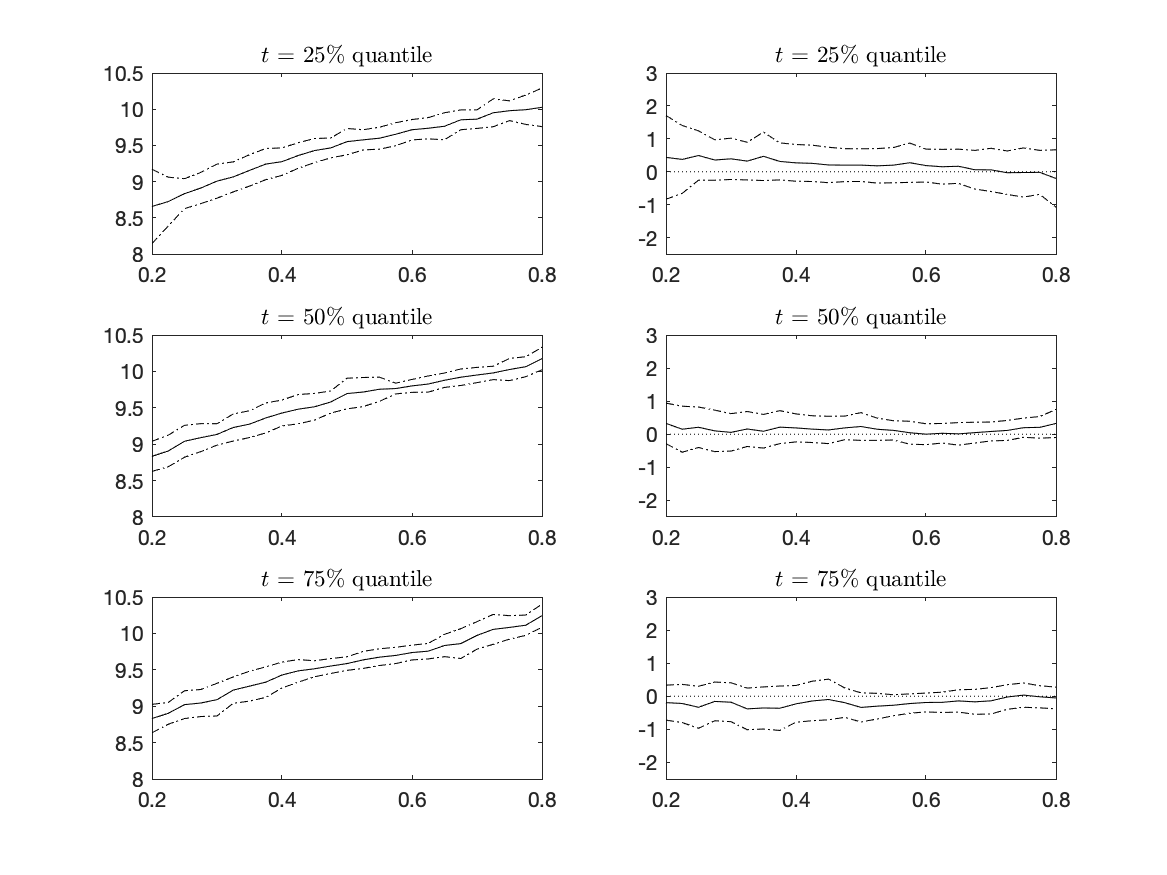}   
	\caption{Empirical results for blacks with small $\lambda$}
\label{fig:comp5_app_black}
\end{figure}

\begin{figure}[H]
	\centering
\includegraphics[scale = 0.55,angle=0]{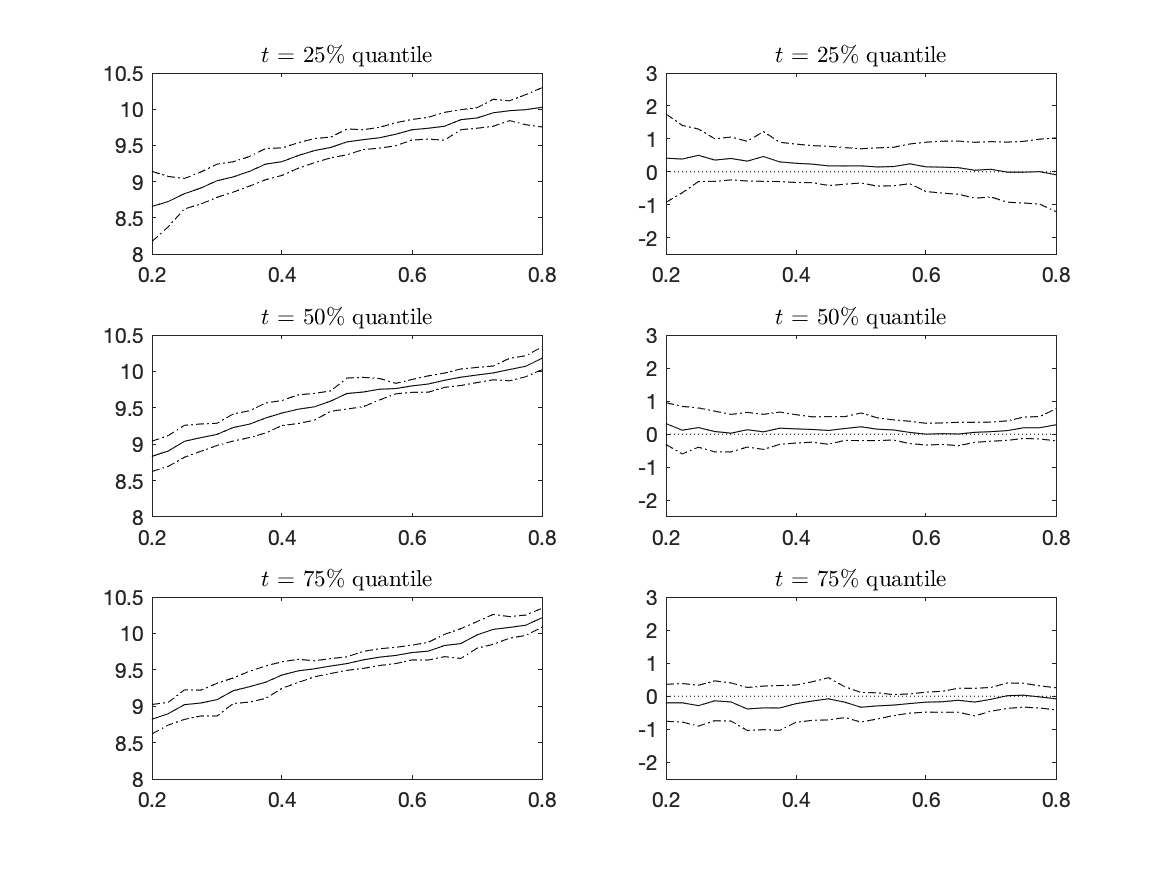}   
	\caption{Empirical results for blacks with large $\lambda$}
\label{fig:comp6_app_black}
\end{figure}

\end{document}